\documentclass[
 reprint,
 amsmath,amssymb,
 aps,
 nofootinbib,
 superscriptaddress
]{article}

\usepackage{amsthm}
\usepackage{graphicx}
\usepackage{xcolor}
\usepackage{braket}
\usepackage{amsmath}
\usepackage{amssymb}
\usepackage{amsthm}
\usepackage{bbm}
\usepackage{enumerate}
\usepackage{booktabs}
\usepackage{txfonts}
\usepackage[subrefformat=parens]{subcaption}
\usepackage{caption}
\usepackage{authblk}
\usepackage{cite}
\usepackage{mathtools}
\usepackage[a4paper,top=25mm,bottom=40mm,left=20mm,right=20mm]{geometry}
\usepackage{algorithm}
\usepackage{algorithmicx}
\usepackage{algpseudocode}
\usepackage{acronym}
\usepackage{makecell}
\usepackage{multirow}
\usepackage{rotating}
\usepackage{tikz}
\usetikzlibrary{quantikz2}
\usetikzlibrary{decorations.pathreplacing}
\usepackage[colorlinks=true
,linktocpage=true
,pdfproducer=medialab
,pdfa=true
,hyperfootnotes=false
]{hyperref}

\newtheorem{theorem}{Theorem}
\newtheorem{lemma}{Lemma}
\newtheorem{corollary}{Corollary}
\newtheorem{assum}{Assumption}
\newtheorem{assumprime}{Assumption}

\newtheorem{assumalp}{Assumption}

\theoremstyle{definition}
\newtheorem{definition}{Definition}
\newtheorem{remark}{Remark}

\newcommand{\drm}{\mathrm{d}}

\begin{document}

\title{Quantum algorithm for solving high-dimensional linear stochastic differential equations via amplitude encoding of the noise term}
\renewcommand{\thefootnote}{\fnsymbol{footnote}}

\author{
Koichi Miyamoto$^{1}$\footnote{miyamoto.kouichi.qiqb@osaka-u.ac.jp}

\small $^{1}$ Center for Quantum Information and Quantum Biology, The University of Osaka, Toyonaka, Osaka, Japan
}

\date{\today}
\maketitle

\acrodef{EM}{Euler-Maruyama}
\acrodef{LCU}{linear combination of unitaries}
\acrodef{PRN}{pseudorandom number}
\acrodef{PRNG}{pseudorandom number generator}
\acrodef{PCG}{permutated congruential generator}
\acrodef{ODE}{ordinary differential equation}
\acrodef{SDE}{stochastic differential equation}
\acrodef{OU}{Ornstein–Uhlenbeck}

\renewcommand{\thefootnote}{\arabic{footnote}}

\begin{abstract}

This work studies quantum algorithms to solve high-dimensional stochastic differential equations (SDEs) $\drm \mathbf{X}_t = A(t) \mathbf{X}_t \drm t + B(t) \drm \mathbf{W}_t$.
Aiming for a speed-up in the dimension $N$ of $\mathbf{X}_t$, we generate quantum states that encode $\mathbf{X}_t$ in the amplitudes, while most of the existing quantum methods for SDEs employ binary encoding.
A key challenge is the amplitude encoding of the noise term, and we address this by utilizing the quantum circuit implementation of a pseudorandom number generator (PRNG).
We propose two methods: the Dyson series-based method and the Euler-Maruyama (EM)-based method.
In the former, we express the noise term via the Dyson series approximation of the time evolution operator, while in the latter, it is approximated using the EM time discretization.
Both methods use the quantum linear systems solver to generate the amplitude-encoding state of $\mathbf{X}_t$, making only ${\rm polylog}(N)$ queries to the PRNG circuit and the block-encodings of $A$ and $B$.
Additionally, going beyond state preparation, we present methods to estimate expectations of functions of $\mathbf{X}_t$ using the state.
\end{abstract}

\tableofcontents

\section{Introduction}

While expectations for quantum computers to speed up demanding computational tasks are rising, particular attention has been paid to quantum algorithms for high-dimensional numerical problems, such as linear equation systems \cite{HHL,Childs2017QLSS,Subasi2019,Lin2020optimalpolynomial,Martyn2021,An2022QLSS,costa2022optimal} (see Ref. \cite{morales2024quantum} as a survey) and \acp{ODE} \cite{Berry_2014,berry2017quantum,childs2020quantum,Liu2021,Fang2023timemarchingbased,Krovi2023improvedquantum,An2023LCHS,Berry2024quantumalgorithm,Jennings2024costofsolvinglinear,low2025optimal,an2025quantum,an2026fast,an2026quantum}.
Quantum algorithms may provide an exponential speed-up with respect to the dimension $N$ of these problems and lead to major breakthroughs in the field of large-scale numerical computation.

In this paper, we focus on another problem, solving high-dimensional \acp{SDE}, which has attracted less attention in the context of quantum computing but is gaining interest recently. 
\acp{SDE} \cite{KloedenPlaten} are a fundamental tool for modeling phenomena involving randomness across fields in science and engineering.
High-dimensional \acp{SDE} appear in simulations of various kinds of phenomena ranging from chemical reactions \cite{Gillespie2007,Higham2008} to cosmic inflation \cite{Mizuguchi_2024,murata2026stochastic}.
Concretely, in this paper, letting $N$ and $m$ be positive integers satisfying $N \ge m$\footnote{We can assume this for the following reason. Suppose that $m>N$ holds. Consider the singular value decomposition $B(t)=U(t)\Sigma(t)V^\dagger(t)$, where $U(t)$ and $V(t)$ are $N \times N$ and $m \times m$ orthogonal matrices, respectively, and $\Sigma(t)$ is an $N \times m$ matrix with only its $(1,1)$-th, ..., $(N,N)$-th entries being possibly nonzero.
Then, Eq.~\eqref{eq:TgtSDEGen} can be written as $\drm \mathbf{X}_t = A(t) \mathbf{X}_t \drm t + B^\prime(t) \drm \mathbf{W}^\prime_t$ with $B^\prime(t) \coloneqq U(t)\Sigma(t) \in \mathbb{R}^{N \times m}$ and an $m$-dimensional Brownian motion $\mathbf{W}^\prime_t \coloneqq \int_0^t V^\dagger(s) \drm \mathbf{W}_s$.
Noting that all the entries in the $(N+1)$-th to $m$-th columns of $B^\prime(t)$ are zero, we can further rewrite Eq.~\eqref{eq:TgtSDEGen} as $\drm \mathbf{X}_t = A(t) \mathbf{X}_t \drm t + B^{\prime\prime}(t) \drm \mathbf{W}^{\prime\prime}_t$, where $B^{\prime\prime}(t)\in\mathbb{R}^{N \times N}$ comprises the first to $N$-th columns of $B^\prime(t)$, and  $\mathbf{W}^{\prime\prime}_t$ comprises the first to $N$-th entries of $\mathbf{W}^\prime_t$.}, we consider a probability space $(\Omega,\mathcal{F},\mathbb{P})$ supporting an $m$-dimensional Brownian motion $\mathbf{W}_t=(W^1_t,\ldots,W^N_t)^T$ and a $\mathbb{R}^N$-valued random process $\mathbf{X}_t=(X^1_t,\ldots,X^N_t)^T$ that follows the \ac{SDE}
\begin{align}
\drm \mathbf{X}_t = A(t) \mathbf{X}_t \drm t + B(t) \drm \mathbf{W}_t
\label{eq:TgtSDEGen}
\end{align}
with the deterministic initial value $\mathbf{X}_0=\mathbf{x}_0=(x_{0,1},\ldots,x_{0,N})^T$, where $A:[0,T] \rightarrow \mathbb{R}^{N \times N}$ and $B:[0,T] \rightarrow \mathbb{R}^{N \times m}$ are sufficiently smooth matrix-valued functions.

In fact, quantum algorithms for SDEs have been actively considered, especially in finance \cite{Rebentrost2018,Stamatopoulos2020optionpricingusing,Chakrabarti2021thresholdquantum,kaneko2022quantum}.
However, situations with large dimension $N$ have not received much attention.
To achieve the quantum speed-up with respect to $N$, it seems to be key to use {\it amplitude encoding} of $\mathbf{X}_t$ instead of {\it binary encoding}, which is used in most of the previous works.
Here, binary encoding and amplitude encoding are the ways to represent numbers in a quantum computer.
In binary encoding, a real number is represented in binary using a fixed number of bits, and the corresponding bit string is mapped to a computational basis state on qubits.
To encode an $N$-dimensional real vector, we use one register for each entry, resulting in the use of $O(N)$ qubits.
Operations on this system require $O(N)$ quantum gates.
On the other hand, in amplitude encoding, an $N$-dimensional vector $\mathbf{x}=(x_1,\ldots,x_N)^T$ is represented as a quantum state $\ket{\mathbf{x}}$ whose state vector coincides $\mathbf{x}$, namely, except for a normalization factor, $\ket{\mathbf{x}}=\sum_{j=1}^N x_j\ket{j}$, where $\ket{j}$ is the computational basis state on a quantum register corresponding to the binary representation of $j$.
For encoding in this way, only $\lceil \log_2 N \rceil$ qubits are required, which is exponentially fewer than binary encoding, and operations on the system may require only ${\rm polylog}(N)$ gates.

Quantum algorithms for solving \acp{ODE} $d\mathbf{X}(t)/dt = A(t)\mathbf{X}(t) + \mathbf{b}(t)$ indeed utilize amplitude encoding: they generate quantum states that encode the solution $\mathbf{X}(t)$ in amplitudes.
Noting that the \ac{SDE} in Eq. \eqref{eq:TgtSDEGen} has a similar form to the \ac{ODE}, with the noise term $B(t) \drm \mathbf{W}_t$ identified with the inhomogeneous term $\mathbf{b}(t)$, we expect that quantum \ac{ODE} solvers can be applied to \acp{SDE}.
However, such an extension is not so trivial, because it is not clear how to realize amplitude encoding of the noise term.
In quantum \ac{ODE} solvers, it is usually assumed that we have access to an oracle to prepare the quantum state that encodes $\mathbf{b}(t)=(b_j(t))_{j=1,\ldots,N}$ in the amplitudes.
Although this is also a nontrivial task, there are many proposals for addressing this.
For example, if $b_j(t)$ is given by some function of $i$ and $t$, it might be possible to use methods for loading functions onto quantum states \cite{grover2002creating,Sanders2019,zoufal2019quantum,Moosa_2024}.
In contrast, we cannot take the same approach to encode the noise term.
For $\Delta \mathbf{W}_t$, the increment of $\mathbf{W}_t$ in a time interval, its entries are $N$ independent random variables, which cannot be regarded as a function of the index $j$, rendering function loading methods inapplicable.

To achieve amplitude encoding of vectors of random variables, we propose to use quantum circuits for \acp{PRNG}.
\Acp{PRN} are deterministic sequences, but they serve as practical alternatives to random numbers because their statistical properties allow them to be regarded as random.
Ref. \cite{Miyamoto2020} presented a quantum circuit for a kind of \ac{PRNG}, which computes the $i$-th entry of the sequence for given $i$, and its use in financial engineering problems was proposed in Refs. \cite{Miyamoto2020,kaneko2021quantum,kaneko2022quantum}.
In this paper, we find that, although these previous works use \acp{PRN} in binary encoding, we can combine the \ac{PRNG} circuit and others to construct a circuit $U_{\rm RV}$ to generate the quantum state encoding \acp{PRN} in the amplitudes, which serve as an alternative to the state encoding true random numbers.

Then, equipped with this circuit, we propose quantum algorithms to generate the states encoding the realizations of $\mathbf{X}_t$ in Eq. \eqref{eq:TgtSDEGen}, in particular, the {\it history state}, which simultaneously encodes $\mathbf{X}_0,\mathbf{X}_1,\ldots$, the values of $\mathbf{X}_t$ at multiple times $t_1,t_2,\ldots$.
Based on Ref. \cite{an2026fast}, the work on quantum \ac{ODE} solvers, we build two algorithms:
\begin{enumerate}
    \item Dyson series-based method

    This method is built upon the exact recurrence relation $\mathbf{X}_{n+1} = \Phi_n \mathbf{X}_n+\boldsymbol{\Delta}_n$, where $\Phi_n$ is the time evolution operator over the time interval $[t_n,t_{n+1}]$, and $\boldsymbol{\Delta}_n$ is the noise in this interval, following a $N$-dimensional normal distribution (the detailed definitions are given later).
    Supposing that the block-encodings of $A$ and $BB^T$ are given, we construct the key component quantum circuits from them. 
    Based on Refs. \cite{Berry2024quantumalgorithm,an2026fast}, we construct the block-encoding of $\Phi_n$ with the Dyson series approximation.
    We then use this to construct the block-encodings of $\boldsymbol{\Delta}_n$'s covariance matrix $\Sigma_n$ and its square root, and by the combination with $U_{\rm RV}$, construct the circuit for amplitude encoding of $\boldsymbol{\Delta}_n$.
    Finally, using the quantum linear systems solver that incorporates these circuits, we generate the history state.

    \item Euler-Maruyama-based method
    
    Since the Dyson series-based method requires $BB^T$ to be full-rank to take the square root of $\Sigma_n$, we propose the second method, which covers cases where this requirement is not satisfied.
    This method is based on the \ac{EM} method to approximate the \ac{SDE} in Eq. \eqref{eq:TgtSDEGen} by $\mathbf{X}_{n+1} \simeq \mathbf{X}_n + A(t_n)\mathbf{X}_n \Delta t + B(t_n) \Delta\mathbf{W}_n$.
    Starting from the block-encodings of $A$ and $B$, we can construct the method to generate the history state based on the quantum linear systems solver, more easily than the Dyson series-based method.
    In return, the \ac{EM} time discretization introduces additional errors. 
\end{enumerate}
The charts that schematically show the outlines of these methods are in Figure \ref{fig:outline}.
Both methods make ${\rm polylog}(N)$ queries to building-block quantum circuits, which indicates exponential quantum speed-up with respect to $N$ compared to classically generating realizations of $\mathbf{X}_t$ by integrating Eq. \eqref{eq:TgtSDEGen}.
Additionally, we consider not only generating the quantum state but also extracting quantities of interest as values from it.
Concretely, we propose a method to estimate expectations of quantities written by polynomials of $\mathbf{X}_0,\mathbf{X}_1,\ldots$ by incorporating the history state generation into the quantum algorithm to estimate the overlap between two states \cite{Knill2007}.
Using the history state enables us to deal with quantities depending on $\mathbf{X}_t$ at multiple times.
We also present the method to estimate expectations of quantities defined at the terminal time only.
We also provide a summary of the assumptions made in this work in Table~\ref{tab:ass}.

\begin{figure}[t]
\centering
    \begin{subfigure}{0.49\textwidth}
    \centering
    \includegraphics[width=\linewidth]{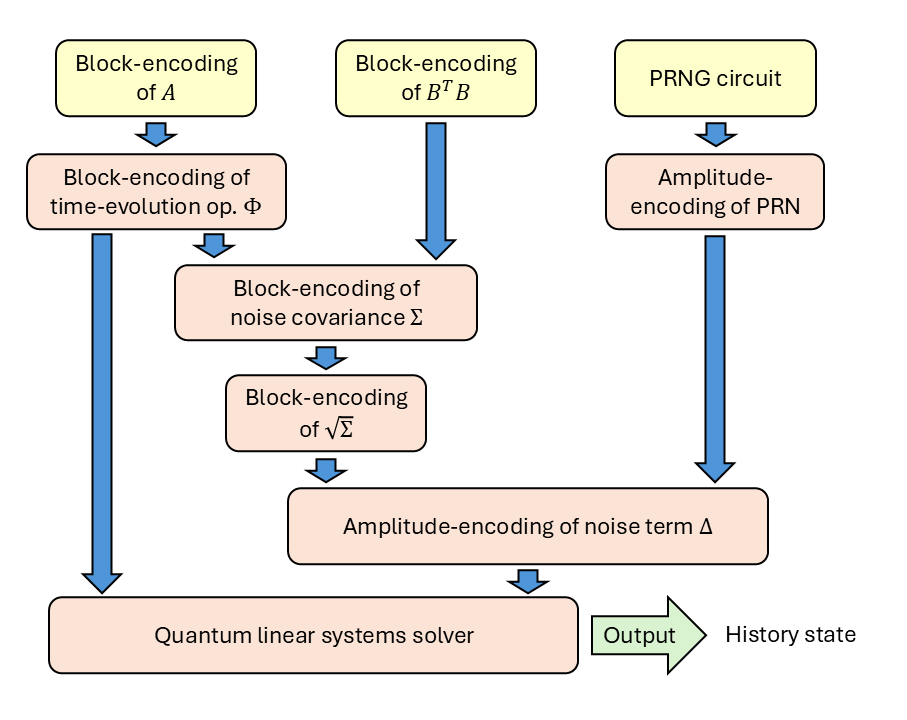}
    \caption{Dyson series-based method}
  \end{subfigure}
  \hfill
  \begin{subfigure}{0.49\textwidth}
    \centering
    \includegraphics[width=\linewidth]{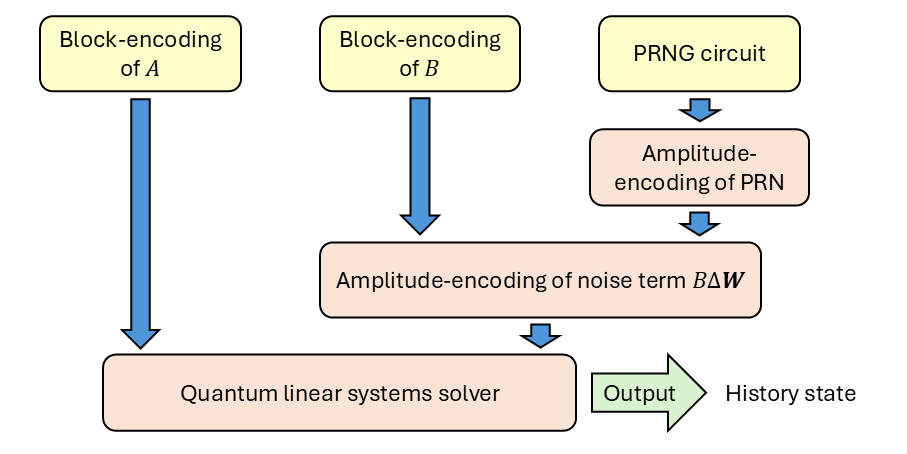}
    \caption{\ac{EM}-based method}
  \end{subfigure}

    \caption{Outlines of the proposed quantum algorithms, the Dyson series-based method and the \ac{EM}-based method. The yellow boxes represent given quantum circuits, and the orange boxes are what we construct from them. The blue arrow means that the object at the start is used to construct the one at the end.}
\label{fig:outline}
\end{figure}

Some recent papers consider quantum algorithms for high-dimensional \acp{SDE}, but to the best of the author's knowledge, existing work does not address the same setting or aim as this paper.
The algorithm proposed in Ref. \cite{bravyi2025quantum} estimates expectations of functions of random processes $\mathbf{X}_t$ not via amplitude encoding of $\mathbf{X}_t$ but by switching from \acp{SDE} to the Kolmogorov equation, which describes the time evolution of expectations, and solving it by Hamiltonian simulation.
It requires $\mathbf{X}_t$ to be random at the initial time, unlike our methods, which can address deterministic initial values\footnote{Although this paper concentrates on deterministic initial values, we can extend our algorithm to the case of the random initial value, if we have access to the oracle to generate the superposition of the amplitude-encoding states of possible initial values.}.
Ref. \cite{wu2026universal} proposes two algorithms: The first one is the Lindblad simulation-based method to estimate expectations of quadratic functions of $\mathbf{X}_t$.
The second one is for amplitude encoding of $\mathbf{X}_t$ with a weak approximation of \acp{SDE}, but reading out expectations from the $\mathbf{X}_t$-encoding state is not discussed.
Both methods have the complexity linearly depending on the dimension $m$ of the Brownian motion\footnote{Ref.~\cite{wu2026universal} evaluates the number of queries to the oracle for Lindblad simulation in their method. For the oracle, Ref.~\cite{wu2026universal} refers to the implementation by Ref.~\cite{li2022simulating}, which takes the $O(m)$ gate cost.}, which requires $m$ to be a few, unlike our methods, which can address $m$ up to $N$.
The quantum algorithm proposed in Ref. \cite{li2026efficient} generates the state encoding the solution of a differential equation with the inhomogeneous term given by an \ac{OU} process, but assumes access to the oracle to generate the state encoding the \ac{OU} process, treating the oracle as a black box\footnote{In other words, the method proposed in this paper can be used as the subroutine for amplitude encoding of the \ac{OU} process in the algorithm in Ref. \cite{li2026efficient}.}.
We also note that none of these papers considers generating the history state and estimating expectations of quantities depending on $\mathbf{X}_t$ at multiple times.
The existing and currently proposed methods for estimating expectations are summarized in Table~\ref{tab:compare}. 

\begin{table}[t]
  \centering
  \begin{tabular}{lp{14cm}}
    Assumption~\ref{ass:PRN} & Availability of a good \ac{PRNG} and a quantum circuit to generate the \ac{PRN}. \\
    Assumption~\ref{ass:UA} & Bounds on the norms of $A$ and $\drm A / \drm t$ and the availability of the block-encoding of $A$. \\
    Assumption~\ref{ass:muA} & Negative logarithmic norm of $A$. \\
    Assumption~\ref{ass:UBB} & Bounds on the condition number of $BB^T$ and the norm of $\drm BB^T/\drm t$, and the availability of the block-encoding of $BB^T$. \\
    Assumption~\ref{ass:UB} & Bounds on the norms of $B$ and $\drm B/\drm t$, and the availability of the block-encoding of $B$. \\
    Assumption~\ref{ass:Ux0} & Availability of a quantum circuit to generate the state that encodes the initial value.\\[5pt]
    Assumption~\ref{ass:arith} & Constant gate and qubit costs for arithmetic circuits.
  \end{tabular}
  \caption{Summary of the assumptions. Assumptions \ref{ass:PRN} to \ref{ass:Ux0} are for the Dyson series-based method. The same ones are also used for the \ac{EM}-based method, except Assumption~\ref{ass:UB} replaces Assumption~\ref{ass:UBB}. Assumption~\ref{ass:arith} is made throughout this paper.}
  \label{tab:ass}
\end{table}

\begin{table}[tp]
  \small
  \centering
  \begin{tabular}{
  |c|c|c|c|c|c|
  }
    \hline
    \multicolumn{2}{|c|}{Method} & \ac{SDE} & Type of $f$ & Query complexity & Assumed oracles\\
    \hline \hline
    \multirow{4}{*}{Quantum} &
    \makecell{\\ Dyson series \\ based \\ (this paper) \\ ~} & Eq.~\eqref{eq:TgtSDEGen}
    & \makecell{Polynomial\\(Multi-time)} & $\tilde{O}\left(\left(\kappa_{BB^T} \alpha_A T\right)^{\frac{d}{2}+1}\frac{1}{\epsilon_{\rm rel}}\right)$ & \multirow[c]{2}{*}{  
    \centering
    \makecell[l]{%
      - Block-encodings of \\
      $A$, $BB^T$ (or $B$) \\
      - State preparation for $\mathbf{x}_0$ \\
      - \ac{PRNG} circuit \\
      - State preparation \\
      for the coefficients of $f$
    }
  }
  \\
    \cline{2-5}
    & \makecell{\\ \ac{EM}-based \\ (this paper) \\ ~} & Eq.~\eqref{eq:TgtSDEGen}
    & \makecell{Polynomial\\(Multi-time)} & $\tilde{O}\left(\left(\max\left\{\frac{\alpha_A^2T}{\eta},\frac{\alpha_A^2T^2}{\epsilon_{\rm rel}}\right\}\right)^{\frac{d}{2}+1}\frac{1}{\epsilon_{\rm rel}}\right)$~\footnotemark\label{fn:EM} &  \\
    \cline{2-6}
    & \makecell{Kolmogorov\\eq.-based \cite{bravyi2025quantum}} &
    \makecell[l]{$dX^i_t=\left(-\lambda_iX^i_t + \sum\nolimits_j b_{ij}X^j_t \right.$ \\ $\left.+\sum\nolimits_{j,k} c_{ijk}X^j_tX^k_t\right)\drm t$\\$  +\sqrt{q}\drm W^i_t$\\w/ div-free condition}
    & \makecell{Polynomial\\(Terminal-time)} & \makecell{${\rm poly}$ with $\log N,T,\epsilon_{\rm rel},J,s,\lambda_1^{-1},q$\\$\exp$ with $\sum_{i=1}\lambda_i x_{0,i}^2/q$} & \makecell{Block-encodings of\\the matrices  in\\Kolmogorov eq.\footnotemark} \\
    \cline{2-6}
    & \makecell{Lindblad\\sim.-based\\ \cite[Algo. I]{wu2026universal}} & \makecell{$\drm \mathbf{X}_t = A(t) \mathbf{X}_t \drm t$ \\ $+ \sum_j B_j(t)\mathbf{X}_t \drm W^j_t$} & \makecell{Quadratic Func.\\(Terminal-time)} & \makecell{$\tilde{O}\left((\alpha_A+\sum_j \sigma_j^2)^2T^2/\epsilon_{\rm rel}\right)$\\$\sigma_j=\max_{t\in[0,T]}\|B_j\|$} & \makecell{Lindblad simulation \\ (Gate cost $=O(m)$ \cite{li2022simulating})}\\
    \hline
    Classical& \makecell{\ac{EM} + \\ Monte Carlo} & Nonlinear\footnotemark\label{fn:common} & Nonlinear Func.{\textsuperscript{7}} & $\tilde{O}\left(\max\left\{\frac{\alpha_A^2T}{\eta},\frac{\alpha_A^2T^2}{\epsilon_{\rm rel}}\right\} \cdot \frac{N(s_A+s_B)}{\epsilon_{\rm rel}^2}\right)$~{\textsuperscript{5}} & --- \\
    \hline
  \end{tabular}
  \caption{Summary of the methods for estimating $\mathbb{E}_\mathbb{P}[f(\{\mathbf{X}_t\})]$, the expectation of a target function $f$ of the solution of the \ac{SDE}. The query complexity means the number of queries to the assumed oracles for the quantum methods and the number of arithmetic operations in the classical method\protect\footnotemark. $\alpha_A$, $\eta$, $\kappa_{BB^T}$, $\epsilon_{\rm rel}$, $d$, $s_A$, and $s_B$ denote $A$'s spectral norm, $A$'s logarithmic norm, $BB^T$'s condition number, the required relative accuracy, the degree of $f$, $A$'s row sparsity, and that of $B$, respectively (the detailed definitions are given later). For the method of Ref. \cite{bravyi2025quantum}, $J$ and $s$ denote the drift strength and the sparsity of the matrices that appear in the Kolmogorov equation, respectively (for the detailed definitions, see Ref.~\cite{bravyi2025quantum}).}
  \label{tab:compare}
\end{table}
\normalsize
\footnotetext[5]{Here, it is assumed that $\eta T \lesssim 1 \lesssim \alpha_A T$, $d\sim1$, and the rate at which $A$ and $B$ vary is small. See Theorem~\ref{th:ExpEM} and Remarks~\ref{rem:CompEMQ} and \ref{rem:CompEMC} for more details.}
\footnotetext[6]{Actually, Ref.~\cite{bravyi2025quantum} considers implementing the time-evolution operators for $A$, $B$, and $C$ not from their block-encodings but from elementary gates, assuming their sparsity. Construction from the block-encodings leads to the same complexity evaluation.}
\footnotetext[7]{For general coefficient functions in the \ac{SDE} and/or general target functions, the standard \ac{EM} method may not work, or the \ac{SDE} may not even have a solution (e.g., non-globally Lipschitz continuous coefficients). Nonetheless, various modified numerical integration schemes have been proposed. See Ref.~\cite{hutzenthaler2015numerical} and the reference therein.}
\footnotetext[8]{On the justification of such a comparison, see Remark~\ref{rem:CompEMC}}

Amplitude encoding of random variables has been considered in previous papers.
Ref. \cite{bouland2023quantum} presents the quantum circuit to encode vectors consisting of $N$ independent normal variables, but it relies on the specific data loading technique called the unary data loader, which requires $O(N)$ qubits and gates.
Ref. \cite{Thaksakronwong2026}, using the method of Ref. \cite{bouland2023quantum}, considers encoding correlated normal variables, which partly overlaps with this paper's aim.
However, it assumes that the covariance $\Sigma$ and its block-encoding (or the oracular access to its entries) are given, unlike this paper, which constructs the block-encoding of $\Sigma_n$ from the block-encodings of $A$ and $BB^T$.

The organization of this paper is as follows.
Section \ref{sec:Pre} gives a brief review on solving \acp{SDE} and quantum algorithms used in this paper as building blocks.
Sections \ref{sec:exact} and \ref{sec:EM} are the main part of this paper, which present the Dyson series-based method and the \ac{EM}-based method, respectively.
We give the detailed procedures for estimating expectations and the complexity evaluations in terms of the number of queries to the given oracles.
Section \ref{sec:sum} summarizes this paper.

\section{Preliminaries \label{sec:Pre}}

\subsection{Notation}

We denote by $\mathbb{R}_+$ the set of all positive real numbers.

For $n\in\mathbb{N}$, we define $[n] \coloneqq \{1,\ldots,n\}$, $[n]_0 \coloneqq \{0,1,\ldots,n\}$.

For $n,n^\prime\in\mathbb{N}$, $\mathbf{0}_n$ is the $n$-dimensional vector with all entries 0, $I_n$ is the $n \times n$ identity matrix, and $O_{n,n^\prime}$ is the $n \times n^\prime$ zero matrix. 
The subscripts may be omitted if the size is obvious from the context or unimportant for the discussion.

For a vector $\mathbf{v}=(v_1,\ldots,v_n)^T\in\mathbb{C}^n$ and a real number $p\ge1$, $\|\mathbf{v}\|_p \coloneqq \left(\sum_{i=1}^n |v_i|^p\right)^{1/p}$ is its $L_p$ norm.
In particular, we denote the Euclid norm $\|\mathbf{v}\|_2$ by $\|\mathbf{v}\|$, omitting the subscript.
For a matrix $A$, $\|A\|$ denotes its spectral norm.
For an invertible matrix $A$, $\kappa(A)=\|A\|\cdot\|A^{-1}\|$ is its condition number.
For a square matrix $A$, $\mu(A) \coloneqq \lim_{h \rightarrow 0^+} \frac{\|I+hA\|-1}{h}$ is its logarithmic norm.
For a $d$-way tensor $C=(C_{i_1,\ldots,i_d})=\mathbb{R}^{\overbrace{n \times \cdots \times n}^d}$, we define $\|C\|_F \coloneqq \left(\sum_{i_1, \ldots, i_d=1}^n \left|C_{i_1,\ldots,i_d}\right|^2\right)^{1/2}$.

For a positive-semidefinite Hermite matrix $A$, we denote by $\sqrt{A}$ the Hermite matrix such that $\left(\sqrt{A}\right)^2=A$.

For $n\in\mathbb{N}$, $\boldsymbol{\mu}\in\mathbb{R}^n$, and $n \times n$ positive-definite symmetric matrix $\Sigma$, $\mathcal{N}_{\boldsymbol{\mu},\Sigma}$ denotes $n$-dimensional normal distribution with mean $\boldsymbol{\mu}$ and covariance matrix $\Sigma$.

We denote by $\mathbbm{1}_C$ the indicator function, which takes 1 if a condition $C$ is satisfied and 0 otherwise. 

We denote by $\log$ the natural logarithm and by $\lg$ that to base 2.

\subsection{How to solve stochastic differential equations}
\subsubsection{Exact solution}
The solution of the SDE \eqref{eq:TgtSDEGen} is given as follows \cite{KloedenPlaten}: for any $t\in[0,T]$,
\begin{align}
    \mathbf{X}_t = \Phi(t,0) \mathbf{x}_0+\boldsymbol{\Delta}(t,0).
    \label{eq:ExactSol}
\end{align}
Here, for $s,t\in[0,T]$ satisfying $s \le t$, the time evolution operator $\Phi(t,s)\in\mathbb{R}^{N \times N}$ is defined as the solution of the matrix differential equation
\begin{align}
    \frac{\partial}{\partial t}\Phi(t,s) = A(t)\Phi(t,s), ~ \Phi(s,s)=I,
\end{align}
and $\boldsymbol{\Delta}(t,s)$ is given by
\begin{align}
    \boldsymbol{\Delta}(t,s) \coloneqq \int_s^t \Phi(t,\tau) B(\tau) \drm \mathbf{W}_\tau.
\end{align} 
It is known that $\Phi(t,s)$ can be written as a Peano–Baker series \cite{brockett2015finite}
\begin{align}
    \Phi(t,s) & \coloneqq \mathcal{T}\exp\left(\int_s^t A(\tau) \drm \tau\right) \nonumber \\
    &= \sum_{n=0}^\infty \frac{1}{n!} \int_s^t \drm \tau_1 \int_s^t \drm \tau_2 \cdots \int_s^t \drm \tau_n \mathcal{T} A(\tau_1) \cdots A(\tau_n), \nonumber \\
    &= \sum_{n=0}^\infty \int_s^t \drm \tau_1 \int_s^{\tau_1} \drm \tau_2 \cdots \int_s^{\tau_{n-1}} \drm \tau_n A(\tau_1) \cdots A(\tau_n).
    \label{eq:PhiDef}
\end{align}
The following properties of $\Phi$ will be used later.
For $s,t\in[0,T]$ such that $s \le t$,
\begin{align}
   \frac{\partial}{\partial s}\Phi(t,s) = -\Phi(t,s)A(s) 
   \label{eq:PhiDot}
\end{align}
holds~\cite[Theorem 5.2 on p.128]{pazy2012semigroups}, and
\begin{align}
    \|\Phi(t,s)\| \le \exp\left(\int_s^t \mu(A(\tau)) \drm \tau\right)
    \label{eq:PhivsmuA}
\end{align}
holds~\cite[Theorem on p.34]{desoer2009feedback}.
Eq. \eqref{eq:PhivsmuA} implies
\begin{align}
    \|\Phi(t,s)\| \le \exp\left(\max_{\tau\in[s,t]}\mu(A(\tau)) \cdot (t-s)\right)
    \label{eq:PhivsmuA2}
\end{align}
and, if $\mu(A(t)) \le 0$ for any $t\in[0,T]$,
\begin{align}
    \|\Phi(t,s)\| \le 1.
    \label{eq:PhiBd}
\end{align}
We also note that, by a property of Ito integrals, $\boldsymbol{\Delta}(t,s)$ follows $\mathcal{N}_{\mathbf{0},\Sigma(t,s)}$, where 
\begin{align}
    \Sigma(t,s) \coloneqq \int_s^t \Phi(t,\tau) B(\tau) B^T(\tau) \Phi^T(t,\tau) \drm \tau.
\end{align}

\subsubsection{The Euler-Maruyama method}
Although Eq. \eqref{eq:ExactSol} is the exact solution, computing $\Phi$ and $\boldsymbol{\Delta}$ by their definitions, which involve the time-ordered exponential of the matrix, is often so demanding on a classical computer that we take other approaches.
A usual way is to discretize the time and calculate $\mathbf{X}$ on the time grid points step by step. 
The \ac{EM} method~\cite{Maruyama1955} is a simple but widely-used way in this direction: we set $r$ time grid points $t_n \coloneqq n \Delta t,n\in[r]_0$ with interval $\Delta t \in T/r$ and recursively define the discrete random process $\mathbf{X}^{\rm EM}_n, n\in[r]_0$ by
\begin{equation}
    \mathbf{X}^{\rm EM}_{n+1} = \mathbf{X}^{\rm EM}_n + A(t_n)\mathbf{X}^{\rm EM}_n \Delta t + B(t_n) \Delta \mathbf{W}_n, \mathbf{X}^{\rm EM}_0=\mathbf{x}_0,
    \label{eq:EM}
\end{equation}
where $\Delta\mathbf{W}_n \coloneqq \mathbf{W}_{t_{n+1}}-\mathbf{W}_{t_n}$, as an approximation of $X_t$.
Note that, since we are now considering the SDE with the diffusion term not depending on $X_t$, the EM method coincides with the Milstein method~\cite{Milstein1975}, which shows the strong convergence of order 1.
Namely, $\sqrt{\mathbb{E}_\mathbb{P}\left[\left\|\mathbf{X}^{\rm EM}_n - \mathbf{X}_{t_n}\right\|^2\right]}$, the $L^2$ error of the \ac{EM} approximation, decays as $O(1/r)$.
The rigorous statement for the currently considered \ac{SDE} is presented in Sec.~\ref{sec:EM}.

\subsection{Quantum building blocks}

\subsubsection{Arithmetic circuits}

In this paper, we consider computation on systems consisting of quantum registers, each of which is a set of qubits. 
We use the fixed-point binary representation for real numbers and, for each $x\in\mathbb{R}$, we denote by $\ket{x}$ the computational basis state on a quantum register that holds the bit string equal to the binary representation of $x$.
In particular, $\ket{0}$ represents the computational basis state in which all the qubits in the register take 0.
We assume that every register has a sufficiently large number of qubits and thus neglect errors caused by finite-precision representation.

We can perform arithmetic operations on numbers represented on registers.
For example, we can implement quantum circuits for four basic arithmetic operations such as addition $O_{\mathrm{add}}:\ket{a}\ket{b}\ket{0}\mapsto\ket{a}\ket{b}\ket{a+b}$ and multiplication $O_{\mathrm{mul}}:\ket{a}\ket{b}\ket{0}\mapsto\ket{a}\ket{b}\ket{ab}$, where $a,b\in\mathbb{Z}$.
For concrete implementations, see \cite{wang2025comprehensive} and the references therein.
In the finite-precision binary representation, these operations are immediately extended to those for real numbers.
Furthermore, using the above circuits, we obtain a quantum circuit $U_p$ to compute a polynomial $p(x)=\sum_{n=0}^N a_n x^n$ on real numbers $x$: $U_p\ket{x}\ket{0}=\ket{x}\ket{p(x)}$.
Also, for $p$ as an elementary function such as $\exp$, $\sin$/$\cos$, and so on, we have a similar circuit to compute it, given a piecewise polynomial approximation of $p$ \cite{haner2018optimizing}.
Besides, we can construct a comparator circuit that acts as $\ket{x}\ket{0} \mapsto \ket{0}\ket{\mathbbm{1}_{x \ge a}}$ for $x\in\mathbb{R}$ and constant $a\in\mathbb{R}$ \cite{Yuan_2023}, and based on it, circuits to compute $\max\{x,a\}$ and $\min\{x,a\}$.
Hereafter, we collectively call these circuits arithmetic circuits.

Later, along with arithmetic circuits, we also use circuits derived from them.
One of them is the circuit $U_{\rm B2A}$ to amplitude-encode the bit-encoded number on qubits, which acts for any $x\in[-1,1]$ and 1-qubit state $\ket{\psi}$ as $U_{\rm A2B}\ket{x}\ket{\psi}=\ket{x}\begin{pmatrix} x & -\sqrt{1-x^2} \\ \sqrt{1-x^2} & x\end{pmatrix}\ket{\psi}$.
This is implemented with a circuit to compute $\arccos(x)$ followed by the Y-rotation gate with controlled angle, which acts as $\ket{\theta}\ket{\psi} \mapsto \ket{\theta}\begin{pmatrix} \cos\theta & -\sin\theta \\  \sin\theta & \cos\theta\end{pmatrix} \ket{\psi}$ and is implemented as presented in Ref. \cite{woerner2019quantum}.
Besides, we use the circuit $U^{\rm SP}_{a,b}$ to generate $\frac{1}{\sqrt{b-a+1}}\sum_{i=a}^b\ket{i}$, the equi-amplitude superposition of integers from $a$ to $b$.
This is a special case of generating a quantum state encoding a probability density function $p$ on $\mathbb{R}$ in its amplitudes \cite{grover2002creating}, which is possible if we have a circuit to compute the integral of $p$ over any interval.
Hereafter, we regard these circuits as a kind of arithmetic circuit, although this classification may not be common.

All of these circuits can be constructed with a polynomial number of elementary gates with respect to the number of bits in the binary representation of real numbers.
If we fix a constant number of bits (say, several tens), the number of gates and qubits for one arithmetic operation also becomes constant.
Thus, in what follows, we make the following assumption, without explicitly stating it each time.
\begin{assumalp}
    Arithmetic circuits use $O(1)$ gates and qubits.
    \label{ass:arith}
\end{assumalp}

\subsubsection{Block-encoding}

Block-encoding \cite{Gilyen2019} is embedding a general matrix into the upper-left block of a unitary matrix.
This enables a quantum computer, which can basically perform only unitary operations, to deal with problems involving a wide range of matrices, and is now the basis of various quantum algorithms.
The mathematical definition of block-encoding is as follows.

\begin{definition}
    Let $A\in\mathbb{C}^{2^s \times 2^s}$ be an $s$-qubit operator, $a\in\mathbb{N}$, $\alpha \ge \|A\|$, and $\epsilon \ge 0$.
    If an $(s+a)$-qubit unitary $U$ satisfies $\|A-\alpha(\bra{0} \otimes I_{2^s}) U (\ket{0} \otimes I_{2^s})\|\le\epsilon$, we say that $U$ is an $(\alpha,a,\epsilon)$-block-encoding of $A$.
    We call $\alpha$, $a$, and $\epsilon$ the normalization factor, the ancilla qubit number, and the accuracy of the block-encoding, respectively.
\end{definition}

For some classes of matrices, the ways for block-encoding have been proposed so far \cite{Gilyen2019}.
A widely-used approach is the \ac{LCU} \cite{childs2012hamiltonian}: for a matrix written by a linear combination of unitary matrices, we can implement its block-encoding, if we have access to an oracle to load the coefficients and the controlled version of the unitaries.
Namely, for an $s$-qubit operator $A\in\mathbb{C}^{2^s \times 2^s}$ written by $A=\sum_{l=1}^L c_l U_l$ with $\mathbf{c}=(c_1,\ldots,c_L)^T\in\mathbb{R}_+^L$ and $L$ $2^s \times 2^s$ unitaries $U_1,\ldots,U_L$, if we have an oracle $V_{\rm prep}$ that acts as
\begin{align}
    V_{\rm prep} \ket{0} = \sum_{l=1}^L \sqrt{\frac{c_l}{\|\mathbf{c}\|_1}} \ket{l}
\end{align}
on an $a$-qubit register, where $a \ge \left\lceil \lg L \right\rceil+1$, and an oracle $V_{\rm sel}$ that acts on an $(s+a)$-qubit system as
\begin{align}
    V_{\rm sel} \ket{l}\ket{\psi} = \ket{l}U_l\ket{\psi}
\end{align}
for any $l\in[L]$ and any $s$-qubit state $\ket{\psi}$,
\begin{align}
    (V_{\rm prep} \otimes I_{2^s})^\dagger V_{\rm sel} (V_{\rm prep} \otimes I_{2^s})
\end{align}
is a $(\|\mathbf{c}\|_1,a,0)$-block-encoding of $A$.

When we have block-encodings of matrices, we can also construct a block-encoding of the product of them.

\begin{figure}[t]
\begin{center}

\begin{quantikz}[wire types={q,q,n,q,q}, transparent]
& \qwbundle{a_1} & \gate[5]{U_{A_1}} & \qw & ~\cdots~ &  & & \qw \\
& \qwbundle{a_2} & \linethrough & \gate[4]{U_{A_2}} & ~\cdots~ &  & & \qw \\
 & \vdots & & & \ddots & & &\\
& \qwbundle{a_n} & \linethrough & \linethrough & ~\cdots~ & \gate[2]{U_{A_n}} & & \qw \\
& \qwbundle{s} & & \qw & ~\cdots~ &  & & \qw
\end{quantikz}
\caption{The quantum circuit for block-encoding of $A_n \cdots A_1$, a product of matrices $A_1,\ldots,A_n$, constructed with block-encodings of them. Here and hereafter, a wire and a box represent a quantum register (set of qubits) and a quantum gate, respectively, and a wire bypassing a box means that the corresponding register is not used in the gate.}
\label{fig:BEProd}
\end{center}
\end{figure}
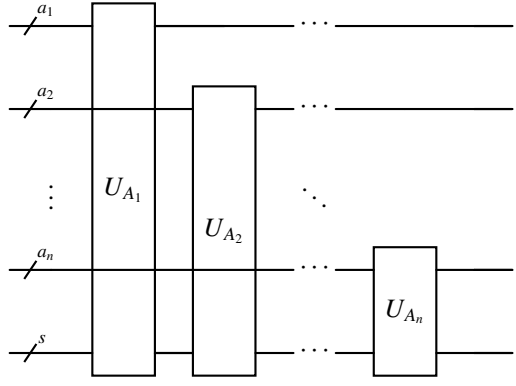

\begin{lemma}[Lemma 53 in the full version of \cite{Gilyen2019}, modified]
    Let quantum circuits $U_{A_1},\ldots,U_{A_n}$ be $(\alpha_1,a_1,0)$-, ..., $(\alpha_n,a_n,0)$-block-encodings of $s$-qubit operators $A_1,...,A_n$, respectively.
    Then, the $\left(\sum_{i=1}^n a_i + s\right)$-qubit quantum circuit $\tilde{U}_{A_n} \cdots \tilde{U}_{A_1}$ shown in Fig.~\ref{fig:BEProd}, where $\tilde{U}_{A_l}$ is $U_{A_l}$ that acts on the system consisting of the $\left(\sum_{i = 1}^{l-1} a_i + 1\right)$-th to $\sum_{i = 1}^l a_i$-th qubits and the last $s$ qubits, is a $\left(\prod_{i=1}^n \alpha_i,\sum_{i=1}^n a_i,0\right)$-block-encoding of $A_n \cdots A_1$.
    \label{lem:BEProd}
\end{lemma}

In this paper, we will some techniques to construct block-encodings of $f(A)$, functions $f$ of a matrix $A$, from a block-encoding of $A$, as is common in recent works on quantum algortihms for problems involving large matrices.
First, we will use the block-encoding of the square root of a matrix, which is presented in Ref.~\cite{chakraborty2019power}.

\begin{theorem}[Lemma 10 in the full version of \cite{chakraborty2019power}, modified]
    Let $\epsilon>0$.
    Let $H$ be a Hermitian matrix with eigenvalues in $[c/\kappa, c]$, where $c>0$ and $\kappa\ge2$.
    Suppose that we have access to a $(\alpha,a,\delta(\epsilon,\kappa,c))$-block-encoding $U_H$ of $H$, where $\alpha \ge c$, $a\in\mathbb{N} \cup \{0\}$, and
    \begin{align}
        \delta(\epsilon,\kappa,c) \coloneqq \frac{\sqrt{c}\epsilon}{16 \pi (J+1)^2 2^J}, J=\left\lceil \lg \left(2\kappa\log\left(\frac{8\sqrt{c}}{\epsilon}\right)\right) \right\rceil.
    \end{align}
    Then, we have access to a $(2\sqrt{c}, a+O(\log\log(1/\epsilon)), \epsilon)$-block-encoding of $\sqrt{H}$ that makes
    \begin{align}
        O\left(\frac{\alpha\kappa}{c} \log^2\left(\frac{\sqrt{c}\kappa}{\epsilon}\right)\right)
    \end{align}
    queries to $U_H$ and
    \begin{align}
        O\left(\frac{a\alpha\kappa}{c} \log^2\left(\frac{\sqrt{c}\kappa}{\epsilon}\right)\right)
    \end{align}
    uses of additional elementary gates.
    \label{th:BESqrt}
\end{theorem}

We also use the block-encoding of the inverse of a matrix, which is the basis of quantum linear systems solvers.
Such solvers date back to the celebrated Harrow-Hassidim-Lloyd algorithm \cite{HHL}, and recent block-encoding-based ones have improved complexities.
Among them, we take the one presented in \cite{Martyn2021}.

\begin{theorem}[Algorithm 5 in \cite{Martyn2021}, modified]
    Let $\epsilon>0$.
    Let $P$ be an invertible matrix with singular values contained in $[\alpha_P/\kappa_P,\alpha_P]$ with $\alpha_P>0,\kappa_P \ge 1$.
    Suppose that we have access to an $(\alpha_P,a_P,0)$-block-encoding $V_P$ of $P$ and its inverse.
    Then, we have access to an $(1,a_P,\epsilon)$-block-encoding $V_{P^{-1}}$ of $\frac{\alpha_P}{2\kappa_P}P^{-1}$, which uses $V_P^\dagger$
    \begin{align}
        O\left(\kappa_P\log\left(\frac{1}{\epsilon}\right)\right)
        \label{eq:BECompMatInv}
    \end{align}
    times and
    \begin{align}
        O\left(a_P \kappa_P\log\left(\frac{1}{\epsilon}\right)\right)
        \label{eq:GateCompMatInv}
    \end{align}
    additional elementary gates.
    \label{eq:MatInvQSVT}
\end{theorem}

The number of additional elementary gates used in $V_{P^{-1}}$, which is not explicitly given in Ref.~\cite{Martyn2021}, follows from the fact that $V_{P^{-1}}$, aside from $V_P^\dagger$, uses operators in the form of $\Pi_\phi \coloneqq \exp\left(2i \phi \Pi\right)$ a number of times given by Eq.~\eqref{eq:BECompMatInv}.
Here, $\phi\in\mathbb{R}$ takes specified values, and $\Pi = \ket{0}\!\bra{0}$ is an operator on the $a_P$-qubit ancillary register of $V_P$.
Since $\Pi_\phi$ is constructed with $O(a_P)$ elementary gates, the number of them in $V_{P^{-1}}$ plies up to Eq.~\eqref{eq:GateCompMatInv}. 

Hereafter, we do not distinguish the number of calls to a quantum circuit from the number of calls to its inverse, since their gate costs, the number of elementary gates for constructing them, are identical. 

\subsubsection{Overlap estimation}

Later, we will present quantum algorithms to solve an SDE, which means generating a quantum state in which a solution of the SDE is encoded in the amplitudes.
In general, most of the quantum solvers for numerical problems such as linear equation systems and differential equations have outputs in this form: they generate solution-encoding quantum states.
However, we eventually want not a quantum state but the solution itself, and extracting the solution from the quantum state raises another issue, which is called the readout problem \cite{aaronson2015read}.
If the solution is a vector of high dimension $N$, obtaining all the $N$ entries from the quantum state inevitably takes $O(N)$ time, which ruins the speed-up with respect to $N$ provided by the quantum solver.
We thus often aim to extract only some quantities that characterize the solution.

To extract quantities of interest as a number from a solution-encoding state generated by our quantum algorithms, we utilize the estimation of the overlap of two quantum states.
Let us consider a quantum state on a two-register system, $\ket{\mathbf{x}}=\sum_{i=1}^n x_i\ket{0}\ket{i}+\ket{0_\perp}$, where the $i$-th entry of a solution $\mathbf{x}=(x_1,\ldots,x_n)^T \in \mathbb{R}^n$ of some problem is encoded as the amplitude of the basis state $\ket{0}\ket{i}$, and $\ket{0_\perp}$ is an unnormalized state such that $(\ket{0}\bra{0} \otimes I)\ket{0_\perp}=0$.\footnote{Hereafter, we repeatedly use the symbol $\ket{0_\perp}$ to denote some unnormalized state orthogonal to the states in the terms before it. Although the represented states may differ between different equations, we adopt this notation because the explicit form of the state is irrelevant, and we expect no confusion.}
We also consider another state $\ket{C}=\frac{1}{\|C\|_F}\sum_{i=1}^n c_i \ket{0}\ket{i}$ encoding a constant vector $C=(c_1, \ldots, c_n)^T \in \mathbb{R}^n$.
Then, $\braket{C|\mathbf{x}}$, the overlap between $\ket{\mathbf{x}}$ and $\ket{C}$, gives $\braket{C|\mathbf{x}}=\frac{1}{\|C\|_F}\sum_{i=1}^n c_i x_i$.
Some quantities characterizing $\mathbf{x}$, for example, a weighted average of its entries, can be written in this form, and thus the estimation of the overlap may give us information of interest.
This approach covers not only linear quantities with respect to $\mathbf{x}$ but also higher-order ones.
For a $s$-way tensor $C=(C_{i_1 \ldots i_s})$, the overlap between $\ket{\mathbf{x}}^{\otimes s}$ and $\ket{C}=\frac{1}{\|C\|_F}\sum_{i_1, \ldots, i_s=1}^n C_{i_1 \ldots i_s}\ket{i_1}\cdots\ket{i_s}$ gives $\frac{1}{\|C\|_F}\sum_{i_1, \ldots, i_s=1}^n C_{i_1 \ldots i_s}x_{i_1} \cdots x_{i_s}$.

In fact, to estimate the overlap between two quantum states, we can use a method called overlap estimation presented in Ref.~\cite{Knill2007}, given quantum circuits to generate those states.

\begin{theorem}
    Let $\epsilon,\delta\in(0,1)$.
    Suppose that we have access to the oracles $U_\psi$ and $U_\phi$ that generate quantum states $\ket{\psi}$ and $\ket{\phi}$, respectively: $U_\psi\ket{0}=\ket{\psi}$ and $U_\phi\ket{0}=\ket{\phi}$.
    Then, there exists a quantum algorithm that outputs $\hat{a}\in(0,1)$ such that $|\hat{a}-\braket{\psi|\phi}|\le\epsilon$ with probability at least $1-\delta$, querying $U_\psi$ and $U_\phi$
    \begin{equation}
        O\left(\frac{1}{\epsilon}\log\left(\frac{1}{\delta}\right)\right)
    \end{equation}
    times.
    \label{th:OE}
\end{theorem}

\subsubsection{Quantum pseudorandom number generator}

Despite the widespread use of algorithms relying on random numbers, such as the Monte Carlo method, it is usually impossible to generate truly random numbers, and we thus substitute \acp{PRN} in practice.
Although they are deterministic sequences and using them in place of random numbers thus introduces errors inevitably, there are various proposals for \acp{PRNG} that pass statistical tests of randomness and serve as practical alternatives to random numbers.
Remarkably, \acp{PRNG} can be useful in quantum algorithms.
Ref. \cite{Miyamoto2020,kaneko2021quantum,kaneko2022quantum} proposed implementing a \ac{PRNG} called \ac{PCG} \cite{oneill:pcg2014} as a quantum circuit and using it in a quantum algorithm for Monte Carlo integration.
Since the recurrence formula of \ac{PCG} is a one-to-one mapping of bit strings, we can implement it as a quantum circuit.
Furthermore, a formula to jump to an arbitrary position in the PCG sequence, namely, to get the $i$-th element of the sequence for specified $i$, is available and implementable as a quantum circuit.
Although PCG generates random binary-encoded integers within a certain range or, after normalization, uniform random numbers in $[0,1]$, we can construct generators for other distributions by transformations of the generated \acp{PRN}.
For example, Ref. \cite{kaneko2022quantum} proposed a quantum circuit to transform uniform random numbers to normal random numbers by inverse transform sampling, using the piecewise polynomial approximation of the inverse cumulative distribution function of the normal distribution given in Ref. \cite{Hormann2003}.

Based on these backgrounds, in this paper, we assume that we have a PRNG that generates numbers close to standard normal random numbers in the following sense and that we can implement it as a quantum circuit.

\begin{assum}
    We have access to a quantum circuit $U_{\rm Rand}^s$ that depends on a randomly chosen integer (seed) $s$ and acts as
    \begin{align}
        U_{\rm Rand}^s\ket{i}\ket{0} = \ket{i}\ket{z^s_i}
        \label{eq:URand}
    \end{align}
    for any $i\in\mathbb{N}$.
    Here, $z^s_1,z^s_2,\ldots\in\mathbb{R}$ satisfy the following properties with $d\in\mathbb{N}$ and
    $\epsilon_{\rm PRN}>0$:
    \begin{enumerate}[(i)]
        \item For any $i \in \mathbb{N}$ and any $u>0$, $\mathbb{P}_s\{z^s_i > u\} < \frac{2}{\sqrt{2\pi}} e^{-u^2/2}$ holds, where $\mathbb{P}_s$ is the probability measure with respect to the randomness of $s$.

        \item For any $k \in [2d]$, any $n,D\in\mathbb{N}$, any $\mathbf{l}=(l_1,\ldots,l_k)\in[D]^{\times k}$, and any $\mathbf{z}^s_1,\ldots,\mathbf{z}^s_D\in\mathbb{R}^n$, vectors formed from nonoverlapping length-$n$ subsequences of $\{z^s_i\}$,
        \begin{enumerate}[a.]
            \item
            \begin{align}
                \left\|\mathbb{E}_s\left[\mathbf{z}^s_{l_1} \otimes \cdots \otimes \mathbf{z}^s_{l_k}\right] - \mathbb{E}_{\boldsymbol{\zeta}}\left[\boldsymbol{\zeta}_{l_1} \otimes \cdots \otimes \boldsymbol{\zeta}_{l_k}\right]\right\| \le (n+2d)^{k/2} \epsilon_{\rm PRN}
                \label{eq:PRNExpSNDiff}
            \end{align}
            holds, where $\boldsymbol{\zeta}_1,\ldots,\boldsymbol{\zeta}_D$ are independent $n$-dimensional standard normal random vectors, and $\mathbb{E}_s$ (resp. $\mathbb{E}_{\boldsymbol{\zeta}}$) is the expectation with respect to the randomness of $s$ (resp. $\boldsymbol{\zeta}_1,\ldots,\boldsymbol{\zeta}_D$).
        
            \item If $\#_1(\mathbf{l}),\ldots,\#_D(\mathbf{l})$ are even, where $\#_l(\mathbf{l})$ represents the number of $l$ in $\mathbf{l}$,
            \begin{align}
                \mathbb{E}_s\left[\left\|\mathbf{z}^s_{l_1} \otimes \cdots \otimes \mathbf{z}^s_{l_k}\right\|\right] \le (n+2d)^{k/2}
                \label{eq:PRNMom}
            \end{align}
            holds.

            \item If either of $\#_1(\mathbf{l}),\ldots,\#_D(\mathbf{l})$ is odd, for any $u>0$,
            \begin{align}
                \left\|\mathbb{E}_s\left[\tilde{\mathbf{z}}^s_{l_1} \otimes \cdots \otimes \tilde{\mathbf{z}}^s_{l_k}\right]\right\| \le (n+2d)^{k/2} \epsilon_{\rm PRN}
                \label{eq:PRNClip}
            \end{align}
            holds, where $\tilde{\mathbf{z}}^s_l$ is given by applying $F^{\rm Bd}_u:\mathbb{R} \rightarrow \mathbb{R}, x \mapsto \max\{\min\{x_j,u\},-u\}$ to $\mathbf{z}^s_l$ entrywise: denoting $\mathbf{z}^s_l=(z^s_{l,1},\ldots,z^s_{l,n})^T$,
            \begin{align}
                \tilde{\mathbf{z}}^s_l =
                \begin{pmatrix}
                F^{\rm Bd}_u(z^s_{l,1}) \\ \vdots \\ F^{\rm Bd}_u(z^s_{l,N})
                \end{pmatrix}.
            \end{align}
        \end{enumerate}
    \end{enumerate}
    \label{ass:PRN}
\end{assum}

Assumption (i) means that each $z_i$ has the Gaussian tail.
For a standard normal variable $\zeta$ and $u>0$, the probability that $|\zeta|>u$ is $2\int_u^\infty \frac{1}{\sqrt{2\pi}}e^{-x^2/2} \drm x$, which is bounded by $\frac{2}{\sqrt{2\pi}}e^{-u^2/2}$, and we assume that this also bounds $\mathbb{P}_s\{z^s_i > u\}$.

The background of Assumption (ii) is as follows.
The squared norm $\left\|\boldsymbol{\zeta}\right\|^2$ of an $n$-dimensional standard normal vector $\boldsymbol{\zeta}$ follows the chi-squared distribution with $n$ degrees of freedom, and its $k$-th moment is  $\mathbb{E}_{\boldsymbol{\zeta}}\left[\left\|\boldsymbol{\zeta}\right\|^{2k}\right]=2^k\Gamma\left(k+\frac{n}{2}\right)/\Gamma\left(\frac{n}{2}\right)$ \cite{simon2002probability}, which is bounded by $(n+2d)^k$.
When $\#_1(\mathbf{l}),\ldots,\#_D(\mathbf{l})$ are even, by the independence of $\boldsymbol{\zeta}_1,\ldots,\boldsymbol{\zeta}_D$, we have
\begin{align}
    \left\|\mathbb{E}_{\boldsymbol{\zeta}}\left[\boldsymbol{\zeta}_{l_1} \otimes \cdots \otimes \boldsymbol{\zeta}_{l_k}\right]\right\| \le \mathbb{E}_{\boldsymbol{\zeta}}\left[\left\|\boldsymbol{\zeta}_{l_1} \otimes \cdots \otimes \boldsymbol{\zeta}_{l_k}\right\|\right] \le\prod_{l=1}^D \mathbb{E}_{\boldsymbol{\zeta}}\left[\|\boldsymbol{\zeta}_l\|^{\#_l(\mathbf{l})}\right] \le (n+2d)^k.
\end{align}
Assumption (ii)b requires that $\mathbf{z}^s_1,\ldots,\mathbf{z}^s_D$ also satisfies this bound.
Assumption (ii)a requires $\mathbb{E}_s\left[\mathbf{z}^s_{l_1} \otimes \cdots \otimes \mathbf{z}^s_{l_k}\right]$ to be close to $\mathbb{E}_{\boldsymbol{\zeta}}\left[\boldsymbol{\zeta}_{l_1} \otimes \cdots \otimes \boldsymbol{\zeta}_{l_k}\right]$ with an error relative to this bound.
In particular, when either of $\#_1(\mathbf{l}),\ldots,\#_D(\mathbf{l})$ is odd, $\mathbb{E}_{\boldsymbol{\zeta}}\left[\boldsymbol{\zeta}_{l_1} \otimes \cdots \otimes \boldsymbol{\zeta}_{l_k}\right] = \mathbf{0}$ holds, and thus $\left\|\mathbb{E}_s\left[\mathbf{z}^s_{l_1} \otimes \cdots \otimes \mathbf{z}^s_{l_k}\right]\right\| \le (n+2d)^{k/2} \epsilon_{\rm PRN}$ is required.
Assumption (ii)c requires that the same bound also holds for symmetrically clipped $\mathbf{z}^s_l$, while for clipped $\boldsymbol{\zeta}_l$, the corresponding expectations are $\mathbf{0}$. 

Hereafter, for clean presentation, we omit the superscript $s$ of $\mathbf{z}^s_l$, assuming that $s$ is chosen in advance.
We also omit the subscript $s$ of $\mathbb{E}_s[\cdot]$ and write it as $\mathbb{E}[\cdot]$ simply.

\section{Proposed method 1: Dyson series-based method \label{sec:exact}}

Now, let us present the proposed quantum algorithms.
The first one is named the Dyson series-based method, in which we aim to generate a quantum state that encodes the solution in Eq. \eqref{eq:ExactSol} using the Dyson series.
We take the same approach as the method in Ref. \cite{an2026fast} for solving \acp{ODE}.
We go through the block-encoding of the time evolution operator $\Phi$, whose construction inevitably introduces errors, and to control them, we make a sufficiently fine time discretization and formulate the problem as a linear equation system.

\subsection{Linear equation system}

Before presenting the linear equation system, we make some assumptions.
First, we hereafter assume that $N$ is a power of 2, namely, $N=2^n$ for some $n\in\mathbb{N}$.
This is just for the concise presentation of the proposed algorithms and imposes no essential limitations on the algorithms\footnote{In fact, we can extend the dimension of $\mathbf{X}_t$ to a power of 2 as follows. Letting $s = \left\lceil \lg N \right\rceil$, we replace the matrices $A(t)$ and $B(t)$ with $2^s \times 2^s$ matrices $\begin{pmatrix} A(t) & O \\ O & I \end{pmatrix}$ and $\begin{pmatrix} \tilde{B}(t) & O \\ O & I \end{pmatrix}$, respectively, where $\tilde{B}(t)$ is the $N \times N$ matrix in the form of $\tilde{B}(t) = \left( B(t) ~ O \right)$. Besides, we append dummy variables $X^{N+1}_t,\ldots,X^{2^s}_t$ to $\mathbf{X}_t$ with their initial values set to 0, and $W^{N+1}_t,\ldots,W^{2^s}_t$, which are constantly 0, to $\mathbf{W}_t$. The extended \ac{SDE} has the $\mathbb{R}^{2^s}$-valued solution with its first to $N$-th entries equal to the original solution and its $(N+1)$-th to $2^s$-th entries being 0.
}.
Besides, we make the following assumptions on the properties of the matrix $A$, which make evaluating errors and complexities possible, along with the assumption of access to its block-encoding.

\begin{assum}
    \begin{align}
        \max_{t \in [0,T]} \|A(t)\| \le \alpha_A
        \label{eq:alphaA}
    \end{align}
    and 
    \begin{align}
        \max_{t\in[0,T]}\left\|\frac{\drm A}{\drm t}(t)\right\| \le \alpha_{\drm A/\drm t}
    \end{align}
    holds for $\alpha_A,\alpha_{\drm A/\drm t}>0$.
    Letting $a_A$ be a non-negative integer, we have access to the oracle $U_A$ that acts as
    \begin{align}
        U_A \ket{t} \ket{\psi} = \ket{t} V_{A(t)}\ket{\psi}, 
    \end{align}
    for any $t\in[0,T]$ and any $2^{a_A}N$-dimensional state vector $\ket{\psi}$, where $V_{A(t)}$ is a $(\alpha_A,a_A,0)$-block-encoding of $A(t)$.
    \label{ass:UA}
\end{assum}

\begin{assum}
    For some $\eta\in\mathbb{R}_+$,
    \begin{align}
        \max_{t\in[0,T]}\mu(A(t))\le -\eta.
    \end{align}
    \label{ass:muA}
\end{assum}

Because the relationship
\begin{align}
    -\|M\| \le \mu(M) \le \|M\|
\end{align}
holds for any square matrix $M$~\cite[Theorem on p.31]{desoer2009feedback}, we have
\begin{align}
    \eta \le \alpha_A.
\end{align}
Besides, Eq.~\eqref{eq:PhivsmuA2} implies that under Assumption \ref{ass:muA},  
\begin{align}
    \|\Phi(t,s)\| \le e^{-\eta(t-s)}
    \label{eq:PhiUBEta}
\end{align}
holds for any $t$ and s such that $0 \le s \le t \le T$.

Now, we set the time grid points $t_0<t_1<\ldots<t_r$ by $t_n=n\Delta t$ for $n\in[r]_0$, where $r \in \mathbb{N}$ and $\Delta t \coloneqq T/r$.
We define $\mathbf{X}_n\coloneqq\mathbf{X}_{t_n}$, which satisfies
\begin{align}
    \mathbf{X}_{n+1} = \Phi_n \mathbf{X}_n+\boldsymbol{\Delta}_n
    \label{eq:XnRec}
\end{align}
with
\begin{align}
    \Phi_n \coloneqq \Phi(t_{n+1},t_n),
    \boldsymbol{\Delta}_n \coloneqq \boldsymbol{\Delta}(t_{n+1},t_n)
\end{align}

We can summarize the recursive relation in Eq.~\eqref{eq:XnRec} in the form of a linear equation system with a larger size:
\begin{align}
    \mathcal{A} \mathbf{X}_{\rm hist} = \mathbf{B}
    \label{eq:AXhistB}
\end{align}
where
\begin{align}
    \mathcal{A} \coloneqq
    \begin{pmatrix}
    I & & & \\
    -\Phi_0 & I & & \\
    & \ddots & \ddots & \\
    & & -\Phi_{r-1} & I
    \end{pmatrix}
    ,
    \mathbf{X}_{\rm hist} \coloneqq
    \begin{pmatrix}
        \mathbf{X}_0 \\ \mathbf{X}_1 \\ \vdots \\ \mathbf{X}_r 
    \end{pmatrix}
    ,
    \mathbf{B} \coloneqq
    \begin{pmatrix}
    \mathbf{x}_0 \\ \boldsymbol{\Delta}_0 \\ \vdots \\ \boldsymbol{\Delta}_{r-1} 
    \end{pmatrix}
    .
    \label{eq:calAXHistBDef}
\end{align}
Here, $\boldsymbol{\Delta}_0,\ldots,\boldsymbol{\Delta}_{r-1}$ are independent $\mathbb{R}^N$-valued random variables such that $\boldsymbol{\Delta}_n \sim \mathcal{N}_{\mathbf{0},\Sigma_n}$ with the covariance matrices
\begin{align}
    \Sigma_n \coloneqq \Sigma(t_{n+1},t_n).
\end{align}
We can express $\boldsymbol{\Delta}_n$ by
\begin{align}
    \boldsymbol{\Delta}_n = S_n \mathbf{Z}_n,
\end{align}
using $\mathbf{Z}_n\sim\mathcal{N}_{\mathbf{0},I}$ and a symmetric matrix $S_n$ such that $S_n^2=\Sigma_n$.
$S_n$ is given by $S_n=W_n \sqrt{D_n} W_n^T$ via singular value decomposition $\Sigma_n=W_n D_n W_n^T$, where $W_n$ is some orthogonal matrix and $D_n \coloneqq {\rm diag}(\lambda_{n,1},\ldots,\lambda_{n,N})$ with $\lambda_{n,1},\ldots,\lambda_{n,N}$ being $\Sigma_n$'s singular values (or, equivalently, the eigenvalues, because $\Sigma_n$ is Hermite and positive-semidefinite).

We also introduce an extension of the linear equation system in Eq. \eqref{eq:AXhistB} with padding: for $R\in\{0\}\cup\mathbb{N}$, we consider
\begin{align}
    \mathcal{A}_+ \mathbf{X}_{\rm hist+} = \mathbf{B}_+
    \label{eq:AXhistB+}
\end{align}
where
\begin{align}
    \mathcal{A}_+ \coloneqq
    \begin{tikzpicture}[baseline=(m.center)]
    \node (m) {
    $
    \begin{pmatrix}
    I & & & & & & & \\
    -\Phi_0 & I & & & & & & \\
    & \ddots & \ddots & & & & & \\
    & & -\Phi_{r-1} & I & & & & \\
    & & & -I & I & & & \\
    & & & & \ddots & \ddots & & \\
    & & & & & -I & I & 
    \end{pmatrix}
    $
    };
    \draw[decorate,decoration={brace,amplitude=5pt}]
      (1.2,-0.1) -- (2.7,-1.4)
      node[midway,right=3pt,yshift=6pt] {$\text{$R$}$};
    \end{tikzpicture}
    ,
    \mathbf{X}_{\rm hist+} \coloneqq
    \begin{tikzpicture}[baseline=(m.center)]
    \node (m) {
    $
    \begin{pmatrix}
        \mathbf{X}_0 \\ \mathbf{X}_1 \\ \vdots \\ \mathbf{X}_r  \\ \mathbf{X}_r  \\ \vdots \\ \mathbf{X}_r 
    \end{pmatrix}
    $
    };
    \draw[decorate,decoration={brace,amplitude=5pt}]
      (0.4,-0.3) -- (0.4,-1.6)
      node[midway,right=3pt] {$\text{$R$}$};
    \end{tikzpicture}
    ,
    \mathbf{B}_+ \coloneqq
    \begin{tikzpicture}[baseline=(m.center)]
    \node (m) {
    $
    \begin{pmatrix}
    \mathbf{x}_0 \\ \boldsymbol{\Delta}_0 \\ \vdots \\ \boldsymbol{\Delta}_{r-1} \\ \mathbf{0}  \\ \vdots \\ \mathbf{0}
    \end{pmatrix}
    $
    };
    \draw[decorate,decoration={brace,amplitude=5pt}]
      (0.5,-0.3) -- (0.5,-1.6)
      node[midway,right=3pt] {$\text{$R$}$};
    \end{tikzpicture}
    .
    \label{eq:calAXHistB+Def}
\end{align}
Note that by adding $R$ blocks below $\mathcal{A}$ and $\mathbf{B}$, the solution $\mathbf{X}_{\rm hist+}$ has the $R$ copies of $\mathbf{X}_r$ as the additional blocks appended to $\mathbf{X}_{\rm hist}$.
This trick makes estimating expectations of functions of $\mathbf{X}_r$ efficient, as seen later.

We present a lemma on the inverse and norm of $\mathcal{A}_+$, including $\mathcal{A}$ as a special case with $R=0$.
\begin{lemma}
    For $\mathcal{A}_+$ defined by Eq.~\eqref{eq:calAXHistB+Def},
    \begin{align}
        \mathcal{A}_+^{-1} =
        \begin{pmatrix}
            (\mathcal{A}^{-1}_+)_{0,0} & \cdots & (\mathcal{A}_+^{-1})_{r+R,0} \\
            \vdots & \ddots & \vdots \\
            (\mathcal{A}_+^{-1})_{r+R,0} & \cdots & (\mathcal{A}_+^{-1})_{r+R,r+R}
        \end{pmatrix}
        ,
        \label{eq:calAInv}
    \end{align}
    where each block $(\mathcal{A}_+^{-1})_{n,n^\prime}\in \mathbb{R}^{N \times N},n,n^\prime\in[r+R]_0$ is given by
    \begin{align}
        (\mathcal{A}_+^{-1})_{n,n^\prime} =
        \begin{cases}
            I &  ; ~ (n = n^\prime) \lor (r \le n^\prime < n)\\
            \Phi_{n-1} \cdots \Phi_{n^\prime} & ; ~ n^\prime < n \le r \\
            \Phi_{r-1} \cdots \Phi_{n^\prime} & ; ~ n^\prime < r< n \\
            O & ; ~ n < n^\prime
        \end{cases}
        .
        \label{eq:calAInvBlock}
    \end{align}
    Besides, if Assumption \ref{ass:muA} and $\eta \Delta t \le 1$ hold,
    \begin{align}
        \left\|\mathcal{A}_+^{-1}\right\| \le \frac{2r}{\max\{1, \eta T\}} + R.
        \label{eq:calAInvNorm}
    \end{align}
    \label{lem:calAInvNorm}
\end{lemma}

\begin{proof}
Consider a linear equation system $\mathcal{A}_+ \mathbf{y} = \mathbf{a}$, with $\mathbf{y},\mathbf{a}\in\mathbb{R}^{(r+R+1)N}$.
Writing $\mathbf{y}=(\mathbf{y}_0^T,\ldots,\mathbf{y}_{r+R}^T)^T$ with $\mathbf{y}_0,\ldots,\mathbf{y}_{r+R}\in\mathbb{R}^{N}$ and $\mathbf{a}=(\mathbf{a}_0^T,\ldots,\mathbf{a}_{r+R}^T)^T$ with $\mathbf{a}_0,\ldots,\mathbf{a}_{r+R}\in\mathbb{R}^{N}$, we have \begin{align}
    \mathbf{y}_n =
    \begin{cases}
        \mathbf{a}_0 & ; ~ n = 0 \\
        \Phi_{n-1} \mathbf{y}_{n-1} + \mathbf{a}_n & ; ~ 1 \le n \le r \\
        \mathbf{y}_{n-1} + \mathbf{a}_n & ; ~ n > r
    \end{cases}
    .
\end{align}
By induction, we get
\begin{align}
    \mathbf{y}_n =
    \begin{cases}
        \sum_{n^\prime=0}^{n-1} \Phi_{n-1} \cdots \Phi_{n^\prime} \mathbf{a}_{n^\prime} + \mathbf{a}_n & ; ~ 1 \le n \le r \\
        \sum_{n^\prime=0}^{r-1} \Phi_{r-1} \cdots \Phi_{n^\prime} \mathbf{a}_{n^\prime} + \sum_{n^\prime=r}^n \mathbf{a}_{n^\prime} & ; ~ n > r
    \end{cases}
    ,
\end{align}
which means that Eq.~\eqref{eq:calAInv} holds.

Lemma 6 in Ref. \cite{an2026fast} implies
\begin{align}
    \left\|\mathcal{A}_+^{-1}\right\| \le \sqrt{\max_{n\in[r+R]_0} \sum_{n^\prime=0}^{r+R}\left\|(\mathcal{A}_+^{-1})_{n,n^\prime}\right\| \times \max_{n^\prime\in[r+R]_0} \sum_{n=0}^{r+R}\left\|(\mathcal{A}_+^{-1})_{n,n^\prime}\right\|}.
    \label{eq:calAInvNormTemp}
\end{align}
Combining Eqs. \eqref{eq:PhiUBEta} and \eqref{eq:calAInv} yields
\begin{align}
    \left\|(\mathcal{A}_+^{-1})_{n,n^\prime}\right\| \le
    \begin{cases}
        1 &  ; ~ (n = n^\prime) \lor (r \le n^\prime < n)\\
        e^{-\eta(n-n^\prime)\Delta t} & ; ~ n^\prime < n \le r \\
        e^{-\eta(r-n^\prime)\Delta t} & ; ~ n^\prime < r <n \\
        0 & ; ~ n < n^\prime
    \end{cases}
    ,
    \label{eq:calA+BlInv}
\end{align}
and using this in Eq.~\eqref{eq:calAInvNormTemp} yields
\begin{align}
    \left\|\mathcal{A}_+^{-1}\right\| \le \sum_{k=0}^{r} e^{-\eta k \Delta t} + R.
    \label{eq:calAInvNormTemp2}
\end{align}
We then have Eq.~\eqref{eq:calAInvNorm}, since the first term in the RHS in Eq.~\eqref{eq:calAInvNormTemp2} is upper-bounded by $\min\{2r,2r/\eta T\}$.
This bound follows from
\begin{align}
    \sum_{k=0}^L e^{-ka} \le \max\left\{L+1, \frac{2}{a}\right\}\le\max\left\{2L, \frac{2}{a}\right\},
    \label{eq:expSum}
\end{align}
which holds for any $L\in\mathbb{N}$ and $a\in(0,1]$.
This is seen from $\sum_{k=0}^L e^{-ka} \le L+1$ and
\begin{align}
    \sum_{k=0}^L e^{-ka} \le \sum_{k=0}^{\infty} e^{-ka} = \frac{1}{1-e^{-a}} \le \frac{2}{a},
    \label{eq:1-e-aInvUB}
\end{align}
where we use
\begin{align}
    \forall x\in[0,1], 1-e^{-x} \ge \frac{x}{2}.
    \label{eq:1-e-xLB}
\end{align}
\end{proof}

\subsection{Block-encoding of the time evolution operator}

Let us consider the approximation of $\Phi$ by the truncated Dyson series.
For later convenience, we introduce additional time grid points and consider $\Phi$ between these points.
Concretely, we set $M$ points in each segment $[t_n,t_{n+1}]$,
\begin{align}
    s_{n,j} \coloneqq t_n + j\delta t, \delta t \coloneqq \frac{\Delta t}{M}, j \in [M]_0,
\end{align}
and set $M$ points between $s_{n,j}$ and $t_{n+1}$,
\begin{align}
    \tau_{n,j,l} \coloneqq s_{n,j} + l\frac{t_{n+1}-s_{n,j}}{M}, l \in [M]_0.
\end{align}
Then, taking another integer $K\in\mathbb{N}$, we define
\begin{align}
    \tilde{\Phi}^{K,r,M}_{n,j} \coloneqq \sum_{k=0}^K \frac{1}{k!}\left(\frac{t_{n+1}-s_{n,j}}{M}\right)^k \sum_{l_1,\ldots,l_k=0}^{M-1} \mathcal{T} A(\tau_{n,j,l_1}) \cdots A(\tau_{n,j,l_k})
    \label{eq:PhiTilde}
\end{align}
as an approximation of $\Phi(t_{n+1},s_{n,j})$.

Its approximation error is bounded as follows.

\begin{lemma}
    Let $\epsilon>0$.
    Under Assumption \ref{ass:UA}, we have
    \begin{align}
        \left\|\tilde{\Phi}^{K,r,M}_{n,j} - \Phi(t_{n+1},s_{n,j})\right\| \le \epsilon
        \label{eq:PhiErrEps}
    \end{align}
    for $K$, $r$, and $M$ such that
    \begin{align}
        K\ge\max\left\{\log\left(\frac{3}{\epsilon}\right),7\right\},r\ge \alpha_A T ,M\ge \frac{4\alpha_{\drm A/\drm t}}{\alpha_A^2 \epsilon}.
        \label{eq:KrM}
    \end{align}
    \label{lem:PhiErr}
\end{lemma}

\begin{proof}
    According to Eqs. (6) and (124) in Ref.~\cite{Berry2024quantumalgorithm}, for $\Delta t=T/r$ such that $\alpha_A \Delta t \le 1$,
    \begin{align}
        \left\|\tilde{\Phi}^{K,r,M}_{n,j} - \Phi(t_{n+1},s_{n,j})\right\| \le \frac{3}{(K+1)!}\left(\frac{\alpha_AT}{r}\right)^{K+1}+\frac{2\alpha_{\drm A/\drm t} T^2}{r^2M}.
        \label{eq:PhiErrBd}
    \end{align}
    Noting that \cite{robbins1955remark}
    \begin{align}
        (K+1)! \ge \left(\frac{K+1}{e}\right)^{K+1}
    \end{align}
    and thus $(K+1)! \ge e^{K+1}$ for $K\ge7$, we plug $K$, $r$, and $M$ in Eq.~\eqref{eq:KrM} into the right-hand side of Eq.~\eqref{eq:PhiErrBd} and do some algebra to see that it is smaller than $\epsilon$ and that Eq.~\eqref{eq:PhiErrEps} thus holds.   
\end{proof}

Then, we can construct the block-encoding of this approximated time evolution operator.
The following construction is similar to those given in Refs. \cite{Kieferova2019,Berry2024quantumalgorithm}, but we present the details for completeness.

\begin{lemma}
    Let $\epsilon\in(0,1)$.
    Suppose that Assumption \ref{ass:UA} holds.
    Take $K$, $r$, and $M$ as Eq.~\eqref{eq:KrM}.
    Then, we have access to an oracle $U_{\tilde{\Phi}}$ that acts as
    \begin{align}
        U_{\tilde{\Phi}}\ket{n}\ket{j}\ket{\psi}=\ket{n}\ket{j}V_{\tilde{\Phi}^{K,r,M}_{n,j}}\ket{\psi}
        \label{eq:UTildePhi}
    \end{align}
    for any $n\in[r-1]_0$, any $j\in[M-1]_0$, and any state vector $\ket{\psi}$ on the space on which $V_{\tilde{\Phi}^{K,r,M}_{n,j}}$ acts.
    Here, $V_{\tilde{\Phi}^{K,r,M}_{n,j}}$ is an $(e, O(K(a_A+\log M+\log K)), 0)$-block-encoding of $\tilde{\Phi}^{K,r,M}_{n,j}$.
    $U_{\tilde{\Phi}}$ makes
    \begin{align}
        O\left(K\right)
    \end{align}
    calls to the controlled $U_A$ and additionally
    \begin{align}
        O\left(K \log K \log M\right)
    \end{align}
    uses of elementary gates.

    \label{lem:BETildePhi}
\end{lemma}

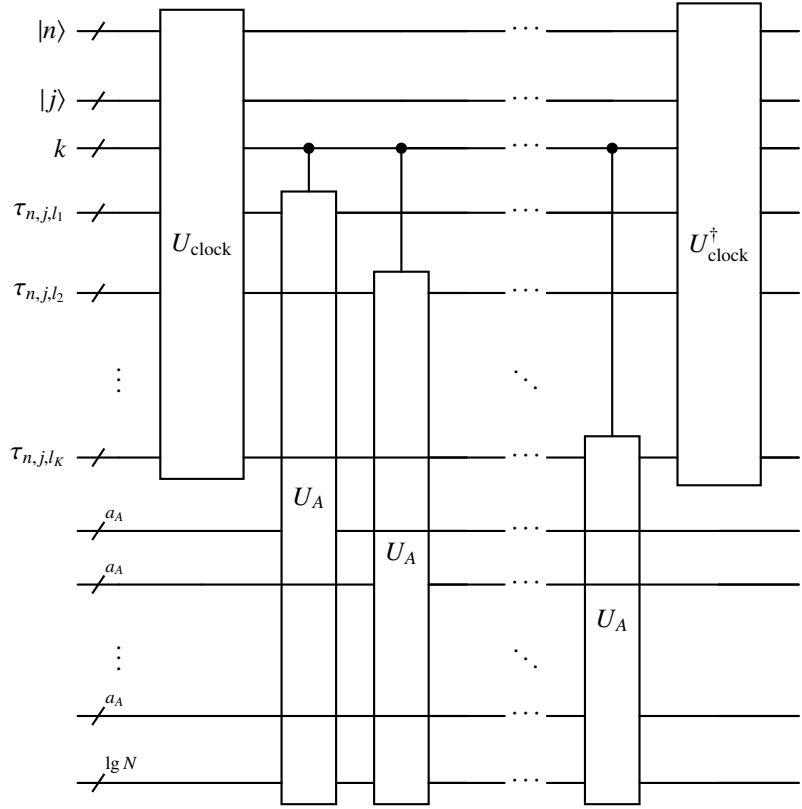
\begin{figure}[t]
\begin{center}

\begin{quantikz}[wire types={q,q,q,q,q,n,q,q,q,n,q,q}, transparent]
\lstick{$\ket{n}$}& \qwbundle{} & \gate[7, nwires=6]{U_{\rm clock}} & \qw & \qw & \qw & ~\cdots~ & \qw & \gate[7, nwires=6]{U_{\rm clock}^\dagger} & \qw \\ 
\lstick{$\ket{j}$}& \qwbundle{} & & \qw & \qw & \qw & ~\cdots~ & \qw & & \qw\\ 
\lstick{$k$}& \qwbundle{} & & \ctrl{1} & \ctrl{2} & \qw & ~\cdots~ & \ctrl{4} &  & \qw\\ 
\lstick{$\tau_{n,j,l_1}$} & \qwbundle{} & & \gate[9]{U_A} & \qw & \qw & ~\cdots~ & \qw & & \qw \\ 
\lstick{$\tau_{n,j,l_2}$} & \qwbundle{} & & \linethrough & \gate[8,label style={yshift=-0.2cm}]{U_A} & \qw & ~\cdots~ & \qw & & \qw \\ 
 & \vdots & & & & & \ddots & & & &\\
\lstick{$\tau_{n,j,l_K}$} & \qwbundle{} & & \linethrough & \linethrough & \qw & ~\cdots~  & \gate[6]{U_A} & & \qw \\
& \qwbundle{a_A} & & & \linethrough & \qw & ~\cdots~ & \linethrough & & \qw \\
& \qwbundle{a_A} & & \linethrough &  & \qw & ~\cdots~ & \linethrough  & & \qw \\
 & \vdots & & & & & \ddots & & & &\\
& \qwbundle{a_A} & & \linethrough & \linethrough & \qw & ~\cdots~ & & & \qw \\
& \qwbundle{\lg N} & & & & \qw & ~\cdots~ &  & & \qw
\end{quantikz}
\caption{The quantum circuit implementation of $U_{\tilde{\Phi}}^\prime$.}
\label{fig:UTildePhi}
\end{center}
\end{figure}

\begin{proof}

Let us consider the circuit $U_{\tilde{\Phi}}^\prime$ shown as the diagram in Fig.~\ref{fig:UTildePhi}.
In this circuit, we use $U_{\rm clock}$, a slight modification of the "prepare $\ket{\rm clock}$" circuit in Ref. \cite{Kieferova2019}, which acts as
\begin{align}
    U_{\rm clock} \ket{n}\ket{j}\ket{0}^{\otimes (K+1)} = \frac{1}{\sqrt{\mathcal{N}_{n,j}}}\sum_{(k,\mathbf{l})\in\mathcal{I}} \sqrt{\beta_{n,j,\mathbf{l}}} \ket{n}\ket{j}\ket{k}\ket{\tau_{n,j,l_1}}\cdots\ket{\tau_{n,j,l_K}}.
    \label{eq:Uclock}
\end{align}
Here, $\mathcal{I}$ is the set of tuples $(k,\mathbf{l})$, where $k\in[K]_0$ and $\mathbf{l}$ is a $k$-tuple of integers $l_1,\ldots,l_k\in[M-1]_0$.
For a fixed $k$, $\tau_{n,j,l_{k+1}},...,\tau_{n,j,l_K}$ are not used, and they can be set to arbitrary numbers, say, $-1$.
$\beta_{n,j,\mathbf{l}}$ is defined by
\begin{align}
    \beta_{n,j,\mathbf{l}} = \frac{\alpha_A^k}{k_1! \cdots k_{\varsigma}!} \left(\frac{t_{n+1}-s_{n,j}}{M}\right)^k \mathbbm{1}_{l_1 \le \cdots \le l_k},
\end{align}
where $\varsigma$ is the number of distinct values in $\mathbf{l}$ and $k_1,\ldots, k_{\varsigma}$ are the number of repetitions for each of the distinct values.
$\mathcal{N}_{n,j}\coloneqq\sum_{(k,\mathbf{l})\in\mathcal{I}} \beta_{n,j,\mathbf{l}}$ is the normalization factor, for which some algebra yields
\begin{align}
    \mathcal{N}_{n,j} = \sum_{k=0}^{K} \frac{\alpha_A^k(t_{n+1}-s_{n,j})^k}{k!}
\end{align}
and, for current $r$,
\begin{align}
    \mathcal{N}_{n,j} \le e.
\end{align}
After operating $U_{\rm clock}$, we operate a sequence of $K$ controlled versions of $U_A$, the $\ell$-th of which is
\begin{align}
    \sum_{k \ge \ell} \ket{k}\!\bra{k} \otimes U^{(\ell)}_A + \sum_{k < \ell} \ket{k}\!\bra{k} \otimes I,
    \label{eq:UPhiPr}
\end{align}
where $U^{(\ell)}_A$ is $U_A$ that acts on the system consisting of the $(2K-\ell+2)$-th-to-last register, the $(K-\ell+2)$-th-to-last register, and the last register.
Then, lastly, $U_{\rm clock}^\dagger$ is operated.
Since
\begin{align}
\sum_{(k,\mathbf{l})\in\mathcal{I}} \beta_{n,j,\mathbf{l}} \frac{A(\tau_{n,j,l_1})}{\alpha_A} \cdots \frac{A(\tau_{n,j,l_k})}{\alpha_A}=\tilde{\Phi}^{K,r,M}_{n,j},   
\end{align}
based on the \ac{LCU} technique, $U_{\tilde{\Phi}}^\prime$ acts as
\begin{align}
    U_{\tilde{\Phi}}^\prime\ket{n}\ket{j}\ket{0}\ket{\psi}=\frac{1}{\mathcal{N}_{n,j}}\ket{n}\ket{j}\ket{0}\tilde{\Phi}^{K,r,M}_{n,j}\ket{\psi} + \ket{0_\perp},
\end{align}
where the third ket corresponds to the system consisting of the third through second-to-last registers in Fig.~\ref{fig:UTildePhi}, $\ket{\psi}$ is an arbitrary state on the last register in Fig.~\ref{fig:UTildePhi}, and $\ket{0_\perp}$ is a state such that $(I \otimes I \otimes \ket{0}\!\bra{0} \otimes I)\ket{0_\perp}=0$.
By adding one more ancillary qubit after the register for $\ket{j}$ and combining $U_{\tilde{\Phi}}^\prime$ with the quantum circuit $U_{\mathcal{N}}$ acting as
\begin{align}
    U_{\mathcal{N}} \ket{n}\ket{j}\ket{0} = \ket{n}\ket{j}\left(\frac{\mathcal{N}_{n,j}}{e}\ket{0} + \sqrt{1-\frac{\mathcal{N}_{n,j}^2}{e^2}}\ket{1}\right),
\end{align}
which can be implemented as a combination of the circuit for computing $\mathcal{N}_{n,j}$ and $U_{\rm B2A}$, we finally construct the quantum circuit $U_{\tilde{\Phi}}$ that acts as
\begin{align}
    U_{\tilde{\Phi}}\ket{n}\ket{j}\ket{0}\ket{0}\ket{\psi}=\frac{1}{e}\ket{n}\ket{j}\ket{0}\ket{0}\tilde{\Phi}^{K,r,M}_{n,j}\ket{\psi} + \ket{0_\perp},
\end{align}
where $\ket{0_\perp}$ now represents a state such that $(I \otimes I \otimes \ket{0}\!\bra{0} \otimes \ket{0}\!\bra{0} \otimes I)\ket{0_\perp}=0$.
This is nothing but $U_{\tilde{\Phi}}$ as stated.

In the circuit shown in Fig.~\ref{fig:UTildePhi}, the number of queries to the controlled versions of $U_A$ is $O(K)$, and the number of other gates is as stated, which is the same as the gate complexity of the "prepare $\ket{\rm clock}$" circuit~\cite{Kieferova2019}.
Since $\mathcal{N}_{n,j}$ can be computed with $O(K)$ additions and multiplications by Horner's method~\cite[Sec. 5.1]{higham2002accuracy}, $U_{\mathcal{N}}$ can be implemented by $O(K)$ elementary gates, making no change to the order of the complexity.

\end{proof}

\subsection{Block-encoding of the covariance matrix and its square root}

Next, we consider the block-encodings of the covariance matrix $\Sigma_n$ and its square root.
We begin with presenting an assumption on $B(t)B(t)^T$, which represents the covariance of the increment ${\rm d} \mathbf{X}_t$ of $\mathbf{X}_t$.

\begin{assum}
    For any $t\in[0,T]$, the eigenvalues of $B(t)B(t)^T$ are in
    $[\sigma^2/\kappa_{BB^T}, \sigma^2]$, where $\sigma>0$ and $\kappa_{BB^T} \ge 1$.
    Besides,
    \begin{align}
        \max_{t\in[0,T]}\left\|\frac{\drm}{\drm t}\left(B(t)B^T(t)\right)\right\| \le \alpha_{\drm BB^T/\drm t}
        \label{eq:alphaBBDot}
    \end{align}
    holds for $\alpha_{\drm BB^T/\drm t}>0$.
    Furthermore, letting $a_{BB^T}$ be a non-negative integer,
    we have access to the oracle $U_{BB^T}$ that acts as
    \begin{align}
        U_{BB^T} \ket{t} \ket{\psi} = \ket{t} V_{B(t)B^T(t)}\ket{\psi}, 
    \end{align}
    for any $t\in[0,T]$ and any $2^{a_{BB^T}}N$-dimensional state vector $\ket{\psi}$, where $V_{B(t)B^T(t)}$ is a $(\sigma^2,a_{BB^T},0)$-block-encoding of $B(t)B^T(t)$.
    \label{ass:UBB}
\end{assum}
Under this assumption,
\begin{align}
    \max_{t\in[0,T]}\left\|B(t)B(t)^T\right\| \le \sigma^2
    \label{eq:BBnorm}
\end{align}
immediately follows.
We also note that this assumption implies that $B(t)B(t)^T$ is full-rank and that $m=N$ holds. 

We define
\begin{align}
\tilde{\Sigma}^{K,r,M}_n \coloneqq \sum_{j=0}^{M-1} \tilde{\Phi}^{K,r,M}_{n,j} B(s_{n,j}) B^T(s_{n,j}) \left(\tilde{\Phi}^{K,r,M}_{n,j}\right)^T \delta t
\end{align} 
as an approximation of $\Sigma_n$, whose error is bounded as follows.

\begin{lemma}
    Let $\epsilon\in(0,1]$.
    Under Assumptions \ref{ass:muA} and \ref{ass:UBB},
    \begin{align}
        \left\|\tilde{\Sigma}^{K,r,M}_n - \Sigma_n\right\| \le \epsilon \sigma^2 \Delta t
        \label{eq:tilSigDifEps}
    \end{align}
    holds for $K$, $r$, and $M$ such that
    \begin{align}
        K\ge\max\left\{\log\left(\frac{18}{\epsilon}\right),7\right\},
        r\ge \alpha_A T ,
        M\ge \max\left\{\frac{24\alpha_{\drm A/\drm t}}{\alpha_A^2\epsilon},\frac{2\left(2\alpha_A + \frac{\alpha_{\drm BB^T/\drm t}}{\sigma^2}\right)\Delta t}{\epsilon}\right\}.
        \label{eq:KrM2}
    \end{align}
    \label{lem:SigmaErr}
\end{lemma}

\begin{proof}
    $ \left\|\tilde{\Sigma}^{K,r,M}_n - \Sigma_n\right\|$ can be divided as
    \begin{align}
         \left\|\tilde{\Sigma}^{K,r,M}_n - \Sigma_n\right\| \le \left\|\hat{\Sigma}^{K,r,M}_n - \Sigma_n\right\| + \left\|\tilde{\Sigma}^{K,r,M}_n - \hat{\Sigma}^{K,r,M}_n\right\|,
    \end{align}
    where
    \begin{align}
        \hat{\Sigma}^{K,r,M}_n \coloneqq \sum_{j=0}^{M-1} \Phi(t_{n+1},s_{n,j}) B(s_{n,j}) B^T(s_{n,j}) \Phi^T(t_{n+1},s_{n,j}) \delta t. 
    \end{align}
    
    First, we consider $\left\|\hat{\Sigma}^{K,r,M}_n - \Sigma_n\right\|$.
    We have
    \begin{align}
        \left\|\hat{\Sigma}^{K,r,M}_n - \Sigma_n\right\|
        \le \sum_{j=0}^{M-1} \int_{s_{n,j}}^{s_{n,j+1}}\left\|\Phi(t_{n+1},s_{n,j}) B(s_{n,j}) B^T(s_{n,j}) \Phi^T(t_{n+1},s_{n,j})-\Phi(t_{n+1},\tau) B(\tau) B^T(\tau) \Phi^T(t_{n+1},\tau)\right\|\drm \tau.
        \label{eq:SigmaHatSigmaDiff}
    \end{align}
    For $\tau\in[s_{n,j},s_{n,j+1}]$, we have
    \begin{align}
        & \left\|\Phi(t_{n+1},s_{n,j}) B(s_{n,j}) B^T(s_{n,j}) \Phi^T(t_{n+1},s_{n,j})-\Phi(t_{n+1},\tau) B(\tau) B^T(\tau) \Phi^T(t_{n+1},\tau)\right\| \nonumber \\
        \le & \left\|\Phi(t_{n+1},s_{n,j})-\Phi(t_{n+1},\tau)\right\| \cdot \left\|B(s_{n,j}) B^T(s_{n,j}) \right\| \cdot \left\|\Phi^T(t_{n+1},s_{n,j})\right\| \nonumber \\
        & + \left\|\Phi(t_{n+1},\tau)\right\| \cdot \|B(s_{n,j}) B^T(s_{n,j}) - B(\tau) B^T(\tau)\| \cdot \left\|\Phi^T(t_{n+1},s_{n,j})\right\|  \nonumber \\
        & +  \left\|\Phi(t_{n+1},\tau)\right\| \cdot \|B(\tau) B^T(\tau)\| \cdot \left\|\Phi^T(t_{n+1},s_{n,j}) - \Phi^T(t_{n+1},\tau)\right\|   \nonumber \\
        \le & 2 \sigma^2 \left\|\Phi(t_{n+1},s_{n,j})-\Phi(t_{n+1},\tau)\right\| + \|B(s_{n,j}) B^T(s_{n,j}) - B(\tau) B^T(\tau)\|,
        \label{eq:PhiBBPhiErrTemp}
    \end{align}
    where we use Eqs. \eqref{eq:PhiBd} and \eqref{eq:BBnorm}.
    We note that
    \begin{align}
        \left\|\Phi(t_{n+1},\tau) - \Phi(t_{n+1},s_{n,j})\right\| 
        &= \left\|\int_{s_{n,j}}^{\tau} \frac{\partial}{\partial \tau^\prime}\Phi(t_{n+1},\tau^\prime) \drm \tau^\prime \right\| \nonumber \\
        &= \left\|\int_{s_{n,j}}^{\tau} \Phi(t_{n+1},\tau^\prime)A(\tau^\prime) \drm \tau^\prime \right\| \nonumber \\
        &\le \int_{s_{n,j}}^{\tau} \left\|A(\tau^\prime)\right\| \cdot \left\|\Phi(t_{n+1},\tau^\prime)\right\| \drm \tau^\prime \nonumber \\
        &\le \alpha_A \delta t,
    \end{align}
    where we use Eqs.~\eqref{eq:PhiDot}, \eqref{eq:PhiBd}, and \eqref{eq:alphaA}, and that
    \begin{align}
        \left\|B(s_{n,j}) B^T(s_{n,j}) - B(\tau) B^T(\tau)\right\| 
        &= \left\|\int_{s_{n,j}}^{\tau} \frac{\drm}{\drm \tau^\prime}\left(B(\tau^\prime) B^T(\tau^\prime)\right) \drm \tau^\prime \right\| \nonumber \\
        &\le \int_{s_{n,j}}^{\tau} \left\| \frac{\drm}{\drm \tau^\prime}\left(B(\tau^\prime) B^T(\tau^\prime)\right) \right\| \drm \tau^\prime \nonumber \\
        &\le \alpha_{\drm BB^T/\drm t} \delta t,
    \end{align}
    where we use Eq.~\eqref{eq:alphaBBDot}.
    By plugging these into Eq.~\eqref{eq:PhiBBPhiErrTemp}, we get
    \begin{align}
        \left\|\Phi(t_{n+1},s_{n,j}) B(s_{n,j}) B^T(s_{n,j}) \Phi^T(t_{n+1},s_{n,j})-\Phi(t_{n+1},\tau) B(\tau) B^T(\tau) \Phi^T(t_{n+1},\tau)\right\| \le (2\alpha_A\sigma^2 + \alpha_{\drm BB^T/\drm t}) \delta t,
    \end{align}
    and by using this in Eq. \eqref{eq:SigmaHatSigmaDiff}, we obtain
    \begin{align}
        \left\|\hat{\Sigma}^{K,r,M}_n - \Sigma_n\right\| \le (2\alpha_A\sigma^2 + \alpha_{\drm BB^T/\drm t}) \delta t\Delta t.
    \end{align}
    For $r$ and $M$ in Eq.~\eqref{eq:KrM2}, this becomes
    \begin{align}
        \left\|\hat{\Sigma}^{K,r,M}_n - \Sigma_n\right\| \le \frac{\epsilon \sigma^2 \Delta t}{2}.
        \label{eq:hatSigDiff}
    \end{align}

    Next, let us consider $\left\|\tilde{\Sigma}^{K,r,M}_n - \hat{\Sigma}^{K,r,M}_n\right\|$.
    We have
    \begin{align}
        \left\|\tilde{\Sigma}^{K,r,M}_n - \hat{\Sigma}^{K,r,M}_n\right\|\le\sum_{j=0}^{M-1} \left\|\tilde{\Phi}^{K,r,M}_{n,j} B(s_{n,j}) B^T(s_{n,j}) \left(\tilde{\Phi}^{K,r,M}_{n,j}\right)^T-\Phi(t_{n+1},s_{n,j}) B(s_{n,j}) B^T(s_{n,j}) \Phi^T(t_{n+1},s_{n,j})\right\|\delta t.
        \label{eq:tilhatSigDiff}
    \end{align}
    The norm in the sum is bounded as
    \begin{align}
        & \left\|\tilde{\Phi}^{K,r,M}_{n,j} B(s_{n,j}) B^T(s_{n,j}) \left(\tilde{\Phi}^{K,r,M}_{n,j}\right)^T-\Phi(t_{n+1},s_{n,j}) B(s_{n,j}) B^T(s_{n,j}) \Phi^T(t_{n+1},s_{n,j})\right\| \nonumber \\
        \le & \left\|\tilde{\Phi}^{K,r,M}_{n,j} - \Phi(t_{n+1},s_{n,j})\right\|\cdot \|B(s_{n,j}) B^T(s_{n,j})\| \cdot \left\|\tilde{\Phi}^{K,r,M}_{n,j}\right\| + \left\|\tilde{\Phi}^{K,r,M}_{n,j} - \Phi(t_{n+1},s_{n,j})\right\|\cdot \|B(s_{n,j}) B^T(s_{n,j})\| \cdot \|\Phi(t_{n+1},s_{n,j})\| \nonumber \\
        \le & \sigma^2 \left\|\tilde{\Phi}^{K,r,M}_{n,j} - \Phi(t_{n+1},s_{n,j})\right\| \left(\left\|\tilde{\Phi}^{K,r,M}_{n,j} - \Phi(t_{n+1},s_{n,j})\right\| + \left\|\Phi(t_{n+1},s_{n,j})\right\| + \left\|\Phi(t_{n+1},s_{n,j})\right\|\right) \nonumber \\
        \le & \sigma^2 \left\|\tilde{\Phi}^{K,r,M}_{n,j} - \Phi(t_{n+1},s_{n,j})\right\| \left(2 + \left\|\tilde{\Phi}^{K,r,M}_{n,j} - \Phi(t_{n+1},s_{n,j})\right\|\right) \nonumber \\
        \le & \sigma^2 \cdot \frac{\epsilon}{6} \cdot 3 \nonumber \\
        \le & \frac{1}{2}\sigma^2\epsilon.
        \label{eq:PhiBBPhiDiff}
    \end{align}
    Here, we use Eq.~\eqref{eq:BBnorm} and
    \begin{align}
        \left\|\tilde{\Phi}^{K,r,M}_{n,j} - \Phi(t_{n+1},s_{n,j})\right\| \le \frac{ \epsilon}{6}<1,
    \end{align}
    which is implied by Lemma \ref{lem:PhiErr} for $K$, $r$, and $M$ in Eq.~\eqref{eq:KrM2}.
    Plugging Eq.~\eqref{eq:PhiBBPhiDiff} into Eq.~\eqref{eq:tilhatSigDiff} yields
    \begin{align}
        \left\|\tilde{\Sigma}^{K,r,M}_n - \hat{\Sigma}^{K,r,M}_n\right\|\le\frac{\epsilon \sigma^2 \Delta t}{2}.
        \label{eq:tilhatSigDiff2}
    \end{align}
    
    Combining Eqs.~\eqref{eq:hatSigDiff} and \eqref{eq:tilhatSigDiff2}, we obtain Eq.~\eqref{eq:tilSigDifEps}.
\end{proof}

Now, we can construct the block-encoding of this approximated covariance matrix.

\begin{lemma}
    Under Assumptions \ref{ass:UA} and \ref{ass:UBB}, we have access to an oracle $U_{\tilde{\Sigma}}$ that acts as
    \begin{align}
        U_{\tilde{\Sigma}} \ket{n} \ket{\phi} = \ket{n} V_{\tilde{\Sigma}^{K,r,M}_n} \ket{\phi}
        \label{eq:UTildeSig}
    \end{align}
    for any $n\in[r-1]_0$ and any state vector $\ket{\phi}$ on the space on which $V_{\tilde{\Sigma}^{K,r,M}_{n}}$ acts.
    Here, $V_{\tilde{\Sigma}^{K,r,M}_n}$ is a $(e^2\sigma^2 \Delta t, O(K(a_A+\log M+\log K)+a_{BB^T}),0)$-block-encoding of $\tilde{\Sigma}^{K,r,M}_n$.
    $U_{\tilde{\Sigma}}$ makes
    \begin{align}
        O(K)
    \end{align}
    queries to the controlled $U_A$, one query to $U_{BB^T}$, and additionally
    \begin{align}
        O(K \log K \log M)
    \end{align}
    uses of elementary gates.
    \label{lem:BETildeSig}
\end{lemma}

\begin{proof}

Because of Lemma~\ref{lem:BEProd}, given access to $V_{B(s_{n,j})B^T(s_{n,j})}$ and $V_{\tilde{\Phi}^{K,r,M}_{n,j}}$, we can implement an $(e^2\sigma^2, O(K(a_A+\log M+\log K))+a_{BB^T},0)$-block-encoding $V^{\tilde{\Phi}BB^T\tilde{\Phi}^T}_{n,j}$ of $\tilde{\Phi}^{K,r,M}_{n,j}B(s_{n,j})B(s_{n,j})^T\left(\tilde{\Phi}^{K,r,M}_{n,j}\right)^T$.
Thus, we can use $U_{\tilde{\Phi}}$ twice and $U_{BB^T}$ once to implement a quantum circuit $U_{\tilde{\Phi}BB^T\tilde{\Phi}^T}$ that acts as
\begin{align}
    U_{\tilde{\Phi}BB^T\tilde{\Phi}^T}\ket{n}\ket{j}\ket{\psi}=\ket{n}\ket{j}V^{\tilde{\Phi}BB^T\tilde{\Phi}^T}_{n,j}\ket{\psi}
\end{align}
for any $n\in[r-1]_0$, any $j\in[M-1]_0$, and any state vector $\ket{\psi}$ on which $V^{\tilde{\Phi}BB^T\tilde{\Phi}^T}_{n,j}$ can act.

Consider $U^{\rm SP}_{0,M-1}$, which acts as
\begin{align}
    U^{\rm SP}_{0,M-1} \ket{0} = \frac{1}{\sqrt{M}}\sum_{j=0}^{M-1} \ket{j}=\frac{1}{\sqrt{\Delta t}}\sum_{j=0}^{M-1} \sqrt{\delta t} \ket{j}
\end{align}
and can be implemented by $O(\log M)$ elementary gates.
Then, based on the \ac{LCU} technique,
\begin{align}
    \left(I \otimes \left(U^{\rm SP}_{0,M-1}\right)^\dagger \otimes I\right)U_{\tilde{\Phi}BB^T\tilde{\Phi}^T}(I \otimes U^{\rm SP}_{0,M-1} \otimes I)
    \label{eq:UTildeSigImp}
\end{align}
acts as Eq.~\eqref{eq:UTildeSig} with $V_{\tilde{\Sigma}^{K,r,M}_n}$ being an $(e^2\sigma^2\Delta t, O(K(a_A+\log M+\log K))+a_{BB^T},0)$-block-encoding of \\ $\sum_{j=0}^{M-1}\tilde{\Phi}^{K,r,M}_{n,j}B(s_{n,j})B(s_{n,j})^T\left(\tilde{\Phi}^{K,r,M}_{n,j}\right)^T \delta t=\tilde{\Sigma}^{K,r,M}_n$.
This is nothing but $U_{\tilde{\Sigma}}$ in the claim.

Noting that $U_{\tilde{\Sigma}}$ consists of two queries to $U_{\tilde{\Phi}}$, one query to $U_{BB^T}$, and two queries to $U^{\rm SP}_{0,M-1}$, and recalling $U_{\tilde{\Phi}}$'s complexity stated in Lemma~\ref{lem:BETildePhi}, we sum up the complexities of these queries to obtain the evaluations of the complexity of $U_{\tilde{\Sigma}}$ in the claim.
\end{proof}

Then, let us construct the block-encoding of $S_n$, the square root of $\Sigma_n$, using $\tilde{\Sigma}^{K,r,M}_n$ as an approximation of $\Sigma_n$ and Theorem \ref{th:BESqrt}.
For this, it is necessary to confirm that the eigenvalues of $\Sigma_n$ are lower bounded by some nonzero value.
In fact, this holds under Assumption \ref{ass:UBB} and the appropriate setting of $r$.

\begin{lemma}
    Under Assumptions \ref{ass:UA} and \ref{ass:UBB}, all the eigenvalues of $\Sigma_n$ are contained in $[\sigma^2\Delta t/2\kappa_{BB^T}, \sigma^2\Delta t]$ for
    \begin{align}
        r \ge 4\kappa_{BB^T}\alpha_AT.
        \label{eq:rSigKappa}
    \end{align}
    \label{lem:kappaSig}
\end{lemma}

\begin{proof}
    Regarding the upper bound of the eigenvalues, since $\Sigma_n$ is positive-semidefinite Hermitian, it suffices to show
    \begin{align}
        \|\Sigma_n\| \le \sigma^2 \Delta t.
        \label{eq:SigmanNorm}
    \end{align}
    We can see this as
    \begin{align}
        \|\Sigma_n\| & = \left\|\int_{t_n}^{t_{n+1}} \Phi(t_{n+1},\tau) B(\tau) B^T(\tau) \Phi^T(t_{n+1},\tau) \drm \tau\right\| \nonumber \\
        & \le \int_{t_n}^{t_{n+1}} \left\|\Phi(t_{n+1},\tau)\right\| \cdot \left\|B(\tau) B^T(\tau)\right\| \cdot \left\|\Phi^T(t_{n+1},\tau)\right\| \drm \tau  \nonumber \\
        & \le \sigma^2 \Delta t,
    \end{align}
    where we use Eq.~\eqref{eq:PhiBd} and Eq.~\eqref{eq:BBnorm}.

    The lower bound is shown as follows.
    According to the Bauer–Fike theorem (Exercise 26.3 in Ref.~\cite{trefethen2022numerical}, or, originally, Theorem IIIa in Ref.~\cite{Bauer1960}), for any eigenvalue $\lambda$ of $\Sigma_n$, there exists an eigenvalue $\lambda^{(0)}$ of $\Sigma^{(0)}_n\coloneqq\int_{t_n}^{t_{n+1}} B(\tau)B^T(\tau) \drm \tau$ such that
    \begin{align}
        |\lambda - \lambda^{(0)}| \le \left\|\Sigma_n-\Sigma^{(0)}_n\right\|
    \end{align}
    Besides, since Assumption \ref{ass:UBB} implies that $\sigma^2/\kappa_{BB^T} \le \mathbf{u}^T B(t)B^T(t) \mathbf{u} \le \sigma^2$ holds for any $t\in[0,T]$ and any $N$-dimensional real unit vector $\mathbf{u}$, we have
    \begin{align}
        \mathbf{u}^T \Sigma^{(0)}_n\mathbf{u}=\int_{t_n}^{t_{n+1}} \mathbf{u}^T B(t)B^T(t) \mathbf{u} \drm \tau \in\left[\frac{\sigma^2 \Delta t}{\kappa_{BB^T}},\sigma^2 \Delta t\right],
    \end{align}
    which implies the eigenvalues of $\Sigma^{(0)}_n$ are contained in $[\sigma^2 \Delta t/\kappa_{BB^T},\sigma^2 \Delta t]$.
    Thus, we can prove the statement by showing
    \begin{align}
        \left\|\Sigma_n-\Sigma^{(0)}_n\right\| \le \frac{\sigma^2 \Delta t}{2\kappa_{BB^T}}.
    \end{align}
    For this, invoking Eq.~\eqref{eq:PhiDef}, we note
    \begin{align}
        \Sigma_n-\Sigma^{(0)}_n
        = & \int_{t_n}^{t_{n+1}}\drm s\left(\sum_{l=0}^\infty \int_s^{t_{n+1}} \drm \tau_1 \int_s^{\tau_1} \drm \tau_2 \cdots \int_s^{\tau_{l-1}} \drm \tau_n A(\tau_1) \cdots A(\tau_l)\right) B(s)B^T(s) \nonumber \\
        &\qquad\qquad \times \left(\sum_{l=0}^\infty \int_s^{t_{n+1}} \drm \tau_1 \int_s^{\tau_1} \drm \tau_2 \cdots \int_s^{\tau_{l-1}} \drm \tau_n A^T(\tau_l) \cdots A^T(\tau_1)\right) - \int_{t_n}^{t_{n+1}}B(s)B^T(s)\drm s \nonumber \\
        = & \int_{t_n}^{t_{n+1}}\drm s\left(\sum_{l=1}^\infty \int_s^{t_{n+1}} \drm \tau_1 \int_s^{\tau_1} \drm \tau_2 \cdots \int_s^{\tau_{l-1}} \drm \tau_n A(\tau_1) \cdots A(\tau_l)\right) B(s)B^T(s) \nonumber \\
        &\qquad\qquad \times \left(\sum_{l=1}^\infty \int_s^{t_{n+1}} \drm \tau_1 \int_s^{\tau_1} \drm \tau_2 \cdots \int_s^{\tau_{l-1}} \drm \tau_n A^T(\tau_l) \cdots A^T(\tau_1)\right) \nonumber \\
        & + \int_{t_n}^{t_{n+1}}\drm s\left(\sum_{l=1}^\infty \int_s^{t_{n+1}} \drm \tau_1 \int_s^{\tau_1} \drm \tau_2 \cdots \int_s^{\tau_{l-1}} \drm \tau_n A(\tau_1) \cdots A(\tau_l)\right) B(s)B^T(s) \nonumber \\
        & + \int_{t_n}^{t_{n+1}}\drm sB(s)B^T(s)\left(\sum_{l=1}^\infty \int_s^{t_{n+1}} \drm \tau_1 \int_s^{\tau_1} \drm \tau_2 \cdots \int_s^{\tau_{l-1}} \drm \tau_n A^T(\tau_l) \cdots A^T(\tau_1)\right).
    \end{align}
    This leads to
    \begin{align}
        \left\|\Sigma_n-\Sigma^{(0)}_n\right\|
        \le & \int_{t_n}^{t_{n+1}}\drm s\left(\sum_{l=1}^\infty \int_s^{t_{n+1}} \drm \tau_1 \int_s^{\tau_1} \drm \tau_2 \cdots \int_s^{\tau_{l-1}} \drm \tau_n \|A(\tau_1)\| \cdots \|A(\tau_l)\|\right) \|B(s)B^T(s)\| \nonumber \\
        &\qquad\qquad \times \left(\sum_{l=1}^\infty \int_s^{t_{n+1}} \drm \tau_1 \int_s^{\tau_1} \drm \tau_2 \cdots \int_s^{\tau_{l-1}} \drm \tau_n \|A^T(\tau_l)\| \cdots \|A^T(\tau_1)\|\right) \nonumber \\
        & + \int_{t_n}^{t_{n+1}}\drm s\left(\sum_{l=1}^\infty \int_s^{t_{n+1}} \drm \tau_1 \int_s^{\tau_1} \drm \tau_2 \cdots \int_s^{\tau_{l-1}} \drm \tau_n \|A(\tau_1)\| \cdots \|A(\tau_l)\|\right) \|B(s)B^T(s)\| \nonumber \\
        & + \int_{t_n}^{t_{n+1}}\drm s\|B(s)B^T(s)\|\left(\sum_{l=1}^\infty \int_s^{t_{n+1}} \drm \tau_1 \int_s^{\tau_1} \drm \tau_2 \cdots \int_s^{\tau_{l-1}} \drm \tau_n \|A^T(\tau_l)\| \cdots \|A^T(\tau_1)\|\right).
    \end{align}
    Using Eqs. \eqref{eq:alphaA} and \eqref{eq:BBnorm}, we have
    \begin{align}
        \left\|\Sigma_n-\Sigma^{(0)}_n\right\|&\le \sigma^2\int_{t_n}^{t_{n+1}}\drm s \left(\sum_{l=1}^\infty \frac{1}{l!}(\alpha_A (t_{n+1}-s))^l\right)^2 + 2\sigma^2\int_{t_n}^{t_{n+1}}\drm s\sum_{l=1}^\infty \frac{1}{l!}(\alpha_A (t_{n+1}-s))^l \nonumber \\
        &= \sigma^2\int_{t_n}^{t_{n+1}} \left(e^{\alpha_A (t_{n+1}-s)}-1\right)^2\drm s + 2\sigma^2\int_{t_n}^{t_{n+1}} \left(e^{\alpha_A (t_{n+1}-s)}-1\right)\drm s \nonumber \\
        &=\sigma^2 \Delta t \left[\frac{1}{2\alpha_A\Delta t}(e^{2\alpha_A \Delta t}-1)-1\right].
    \end{align}
    Since $\frac{e^x-1}{x}-1 \le x$ holds for any $x\in[0,1/2]$, we have $\left\|\Sigma_n-\Sigma^{(0)}_n\right\| \le \sigma^2 \Delta t\cdot 2\alpha_A \Delta t < \sigma^2 \Delta t/2\kappa_{BB^T}$ if $\alpha_A \Delta t \le 1/4\kappa_{BB^T}$, which holds for $r$ in Eq.~\eqref{eq:rSigKappa}.
    This concludes the proof.
\end{proof}

Now, we can construct the block-encoding of $S_n$.

\begin{lemma}
    Let $\epsilon>0$.
    Suppose that Assumptions \ref{ass:UA} and \ref{ass:UBB} hold.
    Then, for $K$, $r$, and $M$ satisfying
    \begin{align}
        K&\ge\max\left\{\left\lceil\log\left(\frac{1152\pi\left(\lg\left(4\kappa_{BB^T}\log\left(\frac{8}{\epsilon}\right)\right)+1\right)^2\log\left(\frac{8}{\epsilon}\right)}{\epsilon}\right)\right\rceil,7\right\}, \nonumber \\
        r&\ge 4\kappa_{BB^T}\alpha_A T, \nonumber \\
        M&\ge \max\left\{\frac{24\alpha_{\drm A/\drm t}}{\alpha_A^2},\frac{2\left(2\alpha_A + \frac{\alpha_{\drm BB^T/\drm t}}{\sigma^2}\right)}{\alpha_A}\right\} \times \frac{64\pi \left(\lg\left(4\kappa_{BB^T}\log\left(\frac{8}{\epsilon}\right)\right)+1\right)^2\log\left(\frac{8}{\epsilon}\right)}{\epsilon},
    \label{eq:KrM3}
    \end{align}
    we have access to an oracle $U_{\sqrt{\tilde{\Sigma}}}$ that acts as
    \begin{align}
        U_{\sqrt{\tilde{\Sigma}}} \ket{n} \ket{\phi} = \ket{n} V_{S_n} \ket{\phi}
        \label{eq:UTildeSigSqrt}
    \end{align}
    for any $n\in[r-1]_0$ and any state vector $\ket{\phi}$ on the space on which $V_{\sqrt{\Sigma_{n}}}$ acts.
    Here, $V_{S_n}$ is a $(2\sqrt{\sigma^2 \Delta t}, O(K(a_A+\log M+\log K)+a_{BB^T}+\log\log(1/\epsilon)), \epsilon\sqrt{\sigma^2 \Delta t})$-block-encoding of $S_n$.
    $U_{\sqrt{\tilde{\Sigma}}}$ makes
    \begin{align}
        O\left(K\kappa_{BB^T} \log^2\left(\frac{\kappa_{BB^T}}{\epsilon}\right)\right)
        \label{eq:CompUSigmaSqrt1}
    \end{align}
    queries to the controlled $U_A$, 
    \begin{align}
        O\left(\kappa_{BB^T} \log^2\left(\frac{\kappa_{BB^T}}{\epsilon}\right)\right)
        \label{eq:CompUSigmaSqrtBB}
    \end{align}
    query to $U_{BB^T}$, and additionally
    \begin{align}
        O\left(K(\log K\log M+a_A+a_{BB^T})\kappa_{BB^T} \log^2\left(\frac{\kappa_{BB^T}}{\epsilon}\right)\right)
        \label{eq:CompUSigmaSqrt2}
    \end{align}
    uses of elementary gates.
    \label{lem:BETildeSigSqrt}
\end{lemma}

\begin{proof}
    On the one hand, Lemma \ref{lem:kappaSig} implies that for $r$ in Eq.~\eqref{eq:KrM3}, $\Sigma_n$'s eigenvalues are in $[\sigma^2 \Delta t/2\kappa_{BB^T},\sigma^2 \Delta t]$.
    On the other hand, as stated in Lemma \ref{lem:BETildeSig}, $V_{\tilde{\Sigma}^{K,r,M}_n}$ in Eq.~\eqref{eq:UTildeSig} is an $(e^2\sigma^2 \Delta t, O(K(a_A+\log M+\log K)+a_{BB^T}),0)$-block-encoding of $\tilde{\Sigma}^{K,r,M}_n$.
    In light of Lemma \ref{lem:SigmaErr}, this is also an $(e^2\sigma^2 \Delta t, O(K(a_A+\log M+\log K)+a_{BB^T}),\tilde{\epsilon}\sigma^2 \Delta t)$-block-encoding of $\Sigma_n$ for $K$, $r$, and $M$ in Eq.~\eqref{eq:KrM2}, where
    \begin{align}
        \tilde{\epsilon} \coloneqq \frac{\epsilon}{64\pi \left(\lg\left(4\kappa_{BB^T}\log\left(\frac{8}{\epsilon}\right)\right)+1\right)^2\log\left(\frac{8}{\epsilon}\right)}.
    \end{align}
    By simple algebra, we see that $\tilde{\epsilon}\sigma^2 \Delta t$ is smaller than $\delta(\epsilon\sigma^2\Delta t, 2\kappa_{BB^T},\sigma^2 \Delta t)$.
    Then, we see from Theorem~\ref{th:BESqrt} that we can construct $U_{\sqrt{\tilde{\Sigma}}}$ acting as Eq.~\eqref{eq:UTildeSigSqrt} with $V_{\sqrt{\tilde{\Sigma}^{K,r,M}_n}}$ being a $(2\sqrt{\sigma^2 \Delta t}, a_{\tilde{\Sigma}}+O(\log\log(1/\epsilon\sigma^2\Delta t)), \epsilon\sqrt{\sigma^2\Delta t})$-block-encoding of $S_n$, where $a_{\tilde{\Sigma}}=O(K(a_A+\log M+\log K)+a_{BB^T})$ is the ancillary qubit number of $U_{\tilde{\Sigma}}$, making
    \begin{align}
        O\left(\kappa_{BB^T} \log^2\left(\frac{\kappa_{BB^T}}{\epsilon}\right)\right)
    \end{align}
    queries to $U_{\tilde{\Sigma}}$ and additionally
    \begin{align}
        O\left(a_{\tilde{\Sigma}}\kappa_{BB^T} \log^2\left(\frac{\kappa_{BB^T}}{\epsilon}\right)\right)
    \end{align}
    uses of elementary gates.
    Recalling the complexity of $U_{\tilde{\Sigma}}$ stated in Lemma~\ref{lem:BETildeSig}, we obtain the complexity evaluations in Eqs.~\eqref{eq:CompUSigmaSqrt1}, \eqref{eq:CompUSigmaSqrtBB}, and \eqref{eq:CompUSigmaSqrt2}.
\end{proof}

\subsection{Preparation of the state encoding the noise}

We then consider how to prepare the quantum state that approximately encodes the noise $\boldsymbol{\Delta}_n$ in the amplitudes.
For this, we first consider the preparation of the state encoding pseudorandom numbers in the amplitudes using $U_{\rm Rand}$ in Assumption \ref{ass:PRN}.
Note that normal (pseudo)random numbers are originally unbounded, so we clip them for amplitude encoding.

\begin{lemma}
    Let $U_{\rm SN}>0$.
    Under Assumption \ref{ass:PRN}, we have access to the oracle $U_{\rm RV}$ that acts on a three-register system as
    \begin{align}
        U_{\rm RV}\ket{0}\ket{i}\ket{0} = \ket{0}\ket{i}\frac{1}{\sqrt{m}U_{\rm SN}}\sum_{j=1}^m F^{\rm Bd}_{U_{\rm SN}}(z_{(i-1)m+j})\ket{j}+\ket{0_\perp}
        \label{eq:URV}
    \end{align}
    for any $i\in\mathbb{N}$, where $\ket{0_\perp}$ is an unnormalized state such that $(\ket{0}\!\bra{0} \otimes I \otimes I)\ket{0_\perp}=0$.
    $U_{\rm RV}$ uses $U_{\rm Rand}$ once and $O(1)$ arithmetic circuits.
\end{lemma}

\begin{proof}
    Adding three ancillary registers to the three registers in Eq.~\eqref{eq:URV}, we can perform the following operation, where some other ancillary registers are not displayed:
    \begin{align}
        & \ket{0}\ket{i}\ket{0}\ket{0}\ket{0}\ket{0} \nonumber \\
        \rightarrow & \frac{1}{\sqrt{m}}\sum_{j=1}^m\ket{0}\ket{i}\ket{j}\ket{0}\ket{0}\ket{0} \nonumber \\
        \rightarrow & \frac{1}{\sqrt{m}}\sum_{j=1}^m\ket{0}\ket{i}\ket{j}\ket{m(i-1)+j}\ket{0}\ket{0} \nonumber \\
        \rightarrow & \frac{1}{\sqrt{m}}\sum_{j=1}^m\ket{0}\ket{i}\ket{j}\ket{m(i-1)+j}\ket{z_{m(i-1)+j}}\ket{0} \nonumber \\
        \rightarrow & \frac{1}{\sqrt{m}}\sum_{j=1}^m\ket{0}\ket{i}\ket{j}\ket{m(i-1)+j}\ket{z_{m(i-1)+j}}\Ket{\frac{F^{\rm Bd}_{U_{\rm SN}}(z_{m(i-1)+j})}{U_{\rm SN}}} \nonumber \\
        \rightarrow & \frac{1}{\sqrt{m}}\sum_{j=1}^m\left(\frac{F^{\rm Bd}_{U_{\rm SN}}(z_{m(i-1)+j})}{U_{\rm SN}}\ket{0}+\sqrt{1-\left(\frac{F^{\rm Bd}_{U_{\rm SN}}(z_{m(i-1)+j})}{U_{\rm SN}}\right)^2}\ket{1}\right)\ket{i}\ket{j}\ket{m(i-1)+j}\ket{z_{m(i-1)+j}}\Ket{\frac{F^{\rm Bd}_{U_{\rm SN}}(z_{m(i-1)+j})}{U_{\rm SN}}} \nonumber \\
        \rightarrow & \frac{1}{\sqrt{m}}\sum_{j=1}^m\left(\frac{F^{\rm Bd}_{U_{\rm SN}}(z_{m(i-1)+j})}{U_{\rm SN}}\ket{0}+\sqrt{1-\left(\frac{F^{\rm Bd}_{U_{\rm SN}}(z_{m(i-1)+j})}{U_{\rm SN}}\right)^2}\ket{1}\right)\ket{i}\ket{j}\ket{0}\ket{0}\ket{0}.
        \label{eq:URVImpl}
    \end{align}
    The operations corresponding to the arrows are as follows.
    The first arrow is by $U^{\rm SP}_{1,m}$.
    At the second arrow, we use the circuits for addition and multiplication.
    At the third arrow, we use $U_{\rm Rand}$.
    At the fourth arrow, we use the circuits for max, min, and division.
    At the fifth arrow, we use $U_{\rm B2A}$ with the last and first registers being the control register and the target qubit, respectively.
    The final arrow is uncomputation.
    The final state in Eq.~\eqref{eq:URVImpl} is the state in the right-hand side of Eq.~\eqref{eq:URV}.
\end{proof}

Recall that $m=N$ now holds.
$U_{\rm RV}$ with $m<N$ will be used in the \ac{EM}-based method presented in Sec. \ref{sec:EM}.

Then, combining $U_{\rm RV}$ and the block-encoding of $S_n$, we can construct the circuit to approximately encode $\boldsymbol{\Delta}_n$.

\begin{lemma}
    Let $\epsilon>0$.
    Suppose that Assumptions \ref{ass:PRN}, \ref{ass:UA}, and \ref{ass:UBB} hold.
    Let $N_{\rm s}\in\mathbb{N}$ and $U_{\rm SN}>0$, and suppose that $z_1,\ldots,z_{rNN_{\rm s}}$ generated by $U_{\rm Rand}$ satisfy $|z_1|,\ldots,|z_{rNN_{\rm s}}|\le U_{\rm SN}$.
    Then, for $K$, $r$, and $M$ satisfying Eq.~\eqref{eq:KrM3}, we have access to an oracle $U_{\tilde{\boldsymbol{\Delta}}}$ that acts as
    \begin{align}
        U_{\tilde{\boldsymbol{\Delta}}}\ket{0}\ket{i}\ket{n}\ket{0}= \ket{0}\ket{i}\ket{n}\ket{\tilde{\boldsymbol{\Delta}}^{(i)}_{n}} +\ket{0_\perp}
        \label{eq:UTilDelta}
    \end{align}
    for any $i\in[N_{\rm s}]$ and any $n\in[r-1]_0$, where $\ket{0_\perp}$ is an unnormalized state such that $(\ket{0}\!\bra{0} \otimes I \otimes I \otimes I)\ket{0_\perp}=0$, and the unnormalized state $\ket{\tilde{\boldsymbol{\Delta}}^{(i)}_{n}}$ is given by 
    \begin{align}
        \ket{\tilde{\boldsymbol{\Delta}}^{(i)}_n} \coloneqq \frac{1}{2\sqrt{N\sigma^2\Delta t}U_{\rm SN}}\sum_{j=1}^N \tilde{\Delta}^{(i)}_{n,j} \ket{j}
    \end{align}
    with
    \begin{align}
        \tilde{\boldsymbol{\Delta}}^{(i)}_n =
        \begin{pmatrix}
            \tilde{\Delta}^{(i)}_{n,1} \\ \vdots \\ \tilde{\Delta}^{(i)}_{n,N}
        \end{pmatrix}
        \coloneqq \tilde{S}^{K,r,M}_n\mathbf{z}^{(i)}_n, ~
        \mathbf{z}^{(i)}_n \coloneqq
        \begin{pmatrix}
            z_{(i-1)rN+nN+1} \\ \vdots \\ z_{(i-1)rN+(n+1)N}
        \end{pmatrix}
        ,
    \end{align}
    where $\tilde{S}^{K,r,M}_n\in\mathbb{R}^{N \times N}$ satisfies $\left\|\tilde{S}^{K,r,M}_n-S_n\right\|\le \sqrt{\sigma^2\Delta t}\epsilon$.
    $U_{\tilde{\boldsymbol{\Delta}}}$ uses the controlled $U_A$ a number of times given by Eq.~\eqref{eq:CompUSigmaSqrt1}, $U_{BB^T}$ a number of times given by Eq.~\eqref{eq:CompUSigmaSqrtBB}, $U_{\rm Rand}$ once, and additionally elementary gates a number of times given by Eq.~\eqref{eq:CompUSigmaSqrt2}.
    \label{lem:Delta}
\end{lemma}

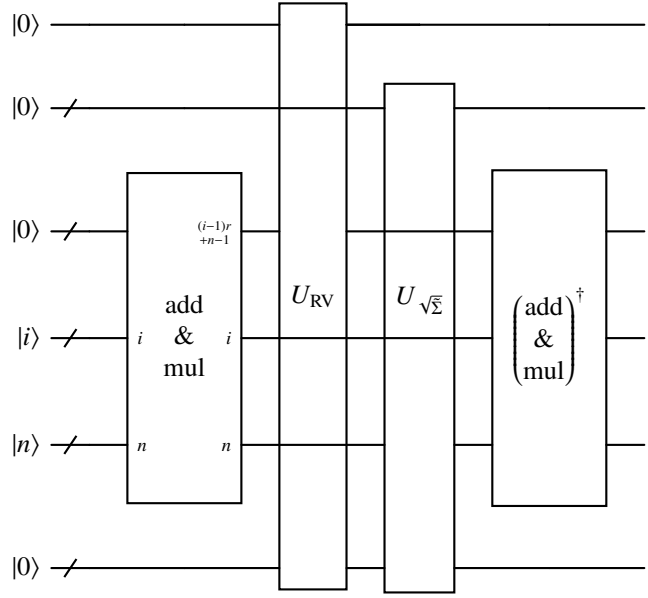
\begin{figure}[t]
\begin{center}

\begin{quantikz}[transparent]
\lstick{$\ket{0}$} & & & \gate[6]{U_{\rm RV}} & \qw  & & \\
\lstick{$\ket{0}$} & \qwbundle{} & & \linethrough & \gate[5,label style={yshift=0.5cm}]{U_{\sqrt{\tilde{\Sigma}}}} & & \\
\lstick{$\ket{0}$} & \qwbundle{} & \gate[3][1.5cm]{\begin{matrix}\text{add} \\ \text{\&} \\ \text{mul}\end{matrix}} \gateoutput{$\substack{(i-1)r \\ +n-1}$} & & \linethrough & \gate[3][1.5cm]{\begin{pmatrix}\text{add} \\ \text{\&} \\ \text{mul}\end{pmatrix}^{\dagger}} & \\
\lstick{$\ket{i}$} & \qwbundle{} & \gateinput{$i$}\gateoutput{$i$} & \linethrough & \linethrough & & \\ 
\lstick{$\ket{n}$} & \qwbundle{} & \gateinput{$n$}\gateoutput{$n$} & \linethrough & & &\\
\lstick{$\ket{0}$} & \qwbundle{} & & & & &
\end{quantikz}
\caption{The quantum circuit implementation of $U_{\tilde{\boldsymbol{\Delta}}}$.}
\label{fig:UTilDelta}
\end{center}
\end{figure}

\begin{proof}
We prepare a system with six registers initialized to $\ket{0}\ket{0}\ket{0}\ket{i}\ket{n}\ket{0}$ and operate the circuit in Fig.~\ref{fig:UTilDelta}, which transforms the quantum state as follows:
\begin{align}
    & \ket{0}\ket{0}\ket{0}\ket{i}\ket{n}\ket{0} \nonumber \\
    & \xrightarrow{\substack{\text{addition \&} \\ \text{multiplication}}} \ket{0}\ket{0}\ket{(i-1)r+n}\ket{i}\ket{n}\ket{0} \nonumber \\
     & \xrightarrow{U_{\rm RV}} \ket{0}\ket{0}\ket{(i-1)r+n}\ket{i}\ket{n}\frac{1}{\sqrt{N}U_{\rm SN}}\sum_{j=1}^N z_{(i-1)rN+nN+j}\ket{j} + \ket{0_\perp} \nonumber \\
     & \xrightarrow{U_{\sqrt{\tilde{\Sigma}}}} \ket{0}\ket{0}\ket{(i-1)r+n}\ket{i}\ket{n}\frac{1}{\sqrt{N}U_{\rm SN}}\sum_{j=1}^N \left(\sum_{k=1}^N \frac{1}{2\sqrt{\sigma^2 \Delta t}} \left(\tilde{S}^{K,r,M}_n\right)_{jk} z_{(i-1)rN+nN+k}\right)\ket{j} +\ket{0_\perp} \nonumber \\
     & \xrightarrow{\text{uncompute}} \ket{0}\ket{0}\ket{0}\ket{i}\ket{n}\frac{1}{\sqrt{N}U_{\rm SN}}\sum_{j=1}^N \left(\sum_{k=1}^N \frac{1}{2\sqrt{\sigma^2 \Delta t}} \left(\tilde{S}^{K,r,M}_n\right)_{jk} z_{(i-1)rN+nN+k}\right)\ket{j} +\ket{0_\perp},
     \label{eq:UtilDelOp}
\end{align}
where $\ket{0_\perp}$ is some unnormalized state, which may differ between different lines and satisfies $(\ket{0}\!\bra{0} \otimes \ket{0}\!\bra{0} \otimes \ket{0}\!\bra{0} \otimes I \otimes I \otimes I)\ket{0_\perp}=0$.
Note that $F^{\rm Bd}_{U_{\rm SN}}(z_l)=z_l$ for every $l\in[rNN_{\rm s}]$ because of the assumption that $|z_l|\le U_{\rm SN}$, and that $U_{\rm RV}$ thus generates the state after the second arrow in Eq.~\eqref{eq:UtilDelOp}.
Regarding the first three registers as one register, we note that the final state in Eq.~\eqref{eq:UtilDelOp} is in the form of Eq.~\eqref{eq:UTilDelta}.
In the operation in Eq. \eqref{eq:UtilDelOp}, we use $U_{\rm RV}$ and $U_{\sqrt{\tilde{\Sigma}}}$ once each and $O(1)$ arithmetic circuits, from which we obtain the query complexity evaluations in the statement, considering the complexity of $U_{\sqrt{\tilde{\Sigma}}}$ given in Lemma \ref{lem:BETildeSigSqrt}.
\end{proof}

\subsection{Preparation of the history state}

We are now ready to consider how to prepare the history state, which encodes $\mathbf{X}_{\rm hist}$, the solution of Eq. \eqref{eq:AXhistB}, in the amplitudes.
We begin with showing the following theorem on preparing such a state, given a block-encoding of $\mathcal{A}$ and a preparation circuit for $\mathbf{B}$. 

\begin{theorem}
    Suppose that Assumption \ref{ass:muA} holds.
    Let $\epsilon>0$.
    Suppose that we have access to the oracles $V_{\tilde{\mathcal{A}}}$ and $V_{\tilde{\mathcal{B}}^{(1)}}$, which satisfies the following with some $U_{\tilde{\mathcal{B}}}\in\mathbb{R}_+$:
    \begin{itemize}
        \item $V_{\tilde{\mathcal{A}}}$: a $(\alpha_{\tilde{\mathcal{A}}},a_{\tilde{\mathcal{A}}},0)$-block-encoding of $\tilde{\mathcal{A}}$, where $\alpha_{\tilde{\mathcal{A}}} \ge \|\tilde{\mathcal{A}}\|$, and $a_{\tilde{\mathcal{A}}}\in\mathbb{N}$. 
        Here, $\tilde{\mathcal{A}}$ is a matrix written by
        \begin{align}
            \tilde{\mathcal{A}} =
            \begin{pmatrix}
            I & & & \\
            -\hat{\Phi}_0 & I & & \\
            & \ddots & \ddots & \\
            & & -\hat{\Phi}_{r-1} & I
            \end{pmatrix}
            \label{eq:calAtil}
        \end{align}
        with matrices $\hat{\Phi}_n,n\in[r-1]_0$ satisfying
        \begin{align}
            \left\|\hat{\Phi}_n - \Phi_n\right\| \le \min\left\{\frac{1}{2}\eta\Delta t e^{-\eta\Delta t},\epsilon_1\right\},~
            \epsilon_1 \coloneqq \frac{\eta^2 \Delta t^2 U_{\tilde{\mathcal{B}}} \epsilon}{32\sqrt{6}}\sqrt{\frac{1}{\|\mathbf{x}_0\|^2\eta\Delta t+\frac{T}{\Delta t}\left(\Delta_{\rm max}^{(1)}\right)^2}},~
            \Delta_{\rm max}^{(1)}\coloneqq\max_{n\in[r-1]_0} \|\boldsymbol{\Delta}^{(1)}_n\|.
            \label{eq:RnPhinDiff}
        \end{align}

        \item $V_{\tilde{\mathcal{B}}^{(1)}}$, which acts on a two-register system as
        \begin{align}
            V_{\tilde{\mathcal{B}}^{(1)}}\ket{0}\ket{0}=\ket{0}\ket{\tilde{\mathcal{B}}^{(1)}} + \ket{0_\perp}.
            \label{eq:VtilB}
        \end{align}
        Here, $\ket{0_\perp}$ is an unnormalized state such that $(\ket{0}\!\bra{0} \otimes I)\ket{0_\perp}=0$, and an unnormalized state $\ket{\tilde{\mathcal{B}}^{(1)}}$ is given by
        \begin{align}
            \ket{\tilde{\mathcal{B}}^{(1)}} \coloneqq \frac{1}{U_{\tilde{\mathcal{B}}}}\tilde{\mathcal{B}}^{(1)},
            \tilde{\mathcal{B}}^{(1)} \coloneqq 
            \begin{pmatrix}
            \mathbf{x}_0 \\ \tilde{\boldsymbol{\Delta}}^{(1)}_0 \\ \vdots \\ \tilde{\boldsymbol{\Delta}}^{(1)}_{r-1} 
            \end{pmatrix}
            \label{eq:ketBTil}
        \end{align}
        with $\tilde{\boldsymbol{\Delta}}^{(1)}_n\in\mathbb{R}^N,n\in[r-1]_0$ satisfying
        \begin{align}
            \left\|\tilde{\boldsymbol{\Delta}}^{(1)}_n - \boldsymbol{\Delta}^{(1)}_n\right\| \le \epsilon_2,~
            \boldsymbol{\Delta}^{(1)}_n \coloneqq S_n \mathbf{z}^{(1)}_n,~
            \epsilon_2 \coloneqq \sqrt{\frac{\Delta t^3}{384T}}\eta U_{\tilde{\mathcal{B}}} \epsilon.
            \label{eq:condTilDel}
        \end{align}
    \end{itemize}
    Then, we have access to an oracle $V^{(1)}_{{\rm hist},\epsilon}$, which acts on a two-register system as
    \begin{align}
        V^{(1)}_{{\rm hist},\epsilon}\ket{0}\ket{0} =   \ket{0}\ket{\tilde{\tilde{\mathbf{X}}}_{\rm hist}} + \ket{0_\perp}.
        \label{eq:V1hist}
    \end{align}
    Here, $\ket{\tilde{\tilde{\mathbf{X}}}_{\rm hist}}$ is an unnormalized state given by
    \begin{align}
        \ket{\tilde{\tilde{\mathbf{X}}}_{\rm hist}} \coloneqq \frac{\alpha_{\tilde{\mathcal{A}}}}{2\kappa_{\tilde{\mathcal{A}}}U_{\tilde{\mathcal{B}}}}\tilde{\tilde{\mathbf{X}}}_{\rm hist}
        \label{eq:XHistKet}
    \end{align}
    with $\kappa_{\tilde{\mathcal{A}}} \coloneqq \frac{4r \alpha_{\tilde{\mathcal{A}}}}{\max\{1,\eta T\}}$ and $\tilde{\tilde{\mathbf{X}}}_{\rm hist}\in\mathbb{R}^{(r+1)N}$ satisfying
    \begin{align}
        \left\| \tilde{\tilde{\mathbf{X}}}_{\rm hist}^{(1)} - 
        \mathbf{X}_{\rm hist}^{(1)}\right\| \le U_{\tilde{\mathcal{B}}} \epsilon,~
        \mathbf{X}_{\rm hist}^{(1)} \coloneqq
    \begin{pmatrix}
        \mathbf{X}_0^{(1)} \\ \vdots \\ \mathbf{X}_r^{(1)}, 
    \end{pmatrix}
        \label{eq:XHistErr}
    \end{align}
    where $\mathbf{X}_0^{(1)},\ldots,\mathbf{X}_r^{(1)}$ are defined by Eq.~\eqref{eq:XnRec} with $\boldsymbol{\Delta}_0=\boldsymbol{\Delta}^{(1)}_0,\ldots,\boldsymbol{\Delta}_{r-1}=\boldsymbol{\Delta}^{(1)}_{r-1}$.
    $V^{(1)}_{{\rm hist},\epsilon}$ uses $V_{\tilde{\mathcal{A}}}$ 
    \begin{align}
        O\left(\frac{r\alpha_{\tilde{\mathcal{A}}}}{\max\{1,\eta T\}}\log\left(\frac{r}{\max\{1,\eta T\} \epsilon}\right)\right)
        \label{eq:CompHistGen}
    \end{align}
    times, $V_{\tilde{\mathcal{B}}^{(1)}}$ once, and
    \begin{align}
        O\left(\frac{a_{\tilde{\mathcal{A}}} r \alpha_{\tilde{\mathcal{A}}}}{\max\{1,\eta T\}}\log\left(\frac{r}{\max\{1,\eta T\} \epsilon}\right)\right)
        \label{eq:CompHistGen2}
    \end{align}
    additional elementary gates.
    \label{th:HistState}
\end{theorem}

\begin{proof}
    The proof of the statement on the accuracy is similar to that of Theorem 8 in Ref. \cite{an2026fast}.
    We define
    \begin{align}
        \tilde{\mathbf{X}}^{(1)}_{n+1} \coloneqq \hat{\Phi}_n \tilde{\mathbf{X}}^{(1)}_n + \tilde{\boldsymbol{\Delta}}^{(1)}_n
        \label{eq:Xtil}
    \end{align}
    for $n\in[r-1]_0$ with $\tilde{\mathbf{X}}^{(1)}_0=\mathbf{x}_0$ and
    \begin{align}
        \tilde{\mathbf{X}}^{(1)}_{\rm hist} \coloneqq \begin{pmatrix}\tilde{\mathbf{X}}^{(1)}_0 \\ \vdots \\ \tilde{\mathbf{X}}^{(1)}_r \end{pmatrix}.
    \end{align}
    We see that this can be written by
    \begin{align}
        \tilde{\mathbf{X}}^{(1)}_{\rm hist} = \tilde{\mathcal{A}}^{-1} \tilde{\mathcal{B}}^{(1)}.
        \label{eq:tilXhist}
    \end{align}
    Besides, similarly to the proof of Theorem 8 in Ref. \cite{an2026fast}, we have
    \begin{align}
        \left\|\tilde{\mathbf{X}}^{(1)}_{\rm hist} - \mathbf{X}^{(1)}_{\rm hist}\right\| \le \sqrt{3}\sqrt{\frac{2\epsilon_1^2 \|\mathbf{x}_0\|^2}{(1-e^{-\eta\Delta t})^3} + \frac{\epsilon_1^2 T \left(\Delta_{\rm max}^{(1)}\right)^2}{(1-e^{-\eta\Delta t/2})^4 \Delta t} + \frac{\epsilon_2^2T}{(1-e^{-\eta\Delta t/2})^2\Delta t}},
    \end{align} 
    which becomes
    \begin{align}
        \left\|\tilde{\mathbf{X}}^{(1)}_{\rm hist} - \mathbf{X}^{(1)}_{\rm hist}\right\| \le \sqrt{\frac{48\epsilon_1^2 \|\mathbf{x}_0\|^2}{\eta^3\Delta t^3} + \frac{768 \epsilon_1^2 T \left(\Delta_{\rm max}^{(1)}\right)^2}{\eta^4 \Delta t^5} + \frac{48\epsilon_2^2}{\eta^2 \Delta t^3}}
    \end{align}
    by using Eq. \eqref{eq:1-e-xLB}.
    This further translates into
    \begin{align}
        \left\|\tilde{\mathbf{X}}^{(1)}_{\rm hist} - \mathbf{X}^{(1)}_{\rm hist}\right\| \le \frac{U_{\tilde{\mathcal{B}}} \epsilon}{2}
        \label{eq:XtilHistErr}
    \end{align}
    by the definitions of $\epsilon_1$ and $\epsilon_2$.
    We also note $\|\tilde{\mathcal{A}}\| \le \alpha_{\tilde{\mathcal{A}}}$, which follows from the definition of block-encoding, and
    \begin{align}
         \|\tilde{\mathcal{A}}^{-1}\| \le \frac{4r}{\max\{1,\eta T\}},
        \label{eq:kappatilA}
    \end{align}
    which is shown at the end of this proof.
    These mean that the singular values of $\tilde{\mathcal{A}}$ lie in $\left[\max\{1,\eta T\}/4r,\alpha_{\tilde{\mathcal{A}}}\right]$, and thus, defining $\epsilon_3\coloneqq\frac{\alpha_{\tilde{\mathcal{A}}}}{4\kappa_{\tilde{\mathcal{A}}}}\epsilon$, we see from Theorem~\ref{eq:MatInvQSVT} that we can construct an $(1,a_{\tilde{\mathcal{A}}},\epsilon_3)$-block-encoding $V_{\tilde{\mathcal{A}}^{-1}}$ of $\frac{\alpha_{\tilde{\mathcal{A}}}}{2\kappa_{\tilde{\mathcal{A}}}}\tilde{\mathcal{A}}^{-1}$ by using $V_{\tilde{\mathcal{A}}}$ $O\left(\kappa_{\tilde{\mathcal{A}}}\log\left(\frac{1}{\epsilon_3}\right)\right)$ times, which is equivalent to Eq.~\eqref{eq:CompHistGen}, and additional elementary gates $O\left(a_{\tilde{\mathcal{A}}}\kappa_{\tilde{\mathcal{A}}}\log\left(\frac{1}{\epsilon_3}\right)\right)$ times, which is equivalent to Eq.~\eqref{eq:CompHistGen2}.
    Then, on a system consisting of two registers, where the first one contains the ancillary qubits for $V_{\tilde{\mathcal{A}}^{-1}}$ and those corresponding to the first register in Eq.~\eqref{eq:VtilB}, we can operate $V_{\tilde{\mathcal{B}}^{(1)}}$ and $V_{\tilde{\mathcal{A}}^{-1}}$ in this order to obtain the state vector written by
    \begin{align}
        \ket{0}\frac{\alpha_{\tilde{\mathcal{A}}}}{2\kappa_{\tilde{\mathcal{A}}}}\tilde{\mathcal{C}}\ket{\tilde{\mathcal{B}}^{(1)}} + \ket{0_\perp},
        \label{eq:ResState}
    \end{align}
    where the matrix $\tilde{\mathcal{C}}$ satisfies
    \begin{align}
        \left\|\frac{\alpha_{\tilde{\mathcal{A}}}}{2\kappa_{\tilde{\mathcal{A}}}}\tilde{\mathcal{C}} - \frac{\alpha_{\tilde{\mathcal{A}}}}{2\kappa_{\tilde{\mathcal{A}}}}\tilde{\mathcal{A}}^{-1}\right\| \le \epsilon_3.
        \label{eq:MatInvErr}
    \end{align}
    We have
    \begin{align}
         \left\|U_{\tilde{\mathcal{B}}}\tilde{\mathcal{C}}\ket{\tilde{\mathcal{B}}^{(1)}} - \mathbf{X}_{\rm hist}^{(1)}\right\| 
        \le & \left\|U_{\tilde{\mathcal{B}}}\tilde{\mathcal{C}}\ket{\tilde{\mathcal{B}}^{(1)}} - U_{\tilde{\mathcal{B}}}\tilde{\mathcal{A}}^{-1}\ket{\tilde{\mathcal{B}}^{(1)}}\right\|+\left\|U_{\tilde{\mathcal{B}}}\tilde{\mathcal{A}}^{-1}\ket{\tilde{\mathcal{B}}^{(1)}}-\mathbf{X}_{\rm hist}^{(1)}\right\| \nonumber \\
        \le & \frac{2\kappa_{\tilde{\mathcal{A}}}U_{\tilde{\mathcal{B}}}}{\alpha_{\tilde{\mathcal{A}}}} \times \left\|\frac{\alpha_{\tilde{\mathcal{A}}}}{2\kappa_{\tilde{\mathcal{A}}}}\tilde{\mathcal{C}} - \frac{\alpha_{\tilde{\mathcal{A}}}}{2\kappa_{\tilde{\mathcal{A}}}}\tilde{\mathcal{A}}^{-1}\right\| \times \left\|\ket{\tilde{\mathcal{B}}^{(1)}}\right\| + \left\|\tilde{\mathcal{A}}^{-1} \tilde{\mathcal{B}}^{(1)}-\mathbf{X}_{\rm hist}^{(1)}\right\| \nonumber \\
        \le & \frac{2\kappa_{\tilde{\mathcal{A}}}U_{\tilde{\mathcal{B}}}}{\alpha_{\tilde{\mathcal{A}}}} \times \epsilon_3 + \frac{U_{\tilde{\mathcal{B}}}\epsilon}{2} \nonumber \\
        \le & U_{\tilde{\mathcal{B}}}\epsilon,
        \label{eq:UBCErr}
    \end{align}
    where we use Eqs. \eqref{eq:tilXhist}, \eqref{eq:XtilHistErr}, and \eqref{eq:MatInvErr}, and $\left\|\ket{\tilde{\mathcal{B}}^{(1)}}\right\|\le1$ at the third arrow.
    Therefore, the state in Eq.~\eqref{eq:ResState} is in the form of Eq.~\eqref{eq:V1hist}.
    Thus, $V^{(1)}_{{\rm hist},\epsilon}$ is constructed by using $V_{\tilde{\mathcal{B}}^{(1)}}$ and $V_{\tilde{\mathcal{A}}^{-1}}$ once each, whose total complexity in terms of the number of queries to $V_{\tilde{\mathcal{A}}}$ and $V_{\tilde{\mathcal{B}}^{(1)}}$ is as stated.

    Lastly, let us prove Eq.~\eqref{eq:kappatilA}.
    For later use, we now see that
    \begin{align}
        \left\|\tilde{\mathcal{A}}_+^{-1}\right\| \le \frac{4r}{\max\{1,\eta T\}} + R
        \label{eq:calAtil+Norm}
    \end{align}
    holds for
    \begin{align}
        \tilde{\mathcal{A}}_+ \coloneqq
            \begin{tikzpicture}[baseline=(m.center)]
            \node (m) {
            $
            \begin{pmatrix}
            I & & & & & & & \\
            -\hat{\Phi}_0 & I & & & & & & \\
            & \ddots & \ddots & & & & & \\
            & & -\hat{\Phi}_{r-1} & I & & & & \\
            & & & -I & I & & & \\
            & & & & \ddots & \ddots & & \\
            & & & & & -I & I & 
            \end{pmatrix}
            $
            };
            \draw[decorate,decoration={brace,amplitude=5pt}]
              (1.2,-0.1) -- (2.7,-1.4)
              node[midway,right=3pt,yshift=6pt] {$\text{$R$}$};
            \end{tikzpicture}
        \label{eq:calAtil+}
    \end{align}
    under the condition $\left\|\hat{\Phi}_n - \Phi_n\right\| \le \frac{1}{2}\eta\Delta t e^{-\eta\Delta t}$ and get Eq.~\eqref{eq:kappatilA} as a special case where $R=0$.
    If we write $\tilde{\mathcal{A}}_+^{-1}$ in the block form as Eq.~\eqref{eq:calAInv}, the $(n,n^\prime)$-th block is given by
    \begin{align}
        (\tilde{\mathcal{A}}_+^{-1})_{n,n^\prime} =
        \begin{cases}
            I &  ; ~ (n = n^\prime) \lor (r \le n^\prime < n)\\
            \hat{\Phi}_{n-1} \cdots \hat{\Phi}_{n^\prime} & ; ~ n^\prime < n \le r \\
            \hat{\Phi}_{r-1} \cdots \hat{\Phi}_{n^\prime} & ; ~ n^\prime < r< n \\
            O & ; ~ n < n^\prime
        \end{cases}
        .
        \label{eq:calATilInvBlock}
    \end{align}
    We also have
    \begin{align}
       \|\hat{\Phi}_n\| \le  \|\Phi_n\| + \left\|\hat{\Phi}_n - \Phi_n\right\| \le e^{-\eta\Delta t} + \frac{1}{2}\eta\Delta t e^{-\eta\Delta t} \le e^{-\eta\Delta t/2}
    \end{align}
    for any $n\in[r-1]_0$, where we use Eqs.~\eqref{eq:PhiUBEta} and \eqref{eq:RnPhinDiff} along with $1+x \le e^x$, which holds for any $x\ge0$.
    We thus get
    \begin{align}
        \left\|(\tilde{\mathcal{A}}_+^{-1})_{n,n^\prime}\right\| =
        \begin{cases}
            1 &  ; ~ (n = n^\prime) \lor (r \le n^\prime < n)\\
            e^{-\eta(n-n^\prime)\Delta t/2} & ; ~ n^\prime < n \le r \\
            e^{-\eta(r-n^\prime)\Delta t/2} & ; ~ n^\prime < r <n \\
            0 & ; ~ n < n^\prime
        \end{cases}
        \label{eq:tilA+InvBlock}
    \end{align}
    similarly to Eq. \eqref{eq:calA+BlInv}.
    By using this in an inequality similar to Eq. \eqref{eq:calAInvNormTemp}, we reach
    \begin{align}
        \left\|\tilde{\mathcal{A}}_+^{-1}\right\| \le \sum_{k=0}^{r} e^{-\eta k \Delta t/2} + R, 
    \end{align}
    which leads to Eq.~\eqref{eq:calAtil+Norm} similarly to Eq.~\eqref{eq:calAInvNorm}.
    This completes the proof.
\end{proof}

Then, using this theorem for the oracles constructed above, the block-encoding of $\Phi_n$ and the preparation circuit for $\boldsymbol{\Delta}_n$, we can realize the preparation of the history state.

\begin{assum}
    We have access to an oracle $U_{\mathbf{x}_0}$ that acts as
    \begin{align}
        U_{\mathbf{x}_0}\ket{0} = \ket{\mathbf{x}_0}\coloneqq \sum_{i=1}^N \frac{x_0^i}{\|\mathbf{x}_0\|} \ket{i}.
    \end{align}
    \label{ass:Ux0}
\end{assum}

\begin{theorem}
    Suppose that Assumptions \ref{ass:PRN}, \ref{ass:UA}, \ref{ass:muA}, \ref{ass:UBB}, and \ref{ass:Ux0} hold.
    Let $\epsilon\in(0,1]$ and $U_{\rm SN}>0$.
    Define
    \begin{align}
        r_{\rm c}\coloneqq\left\lceil 4\kappa_{BB^T}\alpha_A T \right\rceil,
        \label{eq:rReq}
    \end{align}
    and suppose that $z_1,\ldots,z_{r_{\rm c}N}$ generated by $U_{\rm Rand}$ satisfy
    \begin{align}
        |z_1|,\ldots,|z_{r_{\rm c}N}|\le U_{\rm SN}.
        \label{eq:ziUB}
    \end{align}
    Then, we can construct an oracle $V^{(1)}_{{\rm hist},\epsilon}$ as described in Theorem \ref{th:HistState},
    with $\ket{\tilde{\tilde{\mathbf{X}}}_{\rm hist}^{(1)}}$ being
    \begin{align}
        \ket{\tilde{\tilde{\mathbf{X}}}_{\rm hist}^{(1)}} \coloneqq \frac{\max\{1,\eta T\}}{8r_{\rm c} U_{\tilde{\mathcal{B}}}}\tilde{\tilde{\mathbf{X}}}_{\rm hist}^{(1)},
        \label{eq:XHistKet2}
    \end{align}
    where $U_{\tilde{\mathcal{B}}}$ is now
    \begin{align}
        U_{\tilde{\mathcal{B}}} = \sqrt{\|\mathbf{x}_0\|^2+4N\sigma^2T U_{\rm SN}^2},
        \label{eq:UtilB}
    \end{align}
    and $\tilde{\tilde{\mathbf{X}}}_{\rm hist}^{(1)}\in\mathbb{R}^{(r+1)N}$ satisfies Eq.~\eqref{eq:XHistErr}.
    $V^{(1)}_{{\rm hist},\epsilon}$ uses the controlled $U_A$
    \begin{align}
        O\left(\frac{\kappa_{BB^T}\alpha_A T K_{\rm c}}{\max\{1,\eta T\}}\log\left(\frac{\kappa_{BB^T}\alpha_AT}{\max\{1,\eta T\}\epsilon}\right)\right)
    \end{align}
    times, the controlled $U_{BB^T}$
    \begin{align}
        O\left(\frac{\kappa_{BB^T}\alpha_A T}{\max\{1,\eta T\}}\log\left(\frac{\kappa_{BB^T}\alpha_AT}{\max\{1,\eta T\}\epsilon}\right)\right)
    \end{align}
    times, $U_{\mathbf{x}_0}$ and $U_{\rm Rand}$ once each, and
    \begin{align}
        O\left(\frac{\kappa_{BB^T}\alpha_A T K_{\rm c}}{\max\{1,\eta T\}} \left(\log K_{\rm c}\log M_{\rm c}+a_A+a_{BB^T}\right) \log\left(\frac{\kappa_{BB^T}\alpha_AT}{\max\{1,\eta T\}\epsilon}\right)\right)
    \end{align}
    additional elementary gates.
    Here,
    \begin{align}
        K_{\rm c}&\coloneqq\max\left\{\left\lceil\log\left(\frac{1152\pi\left(\lg\left(4\kappa_{BB^T}\log\left(\frac{8}{\varepsilon}\right)\right)+1\right)^2\log\left(\frac{8}{\varepsilon}\right)}{\varepsilon}\right)\right\rceil,7\right\}, \nonumber \\
        M_{\rm c}&\coloneqq \left\lceil\max\left\{\frac{24\alpha_{\drm A/\drm t}}{\alpha_A^2},\frac{2\left(2\alpha_A + \frac{\alpha_{\drm BB^T/\drm t}}{\sigma^2}\right)}{\alpha_A}\right\} \times \frac{64\pi \left(\lg\left(4\kappa_{BB^T}\log\left(\frac{8}{\varepsilon}\right)\right)+1\right)^2\log\left(\frac{8}{\varepsilon}\right)}{\varepsilon}\right\rceil, \nonumber \\
        \varepsilon & \coloneqq \min\left\{\frac{\eta^2T^2 \epsilon}{32\sqrt{6}r_{\rm c}^2},\sqrt{\frac{\|\mathbf{x}_0\|^2+4N\sigma^2TU_{\rm SN}^2}{384N\sigma^2TU_{\rm SN}^2}}\frac{\eta T}{r_{\rm c}}\epsilon\right\}.
    \label{eq:KrMComb}
    \end{align}
    \label{th:HistStateConc}
\end{theorem}

\begin{proof}
    We set $\hat{\Phi}_n$ to the truncated Dyson series $\tilde{\Phi}^{K,r,M}_{n,0}$, which means that $\tilde{\mathcal{A}}$ in Eq.~\eqref{eq:calAtil} is now
    \begin{align}
    \tilde{\mathcal{A}} =
    \begin{pmatrix}
    I & &  \\
    -\tilde{\Phi}^{K,r,M}_{0,0} & I & & \\
    & \ddots & \ddots & \\
    & & -\tilde{\Phi}^{K,r,M}_{r-1,0} & I
    \end{pmatrix},
    \label{eq:calAtilConc}
    \end{align}
    and $\tilde{\boldsymbol{\Delta}}^{(1)}_n$ to $\tilde{\boldsymbol{\Delta}}^{(i)}_n=\tilde{S}^{K,r,M}_n\mathbf{z}^{(i)}_n$ in Lemma \ref{lem:Delta} with $i=1$.
    In fact, aside from the conditions \eqref{eq:RnPhinDiff} and \eqref{eq:condTilDel}, we can implement $V_{\tilde{\mathcal{A}}}$ and $V_{\tilde{\mathcal{B}}^{(1)}}$ in Theorem \ref{th:HistState}.
    Since we can construct an $(e, O(K(a_A+\log K)), 0)$-block-encoding $\tilde{U}_{\tilde{\Phi}}$ of
    \begin{align}
        \begin{pmatrix}
            \tilde{\Phi}^{K,r,M}_{0,0} & & \\
            & \ddots & \\
            & & \tilde{\Phi}^{K,r,M}_{r-1,0}
        \end{pmatrix}
        \label{eq:PhiTilDiag}
    \end{align}
    similarly to $U_{\tilde{\Phi}}$ in Lemma \ref{lem:BETildePhi}, we can construct a $(1+e, O(K(a_A+\log K)), 0)$-block-encoding $V_{\tilde{\mathcal{A}}}$ of $\tilde{\mathcal{A}}$ in the way given in Ref.~\cite{Berry2024quantumalgorithm}, using $\tilde{U}_{\tilde{\Phi}}$ once.
    Having $U_{\tilde{\boldsymbol{\Delta}}}^{(1)}$ acting as $U_{\tilde{\boldsymbol{\Delta}}}^{(1)}\ket{0}\ket{n}\ket{0}= \ket{0}\ket{n}\ket{\tilde{\boldsymbol{\Delta}}^{(1)}_{n}} +\ket{0_\perp}$, which is a special case of $U_{\tilde{\boldsymbol{\Delta}}}$ in Lemma \ref{lem:Delta} with $N_{\rm s}=1$, we can perform the following operation on a three-register system:
    \begin{align}
        \ket{0}\ket{0}\ket{0} & \rightarrow \ket{0}\left(\sqrt{\frac{\|\mathbf{x}_0\|^2}{\|\mathbf{x}_0\|^2+4N\sigma^2T U_{\rm SN}^2}}\ket{r}+\sum_{n=0}^{r-1}\sqrt{\frac{4N\sigma^2\Delta t U_{\rm SN}^2}{\|\mathbf{x}_0\|^2+4N\sigma^2T U_{\rm SN}^2}}\ket{n}\right)\ket{0} \nonumber \\
        & \rightarrow \ket{0}\left(\sqrt{\frac{\|\mathbf{x}_0\|^2}{\|\mathbf{x}_0\|^2+4N\sigma^2T U_{\rm SN}^2}}\ket{r}\ket{\mathbf{x}_0} + \sum_{n=0}^{r-1}\sqrt{\frac{4N\sigma^2\Delta t U_{\rm SN}^2}{\|\mathbf{x}_0\|^2+4N\sigma^2T U_{\rm SN}^2}}\ket{n}\ket{0}\right) \nonumber \\
        & \rightarrow \ket{0}\left(\sqrt{\frac{\|\mathbf{x}_0\|^2}{\|\mathbf{x}_0\|^2+4N\sigma^2T U_{\rm SN}^2}}\ket{r}\ket{\mathbf{x}_0} + \sum_{n=0}^{r-1}\sqrt{\frac{4N\sigma^2\Delta t U_{\rm SN}^2}{\|\mathbf{x}_0\|^2+4N\sigma^2T U_{\rm SN}^2}}\ket{n}\ket{\tilde{\boldsymbol{\Delta}}^{(1)}_{n}}\right) + \ket{0_\perp} \nonumber \\
        & \rightarrow \ket{0}\left(\sqrt{\frac{\|\mathbf{x}_0\|^2}{\|\mathbf{x}_0\|^2+4N\sigma^2T U_{\rm SN}^2}}\ket{0}\ket{\mathbf{x}_0} + \sum_{n=0}^{r-1}\sqrt{\frac{4N\sigma^2\Delta t U_{\rm SN}^2}{\|\mathbf{x}_0\|^2+4N\sigma^2T U_{\rm SN}^2}}\ket{n+1}\ket{\tilde{\boldsymbol{\Delta}}^{(1)}_{n}}\right) + \ket{0_\perp},
        \label{eq:VtilBOp}
    \end{align}
    where the first arrow is by arithmetic circuits, the second arrow is by the controlled version of $U_{\mathbf{x}_0}$ activated when the value on the second register is $r$, the third arrow is by the controlled version of $U_{\tilde{\boldsymbol{\Delta}}}^{(1)}$ activated when the value on the second register is in $[r-1]_0$, and the last arrow is by addition of 1 modulo $r+1$ on the second register.
    The final state in Eq.~\eqref{eq:VtilBOp} is in the form of \eqref{eq:ketBTil} with $U_{\tilde{\mathcal{B}}}$ in Eq.~\eqref{eq:UtilB}, which implies that the unitary for the operation in Eq.~\eqref{eq:VtilBOp} is nothing but $V_{\tilde{\mathcal{B}}^{(1)}}$.
    Note that this $U_{\tilde{\mathcal{B}}}$ satisfies $U_{\tilde{\mathcal{B}}} \ge \left\|\tilde{\mathcal{B}}^{(1)}\right\|$: since $\tilde{S}^{K,r,M}_n$ is the upper-left block of the block-encoding unitary of $S_n$ with normalization factor $2\sqrt{\sigma^2 \Delta t}$, and thus $\left\|\tilde{S}^{K,r,M}_n\right\| \le 2\sqrt{\sigma^2 \Delta t}$, we have
    \begin{align}
        \left\|\tilde{\mathcal{B}}^{(1)}\right\|^2 = \left\|\mathbf{x}_0\right\|^2 + \sum_{n=0}^{r-1}\left\|\tilde{\boldsymbol{\Delta}}^{(1)}_n\right\|^2 = \left\|\mathbf{x}_0\right\|^2 + \sum_{n=0}^{r-1}\left\|\tilde{S}^{K,r,M}_n\mathbf{z}^{(i)}_n\right\|^2 \le \left\|\mathbf{x}_0\right\|^2 + r \cdot 4\sigma^2 \Delta t \cdot NU_{\rm SN}^2 = U_{\tilde{\mathcal{B}}}^2,
    \end{align}
    where we use Eq.~\eqref{eq:ziUB}.

    Using these $V_{\tilde{\mathcal{A}}}$ and $V_{\tilde{\mathcal{B}}^{(1)}}$, we can generate $\ket{\tilde{\tilde{\mathbf{X}}}_{\rm hist}}$ as stated.
    Note that the prefactor $\frac{\alpha_{\tilde{\mathcal{A}}}}{2\kappa_{\tilde{\mathcal{A}}}U_{\tilde{\mathcal{B}}}}$ in Eq.~\eqref{eq:XHistKet} now becomes that in Eq.~\eqref{eq:XHistKet2} for the current value $r=r_{\rm c}$.
    \ \\

    Now, let us see that setting $K=K_{\rm c}$, $r=r_{\rm c}$, and $M=M_{\rm c}$ makes the conditions \eqref{eq:RnPhinDiff} and \eqref{eq:condTilDel} hold.
    We first consider the condition \eqref{eq:RnPhinDiff}.
    As implied by Lemma~\ref{lem:PhiErr}, the current setting, in which
    \begin{align}
        K \ge \max\left\{\log\left(\frac{3}{\varepsilon}\right),7\right\},r\ge\alpha_A T, M\ge\frac{4\alpha_{\drm A/\drm t}}{\alpha_A^2 \varepsilon}
    \end{align}
    holds, makes $\left\|\tilde{\Phi}^{K,r,M}_{n,0}-\Phi_n\right\|\le\varepsilon$ hold.
    Then, it suffices to show $\varepsilon \le \epsilon_1$ and $\varepsilon \le \frac{1}{2}\eta \Delta t e^{-\eta \Delta t}$.
    We can see the former as
    \begin{align}
        \epsilon_1 & = \frac{\eta^2 \Delta t^2 U_{\tilde{\mathcal{B}}} \epsilon}{32\sqrt{6}}\sqrt{\frac{1}{\|\mathbf{x}_0\|^2\eta\Delta t+\frac{T}{\Delta t}\left(\Delta_{\rm max}^{(1)}\right)^2}} \nonumber \\
        & \ge \frac{\eta^2 \Delta t^2 U_{\tilde{\mathcal{B}}} \epsilon}{32\sqrt{6}}\sqrt{\frac{1}{\|\mathbf{x}_0\|^2\eta\Delta t+\sigma^2TNU_{\rm SN}^2}} \nonumber \\
        & = \frac{\eta^2 \Delta t^2 \epsilon}{32\sqrt{6}}\sqrt{\frac{\|\mathbf{x}_0\|^2+4\sigma^2TNU_{\rm SN}^2}{\|\mathbf{x}_0\|^2\eta\Delta t+\sigma^2TNU_{\rm SN}^2}} \nonumber \\
        & \ge \frac{\eta^2 \Delta t^2 \epsilon}{32\sqrt{6}} \nonumber \\
        & \ge \varepsilon.
        \label{eq:eps1vareps}
    \end{align}
    Here, the first inequality follows from
    \begin{align}
        \left(\Delta_{\rm max}^{(1)}\right)^2 = \max_{n\in[r-1]_0} \|\boldsymbol{\Delta}^{(1)}_n\|^2=\max_{n\in[r-1]_0} (\mathbf{z}^{(1)}_n)^T\Sigma_n\mathbf{z}^{(1)}_n \le \max_{n\in[r-1]_0} \|\Sigma_n\| \cdot \|\mathbf{z}^{(1)}_n\|^2 \le \sigma^2\Delta tNU_{\rm SN}^2,
        \label{eq:DeltaMaxUB}
    \end{align}
    where we use Eqs. \eqref{eq:SigmanNorm} and \eqref{eq:ziUB}, and the second inequality follows from $\eta \Delta t \le 1/4$, which can be seen as
    \begin{align}
        \Delta t = \frac{T}{r} \le \frac{1}{4\kappa_{BB^T}\alpha_A} \le \frac{1}{4\|A\|} \le \frac{1}{4\eta}.
    \end{align}
    $\varepsilon \le \frac{1}{2}\eta \Delta t e^{-\eta \Delta t}$ is obtained by combining $\eta \Delta t \le 1$, $\epsilon \le 1$, and $\varepsilon \le \frac{\eta^2 \Delta t^2 \epsilon}{16\sqrt{6}}$.
    In summary, the condition \eqref{eq:RnPhinDiff} is satisfied.

    Next, let us consider the condition \eqref{eq:condTilDel}.
    Following Lemma \ref{lem:Delta}, we construct $U_{\tilde{\boldsymbol{\Delta}}}$ such that
    \begin{align}
        \left\|\tilde{S}^{K,r,M}_n-S_n\right\| \le \varepsilon\sqrt{\sigma^2 \Delta t}
        \label{eq:Snvareps}
    \end{align}
    holds.
    To have such $U_{\tilde{\boldsymbol{\Delta}}}$, setting $K=K_{\rm c}$, $r=r_{\rm c}$, and $M=M_{\rm c}$ suffices.
    Eq.~\eqref{eq:Snvareps} implies
    \begin{align}
        \left\|\tilde{\boldsymbol{\Delta}}^{(1)}_n-S_n\mathbf{z}^{(1)}_n\right\| \le \left\|\tilde{S}^{K,r,M}_n-S_n\right\|\cdot \left\|\mathbf{z}^{(1)}_n\right\| \le \varepsilon\sqrt{\sigma^2 \Delta t} \cdot \sqrt{N}U_{\rm SN} \le \epsilon_2,
        \label{eq:DelTilErrEps2}
    \end{align}
    where the last inequality is seen by
    \begin{align}
        \frac{\epsilon_2}{\sqrt{N \sigma^2 \Delta t}U_{\rm SN}} = \sqrt{\frac{1}{384\sigma^2 TN U_{\rm SN}^2}}\eta \Delta t U_{\tilde{\mathcal{B}}} \epsilon = \sqrt{\frac{\|\mathbf{x}_0\|^2+4\sigma^2TNU_{\rm SN}^2}{384\sigma^2 TN U_{\rm SN}^2}}\frac{\eta T}{r_{\rm c}} \epsilon \ge \varepsilon.
        \label{eq:eps2vareps}
    \end{align}
    Thus, the condition \eqref{eq:condTilDel} also holds.

    ~ \\

    Lastly, let us see that we have the claimed complexity bounds.
    Theorem \ref{th:HistState} states that $V^{(1)}_{{\rm hist},\epsilon}$ queries $V_{\tilde{\mathcal{A}}}$ a number of times given by Eq.~\eqref{eq:CompHistGen}, which is
    \begin{align}
        O\left(\frac{\kappa_{BB^T}\alpha_A T}{\max\{1,\eta T\}}\log\left(\frac{\kappa_{BB^T}\alpha_A T}{\max\{1,\eta T\}\epsilon}\right)\right)
    \end{align}
    for the current values $r=r_{\rm c}$ and $\alpha_{\tilde{\mathcal{A}}}=1+e$, and $V_{\tilde{\mathcal{B}}^{(1)}}$ once.
    According to Ref.~\cite{Berry2024quantumalgorithm}, a block-encoding $V_{\tilde{\mathcal{A}}}$ of $\tilde{\mathcal{A}}$ is constructed with one query to the controlled version of the block-encoding of Eq.~\eqref{eq:PhiTilDiag}, which, as stated in Lemma \ref{lem:BETildePhi}, is implemented by $O(K_{\rm c})$ uses of the controlled $U_A$.
    $V_{\tilde{\mathcal{B}}^{(1)}}$ is now the unitary for the operation in Eq.~\eqref{eq:VtilBOp}, in which the controlled versions of $U_{\mathbf{x}_0}$ and $U_{\tilde{\boldsymbol{\Delta}}}^{(1)}$ are used once each.
    According to Lemma \ref{lem:Delta}, we can implement $U_{\tilde{\boldsymbol{\Delta}}}^{(1)}$ using the controlled $U_A$ a number of times given by Eq.~\eqref{eq:CompUSigmaSqrt1}, where $K$ and $\epsilon$ are now replaced with $K_{\rm c}$ and $\varepsilon$, respectively, $U_{BB^T}$ once, and $U_{\rm Rand}$ once.
    Accumulating these query numbers, we obtain the stated evaluations of the number of queries to controlled $U_A$, $U_{BB^T}$, $U_{\mathbf{x}_0}$, and $U_{\rm Rand}$.
    We also see that the number of additional elementary gates plies up to Eq.~\eqref{eq:CompUSigmaSqrt2}, recalling Lemmas \ref{lem:BETildePhi} and \ref{lem:Delta}, and Theorem \ref{th:HistState}, with $a_{\tilde{\mathcal{A}}}$ being $O(K_{\rm c}(a_A+\log K_{\rm c}))$.
\end{proof}

The above theorem is about preparing the quantum state that encodes one realization of $\mathbf{X}_{\rm hist}$ yielded by a set of pseudorandom numbers $z_1,\ldots,z_{r_{\rm c}N}$.
We can extend this to preparing states encoding other realizations.

\begin{corollary}
    Suppose that Assumptions \ref{ass:PRN}, \ref{ass:UA}, \ref{ass:muA}, \ref{ass:UBB}, and \ref{ass:Ux0} hold.
    Let $\epsilon\in(0,1]$, $U_{\rm SN}>0$, and $N_{\rm s}\in\mathbb{N}$.
    Define $r_{\rm c}$ as Eq.~\eqref{eq:rReq} and set $r=r_{\rm c}$.
    Define for $i\in\mathbb{N}$, $\mathbf{X}_0^{(i)},\ldots,\mathbf{X}_r^{(i)}$ by Eq.~\eqref{eq:XnRec} with $\boldsymbol{\Delta}_0=\boldsymbol{\Delta}^{(i)}_0,\ldots,\boldsymbol{\Delta}_{r-1}=\boldsymbol{\Delta}^{(i)}_{r-1}$, where for $n\in[r-1]_0$, $\boldsymbol{\Delta}^{(i)}_n \coloneqq S_n \mathbf{z}^{(i)}_n$.
    Then, we can construct an oracle $V_{{\rm hist},\epsilon}$, which acts as
    \begin{align}
        V_{{\rm hist},\epsilon}\ket{i}\ket{0}\ket{0} =   \ket{i}\ket{0}\ket{\tilde{\tilde{\mathbf{X}}}_{\rm hist}^{(i)}} + \ket{0_\perp}
        \label{eq:VHisi}
    \end{align}
    for any $i\in N_{\rm s}$.
    Here, $\ket{0_\perp}$ is some unnormalized state such that $(I\otimes \ket{0}\!\bra{0}\otimes I)\ket{0_\perp}=0$, and $\ket{\tilde{\tilde{\mathbf{X}}}_{\rm hist}^{(i)}}$ is an unnormalized state given by
    \begin{align}
        \Ket{\tilde{\tilde{\mathbf{X}}}_{\rm hist}^{(i)}} \coloneqq \frac{\max\{1,\eta T\}}{8r_{\rm c} U_{\tilde{\mathcal{B}}}}\sum_{j=1}^{(r+1)N}\tilde{\tilde{X}}_{\rm hist}^{(i),j}\ket{j}
    \end{align}
    with $\tilde{\tilde{\mathbf{X}}}_{\rm hist}^{(i)}=\left(\tilde{\tilde{X}}_{\rm hist}^{(i),1},\ldots,\tilde{\tilde{X}}_{\rm hist}^{(i),(r+1)N}\right)^T\in\mathbb{R}^{(r+1)N}$ such that, if $z_1,\ldots,z_{r_{\rm c}NN_{\rm s}}$ generated by $U_{\rm Rand}$ satisfy
    \begin{align}
        |z_1|,\ldots,|z_{r_{\rm c}NN_{\rm s}}|\le U_{\rm SN},
        \label{eq:ziUB2}
    \end{align}
    \begin{align}
        \left\| \tilde{\tilde{\mathbf{X}}}_{\rm hist}^{(i)} - 
        \mathbf{X}_{\rm hist}^{(i)}\right\| \le U_{\tilde{\mathcal{B}}} \epsilon
    \end{align}
    holds for
    \begin{align}
        \mathbf{X}_{\rm hist}^{(i)} \coloneqq
    \begin{pmatrix}
        \mathbf{X}_0^{(i)} \\ \vdots \\ \mathbf{X}_r^{(i)} 
    \end{pmatrix}
    .
    \end{align}
    $V_{{\rm hist},\epsilon}$ queries the controlled $U_A$, $U_{BB^T}$, $U_{\mathbf{x}_0}$, $U_{\rm Rand}$, and additional elementary gates a number of times same as Theorem \ref{th:HistStateConc}.
    \label{cor:HistStateConcSuper}
\end{corollary}

\begin{proof}
    We can straightforwardly promote $V_{{\rm hist},\epsilon}^{(1)}$ in Theorem \ref{th:HistStateConc} to $V_{{\rm hist},\epsilon}$ by the following modifications.
    We add the first register in Eq.~\eqref{eq:VHisi} to the system on which $V_{{\rm hist},\epsilon}^{(1)}$ acts.
    Then, we replace $V_{\tilde{\mathcal{A}}}$ in the proof of Theorem \ref{th:HistStateConc} with $I \otimes V_{\tilde{\mathcal{A}}}$, and $V_{\tilde{\mathcal{B}}^{(1)}}$ with $V_{\tilde{\mathcal{B}}}$.
    Here, $V_{\tilde{\mathcal{B}}}$ is obtained by promoting $U_{\tilde{\boldsymbol{\Delta}}}^{(1)}$ in $V_{\tilde{\mathcal{B}}^{(1)}}$ to $U_{\tilde{\boldsymbol{\Delta}}}$ and acts as
    \begin{align}
        V_{\tilde{\mathcal{B}}}\ket{i}\ket{0}\ket{0}=\ket{i}\ket{0}\ket{\tilde{\mathcal{B}}^{(i)}} + \ket{0_\perp}, \ket{\tilde{\mathcal{B}}^{(i)}} \coloneqq \frac{1}{U_{\tilde{\mathcal{B}}}}\tilde{\mathcal{B}}^{(i)},
        \tilde{\mathcal{B}}^{(i)} \coloneqq 
        \begin{pmatrix}
        \mathbf{x}_0 \\ \tilde{\boldsymbol{\Delta}}^{(i)}_0 \\ \vdots \\ \tilde{\boldsymbol{\Delta}}^{(i)}_{r-1} 
        \end{pmatrix}.
    \end{align}
    
\end{proof}

\subsection{Estimating expectations of multi-time functions}

Having the history state, we now aim to extract quantities of interest as values from the state.
As a concrete example, we consider how to estimate the expectation of $Y\coloneqq \sum_{j_1,\ldots,j_d}C_{j_1,\ldots,j_d}X_{\rm hist}^{j_1} \cdots X_{\rm hist}^{j_d}$, where $C$ is a $d$-way tensor and $X_{\rm hist}^j$ is the $j$-th entry of $\mathbf{X}_{\rm hist}$.
Since this depends on the values of $\mathbf{X}_t$ at multiple time points, we call this a multi-time function.
As seen below, given the superposition state $\ket{\rm histSP}$ that encodes many realizations of $\mathbf{X}_{\rm hist}$ and the state encoding $C$, we can estimate $\mathbb{E}_{\mathbb{P}}\left[Y\right]$ as an overlap of these states, which equals the average of the sample values of $Y$.
This approach, which can be dubbed as the {\it classical Monte Carlo method on a quantum computer}, is similar to that in Ref. \cite{Miyamoto2020,kaneko2021quantum,kaneko2022quantum}, where finance-related Monte Carlo integrations are addressed. 

In this estimation, aside from errors in $\ket{\rm histSP}$ originating from the block-encodings constructed above, we need to handle two types of errors.
One is the error of overlap estimation.
The other is the statistical error, the deviation of the sample average from the true expectation, which depends on the variance of the averaged quantity.
We thus start the discussion by showing the following lemma.

\begin{lemma}
    Let $d\in\mathbb{N}$ and $\epsilon\in\left(0,1\right]$.
    Suppose that Assumptions \ref{ass:UA}, \ref{ass:muA}, \ref{ass:UBB}, and \ref{ass:Ux0} hold.
    Set $r=r_{\rm c}$, $K=K_{\rm c}$, and $M=M_{\rm c}$ with $r_{\rm c}$, $K_{\rm c}$, and $M_{\rm c}$ in Theorem \ref{th:HistStateConc}.
    Suppose that Assumption \ref{ass:PRN} holds with
    \begin{align}
        \epsilon_{\rm PRN} = \frac{(2d-1)!!}{(2r\sqrt{N})^d}.
        \label{eq:epsPRN}
    \end{align}
    Then, for any $i\in\mathbb{N}$, $\mathbf{X}_{\rm hist}^{(i)}$ and $\tilde{\tilde{\mathbf{X}}}_{\rm hist}^{(i)}$ in Corollary \ref{cor:HistStateConcSuper} satisfy
    \begin{align}
        \mathbb{E}\left[\left\|\left(\mathbf{X}_{\rm hist}^{(i)}\right)^{\otimes d}\right\|^2\right]\le 2(2d-1)!!(r+1)^d\left(\|\mathbf{x}_0\|^2+(N+2d)\sigma^2T\right)^d,
        \label{eq:EXhistd2}
    \end{align}
    and
    \begin{align}
        \mathbb{E}\left[\left\|\left(\tilde{\tilde{\mathbf{X}}}_{\rm hist}^{(i)}\right)^{\otimes d}\right\|^2\right]\le 2(2d-1)!!(1+3\epsilon)^{2d}(r+1)^d\left(\|\mathbf{x}_0\|^2+(N+2d)\sigma^2T\right)^d.
        \label{eq:EXhisttiltild2}
    \end{align}
    \label{lem:XhistMom}
\end{lemma}

\begin{proof}
    For any $n\in[r]_0$, we can write
    \begin{align}
    \mathbf{X}_n^{(i)} = \sum_{l=-1}^{r-1} \mathcal{D}_{n,l}  \mathbf{z}^{(i)}_{l},
    \label{eq:XnSum}
    \end{align}
    with
    \begin{align}
        \mathcal{D}_{n,l} \coloneqq
        \begin{cases}
            (\mathcal{A}^{-1})_{n,0} & ; ~ l=-1 \\
            (\mathcal{A}^{-1})_{n,l+1} S_{l} & ; ~ l\in[r-1]_0
        \end{cases}
    \end{align}
    and
    \begin{align}
        \mathbf{z}^{(i)}_{-1} \coloneqq \mathbf{x}_0.
    \end{align}
    Noting
    \begin{align}
        \|(\mathcal{A}^{-1})_{n,l+1}\| \le 1
        \label{eq:calAInvBlockNorm}
    \end{align}
    for $l\in\{-1\}\cup[r-1]_0$, which follows from Eqs.~\eqref{eq:PhiUBEta} and \eqref{eq:calAInvBlock}, and
    \begin{align}
        \|S_l\| \le \sqrt{\sigma^2 \Delta t},
        \label{eq:SnNorm}
    \end{align}
    which follows from Eq.~\eqref{eq:SigmanNorm}, we have
    \begin{align}
        \|\mathcal{D}_{n,l}\| \le D_l \coloneqq
        \begin{cases}
        1 & ; ~ l=-1 \\
        \sqrt{\sigma^2\Delta t}, & ; ~ l\in[r-1]_0
        \end{cases}
        .
        \label{eq:calDNorm}
    \end{align}
    Plugging Eq.~\eqref{eq:XnSum} into $\mathbb{E}\left[\left\|\left(\mathbf{X}_{\rm hist}^{(i)}\right)^{\otimes d}\right\|^2\right]$ yields
    \begin{align}
        \mathbb{E}\left[\left\|\left(\mathbf{X}_{\rm hist}^{(i)}\right)^{\otimes d}\right\|^2\right] =  \mathbb{E}\left[\left\|\mathbf{X}_{\rm hist}^{(i)}\right\|^{2d}\right]= \sum_{n_1,\ldots,n_d=0}^r \sum_{l_1,\ldots,l_{2d}=-1}^{r-1} \mathbb{E}\left[\left(\mathcal{D}_{n_1,l_1}\mathbf{z}^{(i)}_{l_1}\right)^T \mathcal{D}_{n_1,l_2} \mathbf{z}^{(i)}_{l_2} \times \cdots \times \left(\mathcal{D}_{n_d,l_{2d-1}}\mathbf{z}^{(i)}_{l_{2d-1}}\right)^T \mathcal{D}_{n_d,l_{2d}} \mathbf{z}^{(i)}_{l_{2d}}\right].
        \label{eq:Xhistd2Temp}
    \end{align}
    Defining
    \begin{align}
        \mathcal{N}_{2k}^{a,b} \coloneqq \left\{\mathbf{l} \in \{a,\ldots,b\}^{\times 2k}~\middle|~\#_a(\mathbf{l}),\#_{a+1}(\mathbf{l}),\ldots,\#_b(\mathbf{l}) \text{ are even}\right\}
    \end{align}
    for $k\in\mathbb{N}$ and $a,b\in\mathbb{Z}$ and using Eqs.~\eqref{eq:PRNMom} and \eqref{eq:calDNorm}, we get
    \begin{align}
        & \sum_{n_1,\ldots,n_d=0}^r \sum_{(l_1,\ldots,l_{2d})\in\mathcal{N}_{2d}^{-1,r-1}} \mathbb{E}\left[\left(\mathcal{D}_{n_1,l_1}\mathbf{z}^{(i)}_{l_1}\right)^T \mathcal{D}_{n_1,l_2} \mathbf{z}^{(i)}_{l_2} \times \cdots \times \left(\mathcal{D}_{n_d,l_{2d-1}}\mathbf{z}^{(i)}_{l_{2d-1}}\right)^T \mathcal{D}_{n_d,l_{2d}} \mathbf{z}^{(i)}_{l_{2d}}\right] \nonumber \\
        \le & (r+1)^d\sum_{\mathbf{l}\in\mathcal{N}_{2d}^{-1,r-1}} D_{l_1} \cdots D_{l_{2d}} \mathbb{E}\left[\left\|\mathbf{z}^{(i)}_{l_1}\right\|\cdots\left\|\mathbf{z}^{(i)}_{l_{2d}}\right\|\right] \nonumber \\
        \le & (r+1)^d\sum_{k=0}^d \binom{2d}{2k} \|\mathbf{x}_0\|^{2k} \left|\mathcal{N}_{2(d-k)}^{0,r-1}\right| \left(\sigma^2\Delta t\right)^{d-k} \left(N+2d\right)^{d-k}.
        \label{eq:Xhistd2Temp2}
    \end{align}
    Using $\left|\mathcal{N}_{2(d-k)}^{0,r-1}\right| \le (2(d-k)-1)!!r^{(d-k)}$, which is shown in Lemma \ref{lem:calN} in Appendix \ref{sec:bndN}, and $\binom{2d}{2k}(2(d-k)-1)!!\le(2d-1)!!\binom{d}{k}$, which is shown by simple algebra, we can further upper bound Eq.~\eqref{eq:Xhistd2Temp2} by
    \begin{align}
        (2d-1)!!(r+1)^d\left(\|\mathbf{x}_0\|^2+(N+2d)\sigma^2T\right)^d.
    \end{align}
    On the other hand, we have
    \begin{align}
        & \sum_{n_1,\ldots,n_d=0}^r \sum_{(l_1,\ldots,l_{2d})\notin\mathcal{N}_{2d}^{-1,r-1}} \mathbb{E}\left[\left(\mathcal{D}_{n_1,l_1}\mathbf{z}^{(i)}_{l_1}\right)^T \mathcal{D}_{n_1,l_2} \mathbf{z}^{(i)}_{l_2} \times \cdots \times \left(\mathcal{D}_{n_d,l_{2d-1}}\mathbf{z}^{(i)}_{l_{2d-1}}\right)^T \mathcal{D}_{n_d,l_{2d}} \mathbf{z}^{(i)}_{l_{2d}}\right] \nonumber \\
        \le & \sum_{n_1,\ldots,n_d=0}^r\sum_{(l_1,\ldots,l_{2d})\notin\mathcal{N}_{2d}^{-1,r-1}}\left\|\mathcal{D}_{n_1,l_1}^T \mathcal{D}_{n_1,l_2}\right\|_{\rm F} \cdots \left\|\mathcal{D}_{n_d,l_{2d-1}}^T \mathcal{D}_{n_d,l_{2d}}\right\|_{\rm F} \left\|\mathbb{E}\left[\mathbf{z}^{(i)}_{l_1}\otimes\cdots\otimes\mathbf{z}^{(i)}_{l_{2d}}\right]\right\| \nonumber \\
        \le & (r+1)^d \sum_{k=0}^{2d} \binom{2d}{k} r^{2d-k} N^{d/2} \left(\sigma^2 \Delta t\right)^{(2d-k)/2} \|\mathbf{x}_0\|^k (N+2d)^{(2d-k)/2}\epsilon_{\rm PRN} \nonumber \\
        \le & (r+1)^dN^{d/2} \left(\|\mathbf{x}_0\| + \sqrt{(N+2d)r\sigma^2T}\right)^{2d}\epsilon_{\rm PRN} \nonumber \\
        \le & 2^d(r+1)^dr^dN^{d/2}\left(\|\mathbf{x}_0\|^2 + (N+2d)\sigma^2T\right)^d \epsilon_{\rm PRN} \nonumber \\
        \le & (2d-1)!!(r+1)^d\left(\|\mathbf{x}_0\|^2+(N+2d)\sigma^2T\right)^d
    \end{align}
    where we use $\left\|\mathcal{D}_{n_k,l_{2k-1}}^T \mathcal{D}_{n_k,l_{2k}}\right\|_{\rm F} \le \sqrt{N} \left\|\mathcal{D}_{n_k,l_{2k-1}}^T \mathcal{D}_{n_k,l_{2k}}\right\| \le \sqrt{N} D_{l_{2k-1}} D_{l_{2k}}$ and Eq.~\eqref{eq:PRNExpSNDiff} at the second inequality and Eq.~\eqref{eq:epsPRN} at the last inequality.
    Combining these, we get Eq.~\eqref{eq:EXhistd2}.

    ~\\

    The proof of Eq.~\eqref{eq:EXhisttiltild2} goes similarly but requires some preparations.
    First, let us see
    \begin{align}
        \left\|\tilde{\mathcal{A}}^{-1}-\mathcal{A}^{-1}\right\| \le \frac{\epsilon}{2\sqrt{6}}
        \label{eq:AtilInvAInv}
    \end{align}
    for $\tilde{\mathcal{A}}$ introduced in the proof of Theorem \ref{th:HistStateConc}.
    Similarly to $\mathcal{A}^{-1}$, $\tilde{\mathcal{A}}^{-1}$ can be written in block form as Eq.~\eqref{eq:calAInv} with the $(n,n^\prime)$-block being
    \begin{align}
        (\tilde{\mathcal{A}}^{-1})_{n,n^\prime} =
        \begin{cases}
            I &  ; ~ n = n^\prime \\
            \tilde{\Phi}^{K,r,M}_{n-1,0} \cdots \tilde{\Phi}^{K,r,M}_{n^\prime,0} & ; ~ n > n^\prime \\
            O & ; ~ n < n^\prime
        \end{cases}
        .
    \end{align}
    Thus, the norm of the $(n,n^\prime)$-block of $\tilde{\mathcal{A}}^{-1}-\mathcal{A}^{-1}$ is 0 for $n \le n^\prime$ and
    \begin{align}
        \left\|(\tilde{\mathcal{A}}^{-1})_{n,n^\prime}-(\mathcal{A}^{-1})_{n,n^\prime}\right\|&= \left\|\tilde{\Phi}^{K,r,M}_{n-1,0} \cdots \tilde{\Phi}^{K,r,M}_{n^\prime,0}-\Phi_{n-1,0} \cdots \Phi_{n^\prime,0}\right\| \nonumber \\
        & \le (n-n^\prime) \times \left(\max_{l\in[r-1]_0}\left\{\left\|\Phi_l\right\|,\left\|\tilde{\Phi}^{K,r,M}_{l,0}\right\|\right\}\right)^{n-n^\prime-1} \times \max_{l\in[r-1]_0} \left\|\Phi_l-\tilde{\Phi}^{K,r,M}_{l,0}\right\| \nonumber \\
        & \le (n-n^\prime) \left(e^{-\eta \Delta t}+\frac{1}{2}\eta \Delta te^{-\eta \Delta t}\right)^{n-n^\prime-1}\varepsilon \nonumber \\
        & \le (n-n^\prime) e^{-\frac{1}{2}(n-n^\prime-1)\eta \Delta t}\varepsilon
        \label{eq:calAtilInvcalAInvBlock}
    \end{align}
    otherwise.
    Here, we use Eq.~\eqref{eq:PhiUBEta} along with $\left\|\tilde{\Phi}^{K,r,M}_{l,0}-\Phi_l\right\|\le\varepsilon\le\frac{1}{2}\eta \Delta te^{-\eta \Delta t}$, which holds for the current $K$, $r$, and $M$, and $1+x \le e^x$, which holds for any $x\in\mathbb{R}$.
    Lemma 6 in \cite{an2026fast} implies that
    \begin{align}
        \left\|\tilde{\mathcal{A}}^{-1}-\mathcal{A}^{-1}\right\| \le \sqrt{\max_{n\in[r]_0} \sum_{n^\prime=0}^r\left\|(\tilde{\mathcal{A}}^{-1})_{n,n^\prime}-(\mathcal{A}^{-1})_{n,n^\prime}\right\| \times \max_{n^\prime\in[r]_0} \sum_{n=0}^r\left\|(\tilde{\mathcal{A}}^{-1})_{n,n^\prime}-(\mathcal{A}^{-1})_{n,n^\prime}\right\|},
        \label{eq:tilAInvAInvNorm}
    \end{align}
    and using Eq.~\eqref{eq:calAtilInvcalAInvBlock} in this yields 
    \begin{align}
        \left\|\tilde{\mathcal{A}}^{-1}-\mathcal{A}^{-1}\right\| \le \sum_{k=0}^r k e^{-\frac{1}{2}(k-1)\eta \Delta t}\varepsilon \le \sum_{k=0}^\infty k e^{-\frac{1}{2}(k-1)\eta \Delta t}\varepsilon \le \frac{\varepsilon}{\left(1-e^{-\frac{1}{2}\eta\Delta t}\right)^2} \le \frac{16\varepsilon}{(\eta \Delta t)^2} \le \frac{\epsilon}{2\sqrt{6}}
        \label{eq:tilAInvAInvNorm2}
    \end{align}
    where we use Eq. \eqref{eq:1-e-aInvUB} and the definition of $\varepsilon$ in Eq.~\eqref{eq:KrMComb}.
    This proves Eq.~\eqref{eq:AtilInvAInv}.
    Combining Eq.~\eqref{eq:AtilInvAInv} with
    \begin{align}
        \left\|\tilde{\mathcal{C}}-\tilde{\mathcal{A}}^{-1}\right\| \le  \frac{\epsilon}{2},
        \label{eq:CtilAtilInv}
    \end{align}
    which follows from Eq.~\eqref{eq:MatInvErr}, we get
    \begin{align}
        \left\|\tilde{\mathcal{C}}-\mathcal{A}^{-1}\right\| \le \epsilon.
        \label{eq:CtilAInv}
    \end{align}
    Besides, Lemma \ref{lem:Delta} implies
    \begin{align}
        \left\|\tilde{S}^{K,r,M}_n-S_n\right\|\le \sqrt{\sigma^2\Delta t}\epsilon,
        \label{eq:SErrNorm}
    \end{align}
    for the current $K$, $r$, and $M$, and combining this with Eq.~\eqref{eq:SnNorm} yields
    \begin{align}
        \left\|\tilde{S}^{K,r,M}_n\right\| \le \|S_n\|+\left\|\tilde{S}^{K,r,M}_n-S_n\right\| \le (1+\epsilon)\sqrt{\sigma^2\Delta t}.
        \label{eq:StilNorm}
    \end{align}

    Then, let us prove Eq.~\eqref{eq:EXhisttiltild2}.
    We have
    \begin{align}
        \tilde{\tilde{\mathbf{X}}}_n^{(i)} = \sum_{l=-1}^{r-1} \tilde{\mathcal{D}}_{n,l}  \tilde{\mathbf{z}}^{(i)}_l,
        \label{eq:Xtiltilni}
    \end{align}
    where
    \begin{align}
        \tilde{\mathbf{z}}^{(i)}_l = \begin{pmatrix} \tilde{z}^{(i)}_{l,1} \\ \vdots \\ \tilde{z}^{(i)}_{l,N} \end{pmatrix} \coloneqq
        \begin{dcases}
        \mathbf{x}_0 & ;~l=-1 \\
        \begin{pmatrix}
            F^{\rm Bd}_{U_{\rm SN}}(z_{(i-1)rN+lN+1}) \\ \vdots \\ F^{\rm Bd}_{U_{\rm SN}}(z_{(i-1)rN+(l+1)N})
        \end{pmatrix} & ; ~ l\in[r-1]_0
        \end{dcases}
    \end{align}
    and
    \begin{align}
        \tilde{\mathcal{D}}_{n,l} \coloneqq
        \begin{cases}
            \tilde{\mathcal{C}}_{n,0} & ; ~ l=-1 \\
            \tilde{\mathcal{C}}_{n,l+1} \tilde{S}^{K,r,M}_l & ; ~ l\in[r-1]_0
        \end{cases}
        .
    \end{align}
    Here, $\tilde{\mathcal{C}}_{n,n^\prime}\in\mathbb{R}^{N \times N}$ is the $(n,n^\prime)$-th block in the block matrix representation of $\tilde{\mathcal{C}}$ like Eq.~\eqref{eq:calAInv}.
    We have
    \begin{align}
        \left\|\tilde{\mathcal{C}}_{n,n^\prime}\right\| \le \left\|(\mathcal{A}^{-1})_{n,n^\prime}\right\| + \left\|\tilde{\mathcal{C}}_{n,n^\prime}-(\mathcal{A}^{-1})_{n,n^\prime}\right\| \le 1 + \left\|\tilde{\mathcal{C}}-\mathcal{A}^{-1}\right\| \le 1 + \epsilon
        \label{eq:calCBlock}
    \end{align}
    by using Eqs. \eqref{eq:calAInvBlockNorm} and \eqref{eq:CtilAInv}.
    Combining Eqs. \eqref{eq:StilNorm} and \eqref{eq:calCBlock} yields
    \begin{align}
        \|\tilde{\mathcal{D}}_{n,l}\| \le  \tilde{D}_l \coloneqq
        \begin{dcases}
        1+\epsilon & ; ~ l=-1 \\
        (1+3\epsilon)\sqrt{\sigma^2\Delta t}, & ; ~ l\in[r-1]_0
        \end{dcases}
        .
        \label{eq:calDtilNorm}
    \end{align}
     Then, we can bound
    \begin{align}
        \mathbb{E}\left[\left\|\left(\tilde{\tilde{\mathbf{X}}}_{\rm hist}^{(i)}\right)^{\otimes d}\right\|^2\right] = \sum_{n_1,\ldots,n_d=0}^r \sum_{l_1,\ldots,l_{2d}=-1}^{r-1} \mathbb{E}\left[\left(\tilde{\mathcal{D}}_{n_1,l_1}\tilde{\mathbf{z}}^{(i)}_{l_1}\right)^T \tilde{\mathcal{D}}_{n_1,l_2} \tilde{\mathbf{z}}^{(i)}_{l_2} \times \cdots \times \left(\tilde{\mathcal{D}}_{n_d,l_{2d-1}}\tilde{\mathbf{z}}^{(i)}_{l_{2d-1}}\right)^T \tilde{\mathcal{D}}_{n_d,l_{2d}} \tilde{\mathbf{z}}^{(i)}_{l_{2d}}\right]
        \label{eq:Xtiltilhistd2Temp}
    \end{align}
    as Eq.~\eqref{eq:EXhisttiltild2} by the discussion similar to the proof of Eq.~\eqref{eq:EXhistd2}, where we note $\left\|\tilde{\mathbf{z}}^{(i)}_l\right\| \le \left\|\mathbf{z}^{(i)}_l\right\|$ and use Eqs.~\eqref{eq:PRNClip} and \eqref{eq:calDtilNorm} instead of Eqs.~\eqref{eq:PRNExpSNDiff} and \eqref{eq:calDNorm}, respectively.
\end{proof}

Using this lemma, we can bound the error between the sample average and the true expectation.

\begin{lemma}
    Let $d\in\mathbb{N}$, $N_{\rm s}\in\mathbb{N}$, $\epsilon\in\left(0,1\right]$, and $\delta\in(0,1)$.
    Suppose that Assumptions \ref{ass:UA}, \ref{ass:muA}, \ref{ass:UBB}, and \ref{ass:Ux0} hold.
    Set $K$, $r$, and $M$ to $K_{\rm c}$, $r_{\rm c}$, and $M_{\rm c}$ in Theorem \ref{th:HistStateConc}, respectively.
    Suppose that Assumption \ref{ass:PRN} holds with
    \begin{align}
        \epsilon_{\rm PRN} = \min\left\{\frac{(2d-1)!!}{(2r\sqrt{N})^d},\frac{d\sqrt{(2d-3)!!}\epsilon}{(2r)^{d/2}},\frac{d^2(2d-3)!!\delta\epsilon^2}{8(2r)^d}\right\}.
        \label{eq:epsPRN2}
    \end{align}
    Set $U_{\rm SN}$ as
    \begin{align}
        U_{\rm SN} = \max\left\{\sqrt{2 \log \left(\frac{8r N N_{\rm s}}{\sqrt{2\pi}\delta}\right)},1\right\}.
        \label{eq:USN}
    \end{align}
    Define
    \begin{align}
        \label{eq:Y}
        Y \coloneqq \sum_{j_1,\ldots,j_d=1}^{N(r+1)} C_{j_1,\ldots,j_d} X_{\rm hist}^{j_1} \cdots X_{\rm hist}^{j_d},
    \end{align}
    and
    \begin{align}
        \hat{Y} \coloneqq \frac{1}{N_{\rm s}}\sum_{i=1}^{N_{\rm s}}\tilde{\tilde{Y}}^{(i)},~\tilde{\tilde{Y}}^{(i)} \coloneqq \sum_{j_1,\ldots,j_d=1}^{N(r+1)} C_{j_1,\ldots,j_d} \tilde{\tilde{X}}_{\rm hist}^{(i),j_1} \cdots \tilde{\tilde{X}}_{\rm hist}^{(i),j_d},~i\in\mathbb{N},
    \end{align}
    with a $d$-way real tensor $C=(C_{j_1,\ldots,j_d})\in\mathbb{R}^{N(r+1) \times \cdots \times N(r+1)}$, where $X_{\rm hist}^j$ and $\tilde{\tilde{X}}_{\rm hist}^{(i),j}$ are the $j$-th entry of $\mathbf{X}_{\rm hist}$ and $\tilde{\tilde{\mathbf{X}}}^{(i)}_{\rm hist}$, respectively.
    Then,
    \begin{align}
        \left|\hat{Y}-\mathbb{E}_{\mathbb{P}}[Y]\right|\le 8d\sqrt{(2d-3)!!}(1+3\epsilon)^d(r+1)^{d/2}\left(\|\mathbf{x}_0\|^2+(N+2d)\sigma^2T\right)^{d/2} \|C\|_F \left(\epsilon + \sqrt{\frac{4}{\delta N_{\rm s}}}\right)
        \label{eq:YErrProb}
    \end{align}
    holds with probability at least $1-\frac{\delta}{2}$.
    \label{lem:YHat}
\end{lemma}

\begin{proof}
    First, defining
    \begin{align}
        \tilde{\tilde{\mathbf{X}}}_{\rm histNB}^{(i)} \coloneqq
        \begin{pmatrix}
            \tilde{\tilde{\mathbf{X}}}_{{\rm NB},0}^{(i)} \\
            \vdots \\
            \tilde{\tilde{\mathbf{X}}}_{{\rm NB},r}^{(i)}
        \end{pmatrix}
        ,
        \tilde{\tilde{\mathbf{X}}}_{{\rm NB},n}^{(i)} \coloneqq \sum_{l=-1}^{r-1} \tilde{\mathcal{D}}_{n,l}  \mathbf{z}^{(i)}_l,
    \end{align}
    let us show that
    \begin{align}
        \mathbb{E}\left[\left\|\left(\tilde{\tilde{\mathbf{X}}}_{\rm histNB}^{(i)}\right)^{\otimes d}-\left(\mathbf{X}_{\rm hist}^{(i)}\right)^{\otimes d}\right\|\right]\le 6d\sqrt{(2d-3)!!}(1+3\epsilon)^{d-1}\epsilon(r+1)^{d/2}\left(\|\mathbf{x}_0\|^2+(N+2d)\sigma^2T\right)^{d/2}.
        \label{eq:Xhistddiff}
    \end{align}
    To see this, we begin with considering
    \begin{align}
    \tilde{\tilde{\mathbf{X}}}_{{\rm NB},n}^{(i)}-\mathbf{X}_n^{(i)} = \sum_{n=-1}^{r-1} \left(\tilde{\mathcal{D}}_{n,l}- \mathcal{D}_{n,l}\right)\mathbf{z}^{(i)}_{l}.
    \end{align}
    We have
    \begin{align}
        \left\|\tilde{\mathcal{D}}_{n,0}- \mathcal{D}_{n,0}\right\|\le\|\tilde{\mathcal{C}}-\mathcal{A}^{-1}\| \le \epsilon
    \end{align}
    and
    \begin{align}
        \left\|\tilde{\mathcal{D}}_{n,l}- \mathcal{D}_{n,l}\right\|&\le\|\tilde{\mathcal{C}}_{n,l+1}-(\mathcal{A}^{-1})_{n,l+1}\|\cdot \left\|\tilde{S}^{K,r,M}_l\right\|+\left\|(\mathcal{A}^{-1})_{n,l+1}\right\|\cdot\left\|\tilde{S}^{K,r,M}_l-S_l\right\| \nonumber \\
        & \le \left(\epsilon(1+\epsilon) + \epsilon\right)\sqrt{\sigma^2 \Delta t} \nonumber \\
        & \le 3\sqrt{\sigma^2 \Delta t}\epsilon,
    \end{align}
    where we use Eqs. \eqref{eq:calAInvBlockNorm}, \eqref{eq:CtilAInv}, \eqref{eq:SErrNorm}, and \eqref{eq:StilNorm}.
    Thus, by the discussion similar to the above, we get
    \begin{align}
        \mathbb{E}\left[\left\|\tilde{\tilde{\mathbf{X}}}_{\rm histNB}^{(i)}-\mathbf{X}_{\rm hist}^{(i)}\right\|^2\right]\le 18\epsilon^2(r+1)\left(\|\mathbf{x}_0\|^2+(N+2)\sigma^2T\right).
        \label{eq:XtiltilhistNBSqErr}
    \end{align}
    We also have
    \begin{align}
        & \mathbb{E}\left[\left\|\left(\mathbf{X}_{\rm hist}^{(i)}\right)^{\otimes k} \otimes \left(\tilde{\tilde{\mathbf{X}}}_{\rm histNB}^{(i)}\right)^{\otimes k^\prime}\right\|^2\right] \nonumber \\
        = & \sum_{n_1,\ldots,n_{k+k^\prime}=0}^r \sum_{l_1,\ldots,l_{2(k+k^\prime)}=-1}^{r-1} \mathbb{E}\left[\left(\mathcal{D}_{n_1,l_1}\mathbf{z}^{(i)}_{l_1}\right)^T \mathcal{D}_{n_1,l_2} \mathbf{z}^{(i)}_{l_2} \times \cdots \times \left(\mathcal{D}_{n_k,l_{2k-1}}\mathbf{z}^{(i)}_{l_{2k-1}}\right)^T \mathcal{D}_{n_k,l_{2k}} \mathbf{z}^{(i)}_{l_{2k}} \right. \nonumber \\
        &\qquad\qquad\qquad\qquad\times\left.\left(\tilde{\mathcal{D}}_{n_{k+1},l_{2k+1}}\mathbf{z}^{(i)}_{l_{2k+1}}\right)^T \tilde{\mathcal{D}}_{n_{k+1},l_{2k+2}} \mathbf{z}^{(i)}_{l_{2k+2}} \times \cdots \times \left(\tilde{\mathcal{D}}_{n_{k+k^\prime},l_{2(k+k^\prime)-1}}\mathbf{z}^{(i)}_{l_{2(k+k^\prime)-1}}\right)^T \tilde{\mathcal{D}}_{n_{k+k^\prime},l_{2(k+k^\prime)}} \mathbf{z}^{(i)}_{l_{2(k+k^\prime)}}\right] \nonumber \\
        \le & 2(2(k+k^\prime)-1)!!(1+3\epsilon)^{k+k^\prime}(r+1)^{k+k^\prime}\left(\|\mathbf{x}_0\|^2+(N+2d)\sigma^2T\right)^{k+k^\prime}
        \label{eq:XtiltilhistNBXhist}
    \end{align}
    for any $k,k^\prime\in\mathbb{N}$ such that $k+k^\prime \le d$.
    Then, by the triangle inequality and the Cauchy–Schwarz inequality, we get
    \begin{align}
        \mathbb{E}\left[\left\|\left(\tilde{\tilde{\mathbf{X}}}_{\rm histNB}^{(i)}\right)^{\otimes d}-\left(\mathbf{X}_{\rm hist}^{(i)}\right)^{\otimes d}\right\|\right]
        \le & \sum_{k=0}^{d-1} \mathbb{E}\left[\left\|\left(\tilde{\tilde{\mathbf{X}}}_{\rm histNB}^{(i)}\right)^{\otimes (d-k-1)} \otimes \left(\mathbf{X}_{\rm hist}^{(i)}\right)^{\otimes k}\otimes\left(\tilde{\tilde{\mathbf{X}}}_{\rm histNB}^{(i)}-\mathbf{X}_{\rm hist}^{(i)}\right)\right\|\right] \nonumber \\
        \le & \sum_{k=0}^{d-1} \sqrt{\mathbb{E}\left[\left\|\left(\tilde{\tilde{\mathbf{X}}}_{\rm histNB}^{(i)}\right)^{\otimes (d-k-1)} \otimes \left(\mathbf{X}_{\rm hist}^{(i)}\right)^{\otimes k} \right\|^2\right] \cdot \mathbb{E}\left[\left\|\left(\tilde{\tilde{\mathbf{X}}}_{\rm histNB}^{(i)}-\mathbf{X}_{\rm hist}^{(i)}\right)\right\|^2\right]},
    \end{align}
    and using Eqs.~\eqref{eq:XtiltilhistNBSqErr} and \eqref{eq:XtiltilhistNBXhist} in the last line yields Eq.~\eqref{eq:Xhistddiff}.

    We now define
    \begin{align}
        Y^{(i)} \coloneqq \sum_{j_1,\ldots,j_d=1}^{N(r+1)} C_{j_1,\ldots,j_d} X_{\rm hist}^{(i),j_1} \cdots X_{\rm hist}^{(i),j_d}, ~\tilde{\tilde{Y}}^{(i)}_{\rm NB} \coloneqq \sum_{j_1,\ldots,j_d=1}^{N(r+1)} C_{j_1,\ldots,j_d} \tilde{\tilde{X}}_{\rm histNB}^{(i),j_1} \cdots \tilde{\tilde{X}}_{\rm histNB}^{(i),j_d}
    \end{align}
    for each $i\in\mathbb{N}$.
    We see from Eq.~\eqref{eq:Xhistddiff} that
    \begin{align}
        \left|\mathbb{E}\left[\tilde{\tilde{Y}}^{(i)}_{\rm NB}\right]-\mathbb{E}\left[Y^{(i)}\right]\right|
        &\le \mathbb{E}\left[\left|\tilde{\tilde{Y}}^{(i)}_{\rm NB}-Y^{(i)}\right|\right] \nonumber \\
        & = \mathbb{E}\left[\left|\sum_{j_1,\ldots,j_d=1}^{N(r+1)} C_{j_1,\ldots,j_d} \left(\tilde{\tilde{X}}_{\rm histNB}^{(i),j_1} \cdots \tilde{\tilde{X}}_{\rm histNB}^{(i),j_d}-X_{\rm hist}^{(i),j_1} \cdots X_{\rm hist}^{(i),j_d}\right)\right|\right] \nonumber \\
        & \le  \mathbb{E}\left[\|C\|_F \left\|\left(\tilde{\tilde{\mathbf{X}}}_{\rm histNB}^{(i)}\right)^{\otimes d}-\left(\mathbf{X}_{\rm hist}^{(i)}\right)^{\otimes d}\right\|\right] \nonumber \\
        & \le 6d\sqrt{(2d-3)!!}(1+3\epsilon)^{d-1}\epsilon(r+1)^{d/2}\left(\|\mathbf{x}_0\|^2+(N+2d)\sigma^2T\right)^{d/2} \|C\|_F
        \label{eq:BiasYtiltilNBTemp}
    \end{align}
    holds.
    Defining the $d$-way tensor $\mathfrak{C}^{l_1,\ldots,l_d}=\left(\mathfrak{C}^{l_1,\ldots,l_d}_{k_1,\ldots,k_d}\right)_{k_1,\ldots,k_d=1\ldots,N}$ for $l_1,\ldots,l_d\in\{-1,\ldots,r-1\}$ by
    \begin{align}
        \mathfrak{C}^{l_1,\ldots,l_d}_{k_1,\ldots,k_d} = \sum_{j_1,\ldots,j_d=1}^{N(r+1)} C_{j_1,\ldots,j_d} \mathfrak{D}^{l_1}_{j_1,k_1} \cdots \mathfrak{D}^{l_d}_{j_d,k_d}
    \end{align}
    with $N(r+1) \times N$ matrices
    \begin{align}
        \mathfrak{D}^l = \left(\mathfrak{D}^l_{j,k}\right) \coloneqq
        \begin{pmatrix}
            \mathcal{D}_{0,l} \\ \vdots \\ \mathcal{D}_{r,l}
        \end{pmatrix},
        l\in\{-1,\ldots,r-1\},
    \end{align}
    and letting $\mathbf{Z}_l=(Z_{l,1},\ldots,Z_{l,N})^T,l\in[r-1]_0$ be independent standard normal vectors under $\mathbb{P}$, together with $\mathbf{Z}_{-1} \coloneqq \mathbf{x}_0$, we have
    \begin{align}
        \left|\mathbb{E}\left[Y^{(i)}\right]-\mathbb{E}_{\mathbb{P}}\left[Y\right]\right|
        &= \left|\sum_{l_1,\ldots,l_d=-1}^{r-1} \sum_{k_1,\ldots,k_d=1}^N \mathfrak{C}^{l_1,\ldots,l_d}_{k_1,\ldots,k_d} \left(\mathbb{E}\left[z^{(i)}_{l_1,k_1} \cdots z^{(i)}_{l_d,k_d}\right]-\mathbb{E}_\mathbb{P}\left[Z_{l_1,k_1} \cdots Z_{l_d,k_d}\right]\right)\right| \nonumber \\
        & \le \sum_{l_1,\ldots,l_d=-1}^{r-1} \left\|\mathfrak{C}^{l_1,\ldots,l_d}\right\|_{\rm F} \left\|\mathbb{E}\left[\mathbf{z}^{(i)}_{l_1} \otimes \cdots \otimes \mathbf{z}^{(i)}_{l_d}\right]-\mathbb{E}_\mathbb{P}[\mathbf{Z}_{l_1} \otimes \cdots \otimes \mathbf{Z}_{l_d}]\right\| \nonumber \\
        & \le \sum_{k=0}^d \binom{d}{k} r^{d-k} \|C\|_{\rm F} (r+1)^{d/2} \left(\sqrt{\sigma^2 \Delta t}\right)^{d-k}\|\mathbf{x}_0\|^k (N+2d)^{(d-k)/2}\epsilon_{\rm PRN} \nonumber \\
        & \le \|C\|_{\rm F}(r+1)^{d/2}\left(\|\mathbf{x}_0\| + \sqrt{(N+2d)r\sigma^2T}\right)^d\epsilon_{\rm PRN} \nonumber \\
        & \le 2^{d/2}\|C\|_{\rm F}(r+1)^{d/2}r^{d/2}\left(\|\mathbf{x}_0\|^2 + (N+2d)\sigma^2T\right)^{d/2}\epsilon_{\rm PRN} \nonumber \\
        & \le d\sqrt{(2d-3)!!}\|C\|_{\rm F}(r+1)^{d/2}\left(\|\mathbf{x}_0\|^2 + (N+2d)\sigma^2T\right)^{d/2}\epsilon,
        \label{eq:BiasYtiltilNBTemp2}
    \end{align}
    where we use Eqs.~\eqref{eq:PRNExpSNDiff} and \eqref{eq:epsPRN2}, and
    \begin{align}
        \left\|\mathfrak{C}^{l_1,\ldots,l_d}\right\|_{\rm F} \le \|C\|_F \prod_{k=1}^d \|\mathfrak{D}^{l_k}\| \le \|C\|_F \prod_{k=1}^d \sqrt{\sum_{n=0}^r \|\mathcal{D}_{n,l_k}\|^2} \le \|C\|_F (r+1)^{d/2} \prod_{k=1}^d \max_{n\in[r]_0}\|\mathcal{D}_{n,l_k}\|,
        \label{eq:tilCFrob}
    \end{align}
    along with Eq.~\eqref{eq:calDNorm}.
    Combining Eqs.~\eqref{eq:BiasYtiltilNBTemp} and \eqref{eq:BiasYtiltilNBTemp2} yields
    \begin{align}
        \left|\mathbb{E}\left[\tilde{\tilde{Y}}^{(i)}_{\rm NB}\right]-\mathbb{E}_\mathbb{P}\left[Y\right]\right|
        &\le 7d\sqrt{(2d-3)!!}(1+3\epsilon)^{d-1}\epsilon(r+1)^{d/2}\left(\|\mathbf{x}_0\|^2+(N+2d)\sigma^2T\right)^{d/2} \|C\|_F
        \label{eq:BiasYtiltilNB}
    \end{align}
    Besides, we have
    \begin{align}
        \mathbb{E}\left[\left(\tilde{\tilde{Y}}^{(i)}_{\rm NB}\right)^2\right]
        &\le  \mathbb{E}\left[\|C\|_F^2 \left\|\left(\tilde{\tilde{\mathbf{X}}}_{\rm histNB}^{(i)}\right)^{\otimes d}\right\|^2\right] \nonumber \\
        & = 2(2d-1)!!(1+3\epsilon)^{2d}(r+1)^d\left(\|\mathbf{x}_0\|^2+(N+2d)\sigma^2T\right)^d\|C\|_F^2,
        \label{eq:VarYtiltilNB}
    \end{align}
    where we use
    \begin{align}
        \mathbb{E}\left[\left\|\left(\tilde{\tilde{\mathbf{X}}}_{\rm histNB}^{(i)}\right)^{\otimes d}\right\|^2\right]\le 2(2d-1)!!(1+3\epsilon)^{2d}(r+1)^d\left(\|\mathbf{x}_0\|^2+(N+2d)\sigma^2T\right)^d,
        \label{eq:EXhistNBtiltild2}
    \end{align}
    which is obtained in a similar way to Eq.~\eqref{eq:EXhisttiltild2}.
    Furthermore, for $i \ne i^\prime$, we have 
    \begin{align}
        &\left|\mathrm{Cov}\left(\tilde{\tilde{Y}}^{(i)}_{\rm NB},\tilde{\tilde{Y}}^{(i^\prime)}_{\rm NB}\right)\right| \nonumber \\
        \le & \left|\sum_{l_1,\ldots,l_{2d}=-1}^{r-1} \tilde{\mathfrak{C}}^{l_1,\ldots,l_d} \otimes \tilde{\mathfrak{C}}^{l_{d+1},\ldots,l_{2d}} \left(\mathbb{E}\left[\mathbf{z}^{(i)}_{l_1} \otimes \cdots \otimes \mathbf{z}^{(i)}_{l_d} \otimes \mathbf{z}^{(i^\prime)}_{l_{d+1}} \cdots \mathbf{z}^{(i^\prime)}_{l_{2d}}\right]-\mathbb{E}\left[\mathbf{z}^{(i)}_{l_1} \otimes \cdots \otimes \mathbf{z}^{(i)}_{l_d}\right] \otimes \mathbb{E}\left[\mathbf{z}^{(i^\prime)}_{l_{d+1}} \cdots \mathbf{z}^{(i^\prime)}_{l_{2d}}\right]\right)\right|,
        \label{eq:tiltilYiCov}
    \end{align}
    where the $d$-way tensor $\tilde{\mathfrak{C}}^{l_1,\ldots,l_d}=\left(\tilde{\mathfrak{C}}^{l_1,\ldots,l_d}_{k_1,\ldots,k_d}\right)_{k_1,\ldots,k_d=1\ldots,N}$ for $l_1,\ldots,l_d\in\{-1,\ldots,r-1\}$ is given by
    \begin{align}
        \tilde{\mathfrak{C}}^{l_1,\ldots,l_d}_{k_1,\ldots,k_d} = \sum_{j_1,\ldots,j_d=1}^{N(r+1)} C_{j_1,\ldots,j_d} \tilde{\mathfrak{D}}^{l_1}_{j_1,k_1} \cdots \tilde{\mathfrak{D}}^{l_d}_{j_d,k_d}
    \end{align}
    with $N(r+1) \times N$ matrices
    \begin{align}
        \tilde{\mathfrak{D}}^l = \left(\tilde{\mathfrak{D}}^l_{j,k}\right) \coloneqq
        \begin{pmatrix}
            \tilde{\mathcal{D}}_{0,l} \\ \vdots \\ \tilde{\mathcal{D}}_{r,l}
        \end{pmatrix},
        l\in\{-1,\ldots,r-1\}.
    \end{align}
    Noting that replacing $\mathbf{z}^{(i)}_0,\ldots,\mathbf{z}^{(i)}_{r-1},\mathbf{z}^{(i^\prime)}_0,\ldots,\mathbf{z}^{(i^\prime)}_{r-1}$ with independent standard normal vectors in Eq.~\eqref{eq:tiltilYiCov} makes it zero, we have
    \begin{align}
        \left|\mathrm{Cov}\left(\tilde{\tilde{Y}}^{(i)}_{\rm NB},\tilde{\tilde{Y}}^{(i^\prime)}_{\rm NB}\right)\right|
        & \le  \sum_{l_1,\ldots,l_{2d}=-1}^{r-1} \left\|\tilde{\mathfrak{C}}^{l_1,\ldots,l_d}\right\|_{\rm F}^2 \left\| \mathbb{E}\left[\mathbf{z}^{(i)}_{l_1} \otimes \cdots \otimes \mathbf{z}^{(i)}_{l_d} \otimes \mathbf{z}^{(i^\prime)}_{l_{d+1}} \otimes \cdots \otimes \mathbf{z}^{(i^\prime)}_{l_{2d}}\right]-\mathbb{E}\left[\mathbf{z}^{(i)}_{l_1} \otimes \cdots \otimes \mathbf{z}^{(i)}_{l_d}\right] \mathbb{E}\left[\mathbf{z}^{(i^\prime)}_{l_{d+1}} \otimes \cdots \otimes \mathbf{z}^{(i^\prime)}_{l_{2d}}\right]\right\| \nonumber \\
        & \le \sum_{k=0}^{2d} \binom{2d}{k} r^{2d-k} \|C\|_{\rm F}^2 (r+1)^d (1+\epsilon)^k\left((1+3\epsilon)\sqrt{\sigma^2 \Delta t}\right)^{2d-k}2\|\mathbf{x}_0\|^k (N+2d)^{(2d-k)/2}\epsilon_{\rm PRN} \nonumber \\
        & \le 2\|C\|_{\rm F}^2(r+1)^d(1+3\epsilon)^{2d}\left(\|\mathbf{x}_0\| + \sqrt{(N+2d)r\sigma^2T}\right)^{2d}\epsilon_{\rm PRN} \nonumber \\
        & \le 2^{d+1}\|C\|_{\rm F}^2(r+1)^d(1+3\epsilon)^{2d}r^d\left(\|\mathbf{x}_0\|^2 + (N+2d)\sigma^2T\right)^{d}\epsilon_{\rm PRN}  \nonumber \\
        & \le \frac{1}{4}d^2(2d-3)!!\|C\|_{\rm F}^2(1+3\epsilon)^{2d}(r+1)^d\left(\|\mathbf{x}_0\|^2 + (N+2d)\sigma^2T\right)^{d}\delta\epsilon^2,
        \label{eq:tiltilYiCov2}
    \end{align}
    where we use Eq.~\eqref{eq:PRNExpSNDiff} and $\left\|\tilde{\mathfrak{C}}^{l_1,\ldots,l_d}\right\|_{\rm F} \le \|C\|_F (r+1)^{d/2} \prod_{k=1}^d \max_{n\in[r]_0}\|\tilde{\mathcal{D}}_{n,l_k}\|$ obtained similarly to Eq.~\eqref{eq:tilCFrob}.
    By combining Eqs.~\eqref{eq:BiasYtiltilNB}, \eqref{eq:VarYtiltilNB}, and \eqref{eq:tiltilYiCov2}, we see that for 
    \begin{align}
        \hat{Y}_{\rm NB} \coloneqq \frac{1}{N_{\rm s}}\sum_{i=1}^{N_{\rm s}}\tilde{\tilde{Y}}^{(i)}_{\rm NB}, 
    \end{align}
    \begin{align}
        \left|\mathbb{E}[\hat{Y}_{\rm NB}]-\mathbb{E}_\mathbb{P}[Y]\right| \le 7d\sqrt{(2d-3)!!}(1+3\epsilon)^d\epsilon(r+1)^{d/2}\left(\|\mathbf{x}_0\|^2+(N+2d)\sigma^2T\right)^{d/2} \|C\|_F
    \end{align}
    and
    \begin{align}
        {\rm Var}[\hat{Y}_{\rm NB}] \le \left(\frac{2(2d-1)!!}{N_{\rm s}} + \frac{1}{4}d^2(2d-3)!!\delta\epsilon^2\right)(1+3\epsilon)^{2d}(r+1)^d\left(\|\mathbf{x}_0\|^2+(N+2d)\sigma^2T\right)^d\|C\|_F^2
        \label{eq:VarYhatNB}
    \end{align}
    holds.
    Thus, using Chebyshev's inequality and doing some algebra, we see that
    \begin{align}
        \left|\hat{Y}_{\rm NB}-\mathbb{E}_\mathbb{P}[Y]\right| & \le \left|\hat{Y}_{\rm NB}-\mathbb{E}[\hat{Y}_{\rm NB}]\right|+\left|\mathbb{E}[\hat{Y}_{\rm NB}]-\mathbb{E}_\mathbb{P}[Y]\right| \nonumber \\
        & \le 8d\sqrt{(2d-3)!!}(1+3\epsilon)^d(r+1)^{d/2}\left(\|\mathbf{x}_0\|^2+(N+2d)\sigma^2T\right)^{d/2} \|C\|_F \left(\epsilon + \sqrt{\frac{4}{\delta N_{\rm s}}}\right)
    \end{align}
    holds with probability at least $1-\frac{\delta}{4}$.

    We also note that, due to (i) of Assumption~\ref{ass:PRN}, the probability that the absolute value of each entry in the sequence $\{z_i\}$  exceeds the current $U_{\rm SN}$ is less than
    \begin{align}
        \frac{2}{\sqrt{2\pi}} e^{-U_{\rm SN}^2/2} \le \frac{\delta}{4 r_{\rm c} N N_{\rm s}}.
    \end{align}
    Thus, with probability at least $\left(1-\frac{\delta}{4 r_{\rm c} N N_{\rm s}}\right)^{r_{\rm c} N N_{\rm s}} \ge 1-\frac{\delta}{4}$, $|z_1|,\ldots,|z_{r_{\rm c} NN_{\rm s}}| \le U_{\rm SN}$ holds and thus $\hat{Y}_{\rm NB}=\hat{Y}$ holds.
    
    Combining the above findings, we see that the claim holds.

\end{proof}

Finally, we present a quantum algorithm for estimating expectations of multi-time functions.
The following is the theorem on the complexity of the algorithm, and the procedure of the algorithm is presented in the proof.

\begin{theorem}
    Let $d\in\mathbb{N}$, and $\delta\in(0,1)$.
    Suppose that Assumptions \ref{ass:UA}, \ref{ass:muA}, \ref{ass:UBB}, and \ref{ass:Ux0} hold.
    Set $r$ to $r_{\rm c}$ in Eq. \eqref{eq:rReq}.
    Let $\epsilon$ be a positive real number such that
    \begin{align}
        \epsilon^\prime \coloneqq \frac{\epsilon}{2^{d+5}d\sqrt{(2d-3)!!}(r_{\rm c}+1)^{d/2}\left(\|\mathbf{x}_0\|^2+(N+2d)\sigma^2T\right)^{d/2} \|C\|_F}
        \label{eq:epsPr}
    \end{align}
    satisfies $\epsilon^\prime \le \frac{1}{3}$.
    Suppose that Assumption \ref{ass:PRN} holds with
    \begin{align}
        \epsilon_{\rm PRN} = \min\left\{\frac{(2d-1)!!}{(2r\sqrt{N})^d},\frac{d\sqrt{(2d-3)!!}\epsilon^\prime}{(2r)^{d/2}},\frac{d^2(2d-3)!!\delta(\epsilon^\prime)^2}{8(2r)^d}\right\}.
    \end{align}
    Assume that for a $d$-way real tensor $C=(C_{j_1,\ldots,j_d})\in\mathbb{R}^{N(r+1) \times \cdots \times N(r+1)}$, we have access to the oracle $U_C$ that acts on a $d$-register system as
    \begin{align}
        U_C \ket{0}^{\otimes d} = \frac{1}{\|C\|_F} \sum_{i_1,\ldots,i_d=1}^{N(r+1)} C_{i_1,\ldots,i_d} \ket{i_1}\cdots\ket{i_d}.
    \end{align}
    Then, there is a quantum algorithm that, with probability at least $1-\delta$, outputs an estimate $\hat{\mu}_Y$ of $\mathbb{E}_\mathbb{P}[Y]$ with error at most $\epsilon$.
    This algorithm uses the controlled $U_A$
    \begin{align}
    O\left(\left(\frac{32\kappa_{BB^T} \alpha_A T}{\max\{1,\eta T\}}\sqrt{\|\mathbf{x}_0\|^2+N\sigma^2T \log \left(\frac{r_{\rm c} N}{\delta^2 \epsilon^{\prime2}}\right)}\right)^d\frac{d\kappa_{BB^T} \alpha_A T \|C\|_F K_{\rm c}^\prime}{\max\{1,\eta T\}\epsilon}\log\left(\frac{1}{\delta}\right)\log\left(\frac{\kappa_{BB^T}^2\alpha_A^2}{\eta^2\epsilon^\prime}\right)\right)
    \label{eq:CompExaUA}
    \end{align}
    times, the controlled $U_{BB^T}$
    \begin{align}
    O\left(\left(\frac{32\kappa_{BB^T} \alpha_A T}{\max\{1,\eta T\}}\sqrt{\|\mathbf{x}_0\|^2+N\sigma^2T \log \left(\frac{r_{\rm c} N}{\delta^2 \epsilon^{\prime2}}\right)}\right)^d\frac{d\kappa_{BB^T} \alpha_A T \|C\|_F }{\max\{1,\eta T\}\epsilon}\log\left(\frac{1}{\delta}\right)\log\left(\frac{\kappa_{BB^T}^2\alpha_A^2}{\eta^2\epsilon^\prime}\right)\right)
    \label{eq:CompExaUBB}
    \end{align}
    times, $U_C$ and the controlled $U_{\mathbf{x}_0}$ and $U_{\rm Rand}$
    \begin{align}
    O\left(\left(\frac{32\kappa_{BB^T} \alpha_A T}{\max\{1,\eta T\}}\sqrt{\|\mathbf{x}_0\|^2+N\sigma^2T \log \left(\frac{r_{\rm c} N}{\delta^2 \epsilon^{\prime2}}\right)}\right)^d\frac{d\|C\|_F}{\epsilon}\log\left(\frac{1}{\delta}\right)\right)
    \end{align}
    times each, and
    \begin{align}
    O\left(\left(\frac{32\kappa_{BB^T} \alpha_A T}{\max\{1,\eta T\}}\sqrt{\|\mathbf{x}_0\|^2+N\sigma^2T \log \left(\frac{r_{\rm c} N}{\delta^2 \epsilon^{\prime2}}\right)}\right)^d\frac{d\kappa_{BB^T} \alpha_A T \|C\|_F K_{\rm c}^\prime}{\max\{1,\eta T\}\epsilon}\left(\log K_{\rm c}^\prime\log M_{\rm c}^\prime+a_A+a_{BB^T}\right)\log\left(\frac{1}{\delta}\right)\log\left(\frac{\kappa_{BB^T}^2\alpha_A^2}{\eta^2\epsilon^\prime}\right)\right)
    \end{align}
    additional elementary gates.
    Here,
    \begin{align}
    K_{\rm c}^\prime&\coloneqq\max\left\{\left\lceil\log\left(\frac{1152\pi\left(\lg\left(4\kappa_{BB^T}\log\left(\frac{8}{\varepsilon^\prime}\right)\right)+1\right)^2\log\left(\frac{8}{\varepsilon^\prime}\right)}{\varepsilon^\prime}\right)\right\rceil,7\right\}, \nonumber \\
    M_{\rm c}^\prime & \coloneqq \left\lceil\max\left\{\frac{24\alpha_{\drm A/\drm t}}{\alpha_A^2},\frac{2\left(2\alpha_A + \frac{\alpha_{\drm BB^T/\drm t}}{\sigma^2}\right)}{\alpha_A}\right\} \times \frac{64\pi \left(\lg\left(4\kappa_{BB^T}\log\left(\frac{8}{\varepsilon^\prime}\right)\right)+1\right)^2\log\left(\frac{8}{\varepsilon^\prime}\right)}{\varepsilon^\prime}\right\rceil \nonumber \\
    \varepsilon^\prime & \coloneqq \min\left\{\frac{\eta^2 T^2}{8\sqrt{3}\left\lceil 4\kappa_{BB^T}\alpha_A T \right\rceil^2},\sqrt{\frac{\|\mathbf{x}_0\|^2+4N\sigma^2TU_{\rm SN}^2}{48N\sigma^2TU_{\rm SN}^2}}\frac{\eta T}{\left\lceil 4\kappa_{BB^T}\alpha_A T\right\rceil}\right\}\times\epsilon^\prime.
    \label{eq:rcKcprime}
    \end{align}
    \label{th:QAOE}
\end{theorem}

\begin{proof}

The quantum algorithm is as shown in Algorithm \ref{alg:main}.

\begin{algorithm}[t]
\caption{Proposed quantum algorithm for estimating $\mathbb{E}_\mathbb{P}[Y]$}\label{alg:main}
\begin{algorithmic}[1]

\State Set $K=K_{\rm c}^\prime$, $r=r_{\rm c}$, and $M=M_{\rm c}^\prime$.

\State Set
\begin{align}
    N_{\rm s} = \left\lceil \frac{4}{\delta \epsilon^{\prime 2}} \right\rceil.
    \label{eq:Ns}
\end{align}

\State Set $U_{\rm SN}$ a
\begin{align}
    U_{\rm SN} = \max\left\{\sqrt{2 \log \left(\frac{8r N N_{\rm s}}{\sqrt{2\pi}\delta}\right)},1\right\}.
    \label{eq:USNdelPr}
\end{align}

\State Construct the oracle $V_{{\rm hist},\epsilon^\prime}$.

\State By a use of $U^{\rm SP}_{1,N_{\rm s}}$ and $d$ uses of $V_{{\rm hist},\epsilon^\prime}$, construct an oracle $V_{\rm histSP}$ that acts as
\begin{align}
    V_{\rm histSP}\ket{0}\ket{0}\ket{0}^{\otimes d} = \ket{\rm histSP} \coloneqq \frac{1}{\sqrt{N_{\rm s}}}\sum_{i=1}^{N_{\rm s}}  \ket{i}\left(\ket{0}\Ket{\tilde{\tilde{\mathbf{X}}}_{\rm hist}^{(i)}}^{\otimes d} + \ket{0_\perp}\right),
    \label{eq:VHisSup}
\end{align}
where $\ket{0_\perp}$ is some unnormalized state such that $(I \otimes \ket{0}\!\bra{0} \otimes I)\ket{0_\perp}=0$.

\State By using $U^{\rm SP}_{1,N_{\rm s}}$ and $U_C$ once each, construct an oracle $U_C^{\rm SP}$ that acts as
\begin{align}
    U_C^{\rm SP}\ket{0}\ket{0}\ket{0}^{\otimes d}  = \ket{C{\rm SP}} \coloneqq  \frac{1}{\sqrt{N_{\rm s}}}\sum_{i=1}^{N_{\rm s}}  \ket{i}\ket{0}\ket{C}.
    \label{eq:UCSup}
\end{align}

\State By the algorithm in Theorem \ref{th:OE}, get an estimate $\hat{\mu}^\prime$ of $\braket{\mathrm{histSP} | C{\rm SP}}$ with accuracy
\begin{align}
    \epsilon_{\rm OE}=\left(\frac{\max\{1,\eta T\}}{8r_{\rm c} U_{\tilde{\mathcal{B}}}}\right)^d\frac{\epsilon}{2\|C\|_F}
\end{align}
and success probability $1-\frac{\delta}{2}$.

\State Output
\begin{align}
    \hat{\mu}_Y = 
    \left(\frac{8r_{\rm c} U_{\tilde{\mathcal{B}}}}{\max\{1,\eta T\}}\right)^d\|C\|_F \hat{\mu}^\prime.
\end{align}

\end{algorithmic}
\end{algorithm}

Let us prove the accuracy of this algorithm.
We note that
\begin{align}
    \left(\frac{8r U_{\tilde{\mathcal{B}}}}{\max\{1,\eta T\}}\right)^d\|C\|_F\braket{\mathrm{histSP} | C{\rm SP}} = \hat{Y}.
\end{align}
Lemma \ref{lem:YHat} implies that in the current setting, this satisfies
\begin{align}
    \left|\hat{Y}-\mathbb{E}_\mathbb{P}[Y]\right|&\le 8d\sqrt{(2d-3)!!}(1+3\epsilon^\prime)^d(r+1)^{d/2}\left(\|\mathbf{x}_0\|^2+(N+2d)\sigma^2T\right)^{d/2} \|C\|_F \left(\epsilon^\prime + \sqrt{\frac{4}{\delta N_{\rm s}}}\right) \nonumber \\
    &\le \frac{\epsilon}{2}
\end{align}
with probability at least $1-\frac{\delta}{2}$.
Theorem \ref{th:OE} implies that with probability at least $1-\frac{\delta}{2}$, $\hat{\mu}^\prime$ satisfies
\begin{align}
    \left|\hat{\mu}^\prime-\braket{\mathrm{histSP} | C{\rm SP}}\right|\le \epsilon_{\rm OE},
\end{align}
and thus
\begin{align}
    \left|\hat{\mu}_Y - \hat{Y}\right| \le \left(\frac{8r U_{\tilde{\mathcal{B}}}}{\max\{1,\eta T\}}\right)^d\|C\|_F\epsilon_{\rm OE} \le \frac{\epsilon}{2}
\end{align}
holds.
Combining these, we see that $\hat{\mu}_Y$ satisfies $\left|\hat{\mu}_Y - \mathbb{E}_\mathbb{P}[Y]\right|\le \epsilon$ with probability at least $1-\delta$.

Lastly, let us evaluate the number of uses of various oracles and additional elementary gates.
The overlap estimation in Step 7 uses $V_{\rm histSP}$ and $U^{\rm SP}_C$ $O\left(\frac{1}{\epsilon_{\rm OE}}\log\left(\frac{1}{\delta}\right)\right)$ times each, thus makes $O\left(\frac{1}{\epsilon_{\rm OE}}\log\left(\frac{1}{\delta}\right)\right)$ uses $V_{{\rm hist},\epsilon}$ and $U_C$.
According to Corollary \ref{cor:HistStateConcSuper}, $V_{{\rm hist},\epsilon}$ uses the controlled $U_A$
\begin{align}
    O\left(\frac{\kappa_{BB^T} \alpha_A T K_{\rm c}^\prime}{\max\{1,\eta T\}}\log\left(\frac{\kappa_{BB^T}\alpha_A T}{\max\{1,\eta T\}\epsilon^\prime}\right)\right)
\end{align}
times, the controlled $U_{BB^T}$
\begin{align}
    O\left(\frac{\kappa_{BB^T} \alpha_A T }{\max\{1,\eta T\}}\log\left(\frac{\kappa_{BB^T}\alpha_A T}{\max\{1,\eta T\}\epsilon^\prime}\right)\right)
\end{align}
times, the controlled $U_{\mathbf{x}_0}$ and $U_{\rm Rand}$ once each, and 
\begin{align}
    O\left(\frac{\kappa_{BB^T} \alpha_A T K_{\rm c}^\prime}{\max\{1,\eta T\}}\left(\log K_{\rm c}^\prime\log M_{\rm c}^\prime+a_A+a_{BB^T}\right)\log\left(\frac{\kappa_{BB^T}\alpha_A T}{\max\{1,\eta T\}\epsilon^\prime}\right)\right)
\end{align}
additional elementary gates.
By combining these evaluations and doing some algebra, we obtain the claimed evaluations.
\end{proof}

\begin{remark}
Although the complexity evaluations in Theorem \ref{th:QAOE} are rather complicated, hiding logarithmic factors simplifies them.
The number of queries to $U_A$ in Eq. \eqref{eq:CompExaUA}, which is largest among the oracles used in Algorithm \ref{alg:main}, is simplified to
\begin{align}
    \tilde{O}\left(\left(\frac{\kappa_{BB^T} \alpha_A T}{\max\{1,\eta T\}}\right)^{d+1}\left(\|\mathbf{x}_0\|^2+N\sigma^2T\right)^{d/2}\frac{\|C\|_F}{\epsilon}\right),
    \label{eq:CompExaUASimp}
\end{align}
where we assume that $d$ is a few and consider that $O(1)$ quantities to the power of $d$ are also $O(1)$ and can be hidden.
This expression has a ${\rm poly}(N)$ factor, raising the concern that the proposed method may have a ${\rm poly}(N)$ complexity.
Here, regarding the root mean square $\sqrt{\mathbb{E}_\mathbb{P}[Y^2]}$ as the typical magnitude of $Y$, we accordingly set the accuracy $\epsilon$ relative to it.
Since $\sqrt{\mathbb{E}_\mathbb{P}[Y^2]}$ can be estimated as $\sqrt{\mathbb{E}_\mathbb{P}[Y^2]} \sim r_{\rm c}^{d/2}\left(\|\mathbf{x}_0\|^2+N\sigma^2T\right)^{d/2}\|C\|_F$ like $\sqrt{\mathbb{E}\left[\left(\tilde{\tilde{Y}}^{(i)}_{\rm NB}\right)^2\right]}$, we have
\begin{align}
    \sqrt{\mathbb{E}_\mathbb{P}[Y^2]} \lesssim r_{\rm c}^{d/2}\left(\|\mathbf{x}_0\|^2+N\sigma^2T\right)^{d/2}\|C\|_F.
    \label{eq:RMSYMT}
\end{align}
If we set an accuracy relative to this, namely
\begin{align}
\epsilon=\left(r_{\rm c}\left(\|\mathbf{x}_0\|^2+N\sigma^2T\right)\right)^{d/2}\|C\|_F\epsilon_{\rm rel}    
\end{align}
with $\epsilon_{\rm rel} \in (0,1)$, Eq. \eqref{eq:CompExaUASimp} becomes
\begin{align}
    \tilde{O}\left(\left(\kappa_{BB^T} \alpha_A T\right)^{\frac{d}{2}+1}\frac{1}{\epsilon_{\rm rel}}\right),
    \label{eq:CompExaUASimp2}
\end{align}
where we also assume that $\eta T \lesssim 1$.
Eq. \eqref{eq:CompExaUASimp2} has no ${\rm poly}(N)$ factor.
Of course, $\mathbb{E}_\mathbb{P}[Y]$ can be accidentally very close to zero compared to Eq.~\eqref{eq:RMSYMT}, and if we nonetheless require such high accuracy that the small value of $\mathbb{E}_\mathbb{P}[Y]$ can be resolved, estimating $\mathbb{E}_\mathbb{P}[Y]$ by the proposed algorithm may incur an $O({\rm poly}(N))$ computational cost.
Indeed, such an estimation is also costly in classical Monte Carlo methods, incurring an $O({\rm poly}(N))$ cost.
\label{rem:CompExaMT}
\end{remark}

\subsection{Estimating expectations of terminal-time functions}

Although expectations of multi-time functions are often of interest, we also want expectations of terminal-time functions to the same extent or even more.
Here, by terminal-time functions, we mean quantities in the form of $Y\coloneqq \sum_{j_1,\ldots,j_d}C_{j_1,\ldots,j_d}X_r^{j_1} \cdots X_r^{j_d}$, which involve only the value of $\mathbf{X}_t$ at the terminal time $T$.
To estimate $\mathbb{E}_\mathbb{P}\left[Y\right]$, we go through preparing the quantum state encoding not $\mathbf{X}_{\rm hist}$ but $\mathbf{X}_{\rm hist+}$ in Eq. \eqref{eq:calAXHistB+Def}.
Roughly speaking, since $\mathbf{X}_r$ repeatedly appears in $\mathbf{X}_{\rm hist+}$, using the state encoding $\mathbf{X}_{\rm hist+}$ makes estimating $\mathbb{E}_\mathbb{P}\left[Y\right]$ more efficient.

Although intermediate results and their proofs are similar to those for the estimation of expectations of multi-time functions, we present them for completeness.
We begin with presenting the following theorem, which corresponds to Theorem \ref{th:HistState}.

\begin{theorem}
    Suppose that Assumption \ref{ass:muA} holds.
    Let $\epsilon>0$.
    Suppose that we have access to the oracles $V_{\tilde{\mathcal{A}}_+}$ and $V_{\tilde{\mathcal{B}}^{(1)}_+}$, which satisfies the following with some $U_{\tilde{\mathcal{B}}_+}\in\mathbb{R}_+$:
    \begin{itemize}
        \item $V_{\tilde{\mathcal{A}}_+}$: a $(\alpha_{\tilde{\mathcal{A}}_+},a_{\tilde{\mathcal{A}}_+},0)$-block-encoding of $\tilde{\mathcal{A}_+}$, where $\alpha_{\tilde{\mathcal{A}}_+} \ge \|\tilde{\mathcal{A}}_+\|$, and $a_{\tilde{\mathcal{A}}_+}\in\mathbb{N}$. 
        Here, $\tilde{\mathcal{A}}_+$ is defined by Eq.~\eqref{eq:calAtil+} with matrices $\hat{\Phi}_n,n\in[r-1]_0$ satisfying
        \begin{align}
            \left\|\hat{\Phi}_n - \Phi_n\right\| \le \min\left\{\frac{1}{2}\eta\Delta t e^{-\eta\Delta t},\epsilon_{1+}\right\},~\epsilon_{1+} \coloneqq \frac{\eta \Delta t U_{\tilde{\mathcal{B}}_+} \epsilon}{8\left( \|\mathbf{x}_0\| + \frac{8\Delta_{\rm max}^{(1)}}{\eta \Delta t}\right)}.
            \label{eq:RnPhinDiff+}
        \end{align}

        \item $V_{\tilde{\mathcal{B}}^{(1)}_+}$, which acts on a two-register system as
        \begin{align}
            V_{\tilde{\mathcal{B}}^{(1)}_+}\ket{0}\ket{0}=\ket{0}\ket{\tilde{\mathcal{B}}^{(1)}_+} + \ket{0_\perp}.
            \label{eq:VtilB+}
        \end{align}
        Here, $\ket{0_\perp}$ is some unnormalized state on the first register such that $(\ket{0}\!\bra{0} \otimes I)\ket{0_\perp}=0$, and an unnormalized state $\ket{\tilde{\mathcal{B}}^{(1)}_+}$ is given by
        \begin{align}
            \ket{\tilde{\mathcal{B}}_+^{(1)}} \coloneqq \frac{1}{U_{\tilde{\mathcal{B}}_+}}\tilde{\mathcal{B}}_+^{(1)},
            \tilde{\mathcal{B}}_+^{(1)} \coloneqq 
            \begin{tikzpicture}[baseline=(m.center)]
            \node (m) {
            $
            \begin{pmatrix}
            \mathbf{x}_0 \\ \tilde{\boldsymbol{\Delta}}^{(1)}_{+,0} \\ \vdots \\ \tilde{\boldsymbol{\Delta}}^{(1)}_{+,r-1} \\ \mathbf{0}  \\ \vdots \\ \mathbf{0}
            \end{pmatrix}
            $
            };
            \draw[decorate,decoration={brace,amplitude=5pt}]
              (0.6,-0.4) -- (0.6,-1.7)
              node[midway,right=3pt] {$\text{$R$}$};
            \end{tikzpicture}
            ,
            \label{eq:ketBTil+}
        \end{align}
        where for $n\in[r-1]_0$, $\tilde{\boldsymbol{\Delta}}^{(1)}_{+,n}\in\mathbb{R}^N$ satisfies
        \begin{align}
            \left\|\tilde{\boldsymbol{\Delta}}^{(1)}_{+,n} - \boldsymbol{\Delta}^{(1)}_n\right\| \le \epsilon_{2+},~\epsilon_{2+} \coloneqq \frac{\eta \Delta t U_{\tilde{\mathcal{B}}_+} \epsilon}{16}.
            \label{eq:condTilDel+}
        \end{align}
    \end{itemize}
    Then, we have access to an oracle $V^{(1)}_{{\rm hist+},\epsilon}$, which acts as
    \begin{align}
        V^{(1)}_{{\rm hist+},\epsilon}\ket{0}\ket{0} =   \ket{0}\ket{\tilde{\tilde{\mathbf{X}}}_{\rm hist+}^{(1)}} + \ket{0_\perp}.
        \label{eq:V1hist+}
    \end{align}
    Here, $\ket{\tilde{\tilde{\mathbf{X}}}_{\rm hist+}^{(1)}}$ is an unnormalized state given by
    \begin{align}
        \ket{\tilde{\tilde{\mathbf{X}}}_{\rm hist+}^{(1)}} \coloneqq \frac{\alpha_{\tilde{\mathcal{A}}_+}}{2\kappa_{\tilde{\mathcal{A}}_+}U_{\tilde{\mathcal{B}}_+}}\tilde{\tilde{\mathbf{X}}}_{\rm hist+}^{(1)},
        \tilde{\tilde{\mathbf{X}}}_{\rm hist+}^{(1)}=\begin{pmatrix}\tilde{\tilde{\mathbf{X}}}_{+,0}^{(1)} \\ \vdots \\ \tilde{\tilde{\mathbf{X}}}_{+,r+R}^{(1)}\end{pmatrix}
        \label{eq:XHistKet+}
    \end{align}
    with
    $
        \kappa_{\tilde{\mathcal{A}}_+} \coloneqq \alpha_{\tilde{\mathcal{A}}_+}\left(\frac{4r}{\max\{1,\eta T\}}+R\right)
    $
    and $\tilde{\tilde{\mathbf{X}}}_{+,n}^{(1)}\in\mathbb{R}^N,n\in[r+R]_0$ satisfying
    \begin{align}
        \left\| \tilde{\tilde{\mathbf{X}}}_{+,n}^{(1)} - 
        \mathbf{X}_r^{(1)}\right\| \le U_{\tilde{\mathcal{B}}_+} \epsilon, n\in\{r,\ldots,r+R\},
        \label{eq:XrErr}
    \end{align}
    where $\mathbf{X}_r^{(1)}$ is defined in Theorem \ref{th:HistState}.
    $V^{(1)}_{{\rm hist+},\epsilon}$ uses $V_{\tilde{\mathcal{A}}_+}$ 
    \begin{align}
        O\left(\alpha_{\tilde{\mathcal{A}}_+}\left(\frac{r}{\max\{1,\eta T\}}+R\right)\log\left(\frac{1}{\epsilon}\left(\frac{r}{\max\{1,\eta T\}}+R\right)\right)\right)
        \label{eq:CompHistGen+}
    \end{align}
    times, $V_{\tilde{\mathcal{B}}^{(1)}}$ once, and
    \begin{align}
        O\left(a_{\tilde{\mathcal{A}}_+}\alpha_{\tilde{\mathcal{A}}_+}\left(\frac{r}{\max\{1,\eta T\}}+R\right)\log\left(\frac{1}{\epsilon}\left(\frac{r}{\max\{1,\eta T\}}+R\right)\right)\right)
        \label{eq:CompHistGen2+}
    \end{align}
    additional elementary gates.
    \label{th:HistState+}
\end{theorem}

\begin{proof}
    According to the proof of Theorem 10 in Ref. \cite{an2026fast}, $\tilde{\mathbf{X}}^{(1)}_{+,r}$ defined by $\tilde{\mathbf{X}}^{(1)}_{+,n+1} \coloneqq \hat{\Phi}_n \tilde{\mathbf{X}}^{(1)}_{+,n} + \tilde{\boldsymbol{\Delta}}^{(1)}_{+,n},\tilde{\mathbf{X}}^{(1)}_{+,0}=\mathbf{x}_0$ satisfies
    \begin{align}
        \left\|\tilde{\mathbf{X}}^{(1)}_{+,r} - \mathbf{X}^{(1)}_r\right\| \le \frac{2\epsilon_{1+}}{\eta \Delta t} \left( \|\mathbf{x}_0\| + \frac{8\Delta_{\rm max}^{(1)}}{\eta \Delta t}\right) + \frac{4\epsilon_{2+}}{\eta \Delta t},
    \end{align}   
    which becomes
    \begin{align}
        \left\|\tilde{\mathbf{X}}^{(1)}_{+,r} - \mathbf{X}^{(1)}_r\right\| \le \frac{U_{\tilde{\mathcal{B}}_+} \epsilon}{2}
        \label{eq:XtilrErr}
    \end{align}
    by the definition of $\epsilon_{1+}$ and $\epsilon_{2+}$.
    We note
    \begin{align}
        \tilde{\mathbf{X}}^{(1)}_{\rm hist+} \coloneqq
        \begin{tikzpicture}[baseline=(m.center)]
        \node (m) {
        $
        \begin{pmatrix}
        \tilde{\mathbf{X}}^{(1)}_{+,0} \\ \vdots \\ \tilde{\mathbf{X}}^{(1)}_{+,r} \\ \tilde{\mathbf{X}}^{(1)}_{+,r} \\ \vdots \\ \tilde{\mathbf{X}}^{(1)}_{+,r} 
        \end{pmatrix}
        $
        };
        \draw[decorate,decoration={brace,amplitude=5pt}]
          (0.5,-0.1) -- (0.5,-1.5)
          node[midway,right=3pt] {$\text{$R$}$};
        \end{tikzpicture}
        = \tilde{\mathcal{A}}_+^{-1} \tilde{\mathcal{B}}^{(1)}_+.
        \label{eq:tilXhist+}
    \end{align}
    Since $\|\tilde{\mathcal{A}}_+\| \le \alpha_{\tilde{\mathcal{A}}_+}$ and Eq.~\eqref{eq:calAtil+Norm} hold, the singular values of $\tilde{\mathcal{A}}_+$ lie in $\left[\left(\frac{4r}{\max\{1,\eta T\}}+R\right)^{-1},\alpha_{\tilde{\mathcal{A}}_+}\right]$.
    Thus, defining $\epsilon_{3+}\coloneqq\frac{\alpha_{\tilde{\mathcal{A}}_+}}{4\kappa_{\tilde{\mathcal{A}}_+}}\epsilon$, we see from Theorem~\ref{eq:MatInvQSVT} that we can construct an $\left(1,a_{\tilde{\mathcal{A}}_+},\epsilon_{3+}\right)$-block-encoding $V_{\tilde{\mathcal{A}}_+^{-1}}$ of $\frac{\alpha_{\tilde{\mathcal{A}}_+}}{2\kappa_{\tilde{\mathcal{A}}_+}}\tilde{\mathcal{A}}_+^{-1}$ by using $V_{\tilde{\mathcal{A}}_+}$ $O\left(\kappa_{\tilde{\mathcal{A}}_+}\log\left(\frac{1}{\epsilon_{3+}}\right)\right)$ times, which is equivalent to Eq.~\eqref{eq:CompHistGen+}, and additional elementary gates $O\left(a_{\tilde{\mathcal{A}}_+}\kappa_{\tilde{\mathcal{A}}_+}\log\left(\frac{1}{\epsilon_{3+}}\right)\right)$ times, which is equivalent to Eq.~\eqref{eq:CompHistGen2+}.
    Then, on a system consisting of two registers, where the first one contains the ancillary qubits for $V_{\tilde{\mathcal{A}}_+^{-1}}$ and those corresponding to the first register in Eq.~\eqref{eq:VtilB+}, we can operate $V_{\tilde{\mathcal{B}}_+^{(1)}}$ and $V_{\tilde{\mathcal{A}}_+^{-1}}$ in this order to obtain the state vector written by
    \begin{align}
        \ket{0}\frac{\alpha_{\tilde{\mathcal{A}}_+}}{2\kappa_{\tilde{\mathcal{A}}_+}}\tilde{\mathcal{C}}_+\ket{\tilde{\mathcal{B}}_+^{(1)}} + \ket{0_\perp},
        \label{eq:ResState+}
    \end{align}
    where the matrix $\tilde{\mathcal{C}}_+$ satisfies
    \begin{align}
        \left\|\frac{\alpha_{\tilde{\mathcal{A}}_+}}{2\kappa_{\tilde{\mathcal{A}}_+}}\tilde{\mathcal{C}}_+ - \frac{\alpha_{\tilde{\mathcal{A}}_+}}{2\kappa_{\tilde{\mathcal{A}}_+}}\tilde{\mathcal{A}}_+^{-1}\right\| \le \epsilon_{3+}.
        \label{eq:MatInvErr+}
    \end{align}
    Then, identifying $\frac{\alpha_{\tilde{\mathcal{A}}_+}}{2\kappa_{\tilde{\mathcal{A}}_+}}\tilde{\mathcal{C}}_+\ket{\tilde{\mathcal{B}}_+^{(1)}}$ with $\ket{\tilde{\tilde{\mathbf{X}}}_{\rm hist+}}$, we have
    \begin{align}
        \left\| \tilde{\tilde{\mathbf{X}}}_{+,n}^{(1)} - 
        \mathbf{X}_r^{(1)}\right\| & \le \left\| \tilde{\tilde{\mathbf{X}}}_{+,n}^{(1)} - 
        \tilde{\mathbf{X}}_{+,r}^{(1)}\right\| + \left\| \tilde{\mathbf{X}}_{+,r}^{(1)} - 
        \mathbf{X}_r^{(1)}\right\| \nonumber \\
        & \le \left\|\tilde{\tilde{\mathbf{X}}}_{\rm hist+}^{(1)}-\tilde{\mathbf{X}}_{\rm hist+}^{(1)}\right\| + \left\| \tilde{\mathbf{X}}_{+,r}^{(1)} - 
        \mathbf{X}_{+,r}^{(1)}\right\| \nonumber \\
        & \le \frac{2\kappa_{\tilde{\mathcal{A}}_+}U_{\tilde{\mathcal{B}}_+}}{\alpha_{\tilde{\mathcal{A}}_+}} \times \left\|\frac{\alpha_{\tilde{\mathcal{A}}_+}}{2\kappa_{\tilde{\mathcal{A}}_+}}\tilde{\mathcal{C}}_+ - \frac{\alpha_{\tilde{\mathcal{A}}_+}}{2\kappa_{\tilde{\mathcal{A}}_+}}\tilde{\mathcal{A}}_+^{-1}\right\|\times\left\|\ket{\tilde{\mathcal{B}}_+^{(1)}}\right\| + \frac{U_{\tilde{\mathcal{B}}_+} \epsilon}{2} \nonumber \\
        & \le U_{\tilde{\mathcal{B}}_+} \epsilon,
    \end{align}
    where we use Eqs.~\eqref{eq:XtilrErr}, \eqref{eq:tilXhist+}, and \eqref{eq:MatInvErr+}.
    Therefore, the state in Eq.~\eqref{eq:ResState+} is in the form of Eq.~\eqref{eq:V1hist+}.
    In summary, $V^{(1)}_{{\rm hist+},\epsilon}$ is constructed by one use of $V_{\tilde{\mathcal{B}}_+^{(1)}}$ and one use of $V_{\tilde{\mathcal{A}}_+^{-1}}$, in which the number of uses of $V_{\tilde{\mathcal{A}}_+}$ and other elementary gates is as claimed.
\end{proof}

We then have the following theorem on the concrete implementation for preparing the state approximately encoding a realization of $\mathbf{X}_{\rm hist+}$, which corresponds to Theorem \ref{th:HistStateConc}.

\begin{theorem}
    Suppose that Assumptions \ref{ass:PRN}, \ref{ass:UA}, \ref{ass:muA}, \ref{ass:UBB}, and \ref{ass:Ux0} hold.
    Let $\epsilon\in(0,1]$ and $U_{\rm SN}>0$.
    Set $r=r_{\rm c}$, $K=K_{\rm c}$, and $M=M_{\rm c}$ with $r_{\rm c}$, $K_{\rm c}$, and $M_{\rm c}$ in Theorem \ref{th:HistStateConc}.
    Suppose that $z_1,\ldots,z_{r_{\rm c}N}$ generated by $U_{\rm Rand}$ satisfy Eq. \eqref{eq:ziUB}.
    Then, we can construct an oracle $V^{(1)}_{{\rm hist+},\epsilon}$ as described in Theorem \ref{th:HistState+},
    with $\ket{\tilde{\tilde{\mathbf{X}}}_{\rm hist+}^{(1)}}$ being
    \begin{align}
        \ket{\tilde{\tilde{\mathbf{X}}}_{\rm hist+}^{(1)}} \coloneqq \frac{1}{2\left(\frac{4r_{\rm c}}{\max\{1,\eta T\}}+R\right) U_{\tilde{\mathcal{B}}}}\tilde{\tilde{\mathbf{X}}}_{\rm hist+}^{(1)},
        \label{eq:XHistKet2+}
    \end{align}
    where $U_{\tilde{\mathcal{B}}}$ is given by Eq. \eqref{eq:UtilB} and $\tilde{\tilde{\mathbf{X}}}_{\rm hist+}^{(1)}\in\mathbb{R}^{(r + R + 1)N}$ satisfies Eq.~\eqref{eq:XrErr}.
    $V^{(1)}_{{\rm hist+},\epsilon}$ uses the controlled $U_A$
    \begin{align}
        O\left(\left(\frac{\kappa_{BB^T}\alpha_A T}{\max\{1,\eta T\}}+R\right)K_{\rm c}\log\left(\frac{1}{\epsilon}\left(\frac{\kappa_{BB^T}\alpha_A T}{\max\{1,\eta T\}}+R\right)\right)\right)
    \end{align}
    times, the controlled $U_{BB^T}$
    \begin{align}
        O\left(\left(\frac{\kappa_{BB^T}\alpha_A T}{\max\{1,\eta T\}}+R\right)\log\left(\frac{1}{\epsilon}\left(\frac{\kappa_{BB^T}\alpha_A T}{\max\{1,\eta T\}}+R\right)\right)\right)
    \end{align}
    times, the controlled $U_{\mathbf{x}_0}$ and $U_{\rm Rand}$ once each, and
    \begin{align}
        O\left(\left(\frac{\kappa_{BB^T}\alpha_A T}{\max\{1,\eta T\}}+R\right) K_{\rm c} \left(\log K_{\rm c}\log M_{\rm c}+a_A+a_{BB^T}\right)\log\left(\frac{1}{\epsilon}\left(\frac{\kappa_{BB^T}\alpha_A T}{\max\{1,\eta T\}}+R\right)\right)\right)
    \end{align}
    additional elementary gates.
    \label{th:HistStateConc+}
\end{theorem}

\begin{proof}
    The proof goes like Theorem \ref{th:HistStateConc} with slight modifications.
    We set $\hat{\Phi}_n$ to $\tilde{\Phi}^{K,r,M}_{n,0}$, which means that $\tilde{\mathcal{A}}_+$ is now
    \begin{align}
        \tilde{\mathcal{A}}_+ \coloneqq
            \begin{tikzpicture}[baseline=(m.center)]
            \node (m) {
            $
            \begin{pmatrix}
            I & & & & & & & \\
            -\tilde{\Phi}^{K,r,M}_{0,0} & I & & & & & & \\
            & \ddots & \ddots & & & & & \\
            & & -\tilde{\Phi}^{K,r,M}_{r-1,0} & I & & & & \\
            & & & -I & I & & & \\
            & & & & \ddots & \ddots & & \\
            & & & & & -I & I & 
            \end{pmatrix}
            $
            };
            \draw[decorate,decoration={brace,amplitude=5pt}]
              (1.5,-0.2) -- (3.0,-1.5)
              node[midway,right=3pt,yshift=6pt] {$\text{$R$}$};
            \end{tikzpicture}
            ,
        \label{eq:tilA+}
    \end{align}
    and $\tilde{\boldsymbol{\Delta}}^{(1)}_{+,n}$ to $\tilde{\boldsymbol{\Delta}}^{(i)}_n=\tilde{S}^{K,r,M}_n\mathbf{z}^{(i)}_n$ in Lemma \ref{lem:Delta} with $i=1$.
    The way to construct $V_{\tilde{\mathcal{B}}^{(1)}_+}$ is similar to that for $V_{\tilde{\mathcal{B}}^{(1)}}$ in Theorem \ref{th:HistStateConc}.
    $V_{\tilde{\mathcal{A}}_+}$ can be constructed similarly to $V_{\tilde{\mathcal{A}}}$ except for replacing the matrix \eqref{eq:PhiTilDiag} with
    \begin{align}
        \begin{tikzpicture}[baseline=(m.center)]
            \node (m) {
            $
            \begin{pmatrix}
                \tilde{\Phi}^{K,r,M}_{0,0} & & & & & & \\
                & \ddots & & & & & \\
                & & \tilde{\Phi}^{K,r,M}_{r-1,0} & & & & \\
                & & & I & & & \\
                & & & & \ddots & & \\
                & & & & & I & 
            \end{pmatrix}
            $
            };
            \draw[decorate,decoration={brace,amplitude=5pt}]
              (1,-0.1) -- (2.3,-1.3)
              node[midway,right=3pt,yshift=6pt] {$\text{$R$}$};
        \end{tikzpicture}
        .
    \end{align}
    We can take $\alpha_{\tilde{\mathcal{A}}_+}=1+e$ as $\alpha_{\tilde{\mathcal{A}}}$.
    $U_{\tilde{\mathcal{B}}_+}$ is set to $U_{\tilde{\mathcal{B}}}$ in Theorem \ref{th:HistStateConc}.
    Then, combining $V_{\tilde{\mathcal{A}}_+}$ and $V_{\tilde{\mathcal{B}}^{(1)}_+}$ gives $V^{(1)}_{{\rm hist+},\epsilon}$ as described in Theorem \ref{th:HistState+}.
    Note that the prefactor $\frac{\alpha_{\tilde{\mathcal{A}}_+}}{2\kappa_{\tilde{\mathcal{A}}_+}U_{\tilde{\mathcal{B}}_+}}$ in Eq. \eqref{eq:XHistKet+} now becomes that in Eq. \eqref{eq:XHistKet2+} in the current setting that $r=r_{\rm c}$.
    
    We now check that the stated setting $K=K_{\rm c}$, $r=r_{\rm c}$, and $M=M_{\rm c}$ actually leads to the conditions \eqref{eq:RnPhinDiff+} and \eqref{eq:condTilDel+}.
    We first consider the condition \eqref{eq:RnPhinDiff+}.
    Lemma~\ref{lem:PhiErr} implies that $\left\|\tilde{\Phi}^{K,r,M}_{n,0}-\Phi_n\right\|\le\varepsilon$ holds, so it suffices to show $\varepsilon \le \epsilon_{1+}$ and $\varepsilon \le \frac{1}{2}\eta \Delta t e^{-\eta \Delta t}$.
    The latter has been already seen in the proof of Theorem \ref{th:HistStateConc}, and we can see the former as
    \begin{align}
        \epsilon_{1+} & = \frac{\eta \Delta t U_{\tilde{\mathcal{B}}} \epsilon}{8\left( \|\mathbf{x}_0\| + \frac{8\Delta_{\rm max}^{(1)}}{\eta \Delta t}\right)} \nonumber \\
        & \ge \frac{\eta \Delta t U_{\tilde{\mathcal{B}}} \epsilon}{8\sqrt{2} \sqrt{\|\mathbf{x}_0\|^2 + \left(\frac{8\Delta_{\rm max}^{(1)}}{\eta \Delta t}\right)^2}} \nonumber \\
        & \ge \frac{\eta^2 \Delta t^2 \epsilon}{8\sqrt{2}}\sqrt{\frac{\|\mathbf{x}_0\|^2+4\sigma^2TNU_{\rm SN}^2}{\|\mathbf{x}_0\|^2(\eta\Delta t)^2+64\sigma^2 \Delta t NU_{\rm SN}^2}} \nonumber \\
        & \ge \frac{\eta^2 \Delta t^2 \epsilon}{8\sqrt{2}}\min\left\{\frac{1}{\eta \Delta t},\frac{\sqrt{r}}{4}\right\} \nonumber \\
        & \ge \frac{\eta^2 T^2 \epsilon}{32\sqrt{2} r_{\rm c}^2} \nonumber \\
        & \ge \varepsilon.
        \label{eq:eps1vareps+}
    \end{align}
    Here, we use $(x+y)^2\le 2(x^2+y^2)$, which holds for any $x,y\in\mathbb{R}$, at the first inequality.
    At the second inequality, we use Eqs \eqref{eq:UtilB} and $\left(\Delta_{\rm max}^{(1)}\right)^2\le \sigma^2\Delta tNU_{\rm SN}^2$, which is seen similarly to Eq. \eqref{eq:DeltaMaxUB}.

    As for Eq. \eqref{eq:condTilDel+}, we get it by identifying $\boldsymbol{\Delta}^{(1)}_n$ with $S_n\mathbf{z}^{(1)}_n$ and combining Eq. \eqref{eq:DelTilErrEps2} and $\epsilon_{2+} \ge \epsilon_2$, which we can see easily.

    The evaluation of query complexities can also be done similarly to Theorem \ref{th:HistStateConc}, using Eqs. \eqref{eq:CompHistGen+} and \eqref{eq:CompHistGen2+} instead of Eqs. \eqref{eq:CompHistGen} and \eqref{eq:CompHistGen2}, respectively, and noting $\alpha_{\tilde{\mathcal{A}}_+}=1+e$.
\end{proof}

Then, we can construct an oracle to generate the quantum state that approximately encodes the $i$-th realization of $\mathbf{X}_{\rm hist+}$ for specified $i$.
The following corollary corresponds to Corollary \ref{cor:HistStateConcSuper} and can be proven similarly, so we present it without proof.

\begin{corollary}
    In the same setting as Corollary \ref{cor:HistStateConcSuper}, we can construct an oracle $V_{{\rm hist},\epsilon}$, which acts on a three-register system as
    \begin{align}
        V_{{\rm hist+},\epsilon}\ket{i}\ket{0}\ket{0} =   \ket{i}\ket{0}\ket{\tilde{\tilde{\mathbf{X}}}_{\rm hist+}^{(i)}} + \ket{0_\perp}
        \label{eq:VHisi+}
    \end{align}
    for any $i\in N_{\rm s}$.
    Here, $\ket{0_\perp}$ is some unnormalized state such that $(I\otimes \ket{0}\!\bra{0}\otimes I)\ket{0_\perp}=0$, and $\ket{\tilde{\tilde{\mathbf{X}}}_{\rm hist+}^{(i)}}$ is an unnormalized state given by
    \begin{align}
        \Ket{\tilde{\tilde{\mathbf{X}}}_{\rm hist+}^{(i)}} \coloneqq \frac{1}{2\left(\frac{4r_{\rm c}}{\max\{1,\eta T\}}+R\right) U_{\tilde{\mathcal{B}}}}\sum_{j=1}^{(r+R+1)N}\tilde{\tilde{X}}_{\rm hist+}^{(i),j}\ket{j}
    \end{align}
    with $\tilde{\tilde{\mathbf{X}}}_{\rm hist+}^{(i)}=\left(\tilde{\tilde{X}}_{\rm hist+}^{(i),1},\ldots,\tilde{\tilde{X}}_{\rm hist+}^{(i),(r+R+1)N}\right)^T\in\mathbb{R}^{(r+R+1)N}$ such that, if $z_1,\ldots,z_{r_{\rm c}NN_{\rm s}}$ generated by $U_{\rm Rand}$ satisfy Eq. \eqref{eq:ziUB2},
    \begin{align}
        \left\| \tilde{\tilde{\mathbf{X}}}_{\rm hist+}^{(i)} - 
        \mathbf{X}_{\rm hist+}^{(i)}\right\| \le U_{\tilde{\mathcal{B}}} \epsilon
    \end{align}
    holds for
    \begin{align}
        \mathbf{X}_{\rm hist+}^{(i)} \coloneqq
        \begin{tikzpicture}[baseline=(m.center)]
        \node (m) {
        $
        \begin{pmatrix}
            \mathbf{X}_0^{(i)} \\ \vdots \\ \mathbf{X}_r^{(i)}  \\ \mathbf{X}_r^{(i)}  \\ \vdots \\ \mathbf{X}_r^{(i)} 
        \end{pmatrix}
        $
        };
        \draw[decorate,decoration={brace,amplitude=5pt}]
          (0.45,-0.1) -- (0.45,-1.5)
          node[midway,right=3pt] {$\text{$R$}$};
        \end{tikzpicture}
    .
    \end{align}
    $V_{{\rm hist+},\epsilon}$ queries the controlled $U_A$, $U_{BB^T}$, $U_{\mathbf{x}_0}$, $U_{\rm Rand}$, and additional elementary gates a number of times same as Theorem \ref{th:HistStateConc+}.
    \label{cor:HistStateConc+Super}
\end{corollary}

Again, the estimation of expectations is based on the sample average, so we present the following lemmas, which correspond to Lemmas \ref{lem:XhistMom} and \ref{lem:YHat}.

\begin{lemma}
    Let $d\in\mathbb{N}$ and $\epsilon\in\left(0,1\right]$.
    Suppose that Assumptions \ref{ass:UA}, \ref{ass:muA}, \ref{ass:UBB}, and \ref{ass:Ux0} hold.
    Set $r=r_{\rm c}$, $K=K_{\rm c}$, and $M=M_{\rm c}$ with $r_{\rm c}$, $K_{\rm c}$, and $M_{\rm c}$ in Theorem \ref{th:HistStateConc+}.
    Suppose that Assumption \ref{ass:PRN} holds with $\epsilon_{\rm PRN}$ in Eq.~\eqref{eq:epsPRN}.
    Set $R$ by 
    \begin{align}
        R=\left\lceil\frac{r_{\rm c}}{\max\{1,\eta T\}}\right\rceil.
        \label{eq:r+}
    \end{align}
    Then, for any $i\in\mathbb{N}$, $\mathbf{X}_{\rm hist+}^{(i)}$ and $\tilde{\tilde{\mathbf{X}}}_{\rm hist+}^{(i)}$ in Corollary \ref{cor:HistStateConc+Super} satisfy
    \begin{align}
        \mathbb{E}\left[\left\|\left(\mathbf{X}_r^{(i)}\right)^{\otimes d}\right\|^2\right]\le 2(2d-1)!!\left(\|\mathbf{x}_0\|^2+(N+2d)\sigma^2T\right)^d
        \label{eq:EXr2}
    \end{align}
    and, for any $n\in\{r,\ldots,r+R\}$,
    \begin{align}
        \mathbb{E}\left[\left\|\left(\tilde{\tilde{\mathbf{X}}}_{+,n}^{(i)}\right)^{\otimes d}\right\|^2\right]\le 2(2d-1)!!(1+3\epsilon)^{2d}\left(\|\mathbf{x}_0\|^2+(N+2d)\sigma^2T\right)^d,
        \label{eq:EXrtiltil2}
    \end{align}
    where $\tilde{\tilde{\mathbf{X}}}_{+,n}^{(i)}\in\mathbb{R}^N,n\in[r+R]_0$ is the $n$-th block of $\mathbf{X}_{\rm hist+}^{(i)}$, namely
    $
        \tilde{\tilde{\mathbf{X}}}_{\rm hist+}^{(i)} =
        \begin{pmatrix}
            \tilde{\tilde{\mathbf{X}}}_{+,0}^{(i)} \\ \vdots \\ \tilde{\tilde{\mathbf{X}}}_{+,r+R}^{(i)}
        \end{pmatrix}
    $.
    \label{lem:XrMom}
\end{lemma}

\begin{proof}
    Writing $\mathbf{X}_n^{(i)}$ as Eq. \eqref{eq:XnSum}, we can bound $\mathbb{E}\left[\left\|\left(\mathbf{X}_r^{(i)}\right)^{\otimes d}\right\|^2\right]$ similarly to the discussion in Lemma \ref{lem:XhistMom}:
    \begin{align}
        \mathbb{E}\left[\left\|\left(\mathbf{X}_r^{(i)}\right)^{\otimes d}\right\|^2\right] =  \mathbb{E}\left[\left\|\mathbf{X}_r^{(i)}\right\|^{2d}\right]
        = \sum_{l_1,\ldots,l_{2d}=-1}^{r-1} \mathbb{E}\left[\left(\mathcal{D}_{r,l_1}\mathbf{z}^{(i)}_{l_1}\right)^T \mathcal{D}_{r,l_2} \mathbf{z}^{(i)}_{l_2} \times \cdots \times \left(\mathcal{D}_{r,l_{2d-1}}\mathbf{z}^{(i)}_{l_{2d-1}}\right)^T \mathcal{D}_{r,l_{2d}} \mathbf{z}^{(i)}_{l_{2d}}\right]
        \label{eq:Xrd2Temp}
    \end{align}
    is transformed into Eq.~\eqref{eq:EXr2} by a discussion similar to that yielding Eq. \eqref{eq:EXhistd2}.

    The proof of Eq. \eqref{eq:EXrtiltil2} goes similarly to Eq. \eqref{eq:EXhisttiltild2}.
    Recall that the blocks of $\mathcal{A}_+^{-1}$ and $\tilde{\mathcal{A}}_+^{-1}$ are given by Eqs. \eqref{eq:calAInvBlock} and \eqref{eq:calATilInvBlock}, respectively, where $\hat{\Phi}_n$ is now $\tilde{\Phi}^{K,r,M}_{n,0}$.
    Thus, the norm of the $(n,n^\prime)$-block of $\tilde{\mathcal{A}}_+^{-1}-\mathcal{A}_+^{-1}$ is bounded by
    \begin{align}
        \left\|(\tilde{\mathcal{A}}_+^{-1})_{n,n^\prime}-(\mathcal{A}_+^{-1})_{n,n^\prime}\right\| \le
        \begin{cases}
            0 & ; ~ (n \le n^\prime) \lor (r \le n^\prime < n)\\
            (n-n^\prime) e^{-\frac{1}{2}(n-n^\prime-1)\eta \Delta t}\varepsilon & ; ~ n^\prime < n \le r \\
            (r-n^\prime) e^{-\frac{1}{2}(r-n^\prime-1)\eta \Delta t}\varepsilon & ; ~ n^\prime < r < n
        \end{cases},
    \end{align}
    which is obtained similarly to Eq. \eqref{eq:calAtilInvcalAInvBlock}.
    Then, using an inequality similar to Eq. \eqref{eq:tilAInvAInvNorm}, we have
    \begin{align}
        \left\|\tilde{\mathcal{A}}_+^{-1}-\mathcal{A}_+^{-1}\right\| & \le \sqrt{\max_{n\in[r + R]_0} \sum_{n^\prime=0}^{r + R}\left\|(\tilde{\mathcal{A}}_+^{-1})_{n,n^\prime}-(\mathcal{A}_+^{-1})_{n,n^\prime}\right\| \times \max_{n^\prime\in[r + R]_0} \sum_{n=0}^{r + R}\left\|(\tilde{\mathcal{A}}_+^{-1})_{n,n^\prime}-(\mathcal{A}_+^{-1})_{n,n^\prime}\right\|} \nonumber \\
        & \le \sqrt{\left(\sum_{k=0}^r ke^{-\frac{1}{2}(k-1)\eta \Delta t} \right) \times \max_{n^\prime\in[r]_0}\left(\sum_{n=n^\prime}^r (n-n^\prime) e^{-\frac{1}{2}(n-n^\prime-1)\eta \Delta t} + R (r-n^\prime) e^{-\frac{1}{2}(r-n^\prime-1)\eta \Delta t}\right)} \varepsilon  \nonumber \\
        & \le \sqrt{\frac{16}{(\eta \Delta t)^2} \times \left(\frac{16}{(\eta \Delta t)^2} + \frac{2R}{\eta \Delta t}\right)} \varepsilon  \nonumber \\
        & \le \frac{\sqrt{288+32\eta \Delta t}}{(\eta \Delta t)^2} \varepsilon \nonumber \\
        & \le \frac{\sqrt{30}}{24}\epsilon.
        \label{eq:AtilInvAInv+}
    \end{align}
    Here, at the third inequality, we use $\sum_{k=0}^r k e^{-\frac{1}{2}(k-1)\eta \Delta t} \le \frac{16}{(\eta \Delta t)^2}$, which is shown in Eq. \eqref{eq:tilAInvAInvNorm2}, and $\max_{k\in\mathbb{N}} k e^{-\frac{1}{2}(k-1)\eta \Delta t} \le \frac{2}{\eta \Delta t}$, which is seen by elementary analysis.
    The fourth inequality follows from $R\le\frac{1}{\eta \Delta t}+1$, which holds for the current $R$ in Eq. \eqref{eq:r+}, and the last inequality follows from the definition of $\varepsilon$ in Eq.~\eqref{eq:KrMComb} and $\eta\Delta t < \eta/\alpha \le 1$.
    Combining Eq.~\eqref{eq:AtilInvAInv+} with
    \begin{align}
        \left\|\tilde{\mathcal{C}}_+-\tilde{\mathcal{A}}_+^{-1}\right\| \le  \frac{\epsilon}{2},
        \label{eq:CtilAtilInv+}
    \end{align}
    which follows from Eq.~\eqref{eq:MatInvErr+}, we get
    \begin{align}
        \left\|\tilde{\mathcal{C}}_+-\mathcal{A}_+^{-1}\right\| \le \epsilon.
        \label{eq:CtilAInv+}
    \end{align}
    $\tilde{\tilde{\mathbf{X}}}_{+,n}^{(i)}$ can be written by
    \begin{align}
        \tilde{\tilde{\mathbf{X}}}_{+,n}^{(i)} = \sum_{l=-1}^{r-1} \tilde{\mathcal{D}}^+_{n,l}  \tilde{\mathbf{z}}^{(i)}_l
        \label{eq:Xtiltilni+}
    \end{align}
    with
    \begin{align}
        \tilde{\mathcal{D}}^+_{n,l} \coloneqq
        \begin{cases}
            \tilde{\mathcal{C}}^+_{n,0} & ; ~ l=-1 \\
            \tilde{\mathcal{C}}^+_{n,l+1} \tilde{S}^{K,r,M}_l & ; ~ l\in[r-1]_0
        \end{cases},
    \end{align}
    where $\tilde{\mathcal{C}}^+_{n,n^\prime}\in\mathbb{R}^{N \times N}$ is the $(n,n^\prime)$-th block of $\tilde{\mathcal{C}}_+$.
    Using Eq. \eqref{eq:StilNorm} and $\left\|\tilde{\mathcal{C}}^+_{n,n^\prime}\right\| \le 1 + \epsilon$, which can be seen similarly to Eq. \eqref{eq:calCBlock}, we have
    \begin{align}
        \|\tilde{\mathcal{D}}^+_{n,l}\| \le  \tilde{D}_l =
        \begin{dcases}
        1+\epsilon & ; ~ l=-1 \\
        (1+3\epsilon)\sqrt{\sigma^2\Delta t}, & ; ~ l\in[r-1]_0
        \end{dcases}
        .
    \end{align}
    Then,
    \begin{align}
        \mathbb{E}\left[\left\|\left(\tilde{\tilde{\mathbf{X}}}_{+,n}^{(i)}\right)^{\otimes d}\right\|^2\right] & = \sum_{l_1,\ldots,l_{2d}=-1}^{r-1} \mathbb{E}\left[\left(\tilde{\mathcal{D}}^+_{n,l_1}\tilde{\mathbf{z}}^{(i)}_{l_1}\right)^T \tilde{\mathcal{D}}^+_{n,l_2} \tilde{\mathbf{z}}^{(i)}_{l_2} \times \cdots \times \left(\tilde{\mathcal{D}}^+_{n,l_{2d-1}}\tilde{\mathbf{z}}^{(i)}_{l_{2d-1}}\right)^T \tilde{\mathcal{D}}^+_{n,l_{2d}} \tilde{\mathbf{z}}^{(i)}_{l_{2d}}\right],
        \label{eq:Xtiltilr2Temp}
    \end{align}
    is transformed into Eq. \eqref{eq:EXrtiltil2} by a discussion similar to that yielding Eq. \eqref{eq:EXhisttiltild2}.
\end{proof}

\begin{lemma}
    Let $d\in\mathbb{N}$, $N_{\rm s}\in\mathbb{N}$, $\epsilon\in\left(0,1\right]$, and $\delta\in(0,1)$.
    Suppose that Assumptions \ref{ass:UA}, \ref{ass:muA}, \ref{ass:UBB}, and \ref{ass:Ux0} hold.
    Set $K$, $r$, and $M$ as in Theorem \ref{th:HistStateConc+}.
    Suppose that Assumption \ref{ass:PRN} holds with $\epsilon_{\rm PRN}$ in Eq.~\eqref{eq:epsPRN2}.
    Set $R$ as in Lemma \ref{lem:XrMom}.
    Set $U_{\rm SN}$ as Eq. \eqref{eq:USN}.
    Define
    \begin{align}
        Y \coloneqq \sum_{j_1,\ldots,j_d=1}^{N} C_{j_1,\ldots,j_d} X_r^{j_1} \cdots X_r^{j_d},
        \label{eq:Yr}
    \end{align}
    \begin{align}
        \tilde{\tilde{Y}}^{(i)}_n \coloneqq \sum_{j_1,\ldots,j_d=1}^{N} C_{j_1,\ldots,j_d} \tilde{\tilde{X}}_n^{(i),j_1} \cdots \tilde{\tilde{X}}_n^{(i),j_d}, ~ i\in\mathbb{N}, ~ n\in\{r,\ldots,R\},
    \end{align}
    and
    \begin{align}
        \hat{Y} \coloneqq \frac{1}{N_{\rm s}(R+1)}\sum_{i=1}^{N_{\rm s}}\sum_{n=r}^{r+R}\tilde{\tilde{Y}}^{(i)}_n
    \end{align}
    with a $d$-way real tensor $C=(C_{j_1,\ldots,j_d})\in\mathbb{R}^{N \times \cdots \times N}$, where $\tilde{\tilde{X}}_n^{(i),j}$ are the $j$-th entry of $\tilde{\tilde{\mathbf{X}}}^{(i)}_n$.
    Then,
    \begin{align}
        \left|\hat{Y}-\mathbb{E}_\mathbb{P}[Y]\right|\le 8d\sqrt{(2d-3)!!}(1+3\epsilon)^{d}\left(\|\mathbf{x}_0\|^2+(N+2d)\sigma^2T\right)^{d/2}\|C\|_F\left(\epsilon+\sqrt{\frac{4}{\delta N_{\rm s}}}\right)
        \label{eq:YErrProb+}
    \end{align}
    holds with probability at least $1-\frac{\delta}{2}$.
    \label{lem:YrHat}
\end{lemma}

\begin{proof}
    For $i\in\mathbb{N}$ and $n\in\{r,\ldots,r+R\}$, by a discussion similar to that leading to Eq. \eqref{eq:Xhistddiff}, we obtain
    \begin{align}
        \mathbb{E}\left[\left\|\left(\tilde{\tilde{\mathbf{X}}}_{{\rm NB}+,n}^{(i)}\right)^{\otimes d}-\left(\mathbf{X}_r^{(i)}\right)^{\otimes d}\right\|\right]\le 6d\sqrt{(2d-3)!!}(1+3\epsilon)^{d-1}\epsilon\left(\|\mathbf{x}_0\|^2+(N+2d)\sigma^2T\right)^{d/2},
    \end{align}
    where $\tilde{\tilde{\mathbf{X}}}_{{\rm NB}+,n}^{(i)} \coloneqq \sum_{l=-1}^{r-1} \tilde{\mathcal{D}}^+_{n,l}  \mathbf{z}^{(i)}_l$.
    Then, defining
    \begin{align}
        \tilde{\tilde{Y}}^{(i)}_{{\rm NB}+,n} \coloneqq \sum_{j_1,\ldots,j_d=1}^{N} C_{j_1,\ldots,j_d} \tilde{\tilde{X}}_{{\rm NB}+,n}^{(i),j_1} \cdots \tilde{\tilde{X}}_{{\rm NB}+,n}^{(i),j_d}, ~ i\in\mathbb{N},
    \end{align}
    where $\tilde{\tilde{X}}_{{\rm NB}+,n}^{(i),j}$ is the $j$-th entry of $\tilde{\tilde{\mathbf{X}}}_{{\rm NB}+,n}^{(i)}$, we get
    \begin{align}
        \left|\mathbb{E}\left[\tilde{\tilde{Y}}^{(i)}_{{\rm NB}+,n}\right]-\mathbb{E}_\mathbb{P}\left[Y\right]\right| \le 7d\sqrt{(2d-3)!!}(1+3\epsilon)^{d-1}\epsilon\left(\|\mathbf{x}_0\|^2+(N+2d)\sigma^2T\right)^{d/2} \|C\|_F.
        \label{eq:BiasYtiltilNB+}
    \end{align}
    by a discussion similar to that yielding Eq. \eqref{eq:BiasYtiltilNB}.
    Furthermore, since we can obtain
    \begin{align}
        \mathbb{E}\left[\left\|\left(\tilde{\tilde{\mathbf{X}}}_{{\rm NB}+,n}^{(i)}\right)^{\otimes d}\right\|^2\right]\le (2d-1)!!(1+3\epsilon)^{2d}\left(\|\mathbf{x}_0\|^2+(N+2d)\sigma^2T\right)^d
    \end{align}
    similarly to Eq. \eqref{eq:EXhistNBtiltild2}, we have
    \begin{align}
        \mathbb{E}\left[\left|\tilde{\tilde{Y}}^{(i)}_{{\rm NB}+,n} \tilde{\tilde{Y}}^{(i)}_{{\rm NB}+,n^\prime}\right|\right] & \le \|C\|_F^2 \mathbb{E}\left[\left\|\left(\tilde{\tilde{\mathbf{X}}}_{{\rm NB}+,n}^{(i)}\right)^{\otimes d}\right\| \cdot \left\|\left(\tilde{\tilde{\mathbf{X}}}_{{\rm NB}+,n^\prime}^{(i)}\right)^{\otimes d}\right\|\right] \nonumber \\
        & \le \|C\|_F^2 \sqrt{\mathbb{E}\left[\left\|\left(\tilde{\tilde{\mathbf{X}}}_{{\rm NB}+,n}^{(i)}\right)^{\otimes d}\right\|^2\right]} \sqrt{\mathbb{E}\left[\left\|\left(\tilde{\tilde{\mathbf{X}}}_{{\rm NB}+,n^\prime}^{(i)}\right)^{\otimes d}\right\|^2\right]}  \nonumber \\
        & \le (2d-1)!!(1+3\epsilon)^{2d}\left(\|\mathbf{x}_0\|^2+(N+2d)\sigma^2T\right)^d\|C\|_F^2.
        \label{eq:tiltilYVar+}
    \end{align}
    for any $n,n^\prime\in\{r,\ldots,r+R\}$, where we use the Cauchy–Schwarz inequality at the first and second inequalities.
    We also get
    \begin{align}
        \left|\mathrm{Cov}\left(\tilde{\tilde{Y}}^{(i)}_{{\rm NB+},n},\tilde{\tilde{Y}}^{(i^\prime)}_{{\rm NB+},n^\prime}\right)\right| \le \frac{1}{4}d^2(2d-3)!!\|C\|_{\rm F}^2(1+3\epsilon)^{2d}\left(\|\mathbf{x}_0\|^2 + (N+2d)\sigma^2T\right)^{d}\delta\epsilon^2
        \label{eq:tiltilYiCov+}
    \end{align}
    for $i \ne i^\prime$, similarly to Eq.~\eqref{eq:tiltilYiCov2}.
    By combining Eqs.~\eqref{eq:tiltilYVar+} and \eqref{eq:tiltilYiCov+}, we see that for $\hat{Y}_{\rm NB+} \coloneqq \frac{1}{N_{\rm s}(R+1)}\sum_{i=1}^{N_{\rm s}}\sum_{n=r}^{r+R}\tilde{\tilde{Y}}^{(i)}_{{\rm NB}+,n}$,
    \begin{align}
        & {\rm Var}\left[\hat{Y}_{{\rm NB}+}\right] \nonumber \\
        & \le \sum_{i=1}^{N_{\rm s}}\mathbb{E}\left[\left(\frac{1}{N_{\rm s}(R+1)}\sum_{n=r}^{r+R}\tilde{\tilde{Y}}^{(i)}_{{\rm NB}+,n}\right)^2\right] + \frac{1}{N_{\rm s}^2(R+1)^2}\sum_{\substack{i,i^\prime\in[N_{\rm s}] \\ i \ne i^\prime}}\sum_{n,n^\prime=r}^{r+R}{\rm Cov}\left(\tilde{\tilde{Y}}^{(i)}_{{\rm NB+},n},\tilde{\tilde{Y}}^{(i^\prime)}_{{\rm NB+},n^\prime}\right) \nonumber \\
        & \le \frac{1}{N_{\rm s}} (2d-1)!!(1+3\epsilon)^{2d}\left(\|\mathbf{x}_0\|^2+(N+2d)\sigma^2T\right)^d\|C\|_F^2 + \frac{1}{4}d^2(2d-3)!!\|C\|_{\rm F}^2(1+3\epsilon)^{2d}\left(\|\mathbf{x}_0\|^2 + (N+2d)\sigma^2T\right)^{d}\delta\epsilon^2 \nonumber \\
        & \le 
        d^2(2d-3)!!\|C\|_{\rm F}^2(1+3\epsilon)^{2d}\left(\|\mathbf{x}_0\|^2 + (N+2d)\sigma^2T\right)^{d}\left(\frac{1}{N_{\rm s}} + \frac{1}{4}\delta\epsilon^2\right)
        \label{eq:VarYNB+}
    \end{align}
    holds.
    Using Eqs. \eqref{eq:BiasYtiltilNB+} and \eqref{eq:VarYNB+} along with Chebyshev's inequality, we see that
    \begin{align}
        \left|\hat{Y}_{\rm NB}-\mathbb{E}_\mathbb{P}[Y]\right| & \le \left|\hat{Y}_{\rm NB}-\mathbb{E}[\hat{Y}_{\rm NB}]\right| + \left|\mathbb{E}[\hat{Y}_{\rm NB}]-\mathbb{E}_\mathbb{P}[Y]\right| \nonumber \\
        & \le 8d\sqrt{(2d-3)!!}(1+3\epsilon)^{d}\left(\|\mathbf{x}_0\|^2+(N+2d)\sigma^2T\right)^{d/2}\|C\|_F\left(\epsilon+\sqrt{\frac{4}{\delta N_{\rm s}}}\right)
    \end{align}
    holds with probability at least $1-\frac{\delta}{4}$.
    Besides, as seen in the proof of Lemma \ref{lem:YHat}, with probability at least $1-\frac{\delta}{4}$, $|z_1|,\ldots,|z_{r_{\rm c} NN_{\rm s}}| \le U_{\rm SN}$ holds, and thus $\hat{Y}_{\rm NB}=\hat{Y}$ holds.
    Combining the above findings, we see that the claim holds.
\end{proof}

Finally, we present a quantum algorithm to estimate expectations of terminal functions and the evaluation of its complexity.

\begin{theorem}
    Let $d\in\mathbb{N}$, and $\delta\in(0,1)$.
    Suppose that Assumptions \ref{ass:UA}, \ref{ass:muA}, \ref{ass:UBB}, and \ref{ass:Ux0} hold.
    Let $C$ be a $d$-way real tensor $C=(C_{j_1,\ldots,j_d})\in\mathbb{R}^{N \times \cdots \times N}$.
    Set $r$ to $r_{\rm c}$ in Eq. \eqref{eq:rReq}.
    Let $\epsilon$ be a positive real number such that
    \begin{align}
        \epsilon^\prime_{\rm T} \coloneqq \frac{\epsilon}{2^{d+5}d\sqrt{(2d-3)!!}\left(\|\mathbf{x}_0\|^2+(N+2d)\sigma^2T\right)^{d/2}\|C\|_F}
    \end{align}
    satisfies $\epsilon^\prime_{\rm T} \le \frac{1}{3}$.
    Suppose that Assumption \ref{ass:PRN} holds with 
    \begin{align}
        \epsilon_{\rm PRN} = \min\left\{\frac{(2d-1)!!}{(2r\sqrt{N})^d},\frac{d\sqrt{(2d-3)!!}\epsilon^\prime_{\rm T}}{(2r)^{d/2}},\frac{d^2(2d-3)!!\delta(\epsilon^\prime_{\rm T})^2}{8(2r)^d}\right\}.
    \end{align}
    Assume that we have access to the oracle $U_C$ that acts as
    \begin{align}
        U_C \ket{0}^{\otimes d} = \frac{1}{\|C\|_F} \sum_{i_1,\ldots,i_d=1}^{N} C_{i_1,\ldots,i_d} \ket{i_1}\cdots\ket{i_d}.
    \end{align}
    Set $r_{\rm c}$ by Eq. \eqref{eq:rReq}.
    Then, there is a quantum algorithm that, with probability at least $1-\delta$, outputs an estimate $\hat{\mu}_Y$ of $\mathbb{E}_\mathbb{P}[Y]$ with error at most $\epsilon$.
    This algorithm uses the controlled $U_A$
    \begin{align}
    O\left(\left(\frac{44\kappa_{BB^T} \alpha_A T}{\max\{1,\eta T\}}\sqrt{\|\mathbf{x}_0\|^2+N\sigma^2T \log \left(\frac{r_{\rm c} N}{\delta^2 \epsilon^{\prime2}_{\rm T}}\right)}\right)^d
    \sqrt{\frac{\kappa_{BB^T} \alpha_A T}{\max\{1,\eta T\}}}\frac{d \|C\|_F K_{\rm c}^\prime}{\epsilon}
    \log\left(\frac{1}{\delta}\right)\log\left(\frac{\kappa_{BB^T}^2\alpha_A^2}{\eta^2\epsilon^\prime_{\rm T}}\right)\right)
    \label{eq:CompUAESTerm}
    \end{align}
    times, the controlled $U_{BB^T}$
    \begin{align}
    O\left(\left(\frac{44\kappa_{BB^T} \alpha_A T}{\max\{1,\eta T\}}\sqrt{\|\mathbf{x}_0\|^2+N\sigma^2T \log \left(\frac{r_{\rm c} N}{\delta^2 \epsilon^{\prime2}_{\rm T}}\right)}\right)^d
    \sqrt{\frac{\kappa_{BB^T} \alpha_A T}{\max\{1,\eta T\}}}\frac{d \|C\|_F}{\epsilon}
    \log\left(\frac{1}{\delta}\right)\log\left(\frac{\kappa_{BB^T}^2\alpha_A^2}{\eta^2\epsilon^\prime_{\rm T}}\right)\right)
    \label{eq:CompUBBESTerm}
    \end{align}
    times, $U_C$ and the controlled $U_{\mathbf{x}_0}$ and $U_{\rm Rand}$
    \begin{align}
    O\left(\left(\frac{44\kappa_{BB^T} \alpha_A T}{\max\{1,\eta T\}}\sqrt{\|\mathbf{x}_0\|^2+N\sigma^2T \log \left(\frac{r_{\rm c} N}{\delta^2 \epsilon^{\prime2}_{\rm T}}\right)}\right)^d\left(\frac{\kappa_{BB^T} \alpha_A T}{\max\{1,\eta T\}}\right)^{-\frac{1}{2}}\frac{d\|C\|_F}{\epsilon}\log\left(\frac{1}{\delta}\right)\right)
    \end{align}
    times each, and
    \begin{align}
    O\left(\left(\frac{44\kappa_{BB^T} \alpha_A T}{\max\{1,\eta T\}}\sqrt{\|\mathbf{x}_0\|^2+N\sigma^2T \log \left(\frac{r_{\rm c} N}{\delta^2 \epsilon^{\prime2}_{\rm T}}\right)}\right)^d
    \sqrt{\frac{\kappa_{BB^T} \alpha_A T}{\max\{1,\eta T\}}}\frac{d \|C\|_F K_{\rm c}^\prime}{\epsilon}
    \left(\log K_{\rm c}^\prime\log M_{\rm c}^\prime+a_A+a_{BB^T}\right)\log\left(\frac{1}{\delta}\right)\log\left(\frac{\kappa_{BB^T}^2\alpha_A^2}{\eta^2\epsilon^\prime_{\rm T}}\right)\right)
    \end{align}
    additional elementary gates.
    Here, $K_{\rm c}^\prime$ and $M_{\rm c}^\prime$ are as in Theorem \ref{th:QAOE}.
    \label{th:QAOE+}
\end{theorem}

\begin{proof}

The quantum algorithm is as shown in Algorithm \ref{alg:main+}.

\begin{algorithm}[tp]
\caption{Proposed quantum algorithm for estimating the expectation of a terminal-time function}\label{alg:main+}
\begin{algorithmic}[1]

\State Set $K=K_{\rm c}^\prime$, $r=r_{\rm c}$, and $M=M_{\rm c}^\prime$.

\State Set $R$ by Eq. \eqref{eq:r+}.

\State Set $N_{\rm s}$ to 
\begin{align}
    N_{\rm s,T} = \left\lceil \frac{4}{\delta \epsilon^{\prime 2}_{\rm T}} \right\rceil.
\end{align}

\State Set $U_{\rm SN}$ as Eq.~\eqref{eq:USNdelPr}.

\State Construct the oracle $V_{{\rm hist+},\epsilon^\prime_{\rm T}}$.

\State By a use of $U^{\rm SP}_{1,N_{\rm s,T}}$ and a use of $V_{{\rm hist+},\epsilon^\prime_{\rm T}}$, construct an oracle $V_{\rm histSP+}$ that acts as
\begin{align}
    V_{\rm histSP+}\ket{0}\ket{0}\ket{0}^{\otimes d} = \ket{\rm histSP+} \coloneqq \frac{1}{\sqrt{N_{\rm s,T}}}\sum_{i=1}^{N_{\rm s,T}}  \ket{i}\left(\ket{0}\Ket{\tilde{\tilde{\mathbf{X}}}_{\rm hist+}^{(i)}}^{\otimes d} + \ket{0_\perp}\right),
    \label{eq:VHisSup+}
\end{align}
where $\ket{0_\perp}$ is some unnormalized state such that $(I \otimes \ket{0}\!\bra{0} \otimes I^{\otimes d})\ket{0_\perp}=0$.

\State By using $U^{\rm SP}_{1,N_{\rm s,T}}$, $U^{\rm SP}_{r,r+R}$, and $U_C$ once each, and arithmetic circuits, construct an oracle $U_C^{\rm SP+}$ that acts as
\begin{align}
    U_C^{\rm SP+}\ket{0}\ket{0}\ket{0}^{\otimes d}  = \ket{C{\rm SP+}} \coloneqq  \frac{1}{\sqrt{N_{\rm s,T}}}\sum_{i=1}^{N_{\rm s,T}}  \ket{i}\ket{0}\ket{C+},
    \label{eq:UCSup+}
\end{align}
where
\begin{align}
    \ket{C+} \coloneqq \frac{1}{\|C\|_F\sqrt{R + 1}} \sum_{n=r}^{r+R} \sum_{j_1,\ldots,j_d=1}^{N} C_{j_1,\ldots,j_d} \ket{Nn+j_1}\cdots\ket{Nn+j_d}
\end{align}

\State By the algorithm in Theorem \ref{th:OE}, get an estimate $\hat{\mu}^\prime$ of $\braket{\mathrm{histSP+} | C{\rm SP+}}$ with accuracy
\begin{align}
    \epsilon_{\rm OE,T}=
    \frac{\sqrt{R+1}\epsilon}{2\left(2\left(\frac{4r_{\rm c}}{\max\{1,\eta T\}}+R\right) U_{\tilde{\mathcal{B}}}\right)^d\|C\|_F}
\end{align}
and success probability $1-\frac{\delta}{2}$.

\State Output
\begin{align}
    \hat{\mu}_Y = 
    \frac{\left(2\left(\frac{4r_{\rm c}}{\max\{1,\eta T\}}+R\right) U_{\tilde{\mathcal{B}}}\right)^d}{\sqrt{R+1}}\|C\|_F \hat{\mu}^\prime.
\end{align}

\end{algorithmic}
\end{algorithm}

Let us prove the accuracy of this algorithm.
Lemma \ref{lem:YHat} implies that in the current setting,
\begin{align}
    \left|\hat{Y}-\mathbb{E}_\mathbb{P}[Y]\right|
    \le 8d\sqrt{(2d-3)!!}(1+3\epsilon^\prime_{\rm T})^{d}\left(\|\mathbf{x}_0\|^2+(N+2d)\sigma^2T\right)^{d/2}\|C\|_F\left(\epsilon^\prime_{\rm T}+\sqrt{\frac{4}{\delta N_{\rm s,T}}}\right)
    \le \frac{\epsilon}{2}
\end{align}
is satisfied with probability at least $1-\frac{\delta}{2}$.
Besides, Theorem \ref{th:OE} implies that with probability at least $1-\frac{\delta}{2}$, $\hat{\mu}^\prime$ satisfies
\begin{align}
    \left|\hat{\mu}^\prime-\braket{\mathrm{histSP+} | C{\rm SP+}}\right|\le \epsilon_{\rm OE,T},
\end{align}
and thus, noting
\begin{align}
    \hat{Y} = \frac{\|C\|_F}{\sqrt{R + 1}}\left(2\left(\frac{4r_{\rm c}}{\max\{1,\eta T\}}+R\right) U_{\tilde{\mathcal{B}}}\right)^d\braket{\mathrm{histSP+} | C{\rm SP+}},
\end{align}
we have $\left|\hat{\mu}_Y - \hat{Y}\right| \le \frac{\epsilon}{2}$.
Combining these, we see that $\hat{\mu}_Y$ satisfies $\left|\hat{\mu}_Y - \mathbb{E}_\mathbb{P}[Y]\right|\le \epsilon$ with probability at least $1-\delta$.

Lastly, let us evaluate the query complexity.
The overlap estimation in Step 8 uses $V_{\rm histSP}$ and $U^{\rm SP+}_C$ $O\left(\frac{1}{\epsilon_{\rm OE,T}}\log\left(\frac{1}{\delta}\right)\right)$ times each, thus makes $O\left(\frac{1}{\epsilon_{\rm OE,T}}\log\left(\frac{1}{\delta}\right)\right)$ uses $V_{{\rm hist+},\epsilon^\prime_{\rm T}}$ and $U_C$.
The query complexity of $V_{{\rm hist+},\epsilon^\prime_{\rm T}}$ is as stated in Corollary \ref{cor:HistStateConc+Super}.
By combining these evaluations and doing some algebra, we obtain the claimed evaluations.
\end{proof}

\begin{remark}
    Hiding logarithmic factors and $O(1)$ quantities to the power of $d$, the number of queries to $U_A$ in Eq. \eqref{eq:CompUAESTerm} is simplified to
    \begin{align}
    \tilde{O}\left(\left(\frac{\kappa_{BB^T} \alpha_A T}{\max\{1,\eta T\}}\right)^{d+\frac{1}{2}}\left(\|\mathbf{x}_0\|^2+N\sigma^2T\right)^{d/2}\frac{\|C\|_F}{\epsilon}\right).
    \label{eq:CompUAESTermSimp}
    \end{align}
    Comparing Eq. \eqref{eq:CompUAESTermSimp} to Eq. \eqref{eq:CompExaUASimp}, we see the advantage of Algorithm \ref{alg:main+} over Algorithm \ref{alg:main} for estimating expectations of terminal-time functions.
    We can also perform such an estimation using Algorithm \ref{alg:main}, since any terminal-time function can be regarded as a multi-time function with entries in $C$ corresponding to $\mathbf{X}_0,\ldots,\mathbf{X}_{r-1}$ set to 0.
    However, Algorithm \ref{alg:main+}, which involves padding of $\mathbf{X}_r$ in contrast to Algorithm \ref{alg:main}, has a better complexity for estimating expectations of terminal-time functions.

    Again, let us regard the root mean square $\sqrt{\mathbb{E}_\mathbb{P}[Y^2]}$ as the typical magnitude of $Y$ and set the accuracy $\epsilon$ accordingly.
    We can see
    \begin{align}
        \sqrt{\mathbb{E}_\mathbb{P}[Y^2]} \lesssim \left(\|\mathbf{x}_0\|^2+N\sigma^2T\right)^{d/2}\|C\|_F,
    \end{align}
    and setting the accuracy relative to this by
    \begin{align}
    \epsilon=\left(\|\mathbf{x}_0\|^2+N\sigma^2T\right)^{d/2}\|C\|_F\epsilon_{\rm rel}
    \label{eq:epsrelTerm}
    \end{align}
    with $\epsilon_{\rm rel} \in (0,1)$ transforms Eq. \eqref{eq:CompUAESTermSimp} into
    \begin{align}
        \tilde{O}\left(\left(\kappa_{BB^T} \alpha_A T\right)^{d+\frac{1}{2}}\frac{1}{\epsilon_{\rm rel}}\right),
    \end{align}
    where we also assume that $\eta T \lesssim 1$.
    This is better than $\tilde{O}\left(\left(\kappa_{BB^T} \alpha_A T\right)^{d+1}\frac{1}{\epsilon_{\rm rel}}\right)$ for Algorithm \ref{alg:main}, which is obtained by plugging Eq.~\eqref{eq:epsrelTerm} and $\eta T \lesssim 1$ into Eq.~\eqref{eq:CompExaUASimp}.
    \label{rem:CompExaTerm}
\end{remark}

\section{Proposed method 2: Euler-Maruyama-based method \label{sec:EM}}

The Dyson series-based method uses the block-encoding of the square root of the covariance matrix $\Sigma_n$, which assumes that all of its singular values are nonzero.
However, this does not necessarily hold in practice.
It is often the case that the dimension $m$ of $\mathbf{W}_t$, the Brownian motion that drives $\mathbf{X}_t$, is smaller than the dimension $N$ of $\mathbf{X}_t$ itself.
In such a case, $B(t)B(t)^T$ is not full-rank, and thus $\Sigma_n$ may not be either, hindering the use of the method for block-encoding of matrix square roots and the application of the Dyson series-based method.

To address the case of rank-deficient $B(t)B(t)^T$, we propose the second quantum algorithm, which we dub the \ac{EM}-based method.
Based on the \ac{EM} method given by Eq. \eqref{eq:EM}, the noise term is now given by $B(t_n) \Delta \mathbf{W}_n$.
The quantum state encoding this can be constructed using the \ac{PRNG} circuit $U_{\rm Rand}$ and the block-encoding of $B$, as seen below.

In this section, to avoid redundancy, we consider estimating expectations of multi-time functions only.

\subsection{Linear equation system}

To present the \ac{EM}-based method, we begin with making an assumption that replaces Assumption \ref{ass:UBB}.

\setcounter{assumprime}{3}
\begin{assumprime}
    There exists $\sigma,\alpha_{\drm B/\drm t}>0$ such that
    \begin{align}
        \|B(t)\|\le\sigma
        \label{eq:BNorm}
    \end{align}
    and
    \begin{align}
        \left\|\frac{\drm B}{\drm t}(t)\right\|\le\alpha_{\drm B/\drm t}
        \label{eq:dBdtNorm}
    \end{align}
    holds for any $t\in[0,T]$.
    Besides, letting $a_{B}\in\mathbb{N}$,
    we have access to the oracle $U_{B}$ that acts as
    \begin{align}
        U_{B} \ket{t} \ket{\psi} = \ket{t} V_{B(t)}\ket{\psi}, 
    \end{align}
    for any $t\in[0,T]$ and any $2^{a_{B}}N$-dimensional state vector $\ket{\psi}$, where $V_{B(t)}$ is a $(\sigma,a_B,0)$-block-encoding of $\tilde{B}(t)\in\mathbb{R}^{N \times N}$ in the form of $\tilde{B}(t) = \left( B(t) ~ O \right)$.
    \label{ass:UB}
\end{assumprime}

With the above assumption on $B$, we now present a rigorous statement on the \ac{EM} strong error for the current \ac{SDE} in Eq.~\eqref{eq:TgtSDEGen} as the following theorem, whose proof is deferred to Appendix~\ref{sec:PrEM}.

\begin{theorem}
    Under Assumptions \ref{ass:UA} and \ref{ass:UB} and
    \begin{align}
        \Delta t \le \frac{\eta}{4\alpha_A^2},
        \label{eq:DeltatCondEM}
    \end{align}
    we have
    \begin{align}
        \mathbb{E}_{\mathbb{P}}\left[\left\|\mathbf{X}^{\rm EM}_n - \mathbf{X}_n\right\|^2\right] \le 
        \left[\left(\frac{1}{2}(\alpha_A^2+\alpha_{\drm A/\drm t})T + \alpha_A \right)^2T^2(\|\mathbf{x}_0\|^2 + 2m\sigma^2T)+2m\alpha_{\drm B/\drm t}^2T^3\right]\frac{1}{r^2}.
    \label{eq:EMStrongErr}
    \end{align}
    \label{th:EM}
\end{theorem}

Then, let us present the linear equation system that replaces Eq. \eqref{eq:AXhistB}.
We can summarize the recursive relation in Eq.~\eqref{eq:EM} as the following linear equation system:
\begin{align}
    \mathcal{A}_{\rm EM} \mathbf{X}_{\rm histEM} = \mathbf{B}^{\rm EM}
    \label{eq:LinSysEM}
\end{align}
where
\begin{align}
    \mathcal{A}_{\rm EM} \coloneqq
    \begin{pmatrix}
    I & & & & \\
    -I-A(t_0)\Delta t & I & & & \\
    & -I-A(t_1)\Delta t & I & & \\
    & & \ddots & \ddots & \\
    & & & -I-A(t_{r-1})\Delta t & I
    \end{pmatrix}
    ,
    \mathbf{X}_{\rm histEM} \coloneqq
    \begin{pmatrix}
        \mathbf{X}^{\rm EM}_0 \\ \mathbf{X}^{\rm EM}_1 \\ \mathbf{X}^{\rm EM}_2 \\ \vdots \\ \mathbf{X}^{\rm EM}_r 
    \end{pmatrix}
    ,
    \mathbf{B}^{\rm EM} \coloneqq
    \begin{pmatrix}
    \mathbf{x}_0 \\ \boldsymbol{\Delta}^{\rm EM}_0 \\ \boldsymbol{\Delta}^{\rm EM}_1 \\ \vdots \\ \boldsymbol{\Delta}^{\rm EM}_{r-1} 
    \end{pmatrix}
    .
    \label{eq:calAEMXHistB}
\end{align}
Here, $\mathbf{X}^{\rm EM}_0,\ldots,\mathbf{X}^{\rm EM}_{r}\in\mathbb{R}^N$, and
\begin{align}
    \boldsymbol{\Delta}^{\rm EM}_n \coloneqq B(t_n) \sqrt{\Delta t} \mathbf{Z}_n
    \label{eq:DeltaEM}
\end{align}
for each $n\in[r-1]_0$, where $\mathbf{Z}_0,\ldots,\mathbf{Z}_{r-1}$ are now independent $\mathbb{R}^m$-valued random variables following $\mathcal{N}_{\mathbf{0},I}$.
Note that $\sqrt{\Delta t} \mathbf{Z}_n$ corresponds to $\Delta \mathbf{W}_n=\mathbf{W}_{t_{n+1}}-\mathbf{W}_{t_n}$.

Now, we aim to generate the quantum state encoding $\mathbf{X}_{\rm histEM}$ by applying the quantum linear system solver to Eq.~\eqref{eq:LinSysEM}.
We present a lemma that corresponds to Lemma \ref{lem:calAInvNorm}.

\begin{lemma}
    Suppose that Assumption \ref{ass:muA} holds.
    Then, $\mathcal{A}_{\rm EM}$ defined by Eq.~\eqref{eq:calAEMXHistB} satisfies
    \begin{align}
        \left\|\mathcal{A}_{\rm EM}\right\| \le 3,
        \label{eq:calAEMNorm}
    \end{align}
    \begin{align}
        \left\|\mathcal{A}_{\rm EM}^{-1}\right\| \le \frac{4r}{\max\{1, \eta T\}},
        \label{eq:calAEMInvNorm}
    \end{align}
    and
    \begin{align}
        \kappa_{\mathcal{A}_{\rm EM}} \le \frac{12r}{\max\{1, \eta T\}},
        \label{eq:kappacalAEM}
    \end{align}
    for $\Delta t$ satisfying Eq.~\eqref{eq:DeltatCondEM}.
    \label{lem:calAEMInvNorm}
\end{lemma}

\begin{proof}
Since
\begin{align}
    \mathcal{A}_{\rm EM} =
    \begin{pmatrix}
    I & & \\
    & \ddots & \\
    & & I
    \end{pmatrix}
    -
    \begin{pmatrix}
    0 & & & \\
    I & 0 & & \\
    & \ddots & \ddots & \\
    & & I & 0
    \end{pmatrix}
    \times
    \begin{pmatrix}
    I + A(t_0) \Delta t & & \\
    & \ddots & \\
    & & I + A(t_{r-1}) \Delta t
    \end{pmatrix}
    ,
\end{align}
we have
\begin{align}
    \|\mathcal{A}_{\rm EM}\| \le 1 + 1 \cdot (1+\alpha_A \Delta t) = 2 + \alpha_A \Delta t \le 3,
\end{align}
where we use
\begin{align}
    \alpha_A \Delta t \le \frac{\eta}{4\alpha_A}<1,
    \label{eq:alphaDelt1}
\end{align}
satisfied under Eq. \eqref{eq:DeltatCondEM}.
Eq. \eqref{eq:calAEMNorm} thus holds.

As Eq.~\eqref{eq:calAInvBlock} holds, the $(n,n^\prime)$-th block of
$\mathcal{A}_{\rm EM}^{-1}$ is given by
\begin{align}
    (\mathcal{A}_{\rm EM}^{-1})_{n,n^\prime} =
    \begin{cases}
        I &  ; ~(n = n^\prime) \\
        (I+A(t_{n-1})\Delta t) \cdots (I+A(t_{n^\prime})\Delta t) & ; ~ n^\prime < n\\
        O & ; ~ n < n^\prime
    \end{cases}
    .
    \label{eq:calAEMInvBlock}
\end{align}
A discussion similar to that leading to Eq.~\eqref{eq:calAInvNormTemp2} yields
\begin{align}
    \left\|\mathcal{A}_{\rm EM}^{-1}\right\| \le \sum_{k=0}^{r} \left(1-\frac{\eta \Delta t}{4}\right)^k,
\end{align}
where we use 
\begin{align}
    \|I+A(t_l)\Delta t\| \le 1-\frac{\eta \Delta t}{4},
    \label{eq:I+At}
\end{align}
which is proved later.
We immediately have $\sum_{k=0}^{r} \left(1-\frac{\eta \Delta t}{4}\right)^k \le r+1 < 4r$.
We also see
\begin{align}
    \sum_{k=0}^{r} \left(1-\frac{\eta \Delta t}{4}\right)^k \le \sum_{k=0}^{\infty} \left(1-\frac{\eta \Delta t}{4}\right)^k = \frac{4}{\eta \Delta t} = \frac{4r}{\eta T}.
\end{align}
We thus have Eq.~\eqref{eq:calAEMInvNorm}.

Combining Eqs.~\eqref{eq:calAEMNorm} and \eqref{eq:calAEMInvNorm} immediately yields Eq.~\eqref{eq:kappacalAEM}.

Finally, let us see that Eq.~\eqref{eq:I+At} holds under Eq.~\eqref{eq:DeltatCondEM}. 
We note that $I+A\Delta t=e^{A\Delta t}-\sum_{k=2}^\infty \frac{1}{k!}(A\Delta t)^k$ and thus have
\begin{align}
    \|I+A\Delta t\| \le \left\| e^{A\Delta t} \right\| + \left\|\sum_{k=2}^\infty \frac{1}{k!}(A\Delta t)^k\right\|,
    \label{eq:I+AtTemp}
\end{align}
where we abbreviate $A(t_l)$ as $A$.
The first term in the right-hand side of Eq. \eqref{eq:I+AtTemp} is bounded by
\begin{align}
    \left\| e^{A\Delta t} \right\| \le e^{-\eta \Delta t} \le 1-\frac{\eta \Delta t}{2},
    \label{eq:I+AtTemp2}
\end{align}
where the first inequality follows from Eq.~\eqref{eq:PhiUBEta}, and the second inequality follows from Eq. \eqref{eq:1-e-xLB} and $\eta \Delta t \le \frac{\eta^2}{4\alpha_A^2}<1$.
The second term in the right-hand side of \eqref{eq:I+AtTemp} is bounded by
\begin{align}
    \left\|\sum_{m=2}^\infty \frac{1}{m!}(A\Delta t)^m\right\| \le \sum_{m=2}^\infty \frac{1}{m!}(\alpha_A \Delta t)^m = e^{\alpha_A \Delta t} - 1 - \alpha_A \Delta t \le (\alpha_A \Delta t)^2 \le \frac{\eta \Delta t}{4}.
    \label{eq:I+AtTemp3}
\end{align}
Here, the second inequality follows from Eq. \eqref{eq:alphaDelt1} and $e^x-1-x \le x^2$, which holds for $x\in[0,1]$, and the third inequality follows from Eq.~\eqref{eq:DeltatCondEM}.
Combining Eqs.~\eqref{eq:I+AtTemp}, \eqref{eq:I+AtTemp2}, and \eqref{eq:I+AtTemp3} yields Eq.~\eqref{eq:I+At}.
\end{proof}

\subsection{Preparation of the history state}

Then, we can prepare the history state under the \ac{EM} approximation, as stated in the following theorem. 

\begin{theorem}
    Suppose that Assumptions \ref{ass:PRN}, \ref{ass:UA}, \ref{ass:muA}, \ref{ass:UB}, and \ref{ass:Ux0} hold.
    Let $\epsilon\in(0,1]$, $U_{\rm SN}>0$, and $N_{\rm s}\in\mathbb{N}$.
    Set $r$ so that
    \begin{align}
        r\ge\left\lceil  \frac{4\alpha_A^2T}{\eta}\right\rceil,
        \label{eq:rReqEM}
    \end{align}
    and suppose that $z_1,\ldots,z_{r m N_{\rm s}}$ generated by $U_{\rm Rand}$ satisfy Eq.~\eqref{eq:ziUB2}.
    Then, we can construct an oracle $V_{{\rm histEM},\epsilon}$, which acts on a three-register system as
    \begin{align}
        V_{{\rm histEM},\epsilon}\ket{i}\ket{0}\ket{0} =   \ket{i}\ket{0}\ket{\tilde{\mathbf{X}}_{\rm histEM}^{(i)}} + \ket{0_\perp}
        \label{eq:V1histEM}
    \end{align}
    for any $i\in[N_{\rm s}]$.
    Here, $\ket{0_\perp}$ is some unnormalized state such that $(I\otimes \ket{0}\!\bra{0}\otimes I)\ket{0_\perp}=0$, and $\ket{\tilde{\mathbf{X}}_{\rm histEM}^{(i)}}$ is an unnormalized state given by
    \begin{align}
        \ket{\tilde{\mathbf{X}}_{\rm histEM}^{(i)}} \coloneqq \frac{\max\{1,\eta T\}}{8r U_{\mathbf{B}^{\rm EM}}}\tilde{\mathbf{X}}_{\rm histEM}
        \label{eq:XHistEMKet}
    \end{align}
    with $\tilde{\mathbf{X}}_{\rm histEM}^{(i)}\in\mathbb{R}^{(r+1)N}$ satisfying
    \begin{align}
        \left\| \tilde{\mathbf{X}}_{\rm histEM}^{(i)} - 
        \mathbf{X}_{\rm histEM}^{(i)}\right\| \le U_{\mathbf{B}^{\rm EM}} \epsilon,
        \label{eq:XHistEMErr}
    \end{align}
    where
    \begin{align}
        U_{\mathbf{B}^{\rm EM}} \coloneqq \sqrt{\|\mathbf{x}_0\|^2+m\sigma^2T U_{\rm SN}^2},
        \label{eq:UBEM}
    \end{align}
    and $\mathbf{X}_{\rm histEM}^{(i)}$ is equal to $\mathbf{X}_{\rm histEM}$ in Eq. \eqref{eq:LinSysEM} with $\mathbf{Z}_n,n\in[r-1]_0$ set to
    \begin{align}
        \mathbf{z}^{(i)}_n \coloneqq
        \begin{pmatrix}
            z_{(i-1)rm+nm+1} \\ \vdots \\ z_{(i-1)rm+(n+1)m}
        \end{pmatrix}.
    \end{align}
    $V_{{\rm histEM},\epsilon}$ uses the controlled $U_A$
    \begin{align}
        O\left(\frac{r}{\max\{1,\eta T\}}\log\left(\frac{r}{\epsilon\max\{1,\eta T\}}\right)\right)
    \end{align}
    times, the controlled $U_B$, $U_{\mathbf{x}_0}$, and $U_{\rm Rand}$ once each, and
    \begin{align}
        O\left(\frac{a_A r}{\max\{1,\eta T\}}\log\left(\frac{r}{\epsilon\max\{1,\eta T\}}\right)\right)
    \end{align}
    additional elementary gates.
    \label{th:HistStateConcEM}
\end{theorem}

\begin{proof}
    Since we can construct an $(1, a_A, 0)$-block-encoding of
    \begin{align}
        \frac{1}{\alpha_A}
        \begin{pmatrix}
            A(t_0) & & \\
            & \ddots & \\
            & & A(t_{r-1})
        \end{pmatrix}
    \end{align}
    with one query to $U_A$ and $O(1)$ additional elementary gates, we can construct a $(1+\alpha_A \Delta t, a_A+1, 0)$-block-encoding of
    \begin{align}
        \begin{pmatrix}
            I+A(t_0)\Delta t & & \\
            & \ddots & \\
            & & I+A(t_{r-1})\Delta t
        \end{pmatrix}
        ,
    \end{align}
    as implied by Lemma 52 in the full version of Ref.~\cite{Gilyen2019}, and then a $(2+\alpha_A \Delta t, a_A+2, 0)$-block-encoding $V_{\mathcal{A}_{\rm EM}}$ of $\mathcal{A}_{\rm EM}$, with one query to $U_A$ and $O(1)$ additional elementary gates, similarly to $V_{\tilde{\mathcal{A}}}$ in Corollary~\ref{th:HistStateConc}.
    Thus, noting that Lemma~\ref{lem:calAEMInvNorm} implies that Eqs.~\eqref{eq:calAEMNorm} and \eqref{eq:kappacalAEM} hold for the current $r$, we see from Theorem~\ref{eq:MatInvQSVT} that we can construct an $(1,a_A+2,\epsilon_{\rm EM})$-block-encoding $V_{\mathcal{A}_{\rm EM}^{-1}}$ of $\frac{\max\{1,\eta T\}}{8r}\mathcal{A}_{\rm EM}^{-1}$, where
    \begin{align}
        \epsilon_{\rm EM} \coloneqq \frac{\max\{1,\eta T\}}{8r}\epsilon,
    \end{align}
    by using $U_A$
    \begin{align}
        O\left(\frac{r}{\max\{1,\eta T\}}\log\left(\frac{1}{\epsilon_{\rm EM}}\right)\right)
    \end{align}
    times and additional elementary gates
    \begin{align}
    O\left(\frac{a_A r}{\max\{1,\eta T\}}\log\left(\frac{1}{\epsilon_{\rm EM}}\right)\right)
    \end{align}
    times.

    On the other hand, in the same way as $U_{\tilde{\boldsymbol{\Delta}}}$ in Lemma~\ref{lem:Delta}, but with $U_{\sqrt{\tilde{\Sigma}}}$ replaced by $U_B$, we can construct an oracle $U_{\boldsymbol{\Delta}^{\rm EM}}$ that acts as
    \begin{align}
        U_{\boldsymbol{\Delta}^{\rm EM}}\ket{0}\ket{i}\ket{n}\ket{0}= \ket{0}\ket{i}\ket{n}\ket{\boldsymbol{\Delta}^{{\rm EM},(i)}_{n}} +\ket{0_\perp} 
        \label{eq:UDeltaEM}
    \end{align}
    for any $i\in[N_{\rm s}]$ and $n\in[r-1]_0$.
    Here, the unnormalized state $\ket{\boldsymbol{\Delta}^{{\rm EM},(i)}_{n}}$ is given by 
    \begin{align}
        \ket{\boldsymbol{\Delta}^{{\rm EM},(i)}_{n}} \coloneqq \frac{1}{\sqrt{m\sigma^2\Delta t}U_{\rm SN}}\sum_{j=1}^N \Delta^{{\rm EM},(i)}_{n,j} \ket{j}
    \end{align}
    with $N$-dimensional real vectors
    \begin{align}
        \boldsymbol{\Delta}^{{\rm EM},(i)}_{n}=
        \begin{pmatrix}
            \Delta^{{\rm EM},(i)}_{n,1} \\ \vdots \\ \Delta^{{\rm EM},(i)}_{n,N}
        \end{pmatrix}
        \coloneqq B(t_n)\sqrt{\Delta t}\mathbf{z}^{(i)}_n,
    \end{align}
    and $\ket{0_\perp}$ is some unnormalized state satisfying $(\ket{0}\!\bra{0} \otimes I \otimes I \otimes I)\ket{0_\perp}=0$.
    Then, similarly to $V_{\tilde{\mathcal{B}}}$ in Corollary \ref{cor:HistStateConcSuper}, we can construct
    $V_{\mathbf{B}^{\rm EM}}$ acting as
    \begin{align}
        V_{\mathbf{B}^{\rm EM}}\ket{i}\ket{0}\ket{0}=\ket{i}\ket{0}\ket{\mathbf{B}^{{\rm EM},(i)}} + \ket{0_\perp}
    \end{align}
    for any $i\in[N_{\rm s}]$, using the controlled version of $U_{\boldsymbol{\Delta}^{\rm EM}}$ and $U_{\mathbf{x}_0}$ once each and additional $O(1)$ elementary gates.
    Here, $\ket{0_\perp}$ is some unnormalized state such that $(I \otimes \ket{0}\!\bra{0} \otimes I)\ket{0_\perp}=0$, and an unnormalized state $\ket{\mathbf{B}^{{\rm EM},(i)}}$ is given by
    \begin{align}
        \ket{\mathbf{B}^{{\rm EM},(i)}} \coloneqq \frac{1}{U_{\mathbf{B}^{\rm EM}}}\mathbf{B}^{{\rm EM},(i)},
        \mathbf{B}^{{\rm EM},(i)} \coloneqq 
        \begin{pmatrix}
        \mathbf{x}_0 \\ \boldsymbol{\Delta}^{{\rm EM},(i)}_0 \\ \vdots \\ \boldsymbol{\Delta}^{{\rm EM},(i)}_{r-1} 
        \end{pmatrix}
        .
    \end{align}

    Now, by using $V_{\mathbf{B}^{\rm EM}}$ and $V_{\mathcal{A}_{\rm EM}^{-1}}$ once each, we can construct a quantum circuit $V_{{\rm histEM},\epsilon}$ that operates on a initial state $\ket{i}\ket{0}\ket{0}$ to generate the quantum state
    \begin{align}
        \ket{i}\ket{0} \otimes \frac{\max\{1,\eta T\}}{8rU_{\mathbf{B}^{\rm EM}}}\mathcal{C}^{\rm EM}\mathbf{B}^{{\rm EM},(i)}  + \ket{0_\perp},
    \end{align}
    where $\mathcal{C}^{\rm EM}$ is a matrix such that
    \begin{align}
        \left\|\frac{\max\{1,\eta T\}}{8r}\mathcal{C}^{\rm EM}-\frac{\max\{1,\eta T\}}{8r}\mathcal{A}_{\rm EM}^{-1}\right\| \le \epsilon_{\rm EM}.
    \end{align}
    This implies
    \begin{align}
        \left\|\mathcal{C}^{\rm EM}-\mathcal{A}_{\rm EM}^{-1}\right\| \le \epsilon,
        \label{eq:calCcalAInvEM}
    \end{align}
    which, along with
    \begin{align}
        \left\|\mathbf{B}^{{\rm EM},(i)}\right\|^2 =\|\mathbf{x}_0\|^2  +\sum_{n=0}^{r-1}\left\|\boldsymbol{\Delta}^{{\rm EM},(i)}_n\right\|^2  \le \|\mathbf{x}_0\|^2 + \sum_{n=0}^{r-1} \left\|B(t_n)\right\|^2 \Delta t \|\mathbf{z}^{(i)}_n\|^2 \le \|\mathbf{x}_0\|^2 + r\sigma^2 \Delta t m U_{\rm SN}^2=U_{\mathbf{B}^{\rm EM}}^2,
    \end{align}
    leads to
    \begin{align}
        \left\|\mathcal{C}^{\rm EM}\mathbf{B}^{{\rm EM},(i)}-\mathbf{X}_{\rm histEM}^{(i)}\right\| \le \left\|\mathcal{C}^{\rm EM}-\mathcal{A}_{\rm EM}^{-1}\right\| \cdot \left\|\mathbf{B}^{{\rm EM},(i)}\right\| \le U_{\mathbf{B}^{\rm EM}} \epsilon.
    \end{align}
    This means that $V_{{\rm histEM},\epsilon}$ is nothing but the stated one.
    The stated complexity bounds are obtained by summing up the number of queries to each component circuit in $V_{\mathbf{B}^{\rm EM}}$ and $V_{\mathcal{A}_{\rm EM}^{-1}}$.
\end{proof}

\subsection{Estimating expectations of multi-time functions}

Finally, we consider the quantum algorithm to estimate expectations of multi-time functions under the \ac{EM} approximation.
Since the estimation is again via the sample average, we first present the lemma on the variance of $\mathbf{X}_{\rm histEM}$ and its approximation $\tilde{\mathbf{X}}_{\rm histEM}$. 

\begin{lemma}
    Let $\epsilon\in\left(0,1\right)$.
    Suppose that Assumptions \ref{ass:UA}, \ref{ass:muA}, \ref{ass:UB}, and \ref{ass:Ux0} hold.
    Set $r$ as Theorem \ref{th:HistStateConcEM}.
    Suppose that Assumption \ref{ass:PRN} holds with
    \begin{align}
        \epsilon_{\rm PRN} = \frac{(2d-1)!!}{(2r\sqrt{m})^d}.
        \label{eq:epsPRNEM}
    \end{align}
    Then, for any $i\in\mathbb{N}$, $\mathbf{X}_{\rm histEM}^{(i)}$ and $\tilde{\mathbf{X}}_{\rm histEM}^{(i)}$ in Theorem \ref{th:HistStateConcEM} satisfy
    \begin{align}
        \mathbb{E}\left[\left\|\left(\mathbf{X}_{\rm histEM}^{(i)}\right)^{\otimes d}\right\|^2\right]\le 2(2d-1)!!(r+1)^d\left(\|\mathbf{x}_0\|^2+(m+2d)\sigma^2T\right)^d,
        \label{eq:EXhistEMd2}
    \end{align}
    and
    \begin{align}
        \mathbb{E}\left[\left\|\left(\tilde{\mathbf{X}}_{\rm histEM}^{(i)}\right)^{\otimes d}\right\|^2\right]\le 2(2d-1)!!(1+\epsilon)^{2d}(r+1)^d\left(\|\mathbf{x}_0\|^2+(m+2d)\sigma^2T\right)^d.
        \label{eq:EXtilhistEMd2}
    \end{align}
    respectively.
    \label{lem:XhistEMMom}
\end{lemma}

\begin{proof}
    The proof of Eq.~\eqref{eq:EXhistEMd2} is similar to that of Eq.~\eqref{eq:EXhistd2}, except that we consider not $\mathbf{X}_n^{(i)}$ but $\mathbf{X}_n^{{\rm EM},(i)}\in\mathbb{R}^N$, which is the $n$-th block of $\mathbf{X}_{\rm histEM}^{(i)}$:
    \begin{align}
        \mathbf{X}_{\rm histEM}^{(i)} \coloneqq
    \begin{pmatrix}
        \mathbf{X}_0^{{\rm EM},(i)} \\ \vdots \\ \mathbf{X}_r^{{\rm EM}.(i)} 
    \end{pmatrix}
    \end{align}
    $\mathbf{X}_n^{{\rm EM},(i)}$ is written by
     \begin{align}
    \mathbf{X}_n^{{\rm EM},(i)} = \sum_{l=-1}^{r-1} \mathcal{D}^{\rm EM}_{n,l}  \mathbf{z}^{(i)}_{l},
    \label{eq:XnEMSum}
    \end{align}
    with
    \begin{align}
        \mathcal{D}^{\rm EM}_{n,l} \coloneqq
        \begin{cases}
            (\mathcal{A}_{\rm EM}^{-1})_{n,0} & ; ~ l=-1 \\
            (\mathcal{A}_{\rm EM}^{-1})_{n,l+1} B(t_l) \sqrt{\Delta t} & ; ~ l\in[r-1]_0
        \end{cases}
    \end{align}
    and $\mathbf{z}^{(i)}_{-1} = \mathbf{x}_0$.
    Eq.~\eqref{eq:I+At} implies
    \begin{align}
        \|(\mathcal{A}_{\rm EM}^{-1})_{n,n^\prime}\| \le 1,
        \label{eq:calAEMInvBlockNorm}
    \end{align}
    and thus
    \begin{align}
        \|\mathcal{D}^{\rm EM}_{n,l}\| \le
        \begin{cases}
        1 & ; ~ l=-1 \\
        \sqrt{\sigma^2\Delta t}, & ; ~ l\in[r-1]_0
        \end{cases}
        .
        \label{eq:calDEMNorm}
    \end{align}
    Then, we can derive Eq.~\eqref{eq:EXhistEMd2} by the same discussion as the proof of Eq.~\eqref{eq:EXhistd2}.
    
    The proof of Eq.~\eqref{eq:EXtilhistEMd2} also goes similarly.
    We consider $\tilde{\mathbf{X}}_n^{{\rm EM},(i)}$ written as
    \begin{align}
    \tilde{\mathbf{X}}_n^{{\rm EM},(i)} = \sum_{l=-1}^{r-1} \tilde{\mathcal{D}}^{\rm EM}_{n,l}  \tilde{\mathbf{z}}^{(i)}_{l},
    \label{eq:XnEMtilSum}
    \end{align}
    with
    \begin{align}
        \tilde{\mathcal{D}}^{\rm EM}_{n,l} \coloneqq
        \begin{cases}
            \mathcal{C}^{\rm EM}_{n,0} & ; ~ l=-1 \\
            \mathcal{C}^{\rm EM}_{n,l+1} B(t_l) \sqrt{\Delta t} & ; ~ l\in[r-1]_0
        \end{cases}
        .
    \end{align}
    Here, $\mathcal{C}^{\rm EM}_{n,n^\prime}$ is the $(n,n^\prime)$-th block of $\mathcal{C}^{\rm EM}$, for which Eqs.~\eqref{eq:calCcalAInvEM} and \eqref{eq:calAEMInvBlockNorm} imply
    \begin{align}
        \|\mathcal{C}^{\rm EM}_{n,n^\prime}\| \le 1+\epsilon,
    \end{align}
    and thus
    \begin{align}
        \|\tilde{\mathcal{D}}^{\rm EM}_{n,l}\| \le
        \begin{cases}
        1 & ; ~ l=-1 \\
        (1+\epsilon)\sqrt{\sigma^2\Delta t}, & ; ~ l\in[r-1]_0
        \end{cases}
        .
        \label{eq:calDtilEMNorm}
    \end{align}
    Then, using this in the same discussion as the proof of Eq.~\eqref{eq:EXhisttiltild2}, we obtain Eq.~\eqref{eq:EXtilhistEMd2}.

\end{proof}

We then have the following lemma on the error between the sample average under the \ac{EM} approximation and the true expectation.
Note that the error due to time discretization stated in Theorem \ref{th:EM} now appears.

\begin{lemma}
    Let $N_{\rm s}\in\mathbb{N}$, $\epsilon\in\left(0,\frac{1}{2}\right)$, and $\delta\in(0,1)$.
    Suppose that Assumptions \ref{ass:UA}, \ref{ass:muA}, \ref{ass:UB}, and \ref{ass:Ux0} hold.
    Set $r$ to $r_{\rm c}^{\rm EM}$ in Eq.~\eqref{eq:rReqEM}.
    Suppose that Assumption \ref{ass:PRN} holds with
    \begin{align}
        \epsilon_{\rm PRN} = \min\left\{\frac{(2d-1)!!}{(2r\sqrt{m})^d},\frac{d\sqrt{(2d-3)!!}\epsilon}{(2r)^{d/2}},\frac{d^2(2d-3)!!\delta\epsilon^2}{8(2r)^d}\right\}.
        \label{eq:epsPRNEM2}
    \end{align}
    Set
    \begin{align}
        U_{\rm SN} = \max\left\{\sqrt{2 \log \left(\frac{8r m N_{\rm s}}{\sqrt{2\pi}\delta}\right)},1\right\}.
        \label{eq:USNEM}
    \end{align}
    Define $Y$ as Eq.~\eqref{eq:Y}.
    Define
    \begin{align}
        \tilde{Y}^{(i)}_{\rm EM} \coloneqq \sum_{j_1,\ldots,j_d=1}^{N(r+1)} C_{j_1,\ldots,j_d} \tilde{X}_{\rm histEM}^{(i),j_1} \cdots \tilde{X}_{\rm histEM}^{(i),j_d}, i\in\mathbb{N},
    \end{align}
    and
    \begin{align}
        \hat{Y}_{\rm EM} \coloneqq \frac{1}{N_{\rm s}}\sum_{i=1}^{N_{\rm s}}\tilde{Y}_{\rm EM}^{(i)},
    \end{align}
    where $\tilde{X}_{\rm histEM}^{(i),j}$ is the $j$-th entry of $\tilde{\mathbf{X}}^{(i)}_{\rm histEM}$.
    Then,
    \begin{align}
        & \left|\hat{Y}_{\rm EM}-\mathbb{E}_\mathbb{P}[Y]\right| \le \sqrt{(2d-3)!!}(r+1)^{d/2} \left(\|\mathbf{x}_0\|^2+(m+2d)\sigma^2T\right)^{d/2}\|C\|_F \nonumber \\
        & \quad \times \left[
        (1+\epsilon)^d \left(\sqrt{\frac{8(2d-1)}{\delta N_{\rm s}}}+8d\epsilon\right) + \frac{d}{r}\sqrt{2\left(\frac{1}{2}(\alpha_A^2+\alpha_{\drm A/\drm t})T + \alpha_A \right)^2T^2+\frac{2m\alpha_{\drm B/\drm t}^2T^3}{\|\mathbf{x}_0\|^2+(m+2d)\sigma^2T}}\right]
        \label{eq:YErrProbEM}
    \end{align}
    holds with probability at least $1-\frac{\delta}{2}$.
    \label{lem:YHatEM}
\end{lemma}

\begin{proof}
    We define
    \begin{align}
        \tilde{\mathbf{X}}_{\rm histEMNB}^{(i)} \coloneqq
        \begin{pmatrix}
            \tilde{\mathbf{X}}_{{\rm EMNB},0}^{(i)} \\
            \vdots \\
            \tilde{\mathbf{X}}_{{\rm EMNB},r}^{(i)}
        \end{pmatrix}
        ,
        \tilde{\mathbf{X}}_{{\rm EMNB},n}^{(i)} \coloneqq \sum_{l=-1}^{r-1} \tilde{\mathcal{D}}^{\rm EM}_{n,l}  \mathbf{z}^{(i)}_l,
    \end{align}
    and
    \begin{align}
        Y_{\rm EM} \coloneqq \sum_{j_1,\ldots,j_d=1}^{N(r+1)} C_{j_1,\ldots,j_d} X_{\rm histEM}^{j_1} \cdots X_{\rm histEM}^{j_d}, ~
        \tilde{Y}^{(i)}_{\rm EMNB} \coloneqq \sum_{j_1,\ldots,j_d=1}^{N(r+1)} C_{j_1,\ldots,j_d} \tilde{X}_{\rm histEMNB}^{(i),j_1} \cdots \tilde{X}_{\rm histEMNB}^{(i),j_d},
    \end{align}
    where $X_{\rm histEM}^{j}$ and $\tilde{X}_{\rm histEMNB}^{(i),j}$ are the $j$-th entries of $\mathbf{X}_{\rm histEM}$ and $\tilde{\mathbf{X}}_{\rm histEMNB}^{(i)}$, respectively.
    We have
    \begin{align}
        \left\|\tilde{\mathcal{D}}^{\rm EM}_{n,l}- \mathcal{D}^{\rm EM}_{n,l}\right\|\le
        \begin{cases}
            \epsilon & ; ~ l=-1 \\
            \sqrt{\sigma^2 \Delta t}\epsilon & ; ~ l\in[r-1]_0
        \end{cases}
        ,
        \label{eq:calDEMtilDiff}
    \end{align}
    which follows from Eqs. \eqref{eq:BNorm} and \eqref{eq:calCcalAInvEM}.
    Then, using a discussion similar to that leading to Eq.~\eqref{eq:BiasYtiltilNB} along with Eqs.~\eqref{eq:calDEMNorm}, \eqref{eq:calDtilEMNorm}, and 
    \eqref{eq:calDEMtilDiff}, we obtain
    \begin{align}
        \left|\mathbb{E}\left[\tilde{Y}^{(i)}_{\rm EMNB}\right]-\mathbb{E}_\mathbb{P}\left[Y_{\rm EM}\right]\right| \le 7d\sqrt{(2d-3)!!}(1+\epsilon)^{d-1}\epsilon(r+1)^{d/2}\left(\|\mathbf{x}_0\|^2+(m+2d)\sigma^2T\right)^{d/2} \|C\|_F.
        \label{eq:BiasYtilEMNB}
    \end{align}
    
    We also have
    \begin{align}
        & \left|\mathbb{E}_\mathbb{P}\left[Y_{\rm EM}\right]-\mathbb{E}_\mathbb{P}\left[Y\right]\right| \nonumber \\
        &\le \mathbb{E}_\mathbb{P}\left[\|C\|_F \left\|\left(\mathbf{X}_{\rm histEM}\right)^{\otimes d}-\mathbf{X}_{\rm hist}^{\otimes d}\right\|\right] \nonumber \\
        &\le \|C\|_F \sum_{k=0}^{d-1} \mathbb{E}_\mathbb{P}\left[\left\|\mathbf{X}_{\rm histEM}^{\otimes (d-k-1)} \otimes \mathbf{X}_{\rm hist}^{\otimes k} \otimes \left(\mathbf{X}_{\rm histEM}-\mathbf{X}_{\rm hist}\right)\right\|\right]  \nonumber \\
        &\le \|C\|_F \sum_{k=0}^{d-1} \sqrt{\mathbb{E}_\mathbb{P}\left[\left\|\mathbf{X}_{\rm histEM}^{\otimes (d-k-1)} \otimes \mathbf{X}_{\rm hist}^{\otimes k}\right\|^2\right] \cdot \mathbb{E}_\mathbb{P}\left[\left\|\mathbf{X}_{\rm histEM}-\mathbf{X}_{\rm hist}\right\|^2\right]}  \nonumber \\
        &\le  d\sqrt{(2d-3)!!}(r+1)^{\frac{d-1}{2}}\left(\|\mathbf{x}_0\|^2+(m+2d)\sigma^2T\right)^{\frac{d-1}{2}}\|C\|_F\sqrt{\frac{\left(\frac{1}{2}(\alpha_A^2+\alpha_{\drm A/\drm t})T + \alpha_A \right)^2T^2(\|\mathbf{x}_0\|^2 + 2m\sigma^2T)+2m\alpha_{\drm B/\drm t}^2T^3}{r}}.
        \label{eq:EMErrY}
    \end{align}
    Here, we use
    \begin{align}
        \mathbb{E}_\mathbb{P}\left[\left\|\left(\mathbf{X}_{\rm hist}\right)^{\otimes k} \otimes \left(\mathbf{X}_{\rm histEM}\right)^{\otimes k^\prime}\right\|^2\right]\le (2(k+k^\prime)-1)!!(r+1)^{k+k^\prime}\left(\|\mathbf{x}_0\|^2+(m+2d)\sigma^2T\right)^{k+k^\prime},
        \label{eq:XhistEMXhist}
    \end{align}
    which is proved similarly to Eq.~\eqref{eq:XtiltilhistNBXhist}, and
    \begin{align}
        \mathbb{E}_\mathbb{P}\left[\left\|\mathbf{X}_{\rm histEM}-\mathbf{X}_{\rm hist}\right\|^2\right] \le \frac{\left(\frac{1}{2}(\alpha_A^2+\alpha_{\drm A/\drm t})T + \alpha_A \right)^2T^2(\|\mathbf{x}_0\|^2 + 2m\sigma^2T)+2m\alpha_{\drm B/\drm t}^2T^3}{r},
    \end{align}
    which follows from Theorem \ref{th:EM}.

    Now, let us define
    \begin{align}
        \hat{Y}_{\rm EMNB} \coloneqq \frac{1}{N_{\rm s}}\sum_{i=1}^{N_{\rm s}}\tilde{Y}^{(i)}_{\rm EMNB}.
    \end{align}
    By a discussion similar to the one leading to Eq.~\eqref{eq:VarYhatNB}, we get
    \begin{align}
        {\rm Var}[\hat{Y}_{\rm EMNB}] \le \left(\frac{2(2d-1)!!}{N_{\rm s}} + \frac{1}{4}d^2(2d-3)!!\delta\epsilon^2\right)(1+\epsilon)^{2d}(r+1)^d\left(\|\mathbf{x}_0\|^2+(m+2d)\sigma^2T\right)^d\|C\|_F^2.
    \end{align}
    Then, using Chebyshev's inequality along with Eqs. \eqref{eq:BiasYtilEMNB} and \eqref{eq:EMErrY}, we see that with probability at least $1-\frac{\delta}{4}$,
    \begin{align}
        & \left|\hat{Y}_{\rm EMNB}-\mathbb{E}_\mathbb{P}[Y]\right|\nonumber \\
        & \le \left|\hat{Y}_{\rm NB}-\mathbb{E}[\hat{Y}_{\rm EMNB}]\right|+\left|\mathbb{E}[\hat{Y}_{\rm EMNB}]-\mathbb{E}_\mathbb{P}[Y_{\rm EM}]\right| + \left|\mathbb{E}_\mathbb{P}[Y_{\rm EM}]-\mathbb{E}_\mathbb{P}[Y]\right|\nonumber \\
        & \le \sqrt{(2d-1)!!}(1+\epsilon)^d(r+1)^{d/2} \left(\|\mathbf{x}_0\|^2+(m+2d)\sigma^2T\right)^{d/2}\|C\|_F \sqrt{\frac{8}{\delta N_{\rm s}}} \nonumber \\
        & ~ + 8d\sqrt{(2d-3)!!}(1+\epsilon)^d(r+1)^{d/2}\left(\|\mathbf{x}_0\|^2+(m+2d)\sigma^2T\right)^{d/2} \|C\|_F \epsilon \nonumber \\
        & ~ +d\sqrt{(2d-3)!!}(r+1)^{\frac{d-1}{2}}\left(\|\mathbf{x}_0\|^2+(m+2d)\sigma^2T\right)^{\frac{d-1}{2}}\|C\|_F\sqrt{\frac{\left(\frac{1}{2}(\alpha_A^2+\alpha_{\drm A/\drm t})T + \alpha_A \right)^2T^2(\|\mathbf{x}_0\|^2 + 2m\sigma^2T)+2m\alpha_{\drm B/\drm t}^2T^3}{r}}
    \end{align}
    holds, which is rearranged to Eq.~\eqref{eq:YErrProbEM}.

    Besides, as seen in the proof of Lemma \ref{lem:YHat}, with probability at least $1-\frac{\delta}{4}$, $|z_1|,\ldots,|z_{r_{\rm c} mN_{\rm s}}| \le U_{\rm SN}$ holds, and thus $\hat{Y}_{\rm EM}=\hat{Y}_{\rm EMNB}$ holds.
    
    Combining the above findings, we see that the claim holds.

\end{proof}

Finally, the quantum algorithm to estimate expectations of multi-time functions and its complexity are as follows.

\begin{theorem}
    Let $N_{\rm s}\in\mathbb{N}$, $\epsilon>0$, and $\delta\in(0,1)$.
    Let $C$ be as in Theorem \ref{th:QAOE}.
    Suppose that Assumptions \ref{ass:UA}, \ref{ass:muA}, \ref{ass:UB}, and \ref{ass:Ux0} hold.
    Suppose that we can take a positive integer $r^{\rm EM}_{\rm c}$ and a positive real number $\epsilon^\prime_{\rm EM}$ satisfying
    \begin{align}
        r^{\rm EM}_{\rm c} &\ge \max\left\{\frac{4\alpha_A^2T}{\eta},\frac{1}{8\epsilon^\prime_{\rm EM}}\sqrt{2\left(\frac{1}{2}(\alpha_A^2+\alpha_{\drm A/\drm t})T + \alpha_A \right)^2T^2+\frac{2m\alpha_{\drm B/\drm t}^2T^3}{\|\mathbf{x}_0\|^2+(m+2d)\sigma^2T}} \right\}, \label{eq:rEMc} \\
        \epsilon^\prime_{\rm EM} &\le \min\left\{\frac{\epsilon}{2^{d+5}d \sqrt{(2d-3)!!}(r^{\rm EM}_{\rm c}+1)^{d/2} \left(\|\mathbf{x}_0\|^2+(m+2d)\sigma^2T\right)^{d/2}\|C\|_F}, 1\right\}.
        \label{eq:epsPrEM}
    \end{align}
    Suppose that Assumption \ref{ass:PRN} holds with
    \begin{align}
        \epsilon_{\rm PRN} = \min\left\{\frac{(2d-1)!!}{(2r^{\rm EM}_{\rm c}\sqrt{m})^d},\frac{d\sqrt{(2d-3)!!}\epsilon^\prime_{\rm EM}}{(2r^{\rm EM}_{\rm c})^{d/2}},\frac{d^2(2d-3)!!\delta(\epsilon^\prime_{\rm EM})^2}{8(2r^{\rm EM}_{\rm c})^d}\right\}.
    \end{align}
    Then, there is a quantum algorithm that, with probability at least $1-\delta$, outputs an estimate $\hat{\mu}_Y$ of $\mathbb{E}_\mathbb{P}[Y]$ with error at most $\epsilon$.
    This algorithm uses the controlled $U_A$
    \begin{align}
        O\left(\left(\frac{8r^{\rm EM}_{\rm c}}{\max\{1,\eta T\}}\sqrt{\|\mathbf{x}_0\|^2+m\sigma^2T \log \left(\frac{r^{\rm EM}_{\rm c} m}{\delta^2 \epsilon_{\rm EM}^{\prime2}}\right)}\right)^d\frac{dr^{\rm EM}_{\rm c} \|C\|_F}{\max\{1,\eta T\}\epsilon}\log\left(\frac{1}{\delta}\right)\right)
    \label{eq:CompEM}
    \end{align}
    times, $U_C$ and the controlled $U_B$, $U_{\mathbf{x}_0}$, and $U_{\rm Rand}$
    \begin{align}
        O\left(\left(\frac{8r^{\rm EM}_{\rm c}}{\max\{1,\eta T\}}\sqrt{\|\mathbf{x}_0\|^2+m\sigma^2T \log \left(\frac{r^{\rm EM}_{\rm c} m}{\delta^2 \epsilon_{\rm EM}^{\prime2}}\right)}\right)^d\frac{\|C\|_F}{\epsilon}\log\left(\frac{1}{\delta}\right)\right)
    \end{align}
    times each, and
    \begin{align}
        O\left(\left(\frac{8r^{\rm EM}_{\rm c}}{\max\{1,\eta T\}}\sqrt{\|\mathbf{x}_0\|^2+m\sigma^2T \log \left(\frac{r^{\rm EM}_{\rm c} m}{\delta^2 \epsilon_{\rm EM}^{\prime2}}\right)}\right)^d\frac{a_A d r^{\rm EM}_{\rm c} \|C\|_F}{\max\{1,\eta T\}\epsilon}\log\left(\frac{1}{\delta}\right)\right)
    \end{align}
    additional elementary gates.
    \label{th:ExpEM}
\end{theorem}

\begin{proof}

The quantum algorithm is as shown in Algorithm \ref{alg:mainEM}.

\begin{algorithm}[t]
\caption{Proposed quantum algorithm for estimating $\mathbb{E}_\mathbb{P}[Y]$ based on the EM method}\label{alg:mainEM}
\begin{algorithmic}[1]

\State Set $r$ to $r^{\rm EM}_{\rm c}$.

\State Set
\begin{align}
    N_{\rm s}^{\rm EM} = \left\lceil \frac{2d-1}{8 d^2 \delta (\epsilon_{\rm EM}^\prime)^2} \right\rceil.
    \label{eq:NsEM}
\end{align}

\State Set $U_{\rm SN}$ as
\begin{align}
    U_{\rm SN} = \max\left\{\sqrt{2 \log \left(\frac{8r^{\rm EM}_{\rm c} m N_{\rm s}^{\rm EM}}{\sqrt{2\pi}\delta}\right)},1\right\}.
\end{align}

\State Construct the oracle $V_{{\rm histEM},\epsilon^\prime_{\rm EM}}$.

\State By a use of $U^{\rm SP}_{1,N_{\rm s}^{\rm EM}}$ and $d$ uses of $V_{{\rm histEM},\epsilon^\prime_{\rm EM}}$, construct an oracle $V_{\rm histEMSP}$ that acts as
\begin{align}
    V_{\rm histEMSP}\ket{0}\ket{0}\ket{0}^{\otimes d} = \ket{\rm histEMSP} \coloneqq \frac{1}{\sqrt{N_{\rm s}^{\rm EM}}}\sum_{i=1}^{N_{\rm s}^{\rm EM}}  \ket{i}\left(\ket{0}\Ket{\tilde{\mathbf{X}}_{\rm histEM}^{(i)}}^{\otimes d} + \ket{0_\perp}\right),
    \label{eq:VHisEMSup}
\end{align}
where $\ket{0_\perp}$ is some unnormalized state such that $(I \otimes \ket{0}\!\bra{0} \otimes I)\ket{0_\perp}=0$.

\State Construct the oracle $U_C^{\rm SP}$ in Algorithm \ref{alg:main}.

\State By the algorithm in Theorem \ref{th:OE}, get an estimate $\hat{\mu}^\prime$ of $\braket{\mathrm{histEMSP} | C{\rm SP}}$ with accuracy
\begin{align}
    \epsilon^{\rm EM}_{\rm OE}=\left(\frac{\max\{1,\eta T\}}{8r^{\rm EM}_{\rm c} U_{\mathbf{B}^{\rm EM}}}\right)^d\frac{\epsilon}{4\|C\|_F}
\end{align}
and success probability $1-\frac{\delta}{2}$.

\State Output
\begin{align}
    \hat{\mu}_Y = 
    \left(\frac{8r^{\rm EM}_{\rm c} U_{\mathbf{B}^{\rm EM}}}{\max\{1,\eta T\}}\right)^d\|C\|_F \hat{\mu}^\prime.
\end{align}

\end{algorithmic}
\end{algorithm}

Let us prove the accuracy of this algorithm.
We note that
\begin{align}
    \left(\frac{8r^{\rm EM}_{\rm c} U_{\mathbf{B}^{\rm EM}}}{\max\{1,\eta T\}}\right)^d\|C\|_F\braket{\mathrm{histEMSP} | C{\rm SP}} = \hat{Y}_{\rm EM}.
\end{align}
Lemma \ref{lem:YHatEM} implies that with probability at least $1-\frac{\delta}{2}$,
\begin{align}
    & \left|\hat{Y}_{\rm EM}-\mathbb{E}_\mathbb{P}[Y]\right|\le \sqrt{(2d-3)!!}(r^{\rm EM}_{\rm c}+1)^{d/2} \left(\|\mathbf{x}_0\|^2+(m+2d)\sigma^2T\right)^{d/2}\|C\|_F \nonumber \\
    & \quad \times \left[
    (1+\epsilon^\prime_{\rm EM})^d \left(\sqrt{\frac{8(2d-1)}{\delta N_{\rm s}}}+8d\epsilon^\prime_{\rm EM}\right) + \frac{d}{r^{\rm EM}_{\rm c}}\sqrt{2\left(\frac{1}{2}(\alpha_A^2+\alpha_{\drm A/\drm t})T + \alpha_A \right)^2T^2+\frac{2m\alpha_{\drm B/\drm t}^2T^3}{\|\mathbf{x}_0\|^2+(m+2d)\sigma^2T}}\right]
\end{align}
holds, and thus
\begin{align}
    \left|\hat{Y}_{\rm EM}-\mathbb{E}_\mathbb{P}[Y]\right| \le \frac{3\epsilon}{4}
\end{align}
holds under the current $r^{\rm EM}_{\rm c}$, $\epsilon^\prime_{\rm EM}$, and $N_{\rm s}^{\rm EM}$.
Theorem \ref{th:OE} implies that with probability at least $1-\frac{\delta}{2}$, $\hat{\mu}^\prime$ satisfies
\begin{align}
    \left|\hat{\mu}^\prime-\braket{\mathrm{histSP} | C{\rm SP}}\right|\le \epsilon_{\rm OE},
\end{align}
and thus
\begin{align}
    \left|\hat{\mu}_Y - \hat{Y}\right| \le \left(\frac{8r^{\rm EM}_{\rm c} U_{\mathbf{B}^{\rm EM}}}{\max\{1,\eta T\}}\right)^d\|C\|_F\epsilon_{\rm OE} = \frac{\epsilon}{4}
\end{align}
holds.
Combining these, we see that $\hat{\mu}_Y$ satisfies $\left|\hat{\mu}_Y - \mathbb{E}_\mathbb{P}[Y]\right|\le \epsilon$ with probability at least $1-\delta$.

Lastly, let us evaluate the number of uses of various oracles and additional elementary gates.
The overlap estimation in Step 7 uses $V_{\rm histEMSP}$ and $U^{\rm SP}_C$ $O\left(\frac{1}{\epsilon^{\rm EM}_{\rm OE}}\log\left(\frac{1}{\delta}\right)\right)$ times each, thus makes $O\left(\frac{1}{\epsilon^{\rm EM}_{\rm OE}}\log\left(\frac{1}{\delta}\right)\right)$ uses of $V_{{\rm histEM},\epsilon_{\rm EM}^\prime}$ and $U_C$.
According to Theorem \ref{th:HistStateConcEM}, $V_{{\rm histEM},\epsilon_{\rm EM}^\prime}$ uses the controlled $U_A$
\begin{align}
    O\left(\frac{r^{\rm EM}_{\rm c}}{\max\{1,\eta T\}}\log\left(\frac{r^{\rm EM}_{\rm c}}{\max\{1,\eta T\}\epsilon_{\rm EM}^\prime}\right)\right)
\end{align}
times, the controlled $U_B$, $U_{\mathbf{x}_0}$, and $U_{\rm Rand}$ once each, and additional elementary gates
\begin{align}
    O\left(\frac{a_Ar^{\rm EM}_{\rm c}}{\max\{1,\eta T\}}\log\left(\frac{r^{\rm EM}_{\rm c}}{\max\{1,\eta T\}\epsilon_{\rm EM}^\prime}\right)\right).
\end{align}
By combining these evaluations and doing some algebra, we obtain the claimed query number bounds.
\end{proof}

\begin{remark} \label{rem:CompEMQ}
The setting of $r^{\rm EM}_{\rm c}$ and $\epsilon^\prime_{\rm EM}$ in Eqs. \eqref{eq:rEMc} and \eqref{eq:epsPrEM} is seemingly circular.
However, this can be resolved by setting the accuracy $\epsilon$ relative to $\sqrt{\mathbb{E}_\mathbb{P}[Y^2]}$ as in Remarks \ref{rem:CompExaMT} and \ref{rem:CompExaTerm}.
We have
\begin{align}
    \sqrt{\mathbb{E}_\mathbb{P}\left[Y^2\right]}\lesssim \left(r^{\rm EM}_{\rm c}+1\right)^{d/2}\left(\|\mathbf{x}_0\|^2+(m+2d)\sigma^2T\right)^{d/2},
\end{align}
and if we set $\epsilon$ relative to this by
\begin{align}
\epsilon=\left(r^{\rm EM}_{\rm c}+1\right)^{d/2}\left(\|\mathbf{x}_0\|^2+(m+2d)\sigma^2T\right)^{d/2}\|C\|_F\epsilon_{\rm rel}    
\end{align}
with $\epsilon_{\rm rel}\in(0,1)$, we can set
\begin{align}
    \epsilon^\prime_{\rm EM} = \frac{\epsilon_{\rm rel}}{2^{d+5}d \sqrt{(2d-3)!!}}
\end{align}
and then
\begin{align}
    r^{\rm EM}_{\rm c} = \left\lceil\max\left\{\frac{4\alpha_A^2T}{\eta},
    \frac{d}{\epsilon_{\rm EM}^{\prime}}\sqrt{2\left(\frac{1}{2}(\alpha_A^2+\alpha_{\drm A/\drm t})T + \alpha_A \right)^2T^2+\frac{2m\alpha_{\drm B/\drm t}^2T^3}{\|\mathbf{x}_0\|^2+(m+2d)\sigma^2T}}
    \right\}\right\rceil.
\end{align}
Now, by hiding logarithmic factors and factors that depend only on $d$, the query complexity evaluation in Eq. \eqref{eq:CompEM} is simplified into
\begin{align}
\tilde{O}\left(\left(\max\left\{\frac{\alpha_A^2T}{\eta},\frac{1}{\epsilon_{\rm rel}}\sqrt{\left((\alpha_A^2+\alpha_{\drm A/\drm t})T + \alpha_A \right)^2T^2+\frac{m\alpha_{\drm B/\drm t}^2T^3}{\|\mathbf{x}_0\|^2+(m+2d)\sigma^2T}} \right\}\right)^{\frac{d}{2}+1}\frac{1}{\epsilon_{\rm rel}}\right),
\label{eq:CompEMQSimple}
\end{align}
where we also assume $\eta T \lesssim 1$.
\end{remark}

\begin{remark}
    In the classical estimation of $\mathbb{E}_\mathbb{P}[Y]$ by the \ac{EM} approximation and the Monte Carlo method, the complexity, which we now measure by the number of arithmetic operations, is evaluated as follows.
    We may take the above $r^{\rm EM}_{\rm c}$ in classical computing as well: it is the sufficient condition for $\|I+A(t)\Delta t\|<1$ to hold, under which the numerical instability is avoided, as well as for the strong error to be suppressed to within $\epsilon^\prime_{\rm EM}$.
    In a step of time evolution, computing Eq.~\eqref{eq:EM} takes $N(s_A+s_B)$ operations, where $s_A$ (resp. $s_B$) is the maximum number of nonzero entries in a row of $A$ (resp. $B$).
    If we require the relative accuracy as in Remark.~\ref{rem:CompEMQ}, we take $\tilde{O}(1/\epsilon_{\rm rel}^2)$ samples.
    Combining these, we can evaluate the total complexity as
    \begin{align}
        \tilde{O}\left(\max\left\{\frac{\alpha_A^2T}{\eta},\frac{1}{\epsilon_{\rm rel}}\sqrt{\left((\alpha_A^2+\alpha_{\drm A/\drm t})T + \alpha_A \right)^2T^2+\frac{m\alpha_{\drm B/\drm t}^2T^3}{\|\mathbf{x}_0\|^2+(m+2d)\sigma^2T}} \right\} \cdot \frac{N(s_A+s_B)}{\epsilon_{\rm rel}^2}\right).
        \label{eq:CompEMCl}
    \end{align}
    
    Comparing this with Eq.~\eqref{eq:CompEMQSimple}, we see that the quantum \ac{EM}-based method provides the exponential speed-up with respect to $N$ and the quadratic speed-up with respect to $\epsilon_{\rm rel}$ over the classical method, at the cost of having the exponent of the factor arising from $r^{\rm EM}_{\rm c}$ increase 1 to $\frac{d}{2}+1$.

    Comparing Eqs.~\eqref{eq:CompExaUASimp2} and \eqref{eq:CompEMCl}, we see that, in addition to the speed-up with respect to $N$ and $\epsilon_{\rm rel}$, the quantum Dyson series-based method removes the dependence of the complexity on $\alpha_{\drm A/\drm t}$ and $\alpha_{\drm B/\drm t}$, except the logarithmic dependencies.
    At the cost, the complexity of the Dyson series-based method has the factor $\left(\kappa_{BB^T} \alpha_A T\right)^{\frac{d}{2}+1}$.

    Note that we are now comparing the number of queries to the block-encoding quantum circuits and the number of classical arithmetic operations, which are not equivalent.
    Nonetheless, the block-encoding circuit can be efficiently constructed in some cases, which justifies the above comparison.
    For example, if a matrix is sparse and we are given its sparse-access oracles, which provide the positions of nonzero entries and the value of the entry at a specified position, the block-encoding can be constructed with $O(1)$ queries to the oracles and a logarithmic number of additional elementary gates~\cite{Gilyen2019}.
    If the sparse-access oracles for $A$ and $BB^T$ are constructed with $\tilde{O}(1)$ arithmetic circuits, comparing Eqs.~\eqref{eq:CompExaUASimp2} and \eqref{eq:CompEMCl} amounts to comparing the number of quantum arithmetic operations with that of classical ones.
    \label{rem:CompEMC}
\end{remark}

\section{Summary \label{sec:sum}}

In this paper, we considered quantum algorithms to solve the high-dimensional \ac{SDE} in Eq. \eqref{eq:TgtSDEGen} by generating the quantum state that encodes realizations of the solution $\mathbf{X}_t$ in the amplitudes, with the aim of achieving a quantum speed-up with respect to the dimension $N$.
Although quantum \ac{ODE} solvers are expected to be applied to this problem, amplitude encoding of the noise term $\boldsymbol{\Delta}_n$ is an issue.
To address this point, we found that states encoding \acp{PRN} in the amplitudes can be generated using the \ac{PRNG} quantum circuit considered in Refs. \cite{Miyamoto2020,kaneko2021quantum,kaneko2022quantum} and that such states can be used for amplitude encoding of $\boldsymbol{\Delta}_n$.
Then, we constructed two algorithms, the Dyson series-based method and the \ac{EM}-based method, in which the key components oracles are constructed from \ac{PRNG} circuit and the block-encodings of $A$ and $BB^T$ (or $B$), and the quantum linear systems solver incorporating the oracles generates the history state.
We also consider estimating expectations of functions of $\mathbf{X}_t$ by the overlap estimation algorithm incorporating history state generation.

Extending the proposed methods for \acp{SDE} in other forms than Eq. \eqref{eq:TgtSDEGen} would be interesting and meaningful.
For example, we are interested in \ac{SDE}s with the multiplicative noise $\drm \mathbf{X}_t = A(t) \mathbf{X}_t \drm t + B(t) \mathbf{X}_t \drm \mathbf{W}_t$, which is considered in Ref.~\cite{wu2026universal}, and more generally, nonlinear \acp{SDE}\footnote{Ref. \cite{li2026efficient} considers a nonlinear differential equation with its inhomogeneous term being a random process, but its SDE is linear.} $\drm \mathbf{X}_t = \mathbf{f}(\mathbf{X}_t) \drm t + g(\mathbf{X}_t) \drm \mathbf{W}_t$.
Ref.~\cite{bravyi2025quantum} considered nonlinear \acp{SDE} but was limited to the setting of quadratic $\mathbf{f}$ with the divergence-free condition and constant $g$.
We aim to extend our methods to a broad class of nonlinear \acp{SDE} as future work.

\section*{Acknowledgements}

This work is supported by MEXT Quantum Leap Flagship Program (MEXT Q-LEAP) Grant No. JPMXS0120319794, JST COI-NEXT Program Grant No. JPMJPF2014, and JST [Moonshot R\&D] [Grant No. JPMJMS256J].

\bibliography{reference}

@article{bouland2023quantum,
  title={{A quantum spectral method for simulating stochastic processes, with applications to Monte Carlo}},
  author={Bouland, Adam and Dandapani, Aditi and Prakash, Anupam},
  journal={arXiv preprint arXiv:2303.06719},
  year={2023}
}

@article{Kieferova2019,
  title = {{Simulating the dynamics of time-dependent Hamiltonians with a truncated Dyson series}},
  author = {Kieferov\'a, M\'aria and Scherer, Artur and Berry, Dominic W.},
  journal = {Phys. Rev. A},
  volume = {99},
  issue = {4},
  pages = {042314},
  numpages = {13},
  year = {2019},
  month = {Apr},
  publisher = {American Physical Society},
  doi = {10.1103/PhysRevA.99.042314},
  url = {https://link.aps.org/doi/10.1103/PhysRevA.99.042314}
}

@inproceedings{chakraborty2019power,
  title={{The Power of Block-Encoded Matrix Powers: Improved Regression Techniques via Faster Hamiltonian Simulation}},
  author={Chakraborty, Shantanav and Gily{\'e}n, Andr{\'a}s and Jeffery, Stacey},
  booktitle={46th International Colloquium on Automata, Languages, and Programming (ICALP 2019)},
  pages={33--1},
  year={2019},
  organization={Schloss Dagstuhl--Leibniz-Zentrum f{\"u}r Informatik},
note={[Full version] arXiv:1804.01973},
doi={10.4230/LIPIcs.ICALP.2019.33}
}

@inproceedings{Gilyen2019,
author = {Gily\'{e}n, Andr\'{a}s and Su, Yuan and Low, Guang Hao and Wiebe, Nathan},
title = {Quantum singular value transformation and beyond: exponential improvements for quantum matrix arithmetics},
year = {2019},
isbn = {9781450367059},
publisher = {Association for Computing Machinery},
address = {New York, NY, USA},
url = {https://doi.org/10.1145/3313276.3316366},
doi = {10.1145/3313276.3316366},
booktitle = {Proceedings of the 51st Annual ACM SIGACT Symposium on Theory of Computing},
pages = {193–204},
numpages = {12},
location = {Phoenix, AZ, USA},
series = {STOC 2019},
note={[Full version] arXiv:1806.01838}
}

@article{costa2022optimal,
  title={Optimal scaling quantum linear-systems solver via discrete adiabatic theorem},
  author={Costa, Pedro CS and An, Dong and Sanders, Yuval R and Su, Yuan and Babbush, Ryan and Berry, Dominic W},
  journal={PRX quantum},
  volume={3},
  number={4},
  pages={040303},
  year={2022},
  publisher={APS},
  doi={10.1103/PRXQuantum.3.040303}
}

@article{Martyn2021,
  title = {Grand Unification of Quantum Algorithms},
  author = {Martyn, John M. and Rossi, Zane M. and Tan, Andrew K. and Chuang, Isaac L.},
  journal = {PRX Quantum},
  volume = {2},
  issue = {4},
  pages = {040203},
  numpages = {40},
  year = {2021},
  month = {Dec},
  publisher = {American Physical Society},
  doi = {10.1103/PRXQuantum.2.040203},
  url = {https://link.aps.org/doi/10.1103/PRXQuantum.2.040203}
}

@article{an2026fast,
  doi = {10.22331/q-2026-01-27-1986},
  url = {https://doi.org/10.22331/q-2026-01-27-1986},
  title = {Fast-forwarding quantum algorithms for linear dissipative differential equations},
  author = {An, Dong and Onwunta, Akwum and Yang, Gengzhi},
  journal = {{Quantum}},
  issn = {2521-327X},
  publisher = {{Verein zur F{\"{o}}rderung des Open Access Publizierens in den Quantenwissenschaften}},
  volume = {10},
  pages = {1986},
  month = jan,
  year = {2026}
}

@article{Berry2024quantumalgorithm,
  doi = {10.22331/q-2024-06-13-1369},
  url = {https://doi.org/10.22331/q-2024-06-13-1369},
  title = {Quantum algorithm for time-dependent differential equations using {D}yson series},
  author = {Berry, Dominic W. and C. S. Costa, Pedro},
  journal = {{Quantum}},
  issn = {2521-327X},
  publisher = {{Verein zur F{\"{o}}rderung des Open Access Publizierens in den Quantenwissenschaften}},
  volume = {8},
  pages = {1369},
  month = jun,
  year = {2024}
}

@book{pazy2012semigroups,
  title={Semigroups of linear operators and applications to partial differential equations},
  author={Pazy, Amnon},
  volume={44},
  year={2012},
  publisher={Springer Science \& Business Media}
}

@book{desoer2009feedback,
  title={Feedback systems: input-output properties},
  author={Desoer, Charles A and Vidyasagar, Mathukumalli},
  year={2009},
  publisher={SIAM}
}

@article{robbins1955remark,
  title={{A remark on Stirling's formula}},
  author={Robbins, Herbert},
  journal={The American mathematical monthly},
  volume={62},
  number={1},
  pages={26--29},
  year={1955},
  publisher={JSTOR}
}

@book{trefethen2022numerical,
  title={Numerical linear algebra},
  author={Trefethen, Lloyd N and Bau, David},
  year={2022},
  publisher={SIAM}
}

@article{HHL,
  title = {Quantum Algorithm for Linear Systems of Equations},
  author = {Harrow, Aram W. and Hassidim, Avinatan and Lloyd, Seth},
  journal = {Phys. Rev. Lett.},
  volume = {103},
  issue = {15},
  pages = {150502},
  numpages = {4},
  year = {2009},
  month = {Oct},
  publisher = {American Physical Society},
  doi = {10.1103/PhysRevLett.103.150502},
  url = {https://link.aps.org/doi/10.1103/PhysRevLett.103.150502}
}

@article{Knill2007,
  title = {Optimal quantum measurements of expectation values of observables},
  author = {Knill, Emanuel and Ortiz, Gerardo and Somma, Rolando D.},
  journal = {Phys. Rev. A},
  volume = {75},
  issue = {1},
  pages = {012328},
  numpages = {13},
  year = {2007},
  month = {Jan},
  publisher = {American Physical Society},
  doi = {10.1103/PhysRevA.75.012328},
  url = {https://link.aps.org/doi/10.1103/PhysRevA.75.012328}
}

@article{wang2025comprehensive,
  title={A comprehensive study of quantum arithmetic circuits},
  author={Wang, Siyi and Li, Xiufan and Lee, Wei Jie Bryan and Deb, Suman and Lim, Eugene and Chattopadhyay, Anupam},
  journal={Philosophical Transactions A},
  volume={383},
  number={2288},
  pages={20230392},
  year={2025},
  publisher={The Royal Society},
  doi={10.1098/rsta.2023.0392}
}

@article{haner2018optimizing,
  title={Optimizing quantum circuits for arithmetic},
  author={H{\"a}ner, Thomas and Roetteler, Martin and Svore, Krysta M},
  journal={arXiv preprint arXiv:1805.12445},
  year={2018},
}

@article{woerner2019quantum,
  title={Quantum risk analysis},
  author={Woerner, Stefan and Egger, Daniel J},
  journal={npj Quantum Information},
  volume={5},
  number={1},
  pages={15},
  year={2019},
  publisher={Nature Publishing Group UK London},
doi={10.1038/s41534-019-0130-6}
}

@article{Yuan_2023,
doi = {10.1088/1367-2630/acfd52},
url = {https://doi.org/10.1088/1367-2630/acfd52},
year = {2023},
month = {oct},
publisher = {IOP Publishing},
volume = {25},
number = {10},
pages = {103011},
author = {Yuan, Yewei and Wang, Chao and Wang, Bei and Chen, Zhao-Yun and Dou, Meng-Han and Wu, Yu-Chun and Guo, Guo-Ping},
title = {{An improved QFT-based quantum comparator and extended modular arithmetic using one ancilla qubit}},
journal = {New Journal of Physics}
}

@article{grover2002creating,
  title={Creating superpositions that correspond to efficiently integrable probability distributions},
  author={Grover, Lov and Rudolph, Terry},
  journal={arXiv preprint quant-ph/0208112},
  year={2002},
}

@book{simon2002probability,
  title={{Probability distributions involving Gaussian random variables: A handbook for engineers and scientists}},
  author={Simon, Marvin K},
  year={2002},
  publisher={Springer}
}

@article{Bauer1960,
  title = {Norms and exclusion theorems},
  author = {Bauer, F. L. and Fike, C. T.},
  journal = {Numerische Mathematik},
  volume = {2},
  pages = {137},
  year = {1960},
  doi = {10.1007/BF01386217},
}

@article{wu2026universal,
  title={{Universal Dilation of Linear It\^{o} SDEs: Quantum Trajectories and Lindblad Simulation of Second Moments}},
  author={Wu, Hsuan-Cheng and Li, Xiantao},
  journal={arXiv preprint arXiv:2601.05928},
  year={2026}
}

@book{KloedenPlaten,
  title={Numerical Solution of Stochastic Differential Equations},
  author={Peter E. Kloeden and Eckhard Platen},
  year={1992},
  publisher={Springer}
}

@article{Milstein1975,
author = {Mil’shtejn, G. N.},
title = {Approximate Integration of Stochastic Differential Equations},
journal = {Theory of Probability \& Its Applications},
volume = {19},
number = {3},
pages = {557-562},
year = {1975},
doi = {10.1137/1119062}
}

@article{Maruyama1955,
title = {{Continuous Markov processes and stochastic equations}},
journal = {Rend. Circ. Mat. Palermo},
volume = {4},
pages = {48-90},
year = {1955},
doi = {10.1007/BF02846028},
author = {Maruyama, Gisiro}
}

@article{haagerup1981best,
  title={{The best constants in the Khintchine inequality}},
  author={Haagerup, Uffe},
  journal={Studia Mathematica},
  volume={70},
  number={3},
  pages={231--283},
  year={1981},
  publisher={Polska Akademia Nauk. Instytut Matematyczny PAN}
}

@book{brockett2015finite,
  title={Finite dimensional linear systems},
  author={Brockett, Roger W},
  year={2015},
  publisher={SIAM}
}

@article{childs2012hamiltonian,
  title={Hamiltonian Simulation Using Linear Combinations of Unitary Operations},
  author={Childs, Andrew M and Wiebe, Nathan},
  journal={Quant. Inf. Comput.},
  volume={12},
  pages={901--924},
  year={2012},
  doi={10.26421/QIC12.11-12}
}

@article{Miyamoto2020,
  title = {{Reduction of qubits in a quantum algorithm for Monte Carlo simulation by a pseudo-random-number generator}},
  author = {Miyamoto, Koichi and Shiohara, Kenji},
  journal = {Phys. Rev. A},
  volume = {102},
  issue = {2},
  pages = {022424},
  numpages = {12},
  year = {2020},
  month = {Aug},
  publisher = {American Physical Society},
  doi = {10.1103/PhysRevA.102.022424},
  url = {https://link.aps.org/doi/10.1103/PhysRevA.102.022424}
}

@techreport{oneill:pcg2014,
    title = {{PCG: A Family of Simple Fast Space-Efficient Statistically Good Algorithms for Random Number Generation}},
    author = "Melissa E. O'Neill",
    institution = "Harvey Mudd College",
    address = "Claremont, CA",
    number = "HMC-CS-2014-0905",
    year = "2014",
    month = Sep,
    xurl = "https://www.cs.hmc.edu/tr/hmc-cs-2014-0905.pdf",
}

@article{Hormann2003,
author = {H\"{o}rmann, Wolfgang and Leydold, Josef},
title = {Continuous random variate generation by fast numerical inversion},
year = {2003},
issue_date = {October 2003},
publisher = {Association for Computing Machinery},
address = {New York, NY, USA},
volume = {13},
number = {4},
issn = {1049-3301},
url = {https://doi.org/10.1145/945511.945517},
doi = {10.1145/945511.945517},
journal = {ACM Trans. Model. Comput. Simul.},
month = oct,
pages = {347–362},
numpages = {16}
}

@article{kaneko2021quantum,
  title={{Quantum speedup of Monte Carlo integration with respect to the number of dimensions and its application to finance}},
  author={Kaneko, Kazuya and Miyamoto, Koichi and Takeda, Naoyuki and Yoshino, Kazuyoshi},
  journal={Quantum Information Processing},
  volume={20},
  number={5},
  pages={185},
  year={2021},
  publisher={Springer},
  doi={10.1007/s11128-021-03127-8}
}

@article{kaneko2022quantum,
  title={Quantum pricing with a smile: implementation of local volatility model on quantum computer},
  author={Kaneko, Kazuya and Miyamoto, Koichi and Takeda, Naoyuki and Yoshino, Kazuyoshi},
  journal={EPJ Quantum Technology},
  volume={9},
  number={1},
  pages={1--32},
  year={2022},
  doi={10.1140/epjqt/s40507-022-00125-2}
}

@article{Childs2017QLSS,
author = {Childs, Andrew M. and Kothari, Robin and Somma, Rolando D.},
title = {Quantum Algorithm for Systems of Linear Equations with Exponentially Improved Dependence on Precision},
journal = {SIAM Journal on Computing},
volume = {46},
number = {6},
pages = {1920-1950},
year = {2017},
doi = {10.1137/16M1087072}
}

@article{Subasi2019,
  title = {Quantum Algorithms for Systems of Linear Equations Inspired by Adiabatic Quantum Computing},
  author = {Suba{\c{s}}{\i}, Yi{\u{g}}it and Somma, Rolando D. and Orsucci, Davide},
  journal = {Phys. Rev. Lett.},
  volume = {122},
  issue = {6},
  pages = {060504},
  numpages = {5},
  year = {2019},
  month = {Feb},
  publisher = {American Physical Society},
  doi = {10.1103/PhysRevLett.122.060504},
  url = {https://link.aps.org/doi/10.1103/PhysRevLett.122.060504}
}

@article{An2022QLSS,
author = {An, Dong and Lin, Lin},
title = {Quantum Linear System Solver Based on Time-optimal Adiabatic Quantum Computing and Quantum Approximate Optimization Algorithm},
year = {2022},
issue_date = {June 2022},
publisher = {Association for Computing Machinery},
address = {New York, NY, USA},
volume = {3},
number = {2},
url = {https://doi.org/10.1145/3498331},
doi = {10.1145/3498331},
journal = {ACM Transactions on Quantum Computing},
month = mar,
articleno = {5},
numpages = {28}
}

@article{Lin2020optimalpolynomial,
  doi = {10.22331/q-2020-11-11-361},
  url = {https://doi.org/10.22331/q-2020-11-11-361},
  title = {Optimal polynomial based quantum eigenstate filtering with application to solving quantum linear systems},
  author = {Lin, Lin and Tong, Yu},
  journal = {{Quantum}},
  issn = {2521-327X},
  publisher = {{Verein zur F{\"{o}}rderung des Open Access Publizierens in den Quantenwissenschaften}},
  volume = {4},
  pages = {361},
  month = nov,
  year = {2020}
}

@article{morales2024quantum,
  title={Quantum linear system solvers: A survey of algorithms and applications},
  author={Morales, Mauro ES and Pira, Lirand{\"e} and Schleich, Philipp and Koor, Kelvin and Costa, Pedro and An, Dong and Aspuru-Guzik, Al{\'a}n and Lin, Lin and Rebentrost, Patrick and Berry, Dominic W},
  journal={arXiv preprint arXiv:2411.02522},
  year={2024}
}

@article{Berry_2014,
doi = {10.1088/1751-8113/47/10/105301},
url = {https://doi.org/10.1088/1751-8113/47/10/105301},
year = {2014},
month = {feb},
publisher = {IOP Publishing},
volume = {47},
number = {10},
pages = {105301},
author = {Berry, Dominic W},
title = {High-order quantum algorithm for solving linear differential equations},
journal = {Journal of Physics A: Mathematical and Theoretical}
}

@article{childs2020quantum,
  title={Quantum Spectral Methods for Differential Equations},
  author={Childs, Andrew M and Liu, Jin-Peng},
  journal={Communications in Mathematical Physics},
  volume={375},
  number={2},
  pages={1427--1457},
  year={2020},
  publisher={Springer},
  doi={10.1007/s00220-020-03699-z}
}

@article{An2023LCHS,
  title = {{Linear Combination of Hamiltonian Simulation for Nonunitary Dynamics with Optimal State Preparation Cost}},
  author = {An, Dong and Liu, Jin-Peng and Lin, Lin},
  journal = {Phys. Rev. Lett.},
  volume = {131},
  issue = {15},
  pages = {150603},
  numpages = {6},
  year = {2023},
  month = {Oct},
  publisher = {American Physical Society},
  doi = {10.1103/PhysRevLett.131.150603},
  url = {https://link.aps.org/doi/10.1103/PhysRevLett.131.150603}
}

@article{an2026quantum,
  title={Quantum Algorithm for Linear Non-unitary Dynamics with Near-Optimal Dependence on All Parameters},
  author={An, Dong and Childs, Andrew M and Lin, Lin},
  journal={Communications in Mathematical Physics},
  volume={407},
  number={1},
  pages={19},
  year={2026},
  publisher={Springer},
  doi={10.1007/s00220-025-05509-w}
}

@article{low2025optimal,
  title={Optimal quantum simulation of linear non-unitary dynamics},
  author={Low, Guang Hao and Somma, Rolando D},
  journal={arXiv preprint arXiv:2508.19238},
  year={2025}
}

@article{Krovi2023improvedquantum,
  doi = {10.22331/q-2023-02-02-913},
  url = {https://doi.org/10.22331/q-2023-02-02-913},
  title = {Improved quantum algorithms for linear and nonlinear differential equations},
  author = {Krovi, Hari},
  journal = {{Quantum}},
  issn = {2521-327X},
  publisher = {{Verein zur F{\"{o}}rderung des Open Access Publizierens in den Quantenwissenschaften}},
  volume = {7},
  pages = {913},
  month = feb,
  year = {2023}
}

@article{an2025quantum,
  title={Quantum differential equation solvers: Limitations and fast-forwarding},
  author={An, Dong and Liu, Jin-Peng and Wang, Daochen and Zhao, Qi},
  journal={Communications in Mathematical Physics},
  volume={406},
  number={8},
  pages={189},
  year={2025},
  publisher={Springer},
  doi={10.1007/s00220-025-05358-7
}
}

@article{Liu2021,
author = {Jin-Peng Liu  and Herman Øie Kolden  and Hari K. Krovi  and Nuno F. Loureiro  and Konstantina Trivisa  and Andrew M. Childs },
title = {Efficient quantum algorithm for dissipative nonlinear differential equations},
journal = {Proceedings of the National Academy of Sciences},
volume = {118},
number = {35},
pages = {e2026805118},
year = {2021},
doi = {10.1073/pnas.2026805118}
}

@article{Fang2023timemarchingbased,
  doi = {10.22331/q-2023-03-20-955},
  url = {https://doi.org/10.22331/q-2023-03-20-955},
  title = {Time-marching based quantum solvers for time-dependent linear differential equations},
  author = {Fang, Di and Lin, Lin and Tong, Yu},
  journal = {{Quantum}},
  issn = {2521-327X},
  publisher = {{Verein zur F{\"{o}}rderung des Open Access Publizierens in den Quantenwissenschaften}},
  volume = {7},
  pages = {955},
  month = mar,
  year = {2023}
}

@article{berry2017quantum,
  title={Quantum algorithm for linear differential equations with exponentially improved dependence on precision},
  author={Berry, Dominic W and Childs, Andrew M and Ostrander, Aaron and Wang, Guoming},
  journal={Communications in Mathematical Physics},
  volume={356},
  number={3},
  pages={1057--1081},
  year={2017},
  publisher={Springer},
  doi={10.1007/s00220-017-3002-y}
}

@article{Jennings2024costofsolvinglinear,
  doi = {10.22331/q-2024-12-10-1553},
  url = {https://doi.org/10.22331/q-2024-12-10-1553},
  title = {The cost of solving linear differential equations on a quantum computer: fast-forwarding to explicit resource counts},
  author = {Jennings, David and Lostaglio, Matteo and Lowrie, Robert B. and Pallister, Sam and Sornborger, Andrew T.},
  journal = {{Quantum}},
  issn = {2521-327X},
  publisher = {{Verein zur F{\"{o}}rderung des Open Access Publizierens in den Quantenwissenschaften}},
  volume = {8},
  pages = {1553},
  month = dec,
  year = {2024}
}

@article{li2026efficient,
      title={Efficient Quantum Simulation for Nonlinear Stochastic Differential Equations}, 
      author={Xiangyu Li and Ahmet Burak Catli and Ho Kiat Lim and Matthew Pocrnic and Dong An and Jin-Peng Liu and Nathan Wiebe},
      year={2026},
      journal={arXiv preprint arXiv:2603.12398}
}

@article{Higham2008,
author = {Higham, Desmond J.},
title = {Modeling and Simulating Chemical Reactions},
journal = {SIAM Review},
volume = {50},
number = {2},
pages = {347-368},
year = {2008},
doi = {10.1137/060666457}
}

@article{Gillespie2007,
   author = "Gillespie, Daniel T.",
   title = "Stochastic Simulation of Chemical Kinetics", 
   journal= "Annual Review of Physical Chemistry",
   year = "2007",
   volume = "58",
   number = "Volume 58, 2007",
   pages = "35-55",
   doi = "10.1146/annurev.physchem.58.032806.104637"
   }

@article{Mizuguchi_2024,
doi = {10.1088/1475-7516/2024/12/050},
url = {https://doi.org/10.1088/1475-7516/2024/12/050},
year = {2024},
month = {dec},
publisher = {IOP Publishing},
volume = {2024},
number = {12},
pages = {050},
author = {Mizuguchi, Yurino and Murata, Tomoaki and Tada, Yuichiro},
title = {{STOLAS: STOchastic LAttice Simulation of cosmic inflation}},
journal = {Journal of Cosmology and Astroparticle Physics}
}

@article{murata2026stochastic,
  title={{STOchastic LAttice Simulation of hybrid inflation}},
  author={Murata, Tomoaki and Tada, Yuichiro},
  journal={arXiv preprint arXiv:2603.04850},
  year={2026}
}

@article{Rebentrost2018,
  title = {{Quantum computational finance: Monte Carlo pricing of financial derivatives}},
  author = {Rebentrost, Patrick and Gupt, Brajesh and Bromley, Thomas R.},
  journal = {Phys. Rev. A},
  volume = {98},
  issue = {2},
  pages = {022321},
  numpages = {15},
  year = {2018},
  month = {Aug},
  publisher = {American Physical Society},
  doi = {10.1103/PhysRevA.98.022321}
}

@article{Stamatopoulos2020optionpricingusing,
  doi = {10.22331/q-2020-07-06-291},
  url = {https://doi.org/10.22331/q-2020-07-06-291},
  title = {Option {P}ricing using {Q}uantum {C}omputers},
  author = {Stamatopoulos, Nikitas and Egger, Daniel J. and Sun, Yue and Zoufal, Christa and Iten, Raban and Shen, Ning and Woerner, Stefan},
  journal = {{Quantum}},
  issn = {2521-327X},
  publisher = {{Verein zur F{\"{o}}rderung des Open Access Publizierens in den Quantenwissenschaften}},
  volume = {4},
  pages = {291},
  year = {2020}
}

@article{Chakrabarti2021thresholdquantum,
  doi = {10.22331/q-2021-06-01-463},
  url = {https://doi.org/10.22331/q-2021-06-01-463},
  title = {A {T}hreshold for {Q}uantum {A}dvantage in {D}erivative {P}ricing},
  author = {Chakrabarti, Shouvanik and Krishnakumar, Rajiv and Mazzola, Guglielmo and Stamatopoulos, Nikitas and Woerner, Stefan and Zeng, William J.},
  journal = {{Quantum}},
  issn = {2521-327X},
  publisher = {{Verein zur F{\"{o}}rderung des Open Access Publizierens in den Quantenwissenschaften}},
  volume = {5},
  pages = {463},
  month = jun,
  year = {2021}
}

@article{Moosa_2024,
doi = {10.1088/2058-9565/acfc62},
url = {https://dx.doi.org/10.1088/2058-9565/acfc62},
year = {2023},
month = {oct},
publisher = {IOP Publishing},
volume = {9},
number = {1},
pages = {015002},
author = {Mudassir Moosa and Thomas W Watts and Yiyou Chen and Abhijat Sarma and Peter L McMahon},
title = {{Linear-depth quantum circuits for loading Fourier approximations of arbitrary functions}},
journal = {Quantum Science and Technology}
}

@article{Sanders2019,
  title = {{Black-Box Quantum State Preparation without Arithmetic}},
  author = {Sanders, Yuval R. and Low, Guang Hao and Scherer, Artur and Berry, Dominic W.},
  journal = {Phys. Rev. Lett.},
  volume = {122},
  issue = {2},
  pages = {020502},
  numpages = {5},
  year = {2019},
  month = {Jan},
  publisher = {American Physical Society},
  doi = {10.1103/PhysRevLett.122.020502}
}

@article{zoufal2019quantum,
  title={Quantum generative adversarial networks for learning and loading random distributions},
  author={Zoufal, Christa and Lucchi, Aur{\'e}lien and Woerner, Stefan},
  journal={npj Quantum Information},
  volume={5},
  number={1},
  pages={103},
  year={2019},
  publisher={Nature Publishing Group UK London},
  doi={10.48550/arXiv.2205.00519}
}

@article{bravyi2025quantum,
  title={Quantum simulation of a noisy classical nonlinear dynamics},
  author={Bravyi, Sergey and Manson-Sawko, Robert and Zayats, Mykhaylo and Zhuk, Sergiy},
  journal={arXiv preprint arXiv:2507.06198},
  year={2025}
}

@article{aaronson2015read,
  title={Read the fine print},
  author={Aaronson, Scott},
  journal={Nature Physics},
  volume={11},
  number={4},
  pages={291--293},
  year={2015},
  publisher={Nature Publishing Group UK London},
  doi={10.1038/nphys3272}
}

@article{Thaksakronwong2026,
  title={Quantum analog-encoding for correlated Gaussian vectors and their exponentiation with application to rough volatility},
  author={Thaksakronwong, Tassa and Miyamoto, Koichi},
  journal={arXiv preprint arXiv:2604.22463},
  year={2026}
}

@book{higham2002accuracy,
  title={Accuracy and stability of numerical algorithms},
  author={Higham, Nicholas J},
  year={2002},
  publisher={SIAM}
}

@book{hutzenthaler2015numerical,
  title={Numerical approximations of stochastic differential equations with non-globally Lipschitz continuous coefficients},
  author={Hutzenthaler, Martin and Jentzen, Arnulf},
  volume={236},
  number={1112},
  year={2015},
  publisher={American Mathematical Society},
  doi={10.1090/memo/1112}
}

@article{li2022simulating,
  title={Simulating Markovian open quantum systems using higher-order series expansion},
  author={Li, Xiantao and Wang, Chunhao},
  journal={arXiv preprint arXiv:2212.02051},
  year={2022}
}
\bibliographystyle{unsrturl}

\appendix

\section{Proof of Theorem \ref{th:EM} \label{sec:PrEM}}

\begin{proof}[Proof of Theorem \ref{th:EM}]
By induction, we see from Eq.~\eqref{eq:EM} that $\mathbf{X}^{\rm EM}_n$ can be written as
\begin{align}
    \mathbf{X}^{\rm EM}_n = \Phi^{\rm EM}_{n,0} \mathbf{x}_0 + \sum_{k=0}^{n-1}\int_{t_k}^{t_{k+1}} \Phi^{\rm EM}_{n,k+1} B(t_k) \drm \mathbf{W}_\tau,
\end{align}
where for $n,k\in[r]_0$ such that $k \le n$,
\begin{align}
    \Phi^{\rm EM}_{n,k} \coloneqq
    \begin{dcases}
    I & ; ~ n=k \\
    (I + A(t_{n-1})\Delta t) \cdots (I + A(t_k)\Delta t) & ; ~ n=k+1,\ldots,r
    \end{dcases}.
\end{align}
We thus have
\begin{align}
    \mathbf{X}^{\rm EM}_n - \mathbf{X}_n = \left(\Phi^{\rm EM}_{n,0} - \Phi(t_n,0)\right)\mathbf{x}_0 + \sum_{k=0}^{n-1}\int_{t_k}^{t_{k+1}} \left(\Phi^{\rm EM}_{n,k+1} B(t_k) - \Phi(t_n,\tau)B(\tau)\right) \drm \mathbf{W}_\tau,
\end{align}
from which It\^{o}'s isometry yields
\begin{align}
    \mathbb{E}\left[\left\|\mathbf{X}^{\rm EM}_n - \mathbf{X}_n\right\|^2\right] =  \left\|\left(\Phi^{\rm EM}_{n,0} - \Phi(t_n,0)\right)\mathbf{x}_0\right\|^2 + \sum_{k=0}^{n-1}\int_{t_k}^{t_{k+1}} \left\|\Phi^{\rm EM}_{n,k+1} B(t_k) - \Phi(t_n,\tau)B(\tau)\right\|_{\rm F}^2 \drm \tau.
    \label{eq:EXEM-X2Temp}
\end{align}
According to \cite[Sec. F.1.1]{an2026fast},
\begin{align}
    \left\|(I + A(s)(t-s)) - \Phi(t,s)\right\| \le \frac{1}{2}(\alpha_A^2+\alpha_{\drm A/\drm t})(t-s)^2
    \label{eq:EMErrPhiOneStep}
\end{align}
holds for $t,s\in[0,T]$ such that $s \le t$.
Using this together with Eqs. \eqref{eq:PhiUBEta} and \eqref{eq:I+At}, we get
\begin{align}
    \left\|\Phi^{\rm EM}_{n,k} - \Phi(t_n,t_k)\right\| & \le \left\|(I + A(t_{n-1})\Delta t) \cdots (I + A(t_k)\Delta t) - \Phi(t_n,t_{n-1})\cdots\Phi(t_{k+1},t_k)\right\| \nonumber \\
    & \le \sum_{l=k}^{n-1} \left\|(I + A(t_{n-1})\Delta t) \cdots (I + A(t_{l+1})\Delta t)\cdot\left((I + A(t_l)\Delta t) - \Phi(t_{l+1},t_l)\right)\cdot\Phi(t_l,t_{l-1})\cdots\Phi(t_{k+1},t_k)\right\| \nonumber \\
    & \le (n-k) \cdot \frac{1}{2}(\alpha_A^2+\alpha_{\drm A/\drm t})\Delta t^2 \nonumber \\
    & = \frac{1}{2}(\alpha_A^2+\alpha_{\drm A/\drm t})(t_n-t_k)\Delta t.
    \label{eq:EMErrPhi}
\end{align}
Besides, using Eqs. \eqref{eq:EMErrPhiOneStep} and \eqref{eq:EMErrPhi} yields 
\begin{align}
    \left\|\Phi^{\rm EM}_{n,k+1} - \Phi(t_n,\tau)\right\| & \le \left\|\Phi^{\rm EM}_{n,k+1} - \Phi(t_n,t_{k+1})\right\| + \left\|\Phi(t_n,t_{k+1})\left(I - \Phi(t_{k+1},\tau)\right)\right\| \nonumber \\
     & \le \left\|\Phi^{\rm EM}_{n,k+1} - \Phi(t_n,t_{k+1})\right\| + \left\|I + A(\tau)(t_{k+1} - \tau) - \Phi(t_{k+1},\tau)\right\| + \|A(\tau)(t_{k+1} - \tau)\| \nonumber \\
     & \le \frac{1}{2}(\alpha_A^2+\alpha_{\drm A/\drm t})(t_n-t_{k+1})\Delta t + \frac{1}{2}(\alpha_A^2+\alpha_{\drm A/\drm t})(t_{k+1}-\tau)\Delta t + \alpha_A \Delta t \nonumber \\
     & \le \frac{1}{2}(\alpha_A^2+\alpha_{\drm A/\drm t})T\Delta t + \alpha_A \Delta t
\end{align}
for $n,k\in[r]$ such that $n>k$ and $\tau\in[t_k,t_{k+1}]$, which leads to
\begin{align}
    \left\|\Phi^{\rm EM}_{n,k+1} B(t_k) - \Phi(t_n,\tau)B(\tau)\right\|_{\rm F} & \le \sqrt{m} \left\|\Phi^{\rm EM}_{n,k+1} B(t_k) - \Phi(t_n,\tau)B(\tau)\right\| \nonumber \\
    & \le \sqrt{m} \left(\left\|\Phi^{\rm EM}_{n,k+1} - \Phi(t_n,\tau)\right\| \cdot \|B(t_k)\| + \|\Phi(t_n,\tau)\| \cdot \|B(t_k)-B(\tau)\|\right) \nonumber \\
    & \le \sqrt{m}\left(\left(\frac{1}{2}(\alpha_A^2+\alpha_{\drm A/\drm t})T + \alpha_A \right) \Delta t \cdot \sigma + \left\|\int_\tau^{t_k} \frac{\drm B}{\drm t}(s)\drm s\right\| \right)  \nonumber \\
    & \le \sqrt{m}\left(\left(\frac{1}{2}(\alpha_A^2+\alpha_{\drm A/\drm t})T + \alpha_A \right) \sigma + \alpha_{\drm B/\drm t} \right)\Delta t.
    \label{eq:EMErrPhiB}
\end{align}
Plugging Eqs. \eqref{eq:EMErrPhi} and \eqref{eq:EMErrPhiB} into Eq. \eqref{eq:EXEM-X2Temp}, along with some algebra, yields the claim.
\end{proof}

\section{Bound on $\mathcal{N}_{2k}^{0,l-1}$ \label{sec:bndN}}

\begin{lemma}
    For $k,l\in\mathbb{N}$,
    \begin{align}
        \left|\mathcal{N}_{2k}^{0,l-1}\right| \le (2k-1)!!l^k
        \label{eq:calNBnd}
    \end{align}
    holds.
    \label{lem:calN}
\end{lemma}

\begin{proof}
    We note
    \begin{align}
        \left|\mathcal{N}_{2k}^{0,l}\right| = \sum_{s=0}^k \binom{2k}{2s} \left|\mathcal{N}_{2(k-s)}^{0,l-1}\right|.
    \end{align}
    By using induction with respect to $l$ along with $\left|\mathcal{N}_{2k}^{0,0}\right|=1$, we see
    \begin{align}
        \left|\mathcal{N}_{2k}^{0,l-1}\right| = \frac{1}{2^l}\sum_{s=0}^l \binom{l}{s} (2s-l)^{2k}.
    \end{align}
    Note that $\frac{1}{2^l}\binom{l}{s}$ is the probability that $s$ out of $l$ independent Rademacher variables $\chi_1,\ldots,\chi_l$ take 1 and the others take -1, and that $2s-l$ is equal to $\chi_1+\ldots+\chi_l$ in such a case.
    Thus, we see
    \begin{align}
        \left|\mathcal{N}_{2k}^{0,l-1}\right| = \mathbb{E}\left[(\chi_1+\ldots+\chi_l)^{2k}\right],
    \end{align}
    and the Khintchine inequality~\cite{haagerup1981best} implies that this is bounded as Eq.~\eqref{eq:calNBnd}.
\end{proof}

\end{document}